\def\llncs{0}
\def\fullpage{1}
\def\anonymous{0}
\def\authnote{0}
\def\notxfont{0}
\def\submission{0}
\def\reply{0}
\def\cameraready{0}
\def\noaux{0}
\def\displaylabel{0}
\ifnum\submission=1
\def\anonymous{1}
\def\llncs{1} 
\def\displaylabel{0}
\else
\fi

\ifnum\cameraready=1
\def\submission{1}
\def\llncs{1}
\def\anonymous{0}
\def\authnote{0}
\def\displaylabel{0}
\else
\fi

\ifnum\anonymous=1
\def\authnote{0}
\else
\fi

\ifnum\llncs=1
	\documentclass[envcountsect,a4paper,ruhttps://www.overleaf.com/project/62afd95023a007072d4588bfnningheads,10pt]{llncs}
\else
	\documentclass[letterpaper,hmargin=1.05in,vmargin=1.05in]{article}
			\ifnum\fullpage=1
		\usepackage{fullpage}
		\fi
\fi

\ifnum\reply=1
\usepackage{ulem}
\renewcommand{\emph}{\textit}
\else
\fi

\ifnum\reply=0  

\newcommand{\strike}[1]{}
\else

\newcommand{\strike}[1]{\sout{#1}}
\fi

\usepackage[%
  colorlinks=true,
  citecolor=darkgreen,
  linkcolor=darkblue,
  urlcolor=darkblue,
  pagebackref=true
]{hyperref}

\usepackage{amsmath, amsfonts, amssymb, mathtools,amscd}

\usepackage{amsthm}

\usepackage{lmodern}
\usepackage[T1]{fontenc}
\usepackage[utf8]{inputenc}

\usepackage{etex}

\usepackage{arydshln} 
\usepackage{url}
\usepackage{ifthen}
\usepackage{bm}
\usepackage{multirow}
\usepackage[dvips]{graphicx}
\usepackage[usenames]{color}
\usepackage{xcolor,colortbl} 
\usepackage{threeparttable}
\usepackage{comment}
\usepackage{paralist,verbatim}
\usepackage{cases}
\usepackage{booktabs}
\usepackage{braket}
\usepackage{cancel} 
\usepackage{ascmac} 
\usepackage{framed}
\usepackage{authblk}
\usepackage{pifont}
\usepackage{physics}
\definecolor{darkblue}{rgb}{0,0,0.6}
\definecolor{darkgreen}{rgb}{0,0.5,0}
\definecolor{maroon}{rgb}{0.5,0.1,0.1}
\definecolor{dpurple}{rgb}{0.2,0,0.65}
\definecolor{chocolate}{rgb}{0.8,0.4,0.1}

\usepackage[capitalise,noabbrev]{cleveref}
\usepackage[absolute]{textpos}
\usepackage[final]{microtype}
\usepackage[absolute]{textpos}
\usepackage{autonum}

\usepackage{dsfont}
\DeclareMathAlphabet{\mathpzc}{OT1}{pzc}{m}{it}

\ifnum\submission=1
\renewcommand*{\backref}[1]{}
\def\notxfont{1}
\else
\fi

\ifnum\llncs=1
\renewcommand{\subparagraph}{\paragraph}
\else
\ifnum\notxfont=1
\else
\usepackage{mathpazo}
\usepackage{newtxtext}
\usepackage{helvet}
\fi
\fi

\newtheoremstyle{thicktheorem}%
{\topsep}
{\topsep}
{\itshape}{}%
{\bfseries}%
{.}
{ }%
{\thmname{#1}\thmnumber{ #2}%
		\thmnote{ (#3)}%
}

\newtheoremstyle{remark}
{\topsep}
{\topsep}
	{}
	{}
	{}
	{.}
	{ }
	{\textit{\thmname{#1}}\thmnumber{ #2}
			\thmnote{ (#3)}%
	}

\ifnum\llncs=0
	\theoremstyle{thicktheorem}
	\newtheorem{theorem}{Theorem}[section]
	\newtheorem{lemma}[theorem]{Lemma}
	
	\newtheorem{proposition}[theorem]{Proposition}
	\newtheorem{definition}[theorem]{Definition}

	\theoremstyle{remark}
	
	\newtheorem{remark}[theorem]{Remark}

\fi

	\crefname{theorem}{Theorem}{Theorems}
	\crefname{assumption}{Assumption}{Assumptions}
	\crefname{construction}{Construction}{Constructions}
	\crefname{corollary}{Corollary}{Corollaries}
	\crefname{conjecture}{Conjecture}{Conjectures}
	\crefname{definition}{Definition}{Definitions}
	\crefname{exmaple}{Example}{Examples}
	\crefname{experiment}{Experiment}{Experiments}
	\crefname{counterexample}{Counterexample}{Counterexamples}
	\crefname{lemma}{Lemma}{Lemmata}
	\crefname{observation}{Observation}{Observations}
	\crefname{proposition}{Proposition}{Propositions}
	\crefname{remark}{Remark}{Remarks}
	\crefname{claim}{Claim}{Claims}
	\crefname{fact}{Fact}{Facts}
	\crefname{note}{Note}{Notes}

\ifnum\llncs=1
 \crefname{appendix}{App.}{Appendices}
 \crefname{section}{Sec.}{Sections}
\else
\fi

\ifnum\llncs=1
\renewcommand*{\backref}[1]{}
\else
	\renewcommand*{\backref}[1]{(Cited on page~#1.)}
	\ifnum\notxfont=1
	\else
		\usepackage{newtxtext}
	\fi
\fi

\usepackage{etoolbox}
\newcommand*{\keys}[1]{\mathsf{#1}}

\newcommand*{\algo}[1]{\ensuremath{\mathsf{#1}}}
\newcommand*{\qalgo}[1]{\ensuremath{\mathpzc{#1}}}
\newcommand*{\qstate}[1]{\mathpzc{#1}}

\newcounter{expitem}

\usepackage{mathtools}

\newcommand{\chosen}{\leftarrow}

\newcommand{\lrun}{\leftarrow}

\newcommand{\la}{\leftarrow}
\newcommand{\ra}{\rightarrow}
\renewcommand{\gets}{\leftarrow}

\newcommand{\seteq}{\coloneqq}

\newcommand{\concat}{\|}

\newcommand{\setbk}[1]{\{#1\}}

\newcommand{\cB}{\mathcal{B}}
\newcommand{\cC}{\mathcal{C}}

\newcommand{\cF}{\mathcal{F}}

\newcommand{\cK}{\mathcal{K}}

\newcommand{\cM}{\mathcal{M}}

\newcommand{\cO}{\mathcal{O}}

\newcommand{\cQ}{\mathcal{Q}}

\newcommand{\cU}{\mathcal{U}}

\newcommand{\qA}{\qalgo{A}}
\newcommand{\qB}{\qalgo{B}}
\newcommand{\qChal}{\qalgo{C}}

\newcommand{\qD}{\qalgo{D}}

\newcommand{\qL}{\qalgo{L}}

\newcommand{\wtl}[1]{\widetilde{#1}}

\def\makeuppercase#1{
\expandafter\newcommand\csname sf#1\endcsname{\mathsf{#1}}
\expandafter\newcommand\csname frak#1\endcsname{\mathfrak{#1}}
\expandafter\newcommand\csname bb#1\endcsname{\mathbb{#1}}
\expandafter\newcommand\csname bf#1\endcsname{\textbf{#1}}
}

\def\makelowercase#1{
\expandafter\newcommand\csname frak#1\endcsname{\mathfrak{#1}}
\expandafter\newcommand\csname bf#1\endcsname{\textbf{#1}}
}

\newcounter{char}
\loop
   \edef\letter{\alph{char}}
   \edef\Letter{\Alph{char}}
   \expandafter\makelowercase\letter
   \expandafter\makeuppercase\Letter
   \stepcounter{char}
   \unless\ifnum\thechar>26
\repeat

\def\makeuppercase#1{
\expandafter\newcommand\csname tl#1\endcsname{\widetilde{#1}}
}

\def\makelowercase#1{
\expandafter\newcommand\csname tl#1\endcsname{\widetilde{#1}}
}

\newcommand{\N}{\mathbb{N}}

\newcommand{\R}{\mathbb{R}}

\newcommand{\bit}{\{0,1\}}

\newcommand{\Ms}{\mathcal{M}}

\newcommand{\Cs}{\mathcal{C}}
\newcommand{\Ks}{\mathcal{K}}

\newcommand{\secp}{\lambda}

\newcommand{\sig}{\sigma}
\newcommand{\cert}{\keys{cert}}

\newcommand{\aux}{\mathsf{aux}}
\newcommand{\param}{\mathsf{param}}

\newcommand{\state}{\mathsf{st}}

\newcommand{\advB}{\mathcal{B}}

\newcommand{\ind}{\mathsf{ind}}

\newcommand{\cpa}{\mathsf{cpa}}

\newcommand{\sel}{\mathsf{sel}}

\newcommand{\Advt}[2]{\mathsf{Adv}_{#1}^{#2}}

\newcommand{\adva}[2]{\mathsf{Adv}_{#1}^{\mathsf{#2}}}
\newcommand{\advb}[3]{\mathsf{Adv}_{#1}^{\mathsf{#2} \mbox{-} \mathsf{#3}}}
\newcommand{\advc}[4]{\mathsf{Adv}_{#1}^{\mathsf{#2} \mbox{-} \mathsf{#3} \mbox{-} \mathsf{#4}}}
\newcommand{\advd}[5]{\mathsf{Adv}_{#1}^{\mathsf{#2} \mbox{-} \mathsf{#3} \mbox{-} \mathsf{#4} \mbox{-} \mathsf{#5}}}

\newcommand{\expt}[2]{\mathsf{Expt}_{#1}^{\mathsf{#2}}}
\newcommand{\tildeexpt}[2]{\widetilde{\mathsf{Expt}}_{#1}^{\mathsf{#2}}}

\newcommand{\expb}[3]{\mathsf{Exp}_{#1}^{ \mathsf{#2} \mbox{-} \mathsf{#3}}}
\newcommand{\expc}[4]{\mathsf{Exp}_{#1}^{ \mathsf{#2} \mbox{-} \mathsf{#3} \mbox{-} \mathsf{#4}}}
\newcommand{\expd}[5]{\mathsf{Exp}_{#1}^{\mathsf{#2} \mbox{-} \mathsf{#3} \mbox{-} \mathsf{#4} \mbox{-} \mathsf{#5}}}

\newcommand{\hybij}[2]{\mathsf{Hyb}_{#1}^{#2}}

\newcommand*{\pk}{\keys{pk}}
\newcommand*{\sk}{\keys{sk}}
\newcommand*{\dk}{\keys{dk}}
\newcommand*{\ek}{\keys{ek}}
\newcommand*{\vk}{\keys{vk}}

\newcommand*{\fsk}{\keys{fsk}}

\newcommand*{\msk}{\keys{msk}}

\newcommand*{\SK}{\keys{SK}}

\newcommand*{\EK}{\keys{EK}}
\newcommand*{\MPK}{\keys{MPK}}
\newcommand*{\MSK}{\keys{MSK}}

\newcommand*{\pp}{\keys{pp}}

\newcommand*{\ct}{\keys{ct}}

\newcommand*{\tldk}{\widetilde{\dk}}
\newcommand*{\tlct}{\widetilde{\ct}}

\newcommand{\CT}{\keys{CT}}

\newcommand*{\msg}{m}

\newcommand{\qct}{\qstate{ct}}

\newcommand{\qcert}{\qstate{cert}}

\newcommand{\qaux}{\qalgo{aux}}

\newenvironment{boxfig}[2]{\begin{figure}[#1]\fbox{\begin{minipage}{0.97\linewidth}
                        \vspace{0.2em}
                        \makebox[0.025\linewidth]{}
                        \begin{minipage}{0.95\linewidth}
            {{
                        #2 }}
                        \end{minipage}
                        \vspace{0.2em}
                        \end{minipage}}
                        }
                        {\end{figure}}

\newcommand{\pprotocol}[4]{
\begin{boxfig}{h!}{\footnotesize 
\centering{\textbf{#1}}
    #4
\vspace{0.2em} } \caption{\label{#3} #2}
\end{boxfig}
}

\newcommand{\protocol}[4]{
\pprotocol{#1}{#2}{#3}{#4} }

\newcommand{\prf}{\algo{F}}

\newcommand{\Setup}{\algo{Setup}}
\newcommand{\setup}{\algo{Setup}}

\newcommand{\gen}{\algo{Gen}}
\newcommand{\KeyGen}{\algo{KeyGen}}
\newcommand{\keygen}{\algo{KeyGen}}

\newcommand{\Enc}{\algo{Enc}}
\newcommand{\Dec}{\algo{Dec}}

\newcommand{\enc}{\algo{Enc}}
\newcommand{\dec}{\algo{Dec}}

\newcommand{\Vrfy}{\algo{Vrfy}}

\newcommand{\vrfy}{\algo{Vrfy}}

\newcommand{\Sampler}{\algo{Sampler}}

\newcommand{\qEnc}{\qalgo{Enc}}
\newcommand{\qDec}{\qalgo{Dec}}
\newcommand{\qDelete}{\qalgo{Del}}

\newcommand{\qVrfy}{\qalgo{Vrfy}}
\newcommand{\qFake}{\qalgo{Fake}}
\newcommand{\qGarble}{\qalgo{Garble}}

\newcommand\PKE{\algo{PKE}}

\newcommand\ABE{\algo{ABE}}

\newcommand\FE{\algo{FE}}

\newcommand\PKFE{\algo{PKFE}}

\newcommand\SKFE{\algo{SKFE}}

\newcommand{\SKE}{\algo{SKE}}
\newcommand{\ske}{\algo{ske}}

\newcommand{\GC}{\algo{GC}}
\newcommand{\gc}{\mathsf{gc}}
\newcommand{\Garble}{\algo{Grbl}}

\newcommand{\Sim}{\algo{Sim}}

\newcommand{\iO}{i\cO}

\newcommand{\eO}{e\cO}

\newcommand{\PRG}{\algo{PRG}}
\newcommand{\PRF}{\algo{PRF}}

\newcommand{\Eval}{\algo{Eval}}

\newcommand{\negl}{{\mathsf{negl}}}

\newcommand{\zo}[1]{\{0,1\}^{#1}}
\newcommand{\bin}{\{0,1\}}

\newcommand{\xor}{\oplus}

\newcommand{\class}[1]{\mathsf{#1}}

\newcommand{\IP}{\class{IP}}

\newcommand{\Ppoly}{\class{P/poly}}

\newcommand{\NCone}{\class{NC}^1}

\newcommand{\Samp}{\mathrm{Smp}}

\newcommand{\calA}{\mathcal{A}}

\newcommand{\calP}{\mathcal{P}}

\newcommand{\pub}{\mathsf{pub}}
\newcommand{\hyb}{\mathsf{Hyb}}

\newcommand{\msf}[1]{\mathsf{#1}}
\newcommand{\tsf}[1]{\textsf{#1}}

\newcommand{\mcl}[1]{\mathcal{#1}}
\newcommand{\mbb}[1]{\mathbb{#1}}
\newcommand{\mrm}[1]{\mathrm{#1}}

\newcommand{\wt}[1]{\widetilde{#1}}

\newcommand{\wh}[1]{\widehat{#1}}

\makeatletter
\newcommand{\vast}{\bBigg@{3.5}}
\newcommand{\Vast}{\bBigg@{5}}
\makeatother

\newcommand{\cnc}[2]{\mathbf{CC}[#1,#2]}
\newcommand{\CCObf}{\mathsf{CC}.\mathsf{Obf}}
\newcommand{\tlP}{\widetilde{P}}
\newcommand{\qCCObf}{\qalgo{CCObf}}
\newcommand{\tlqP}{\widetilde{\qalgo{P}}}
\newcommand{\CCO}{\algo{CCO}}
\newcommand{\CECCO}{\algo{CECCO}}
\newcommand{\tlfDec}{\widetilde{\fDec}}
\newcommand{\tlI}{\widetilde{I}}
\newcommand{\sfCC}{\mathsf{CC}}
\newcommand{\qObf}{\qalgo{Obf}}
\newcommand{\tlaDec}{\widetilde{\qalgo{aDec}}}

\newcommand{\oneadaFEps}{\algo{1adaFE}}
\newcommand{\oneselFEps}{\algo{1selFE}}
\newcommand{\twoFE}{\algo{2FE}}
\newcommand{\twofe}{\algo{2fe}}

\newcommand{\Modify}{\mathsf{Recover}}

\newcommand{\nad}{\mathsf{nad}}
\newcommand{\NAD}{\mathsf{NAD}}

\newcommand{\Delete}{\algo{Del}}
\newcommand{\skcd}{\mathsf{skcd}}

\newcommand{\sfCD}{\mathsf{CD}}
\newcommand{\sfcd}{\mathsf{cd}}
\newcommand{\qGateGrbl}{\qalgo{GateGrbl}}
\newcommand{\qGateEval}{\qalgo{GateEval}}
\newcommand{\qGateDel}{\qalgo{GateDel}}
\newcommand{\qggate}{\widetilde{\qstate{g}}}
\newcommand{\one}{\mathsf{one}}
\newcommand{\ONE}{\mathsf{ONE}}
\newcommand{\qugate}{\widetilde{\qstate{U}}}

\newcommand{\sm}{\msf{Sim}}

\newcommand*{\qenc}{\qalgo{Enc}}
\newcommand*{\qdec}{\qalgo{Dec}}

\newcommand*{\qdel}{\qalgo{Del}}

\newcommand*{\nce}{\mathsf{nce}}
\newcommand*{\lnce}{\mathsf{nce}}
\newcommand*{\fake}{\mathsf{Fake}}

\newcommand{\Fake}{\algo{Fake}}

\newcommand\NCE{\algo{NCE}}

\newcommand*{\fe}{\mathsf{fe}}
\newcommand{\CED}{\algo{CED}}
\newcommand{\sfpre}{\mathsf{pre}}

\newcommand{\FakeSetup}{\fake\setup}
\newcommand{\FakeCT}{\fake\msf{CT}}
\newcommand{\FakeSK}{\fake\msf{SK}}
\newcommand{\Reveal}{\msf{Reveal}}

\newcommand*{\amsk}{\msf{amsk}}

\newcommand*{\apk}{\msf{apk}}
\newcommand*{\act}{\msf{act}}

\newcommand*{\FHE}{\msf{FHE}}
\newcommand*{\fhe}{\algo{fhe}}
\newcommand*{\fpk}{\msf{fpk}}
\newcommand*{\fct}{\msf{fct}}
\newcommand*{\fDec}{\msf{fDec}}

\newcommand*{\qEval}{\qalgo{Eval}}

\newcommand*{\qobf}{\qalgo{Obf}}

\newcommand*{\lock}{\msf{lock}}
\newcommand*{\qsim}{\qalgo{Sim}}

\newcommand*{\EV}{\msf{EV}\text{-}}
\newcommand*{\C}{\msf{C}\text{-}}

\newcommand*{\TD}{\msf{TD}}
\usepackage{fancyhdr}
\ifnum\displaylabel=1
\usepackage{showkeys}
\else
\fi

\ifnum\llncs=1
\let\oldvec\vec
\let\vec\oldvec
\makeatletter
\renewcommand*\l@author[2]{}
\renewcommand*\l@title[2]{}
\makeatletter
\fi    

\theoremstyle{remark}

\definecolor{darkblue}{rgb}{0,0,0.6}
\definecolor{darkgreen}{rgb}{0,0.5,0}
\definecolor{maroon}{rgb}{0.5,0.1,0.1}
\definecolor{dpurple}{rgb}{0.2,0,0.65}
\definecolor{darkviolet}{RGB}{130,95,141}
\definecolor{darkkhaki}{RGB}{189,183,107}

\ifnum\authnote=1
\newcommand{\ryo}[1]{\textcolor{brown}{[{\footnotesize {\bf RN:} { {#1}}}]}}
\newcommand{\rnnote}[1]{\textcolor{brown}{[{\footnotesize {\bf RN:} { {#1}}}]}}
\newcommand{\tpnote}[1]{\textcolor{red}{[{\footnotesize {\bf TP:} { {#1}}}]}}
\newcommand{\fuyuki}[1]{\textcolor{blue}{[{\footnotesize {\bf FK:} { {#1}}}]}}
\newcommand{\takashi}[1]{\textcolor{darkgreen}{[{\footnotesize {\bf TY:} { {#1}}}]}}
\newcommand{\tomoyuki}[1]{\textcolor{darkkhaki}{[{\footnotesize {\bf MT:} { {#1}}}]}}
\newcommand{\taiga}[1]{\textcolor{dpurple}{[{\footnotesize {\bf TH:} { {#1}}}]}}
\else
\newcommand{\ryo}[1]{}
\newcommand{\rnnote}[1]{}
\newcommand{\tpnote}[1]{}
\newcommand{\fuyuki}[1]{}
\newcommand{\takashi}[1]{}
\newcommand{\tomoyuki}[1]{}
\newcommand{\taiga}[1]{}
\fi

\title{
\textbf{Certified Everlasting Secure\\Collusion-Resistant Functional Encryption, and More}
\ifnum\anonymous=1
\else
\fi
\ifnum\noaux=1
\else
\ifnum\submission=1
\thanks{{\emph{Technical details are omitted due to page limitations in this version. Please, see the full version~\cite{myEprint:HKMNPY23} for the technical details.}}}
\else
\fi
\fi
}

\begin{document}

\ifnum\anonymous=1
\author{\empty}
\ifnum\llncs=1
\institute{\empty}
\else
\fi
\else
%
%
\ifnum\llncs=1
\author{
Taiga Hiroka\inst{1} \and Fuyuki Kitagawa\inst{2,3} \and Tomoyuki Morimae\inst{1} \and Ryo Nishimaki\inst{2,3} \and Tapas Pal\inst{4}\footnote{The research was conducted while the author was a postdoc at NTT Social Informatics Laboratories} \and Takashi Yamakawa\inst{1,2,3}
}
\institute{
Yukawa Institute for Theoretical Physics, Kyoto University, Japan \and	NTT Social Informatics Laboratories, Tokyo, Japan
\and    NTT Research Center for Theoretical Quantum Information
\and    Karlsruhe Institute of Technology, KASTEL Security Research Labs, Germany
}
\else
%
%
\author[$\star$]{Taiga Hiroka}
\author[$\dagger$$\diamondsuit$]{Fuyuki Kitagawa}
\author[$\star$]{Tomoyuki Morimae}
\author[$\dagger$$\diamondsuit$]{\\Ryo Nishimaki}
\author[$\flat$]{Tapas Pal}
\author[$\dagger$$\diamondsuit$$\star$]{Takashi Yamakawa}
\affil[$\star$]{{\small Yukawa Institute for Theoretical Physics, Kyoto University, Japan}\authorcr{\small \{taiga.hiroka,tomoyuki.morimae\}@yukawa.kyoto-u.ac.jp}}
\affil[$\dagger$]{{\small NTT Social Informatics Laboratories, Tokyo, Japan}\authorcr{\small \{fuyuki.kitagawa,ryo.nishimaki,takashi.yamakawa\}@ntt.com}}
\affil[$\diamondsuit$]{{\small NTT Research Center for Theoretical Quantum Information}}
\affil[$\flat$\footnote{The research was conducted while the author was a postdoc at NTT Social Informatics Laboratories}]{{\small Karlsruhe Institute of Technology, KASTEL Security Research Labs, Germany}\authorcr{\small tapas.real@gmail.com}}
\renewcommand\Authands{, }
\fi 
\fi

\ifnum\llncs=1
\date{}
\else
\date{\today}
\fi

\maketitle

\begin{abstract}
We study certified everlasting secure functional encryption (FE) and many other cryptographic primitives in this work.
Certified everlasting security roughly means the following.
A receiver possessing a quantum cryptographic object (such as ciphertext) can issue a certificate showing that the receiver has deleted the cryptographic object and information included in the object (such as plaintext) was lost.
If the certificate is valid, the security is guaranteed even if the receiver becomes computationally unbounded after the deletion.
Many cryptographic primitives are known to be impossible (or unlikely) to have information-theoretical security even in the quantum world.
Hence, certified everlasting security is a nice compromise (intrinsic to quantum).

In this work, we define certified everlasting secure versions of FE, compute-and-compare obfuscation, predicate encryption (PE), secret-key encryption (SKE), public-key encryption (PKE), receiver non-committing encryption (RNCE), and garbled circuits.
We also present the following constructions:
\begin{itemize}
     \item Adaptively certified everlasting secure collusion-resistant public-key FE for all polynomial-size circuits from indistinguishability obfuscation and one-way functions.
     \item Adaptively certified everlasting secure bounded collusion-resistant public-key FE for $\mathsf{NC}^1$ circuits from standard PKE.
     \item Certified everlasting secure compute-and-compare obfuscation from standard fully homomorphic encryption and standard compute-and-compare obfuscation.
     \item Adaptively (resp., selectively) certified everlasting secure PE from standard adaptively (resp., selectively) secure attribute-based encryption and certified everlasting secure compute-and-compare obfuscation.
     \item Certified everlasting secure SKE and PKE from standard SKE and PKE, respectively.
     \item Cetified everlasting secure RNCE from standard PKE.
     \item Cetified everlasting secure garbled circuits from standard SKE.
\end{itemize}
\end{abstract}

\ifnum\llncs=1
\else
\newpage
\setcounter{tocdepth}{2}
\tableofcontents

\newpage
\fi


\section{Introduction}\label{sec:intro}

\subsection{Background}\label{sec:background}

Computational security in cryptography relies on assumptions that some problems are hard to solve.
However, such assumptions could be broken in the future
when revolutionary novel algorithms are discovered, or computing devices are drastically improved.
One solution to the problem of computational security is to construct information-theoretically-secure protocols.
However, many cryptographic primitives are known to be impossible (or unlikely) to satisfy information-theoretical security even in the quantum world~\cite{PRL:LC97,PRL:Mayers97,ARXIV:MW18}.

Good compromises (intrinsic to quantum!) have been studied recently~\allowbreak\cite{JACM:Unruh15,TCC:BroIsl20,ARXIV:KunTan20,AC:HMNY21,C:HMNY22,myITCS:Poremba23}. 
In particular, certified everlasting security, which was introduced in~\cite{C:HMNY22} based on~\cite{JACM:Unruh15,TCC:BroIsl20}, achieves the following security: 
After receiving quantum-encrypted data, a receiver can issue a certificate to prove that (s)he deleted its quantum-encrypted data.
If the certificate is valid, its security is guaranteed even if the receiver becomes computationally unbounded later. A (private or public) verification key for certificates is also generated along with quantum-encrypted data.
This security notion is weaker than information-theoretical security since a malicious receiver could refuse to issue a valid certificate.
However, it is still a useful security notion because, for example, a sender can penalize receivers who do not issue valid certificates.
In addition, certified everlasting security is an intrinsically quantum property because it implies information-theoretical security in the classical world.\footnote{This is because a malicious receiver can copy the encrypted data freely. Hence, the encrypted data must be secure against an unbounded malicious receiver at the point when the receiver obtains the encrypted data.
The same discussion does not go through in the quantum world because even a malicious receiver cannot copy the quantum-encrypted data due to the quantum no-cloning theorem.}

Certified everlasting security can bypass the impossibility of information-theoretical security.
In fact, several cryptographic primitives have been shown to have certified everlasting security, such as commitments and zero-knowledge~\cite{C:HMNY22}. 
An important open problem in this direction is
\begin{center}
\textit{
Which cryptographic primitives can have certified everlasting security?
}
\end{center}

Functional encryption (FE) is one of the most advanced cryptographic primitives and achieves considerable flexibility in controlling encrypted data~\cite{TCC:BonSahWat11}.
In FE, an owner of a master secret key $\MSK$ can generate a functional decryption key $\sk_f$ that hardwires a function $f$.
When a ciphertext $\ct_m$ of a message $m$ is decrypted by $\sk_f$, we can obtain the value $f(m)$, and no information beyond $f(m)$ is leaked.
Information-theoretically secure FE is impossible, and all known constructions are computationally secure~\cite{C:GorVaiWee12,SIAMCOMP:GGHRSW16,EC:AgrPel20,TCC:AnaVai19,STOC:JaiLinSah21,EC:JaiLinSah22}.
A motivating application of FE is analyzing sensitive data and computing new data from personal data without sacrificing data privacy.
In this example, users must store their encrypted data on a remote server since users delegate the computation. At some point, users might request the server to ``forget'' their data (even if they are encrypted). European Union~\cite{GDPR16} and California~\cite{CCPA18} adopted data deletion clauses in legal regulations for such users. Encryption with certified deletion could be useful for implementing the right to be forgotten. However, suppose that FE does not have \emph{certified everlasting security}. In that case, the rapid growth of computational power potentially breaks the privacy of sensitive personal data (such as DNA) in the future.
This risk (``recalling'' in the future) is great because descendants inherit DNA information.
Certified everlasting security is desirable for such practical applications of FE.

Hence, we have the following open problem:
\begin{center}
\textit{
Is it possible to construct certified everlasting secure FE?
}
\end{center}
We note that certified everlasting secure FE is particularly useful compared to certified everlasting secure public key encryption (PKE) (or more generally ``all-or-nothing encryption''\footnote{Such as identity-based encryption (IBE), attribute-based encryption (ABE), fully homomorphic encryption (FHE), or witness encryption (WE).}~\cite{C:GarMahMoh17}) because it ensures security even against an honest receiver who holds a decryption key. That is, we can ensure that a receiver who holds a decryption key $\sk_f$ for a function $f$ cannot learn more than $f(m)$ even if the receiver can run an unbounded-time computation after issuing a valid certificate. 
In contrast, certified everlasting PKE does not ensure any security against an honest receiver since the receiver can simply keep a copy of a plaintext after honestly decrypting a ciphertext.

Another useful advanced cryptographic primitive is obfuscation for compute-and-compare programs~\cite{FOCS:WicZir17} (a.k.a. lockable obfuscation~\cite{FOCS:GoyKopWat17}).
A compute-and-compare obfuscation scheme can obfuscate a compute-and-compare circuit parameterized by a polynomial-time computable circuit $P$ along with a lock value $\lock$ and a message $\msg$. The circuit takes an input $x$ and outputs $\msg$ if $P(x)=\lock$ and $\bot$ otherwise. Point functions, conjunction with wild cards, plaintext checkers, and affine testers are examples of such circuits~\cite{FOCS:GoyKopWat17,FOCS:WicZir17}. Hence, certified everlasting secure compute-and-compare obfuscation achieves certified deletion for obfuscated programs in the restricted class of functionalities. In addition, compute-and-compare obfuscation has many cryptographic applications~\cite{FOCS:GoyKopWat17,FOCS:WicZir17,TCC:CVWWW18,EC:FFMV23,C:AgrYadYam22,EC:AKYY23}. We can generically convert all-or-nothing encryption into anonymous one via compute-and-compare obfuscation. In particular, we can obtain predicate encryption (PE)~\cite{EC:KatSahWat08,C:GorVaiWee15} from ABE and compute-and-compare obfuscation. PE is an attribute-hiding variant of ABE and an intermediate primitive between ABE and FE.
If we can achieve certified everlasting secure compute-and-compare obfuscation, it is possible to achieve certified everlasting secure PE (and anonymous IBE and PKE).

Hence, we have the following second open problem:
\begin{center}
\textit{
Is it possible to construct certified everlasting secure compute-and-compare obfuscation?
}
\end{center}

\subsection{Our Results}\label{sec:result}
We solve the above questions in this work.
Our contributions are as follows.
\ifnum\llncs=0
    \begin{enumerate}
     \item We formally define certified everlasting versions of many cryptographic primitives: FE (\cref{sec:def_FE_CED}), compute-and-compare obfuscation (\cref{sec:def_CCObf_CED}), PE (\cref{sec:def_PE_CED}), secret-key encryption (SKE) (\cref{sec:def_cd_ske_pke}), PKE (\cref{sec:def_cd_ske_pke}), receiver non-committing encryption (RNCE) (\cref{sec:def_rnce}), and a garbling scheme (\cref{sec:def_garbled}).

     \item We construct adaptively certified everlasting secure collusion-resistant public-key FE for $\Ppoly$ from indistinguishability obfuscation (IO) and one-way functions (OWFs) (\cref{sec:CRFE_CED_const}). We also construct adaptively certified everlasting secure bounded collusion-resistant public-key FE for $\NCone$ from standard PKE (\cref{sec:const_multi_fe}).

     \item We construct certified everlasting secure compute-and-compare obfuscation from standard FHE and standard compute-and-compare obfuscation (\cref{sec:CCO_CED_const}). Both building blocks can be instantiated with the learning with errors (LWE) assumption. We also construct adaptively (resp., selectively) certified everlasting secure PE from standard adaptively (resp., selectively) secure ABE and certified everlasting secure compute-and-compare obfuscation (\cref{sec:PE_from_LObf_ABE}).

     \item To achieve adaptively certified everlasting secure bounded collusion-resistant FE, we construct many certified everlasting secure cryptographic primitives:
     \begin{itemize}
      \item Two constructions of certified everlasting secure SKE from standard SKE (\cref{sec:const_ske_rom,sec:const_ske_wo_rom}).
     An advantage of the first construction is that the certificate is classical, but a disadvantage is that the security proof relies on the quantum random oracle model (QROM)~\cite{AC:BDFLSZ11}.
     The security of the second construction holds without relying on the QROM, but the certificate is quantum.
     \item Two constructions of certified everlasting secure PKE with the same properties of the SKE constructions above from standard PKE (\cref{sec:const_pke_rom,sec:const_pke_wo_rom}).
     \item A construction of certified everlasting secure RNCE from certified everlasting PKE (\cref{sec:const_rnce_classic}).
     \item A construction of certified everlasting secure garbling scheme for $\Ppoly$ from certified everlasting SKE (\cref{sec:const_garbling}).
      \end{itemize}
     \end{enumerate}
\else 
 \begin{enumerate}
     \item We formally define certified everlasting versions of many cryptographic primitives: FE, compute-and-compare obfuscation, PE, secret-key encryption (SKE), PKE, receiver non-committing encryption (RNCE), and a garbling scheme.

     \item We construct adaptively certified everlasting secure collusion-resistant public-key FE for $\Ppoly$ from indistinguishability obfuscation (IO) and one-way functions (OWFs). We also construct adaptively certified everlasting secure bounded collusion-resistant public-key FE for $\NCone$ from standard PKE.

     \item We construct certified everlasting secure compute-and-compare obfuscation from standard FHE and standard compute-and-compare obfuscation. Both building blocks can be instantiated with the learning with errors (LWE) assumption. We also construct adaptively (resp., selectively) certified everlasting secure PE from standard adaptively (resp., selectively) secure ABE and certified everlasting secure compute-and-compare obfuscation.

     \item To achieve adaptively certified everlasting secure bounded collusion-resistant FE, we construct many certified everlasting secure cryptographic primitives:
     \begin{itemize}
      \item Two constructions of certified everlasting secure SKE from standard SKE.
     An advantage of the first construction is that the certificate is classical, but a disadvantage is that the security proof relies on the quantum random oracle model (QROM)~\cite{AC:BDFLSZ11}.
     The security of the second construction holds without relying on the QROM, but the certificate is quantum.
     \item Two constructions of certified everlasting secure PKE with the same properties of the SKE constructions above from standard PKE.
     \item A construction of certified everlasting secure RNCE from certified everlasting PKE.
     \item A construction of certified everlasting secure garbling scheme for $\Ppoly$ from certified everlasting SKE.
      \end{itemize}
     \end{enumerate}
\fi
All our constructions are privately verifiable, so we must keep verification keys (for deletion certificate) secret.
It is open to achieving certified everlasting secure bounded collusion-resistant FE for $\Ppoly$ from standard PKE.

We introduce fascinating techniques to achieve certified everlasting secure collusion-resistant FE and certified everlasting secure compute-and-compare obfuscation. We developed an authentication technique for BB84 state to satisfy both the functionality of FE and certified everlasting security. (See~\cref{sec:tech_overview_CRFE} for the detail.)
This authentication technique for BB84 states is of independent interest and we believe that it has further applications.\footnote{Indeed, an application was found by Kitagawa, Nishimaki, and Yamakawa \cite{myTCC:KitNisYam23}. See \Cref{sec:subsequent}.}
We also developed a deferred evaluation technique using dummy lock values to satisfy both the functionality of compute-and-compare obfuscation and certified everlasting security. (See~\cref{sec:tech_overview_cc_obfuscation} for the detail.)

\subsection{Concurrent and Independent Work}\label{sec:concurrent_independent_work}
\label{sec:concurrent}

\paragraph{Certified everlasting secure SKE and PKE.} Recently, Bartusek and Khurana concurrently and independently obtained similar results~\cite{myC:BarKhu23}.
They introduce a generic compiler that can convert several cryptographic primitives to certified everlasting secure ones,
such as PKE, ABE, FHE, WE, and timed-release encryption.
Their constructions via the generic compiler have the advantage that the certificates are classical \emph{and} no QROM is required.
Our constructions of certified everlasting SKE and PKE cannot achieve both: if the certificates are classical, QROM is required, and if QROM is not used, the certificates have to be quantum.
We note that their certified everlasting SKE and PKE can be used as building blocks of our RNCE, garbling, and bounded collusion-resistant FE constructions instead of our SKE and PKE schemes.

While their work focuses on all-or-nothing encryption, our work presents certified everlasting secure garbling and FE, which are not given in their work.
It is unclear how to apply their generic compiler to garbling and FE.

One might think that certified everlasting garbling can be constructed from certified everlasting SKE, which is constructed from their generic compiler.
However, it is non-trivial whether certified everlasting garbling can be immediately constructed from certified everlasting SKE because garbling needs double-encryption.
(For details, see \cref{sec:tech_overview_bounded_FE}.)

Moreover, a direct application of their generic compiler to FE does not work because of the following reason.
If we directly apply their generic compiler to FE, we have a ciphertext consisting of classical and quantum parts. 
The classical part is the original FE ciphertext whose plaintext is $m\oplus r$ with random $r$, and the quantum part is random BB84 states whose computational basis states encode $r$.
The decryption key of the function $f$ consists of functional decryption key $\sk_f$ and the basis of the BB84 states.
However, in this construction, a receiver with the ciphertext and the decryption key cannot obtain $f(m)$,
because what the receiver obtains is only $f(m\oplus r)$ and $r$, which cannot recover $f(m)$.

\paragraph{Bartusek-Khurana's results and our collusion-resistant FE, PE, and compute-and-compare obfuscation.}
While our certified everlasting secure bounded collusion-resistant FE (and its building block SKE, PKE, garbling, and RNCE) schemes are concurrent and independent work, our certified everlasting secure collusion-resistant FE, PE, and compute-and-compare obfuscation schemes use the certified everlasting lemma by Bartusek and Khurana\ifnum\llncs=1.\else (\cref{lem:ce}).\fi\ifnum\anonymous=1\footnote{This is because this paper is a major update version of a paper~\cite{ANO:HMNY22} that appeared right after the first version of~\cite{myC:BarKhu23}. The new additional results are collusion-resistant FE, PE, and compute-and-compare obfuscation. The content in the anonymous authors' paper~\cite{ANO:HMNY22} is a concurrent and independent work of the first version of Bartusek and Khurana's work~\cite{myC:BarKhu23}.}\else\footnote{This is because this paper is a major update version of the paper by Hiroka et al.~\cite{EPRINT:HMNY22} with new additional results (i.e., collusion-resistant FE, PE, and compute-and-compare obfuscation). The content in the work by Hiroka et al.~\cite{EPRINT:HMNY22} is a concurrent and independent work of the work by Bartusek and Khurana~\cite{myC:BarKhu23}.}\fi\
Those three schemes were added after the paper by Bartusek and Khurana was made public.
\emph{Their work does not consider FE, PE, and compute-and-compare obfuscation.}

If we directly apply their generic compiler to PE, we cannot hide the attribute part though we can hide the plaintext part.
Even if we apply the same technique to the attribute part, say, we also set the attribute to $a \oplus r^\prime$ with random $r^\prime$, and put random BB84 states whose computational basis states encode $r^\prime$ in a ciphertext, the idea does not work.
A receiver cannot obtain the plaintext even if $P(a)=1$ because the predicate computes $P(a\xor r^\prime)$ instead of $P(a)$, and the correctness does not hold.

It is non-trivial whether we can obtain certified everlasting compute-and-compare obfuscation by their framework for encryption with certified deletion because we need to hide information about circuits while preserving the functionality. Savvy readers might think it may be possible by applying the framework to the compute-and-compare obfuscation from circular \emph{insecure} FHE by Kluczniak~\cite{myPKC:Kluczniak22}. However, we need compute-and-compare obfuscation to instantiate circular insecure FHE. This is a circular argument.

\paragraph{Certified everlasting secure FE.} Bartusek, Garg, Goyal, Khurana, Malavolta, Raizes, and Roberts~\cite{EPRINT:BGGKMRR23} concurrently and independently obtained adaptively certified everlasting secure collusion-resistant FE for $\Ppoly$ from IO and OWFs. They use subspace coset states~\cite{C:CLLZ21}, while we use BB84 states (with one-time signatures). Hence, the techniques are different. Their scheme is publicly verifiable thanks to the subspace coset state approach. Another technical difference is that they directly rely on adaptively secure multi-input FE (MIFE)~\cite{EC:GGGJKL14,AC:GoyJaiONe16} while we do not. Hence, their scheme incurs an additional sub-exponential loss (from IO to adaptively secure MIFE~\cite{AC:GoyJaiONe16}). Our scheme uses selectively secure MIFE and does not incur sub-exponential loss. We note that selectively secure MIFE and IO are equivalent without any security loss~\cite{EC:GGGJKL14}. They also present several certified everlasting secure primitives that are not considered in our work. However, the results on RNCE, garbled circuits, compute-and-compare obfuscation, and PE are unique to our work.

\subsection{Subsequent Work}\label{sec:subsequent}
A subsequent work by Kitagawa, Nishimaki, and Yamakawa \cite{myTCC:KitNisYam23} shows another application of our authentication technique for BB84 states which we develop for the construction of certified everlasting secure collusion-resistant FE.  
Specifically, they use the technique to construct a generic compiler to add the publicly verifiable deletion property for various kinds of cryptographic primitives solely based on OWFs.   

\ifnum\cameraready=0


\subsection{Technical Overview: Collusion-Resistant FE}\label{sec:tech_overview_CRFE}
\paragraph{Certified everlasting lemma of Bartusek and Khurana.}
Our construction is based on a lemma which we call \emph{certified everlasting lemma} proven by 
Bartusek and Khurana \cite{myC:BarKhu23}, which is described as follows. 

Suppose that $\{\mathcal{Z}(m)\}_{m\in\bit^{\secp+1}}$ is a family of distributions over classical strings such that $\mathcal{Z}(m)$ is computatioally indistinguishable from $\mathcal{Z}(0^{\secp+1})$ for any $m\in \bit^{\secp+1}$. 
Intuitively, $\mathcal{Z}(m)$ can be regarded as an ``encryption'' of $m$. 
For $b\in \bit$ and a QPT adversary, let $\widetilde{\mathcal{Z}}(b)$ be the following experiment: 
\begin{itemize}
		\item The experiment samples $z, \theta \leftarrow \{0, 1\}^{\secp}$. 
		  \item The adversary takes  $\ket{z}_{\theta}$, and $\mathcal{Z}(\theta,b \oplus \bigoplus_{j: \theta_j = 0} z_j)$ as input where $z_j$ is the $j$-th bit of $z$ 
		  and outputs a classical string $z'\in \bit^\secp$ and a quantum state $\rho$.  
		\item The experiment outputs $\rho$ if $z'_j=z_j$ for all $j$ such that $\theta_j=1$ and otherwise outputs a special symbol $\bot$.
	\end{itemize}	
Then for any QPT adversary, the trace distance between $\widetilde{\mathcal{Z}}(0)$ and $\widetilde{\mathcal{Z}}(1)$ is $\negl(\secp)$.\footnote{In fact, we need an ``interactive version'' of the lemma. 
We believe that such an interactive version is implicitly proven and used in \cite{myC:BarKhu23}. 
\ifnum\llncs=1 See the full version for the formal statement of the lemma and a comparison with \cite{myC:BarKhu23}.
\else See \cref{lem:ce_int} for the formal statement of the lemma and \cref{rem:difference_from_BK} for a comparison with \cite{myC:BarKhu23}.\fi}

The above lemma can be regarded as a generic compiler that adds certified everlasting security. 
For example, we can construct a certified everlasting PKE scheme from any plain PKE scheme  
as follows. For encrypting a message $b\in \bit$,  a ciphertext is set to be 
 $\ket{z}_{\theta},\enc(\theta,b \oplus \bigoplus_{j: \theta_j = 0} z_j)$ 
where 
$z,\theta\gets \bit^\secp$ 
and $\enc$ is the encryption algorithm of the underlying PKE scheme. Here, we omit an encryption key for simplicity and keep using a similar convention throughout this subsection. 
The deletion algorithm simply measures $\ket{z}_{\theta}$ in the Hadamard basis to output a certificate $z'$ 
and the verification algorithm checks if $z'_j=z_j$ for all $j$ such that $\theta_j=1$. 
Then the above lemma implies that an adversary's internal state has no information about $b$ 
conditioned on the acceptance, which means
certified everlasting security. 


\paragraph{Public-slot FE.}
Unfortunately, their compiler does not directly work for FE in general. 
The problem is that for a function $f$, there may not exist a function $f'$ such that
$f(m)$ can be recovered from
$f'(m \oplus \bigoplus_{j: \theta_j = 0} z_j,\theta)$ and $z$.
To overcome this issue, we introduce an extension of FE which we call public-slot FE.  
In public-slot FE,
a decryption key is associated with a \emph{two-input} function where the first and second inputs are referred to as the secret and public inputs, respectively. 
Given a ciphertext of a message $m$ and a decryption key for a function $f$, one can compute $f(m,\pub)$ for all public inputs $\pub$. 
Its security is defined similarly to that of plain FE except that the challenge message pair $(m^{(0)},m^{(1)})$ must satisfy $f(m^{(0)},\pub)=f(m^{(1)},\pub)$ for all key queries $f$ and public inputs $\pub$. 

We observe that many existing constructions of FE based on IO (e.g., \cite{SIAMCOMP:GGHRSW16}) can be naturally extended to public-slot FE. In particular, we show that a simple modification of the FE scheme of Ananth and Sahai \cite{TCC:AnaSah16} yields an adaptively secure public-slot FE based on IO. 
\ifnum\llncs=0 See \cref{sec:adaptive_PKFE_public_slot} for details.\fi

\paragraph{First attempt.}
Our first attempt to construct a collusion-resistant FE scheme with certified everlasting security is as follows. 
Let $\enc$ be an encryption algorithm of a public-slot FE scheme. 
A ciphertext for a message $m=m_1\ldots m_n \in \bit^n$ consists of  $\{\ket{z_i}_{\theta_i}\}_{i\in [n]}$ and $\enc(\theta_1,\ldots,\theta_n,\beta_1,\ldots,\beta_n)$
where 
$z_i,\theta_i \gets \bit^{\secp}$ for $i\in [n]$,  
and $\beta_i\seteq m_i \oplus \bigoplus_{j: \theta_{i,j} = 0} z_{i,j}$ where $z_{i,j}$ is the $j$-th bit of $z_i$. 
A decryption key for a function $f$ is a decryption key of the underlying public-slot FE for a two-input function $g[f]$ defined as follows.
The function $g[f]$ takes a secret input  
$(\theta_1,\ldots,\theta_n,\beta_1,\ldots,\beta_n)$ and 
a public input
$(b_1,\ldots,b_n) \in \bit^{\secp\times n}$, 
computes $m_i\seteq \beta_i \oplus \bigoplus_{j:\theta_{i,j}=0}b_{i,j}$ for $i\in[n]$, and outputs $f(m_1,\ldots,m_n)$. To see decryption correctness, we first observe that if we first measure $\{\ket{z_i}_{\theta_i}\}_{i\in [n]}$ in the computational basis to get $(b_1,\ldots,b_n)$, then we have $b_{i,j}=z_{i,j}$ for all $i,j$ such that $\theta_{i,j}=0$. Thus, if we run the decryption algorithm of the public-slot FE scheme with the public input  $(b_1,\ldots,b_n)$, then this yields the correct output $f(m_1,\ldots,m_n)$. 
We remark that the decryption can actually be done without measuring $\{\ket{z_i}_{\theta_i}\}_{i\in [n]}$ by running the above procedure coherently. 
The deletion and verification algorithms can be defined similarly to those for the certified everlasting PKE scheme as explained above: The deletion algorithm simply measures $\{\ket{z_i}_{\theta_i}\}_{i\in [n]}$ in the Hadamard basis to get $\{z'_i\}_{i\in [n]}$ and the verification algorithm checks if $z'_{i,j}=z_{i,j}$ for all $i,j$ such that $\theta_{i,j}=1$.

However, the above scheme is insecure. The problem is that public-slot FE does not 
force an adversary to use a legitimate public input. 
By running the decryption algorithm with different public inputs many times, an adversary can learn more than $f(m_1,...,m_n)$, which would even break security as a plain FE scheme.  
For example, if the adversary uses a public input $(b_1,...,b^\prime_{i},...,b_n)$ such that $b^\prime_i$ is the same as $b_i$ except that $b^\prime_{i,j} \ne b_{i,j}$ for some $j$ such that $\theta_{i,j}=0$, then it can obtain $f(m_1,...,1-m_{i},...,m_n)$.

\paragraph{Certify the public input by one-time signatures.}
Our idea to resolve the above issue is to certify $\{z_i\}_{i\in [n]}$ in the quantum part of the ciphertext  by using one-time signatures. 
Specifically, the encryption algorithm first generates a pair of a verification key $\vk_{i,j}$ and a signing key $\sk_{i,j}$ of a deterministic one-time signature for $i\in [n]$ and $j\in [\secp]$. 
A ciphertext for a message $m=m_1\ldots m_n \in \bit^n$ consists of  $\{\ket{\psi_{i,j}}\}_{i\in [n],j\in[\secp]}$ and $\enc(\{\vk_{i,j}\}_{i\in[n],j\in[\secp]},\theta_1,\ldots,\theta_n,\beta_1,\ldots,\beta_n)$
where 
$z_i,\theta_i \gets \bit^n$ for $i\in [n]$,  
$\beta_i\seteq m_i \oplus \bigoplus_{j: \theta_{i,j} = 0} z_j$, 
and 
\begin{align}
\ket{\psi_{i,j}}\seteq \begin{cases}
    \ket{z_{i,j}}\ket{\sigma_{i,j,z_{i,j}}} &\text{~if~} \theta_{i,j}=0\\ 
    \ket{0}\ket{\sigma_{i,j,0}}+(-1)^{z_{i,j}} \ket{1}\ket{\sigma_{i,j,1}} &\text{~if~} \theta_{i,j}=1
    \end{cases} 
    \end{align}
where $\sigma_{i,j,b}$ is a signature generated by using the signing key $\sk_{i,j}$ on the message $b\in \bit$. 
Note that $\ket{\psi_{i,j}}$ is the state obtained by coherently running the signing algorithm with the signing key $\sk_{i,j}$ on $j$-th qubit of $\ket{z_i}_{\theta_i}$. 
We modify the function $g[f]$ associated with the decryption key of the public-slot FE 
to additionally check the validity of the signatures for $b_{i,j}$ for $i,j$ such that   $\theta_{i,j}=0$. That is,  $g[f]$ takes a secret input  
$(\{\vk_{i,j}\}_{i\in[n],j\in[\secp]},\theta_1,\ldots,\theta_n,\beta_1,\ldots,\beta_n)$ and 
a public input
$(b_1,\ldots,b_n,\sigma_1,\ldots,\sigma_n)$, 
parses $\sigma_{i}=(\sigma_{i,1},\ldots,\sigma_{i,\secp})$ for each $i\in[n]$, 
and checks if $\sigma_{i,j}$ is a valid signature for $b_{i,j}$ (i.e., if $\sigma_{i,j}=\sigma_{i,j,b_{i,j}}$) for all $i,j$ such that $\theta_{i,j}=0$. If it is not the case, it just outputs $\bot$. Otherwise, it
computes $m_i\seteq \beta_i \oplus \bigoplus_{j:\theta_{i,j}=0}b_{i,j}$ for $i\in[n]$ and outputs $f(m_1,\ldots,m_n)$.
Note that $\ket{\psi_{i,j}}$ contains the valid signature $\sigma_{i,j,z_{i,j}}$ on the message $z_{i,j}$ whenever $\theta_{i,j}=0$. 
Thus, the decryption correctness is unaffected. 
In addition, if we measure $\ket{\psi_{i,j}}$ in the Hadamard basis for $i,j$ such that $\theta_{i,j}=1$, then the outcome $(c_{i,j},d_{i,j})$ satisfies $z_{i,j}=c_{i,j}\oplus d_{i,j}(\sigma_{i,j,0}\oplus \sigma_{i,j,1})$. By modifying the verification algorithm to check the above equality, the verification correctness also holds. 
By the security of one-time signatures, an adversary cannot arbitrarily modify the public input when running the decryption algorithm of the underlying public-slot FE. 

While this authentication technique seems to prevent obvious attacks, we still do not know how to prove certified everlasting security of this scheme. 
In particular, we want to rely on the certified everlasting lemma of \cite{myC:BarKhu23}. However, the lemma only enables us to perform bit-wise game hops. For example, if $n=3$ and the challenge messages are $000$ and $111$, we would need to consider hybrid experiments where the challenge message evolves as 
$000\rightarrow 100 \rightarrow 110 \rightarrow 111$.\footnote{Note that an FE scheme with $3$-bit messages itself is trivial to construct from any PKE scheme. We are considering this toy example just to explain a technical difficulty.} However, the restriction on the adversary only ensures $f(000)=f(111)$ for decryption key queries $f$ and does not ensure, say, $f(000)=f(100)$. Without this condition, 
we cannot rely on the security of the underlying public-slot FE. Hence, it seems impossible to prove indistinguishability between neighboring intermediate hybrids.

\paragraph{Redundant encoding.}
Our idea for resolving the above issue is to encode the message in a redundant way so that there is a space for a ``spare message''. 
Specifically, we first encode a message $m=m_1\ldots m_n \in \bit^n$ into a $(2n+1)$-bit string $m_1\ldots m_n\concat 0^{n+1}$. The rest of the scheme is identical to that in the previous paragraph, except that 
$i$'s range is $[2n+1]$ instead of $[n]$ and  
$g[f]$ chooses which part to use for deriving the output depending on the value of the $(2n+1)$-th bit. 
Specifically, $g[f]$ takes a secret input  
$(\{\vk_{i,j}\}_{i\in[2n+1],j\in[\secp]},\theta_1,\ldots,\theta_{2n+1},\beta_1,\ldots,\beta_{2n+1})$ and 
a public input
$(b_1,\ldots,b_{2n+1},\sigma_1,\ldots,\sigma_{2n+1})$ and first checks the validity of the signatures on positions corresponding to $i,j$ such that $\theta_{i,j}=0$ as before.      
Then it computes $m_i\seteq \beta_i \oplus \bigoplus_{j:\theta_{i,j}=0}b_{i,j}$ for $i\in[2n+1]$, and outputs $F(m_1,\ldots,m_{2n+1})$ where $F$ is defined as 
\begin{align}
F(m_1,\ldots,m_{2n+1})\seteq 
\begin{cases}
     f(m_1,\ldots,m_n)
     &\text{~if~} m_{2n+1}=0\\ 
     f(m_{n+1},\ldots,m_{2n})
     &\text{~if~} m_{2n+1}=1
    \end{cases}. 
\end{align}

The decryption correctness is unaffected because we always have $m_{2n+1}=0$ when decrypting an honestly generated message. 
The verification correctness is also unaffected since the way of encoding messages is irrelevant. 
We explain why this enables us to avoid the issue mentioned in the previous paragraph. 
Intuitively, the advantage of such a redundant encoding is that we can ensure that the encoded challenge message contains either of two challenge messages in all intermediate hybrids. 
Let $m^{(0)}$ and $m^{(1)}$ be a pair of challenge messages. Note that they correspond to $m^{(0)}\concat 0^{n+1}$ and $m^{(1)}\concat 0^{n+1}$ after encoding. Then we consider intermediate hybrids where the corresponding challenge messages after the encoding evolves as follows:
\begin{enumerate}
    \item Starting from $m^{(0)}\concat 0^{n+1}$, we change the $(n+1)$-th to $2n$-th bits one-by-one toward  $m^{(0)}\concat m^{(1)}\concat 0$.
    \item Flip the $(2n+1)$-th bit, which results in $m^{(0)}\concat m^{(1)}\concat 1$.
    \item Change the first to $n$ bits one-by-one toward $m^{(1)}\concat m^{(1)}\concat 1$.
    \item Flip the $(2n+1)$-th bit, which results in $m^{(1)}\concat m^{(1)}\concat 0$.
    \item Change the $(n+1)$-th to $2n$-th bits one-by-one toward $m^{(1)}\concat 0^{2n+1}$.
\end{enumerate}

Importantly, the value of $F$ on the encoded challenge message is equal to $f(m^{(0)})=f(m^{(1)})$ at any point of the hybrids. This enables us to rely on the security of the underlying public-slot FE along with certified everlasting lemma in every hybrid. 

\paragraph{What one-time signatures to use?}
Finally, we remark that we have to choose an instantiation of one-time signatures carefully. 
Roughly speaking, the reason why we are using one-time signatures is to prevent an adversary from using ``unauthorized'' $b_{i,j}$, i.e., those for which the valid signature $\sigma_{i,j,b_{i,j}}$ is not given to the adversary.  However, by the correctness of one-time signatures, a valid signature must exist on every message. This means that a valid signature on an ``unauthorized'' $b_{i,j}$ must exist even if it is difficult for an adversary to find. This situation is not compatible with the security definition of public-slot FE. Recall that its security requires that the challenge message pair $(m^{(0)},m^{(1)})$ must satisfy $f(m^{(0)},\pub)=f(m^{(1)},\pub)$ for all key queries $f$ and \emph{all} public inputs $\pub$. That is, the security is not applicable if there is at least one $\pub$ such that $f(m^{(0)},\pub)\neq f(m^{(1)},\pub)$ even if such $\pub$ is difficult to find. 
To overcome this issue, we use Lamport signatures instantiated with a PRG. 
Let $\PRG:\bit^{\secp}\rightarrow \bit^{2\secp}$ be a PRG.  
When the message length is $1$, a signing key is set to be $(u_0,u_1)\in \bit^{\secp\times 2}$ and a verification key is set to be $(v_0=\PRG(u_0),v_1=\PRG(u_1))\in \bit^{2\secp\times 2}$.  
A signature for a bit $b$ is defined to be $u_b$. 
This scheme has a special property in that we can program a verification key so that it does not have a valid signature for a particular message. For example, if we want to ensure that a message $0$ does not have a valid signature, then we can set $v_0$ to be a uniformly random $2\secp$-bit string. Then, with probability $1-2^{-\secp}$, there is no preimage of $v_0$, which means that there is no valid signature on the massage $0$.  By using this property, whenever $b_{i,j}$ is unauthorized,    
we can switch to a hybrid where there is no valid signature for $b_{i,j}$. This effectively resolves the above issue. 

\subsection{Technical Overview: Bounded Collusion-Resistant FE}\label{sec:tech_overview_bounded_FE}
In this subsection, we give a high-level overview of our certified everlasting secure bounded collusion-resistant FE schemes.
It is known that the (bounded collusion-resistant) plain FE is constructed from (standard) PKE, RNCE, and garbling~\cite{C:GorVaiWee12}.
A natural strategy is constructing PKE, RNCE, and garbling with certified everlasting security and using them as building blocks. 
We show that PKE with certified everlasting security can be constructed using the techniques of \cite{JACM:Unruh15,C:HMNY22}.
RNCE with certified everlasting security for {\it classical messages} can be constructed from certified everlasting PKE in the same way as standard RNCE~\cite{C:KNTY19}.
However, such an RNCE scheme is insufficient for our purpose (constructing adaptively-secure FE) because it is not for \emph{quantum messages}.
We also need a new idea to construct garbling with certified everlasting security.
The following explains these ideas and how to construct FE with certified everlasting security.

\paragraph{Certified everlasting garbling for $\Ppoly$ circuits.}
In classical cryptography, it is known that we can construct plain garbling from plain SKE using double-encryption~\cite{FOCS:Yao86,JC:LinPin09}. Double-encryption means we generate a nested ciphertext $\ct_2\la\Enc(\sk',\ct_1)$,
where $\ct_1\la\Enc(\sk,m)$,
$m$ is the message, $\Enc$ is the encryption algorithm of SKE, and $\sk,\sk'$ are secret keys of SKE.
This double-encryption is an essential technique for garbling. However, it is an obstacle to our purpose.
First, we do not know SKE with certified everlasting security for {\it quantum} messages.
Second, even if the first problem is solved, we have another problem:
We can obtain a valid certificate showing that $\ct_1$ has been deleted by running the deletion algorithm on $\ct_2$. However, such a certificate does not necessarily mean the deletion of $m$.
We bypass the problem using XOR secret sharing instead of double-encryption.\footnote{
A similar technique was used by Gentry, Halevi, and Vaikuntanathan~\cite{C:GenHalVai10}.
} 
More precisely, we uniformly randomly sample $p$ and compute $(\vk',\ct')\la\Enc(\sk',p)$ and $(\vk,\ct)\la\Enc(\sk,p\oplus m)$ to encrypt message $m$. 
Here, $\Enc$ is the encryption algorithm of certified everlasting SKE, and $\vk',\vk$ are the verification keys that are used to verify the correctness of deletion certificates.
The receiver with $(\ct',\ct)$ can obtain $m$ only if it has both $\sk'$ and $\sk$, and nothing else otherwise, as in the case of double-encryption.
Furthermore, once the receiver issues the deletion certificate of $(\ct',\ct)$, it can no longer obtain the information of $m$ even if it becomes computationally unbounded.

It is easy to see that we can implement the well-known gate garbling~\cite{FOCS:Yao86,JC:LinPin09} by using the double encryption in the parallel way above instead of the sequential double encryption.
We can prove its computational security via a similar discussion as that in~\cite{JC:LinPin09}. (Although \cite{JC:LinPin09} uses double-encryption, we can show the security for the XOR secret sharing case similarly.)
Furthermore, we can prove its certified everlasting security by using the certified everlasting security of the SKE.
Hence, we can obtain certified everlasting garbling.
\ifnum\llncs=0 The formal construction of our certified everlasting garbling is given in \cref{sec:const_garbling}. For details, see that section.\fi

\paragraph{FE with non-adaptive security.}
Our next task is achieving certified everlasting FE using certified everlasting garbling.
It is known that plain FE with \emph{non-adaptive security} can be constructed by running the encryption algorithm of (plain) PKE on labels of a plain garbling scheme~\cite{CCS:SahSey10}.\footnote{The non-adaptive security means that the adversary can call the key queries only before the challenge encryption query.}
In our certified everlasting garbling scheme (explained in the previous paragraph), the labels are classical bit strings and the deletion algorithm does not take the labels as input.
Therefore, this classical construction for plain FE by Sahai and Seyalioglu~\cite{CCS:SahSey10} can be directly applied to 
the construction of our $1$-bounded certified everlasting FE for $\Ppoly$ circuits with \emph{non-adaptive security}.
\ifnum\llncs=0(The formal construction is given in \cref{sec:const_func_non_adapt}.)\fi

\paragraph{FE with adaptive security.}
Now, we want to convert non-adaptive security to adaptive one.\footnote{The adaptive security means that the adversary can call key queries before and after the challenge encryption query.}
However, the conversion is non-trivial.
Let us first review the conversion for plain FE.
In classical cryptography, we can convert non-adaptively secure FE into adaptively secure FE by using RNCE.
Roughly speaking, RNCE is the same as PKE except that we can generate a fake ciphertext $\widetilde{\ct}\la\Fake(\pk)$ without plaintext and we can generate a fake secret key 
$\widetilde{\sk}\la\Reveal(\pk,m)$ that decrypts $\widetilde{\ct}$ to $m$.
The security of RNCE guarantees that
$(\Enc(\pk,m),\sk)$ and $(\Fake(\pk),\Reveal(\pk,m))$ are computationally indistinguishable, where $\Enc$ is the real encryption algorithm, and $\sk$ is the real secret key.
Adaptively secure FE can be constructed by running the real encryption algorithm $\Enc$ of the RNCE on the ciphertext $\nad.\ct$ of the FE. 
We can prove its adaptive security as follows.
The adversary of adaptive security can send key queries after the challenge encryption query.
However, the sender can simulate the challenge encryption query without generating $\nad.\ct$.
This is because, from the security of RNCE, we can switch $(\Enc(\pk,\nad.\ct),(\sk,\nad.\sk_f))$ to the fake one $(\Fake(\pk),(\Reveal(\pk,\nad.\ct),\nad.\sk_f))$, where $\nad.\sk_f$ is the functional secret key of the non-adaptively secure FE.
Therefore, the sender needs not generate $\nad.\ct$ before generating $\nad.\sk_f$ for the simulation of the adversary's queries, 
which means that we can reduce the adaptive security to the non-adaptive security.

How can we adopt the above classical idea of the conversion to the certified everlasting case?
From the discussion above, a straightforward way is to encrypt the ciphertext $\nad.\ct$ of certified everlasting FE with non-adaptive security using certified everlasting RNCE as follows: $(\vk,\ct)\la\Enc(\pk,\nad.\ct)$, where $\vk$ is the verification key, $\pk$ is the public key, and $\Enc$ is the real encryption algorithm of the certified everlasting RNCE.
However, this idea fails for the following two reasons.
First, $\nad.\ct$ is a quantum state. Our certified everlasting RNCE scheme does not support quantum messages.
Second, even if we can construct RNCE for quantum messages, we have another problem:
A valid certificate of $\ct$ is issued by running the deletion algorithm on $\ct$. However, such a certificate does not necessarily mean the deletion of the plaintext of $\nad.\ct$. 
The first problem is about security, and the second problem is about correctness.

Our idea to resolve the first problem is to use quantum teleportation.
We construct RNCE for quantum messages from RNCE for classical messages by using quantum teleportation.\footnote{A similar technique was used in the context of multi-party quantum 
computation~\cite{C:BCKM21a}.}
(We believe that the idea of using quantum teleportation in the following way will be useful in many other applications beyond RNCE.) 
As the ciphertext and the secret key of adaptively secure FE, we take 
\begin{eqnarray*}
\frac{1}{2^{2N}}\sum_{a,b\in\bit^N}(Z^bX^a(\nad.\ct) X^aZ^b)_{C1} \otimes \Enc(\pk,(a,b))_{C2}\otimes
(\nad.\sk_f,\sk)_{S},
\end{eqnarray*}
where $\nad.\ct$ is an $N$-qubit state, the registers $C_1$ and $C_2$ are the ciphertext, and 
the register $S$ is the secret key. 
Here, $X^a\coloneqq\bigotimes_{j=1}^NX_j^{a_j}$,
$Z^b\coloneqq\bigotimes_{j=1}^N Z_j^{b_j}$, $a_j$ is the $j$th bit of $a$, and
$b_j$ is the $j$th bit of $b$.
Moreover, $\Enc$ is the real encryption algorithm of RNCE for classical messages,
$\nad.\sk_f$ is the secret key of non-adaptively secure FE, and $\sk$ is the real secret key of RNCE for classical messages.
We want to show the adaptive security of the construction by reducing it to the non-adaptive security of the building block FE.
In the first step of hybrids, we switch the state
to 
\begin{eqnarray*}
\frac{1}{2^{2N}}\sum_{a,b}(Z^bX^a(\nad.\ct) X^aZ^b)_{C1}\otimes \Fake(\pk)_{C2}
\otimes
(\nad.\sk_f,\Reveal(\pk,(a,b)))_S
\end{eqnarray*}
by using the property of RNCE for classical messages. 
In the second step of hybrids, we switch the state to 
\begin{eqnarray*}
\frac{1}{2^{2N}}\sum_{x,z\in\bit^N}
\mathcal{T}^{x,z}_{A',A}[\nad.\ct_{A'}\otimes|\Phi_N\rangle\langle\Phi_N|_{A,C1}]\otimes
\Fake(\pk)_{C2}\otimes
(\nad.\sk_f,\Reveal(\pk,(x,z)))_S,
\end{eqnarray*}
where $|\Phi_N\rangle$ is the $N$ Bell pairs between the registers $A$ and $C_1$.
$\mathcal{T}^{x,z}_{A',A}[\nad.\ct_{A'}\otimes|\Phi_N\rangle\langle\Phi_N|_{A,C1}]$ is the state on the register $C_1$ obtained in the following way:
the state $\nad.\ct_{A'}$ on the register $A'$ is coupled with the halves of $N$ Bell pairs on the register $A$,
and the teleportation measurement $\mathcal{T}^{x,z}_{A',A}$ with the result $(x,z)$ is applied on the registers $A$ and $A'$. 
Now, we can generate the states on the registers $C_1$ and $C_2$ without knowing $\nad.\ct$, 
which means that
the sender can simulate the challenge encryption query without $\nad.\ct$.
In other words, the sender does not need to generate $\nad.\ct$ before generating $\nad.\sk_f$ for the simulation of adversary queries. 

This idea solves the first problem.
However, the second problem remains. The receiver with $(Z^bX^a(\nad.\ct) X^aZ^b,\allowbreak\Enc(\pk,(a,b)))$ can issue a deletion certificate of $Z^bX^a(\nad.\ct) X^aZ^b$. The deletion certificate does not necessarily pass the verification algorithm for
the deletion of $\nad.\ct$.
This is an obstacle to achieving correctness.
We solve this problem by introducing an efficient algorithm that we call the modification algorithm.
Let $\nad.\cert^*$ be the deletion certificate of $Z^bX^a(\nad.\ct) X^aZ^b$. 
The modification algorithm takes $(a,b)$ and $\nad.\cert^*$ as input, 
and outputs $\nad.\cert$ that is the deletion certificate of $\nad.\ct$. 
Therefore, by using the modification algorithm, we can convert
the deletion certificate $\nad.\cert^*$ of $Z^bX^a(\nad.\ct) X^aZ^b$
to the deletion certificate $\nad.\cert$ of $\nad.\ct$.
We observe that the modification algorithm exists for many natural constructions, including our construction.\footnote{If the deletion algorithm is the computational-basis measurements 
followed by Clifford gates, the modification algorithm is just modifying the Pauli one-time pad, $X^aZ^b$.
In fact, all known constructions use only Hadamard basis measurements.}

\ifnum\llncs=0The formal explanation of the conversion from non-adaptive to adaptive FE is given in \cref{sec:const_fe_adapt}.\fi

\paragraph{$q$-bounded FE for $\NCone$ circuits.}
Finally, we explain how to convert 1-bounded one to the $q$-bounded one.\footnote{$q$-bounded means that the adversary can call key queries $q$ times
with an a priori bounded polynomial $q$.}
Unfortunately, 
we do not know how to obtain $q$-bounded certified everlasting FE for $\Ppoly$ circuits.
What we can construct in this paper is that only for $\NCone$ circuits. 
\ifnum\llncs=1 (It is an open problem to obtain $q$-bounded certified everlasting FE for $\Ppoly$ circuits.) \else
(It is an open problem to obtain $q$-bounded certified everlasting FE for $\Ppoly$ circuits. For more details, see \cref{sec:const_multi_fe}.)
\fi

Let us explain how to convert
$1$-bounded certified everlasting FE for $\Ppoly$ circuits
to
$q$-bounded certified everlasting FE for $\NCone$ circuits.
In classical cryptography, it is known that~\cite{C:GorVaiWee12}
multi-party computation (MPC) can convert
plain $1$-bounded FE for $\Ppoly$ circuits to plain $q$-bounded FE for $\NCone$ circuits.
The idea is, roughly speaking, the view of each party in the MPC protocol is encrypted using $1$-bounded FE scheme.
In this classical construction, no encryption is done on the ciphertexts of plain FE,
and therefore this classical construction can be directly applied to our certified everlasting case.
(It is an open problem to obtain $q$-bounded certified everlasting FE for $\Ppoly$ circuits.
 \ifnum\llncs=0 The formal construction is given in \cref{sec:const_multi_fe}. \fi


\subsection{Technical Overview: Compute-and-Compare Obfuscation}\label{sec:tech_overview_cc_obfuscation}

This section provides a high-level overview of our certified everlasting compute-and-compare obfuscation. Recall that a compute-and-compare obfuscation scheme obfuscates a circuit $P$ along with a lock value $\lock$ and a message $\msg$ and outputs an obfuscated circuit $\tlP$. In the evaluation phase, one can recover $\msg$ from $\tlP$ using an input $x$ to the circuit such that $P(x) = \lock$. A certified everlasting compute-and-compare obfuscation scheme additionally generates a verification key $\vk$ while obfuscating circuit $P$. A user can generate a deletion certificate $\cert$ from $\tlP$. If we have $\vk$, we can verify whether the certificate is valid or not. The certified everlasting security ensures that no information about $P, \lock$ and $\msg$ is available to the user after producing a valid certificate. This means that the user actually deleted the obfuscated circuit.


\paragraph{Compute-and-compare obfuscation without a message.} We first explain our idea to construct a certified everlasting compute-and-compare obfuscation without any message. That is, the evaluation returns $1$ if $P(x) = \lock$ holds. Let $\CCObf$ be the obfuscation algorithm of a standard compute-and-compare obfuscation scheme and $\enc, \dec$ be the encryption, and decryption algorithms of FHE. The main idea is to compute an FHE ciphertext $\ct_P$ encrypting the circuit $P$ and use $\CCObf$ to produce an obfuscated circuit $\widetilde{\dec}$ of the decryption circuit of FHE with lock value $\lock$ and message $1$. The obfuscated circuit $\tlP$ consists of $\ct_P$ and $\widetilde{\dec}$. Given an input $x$, we first apply the evaluation procedure of FHE to get a ciphertext $\ct_{P(x)}=\Enc(P(x))$ (we omit the encryption key) and then run the evaluation algorithm of the compute-and-compare obfuscation with input $\ct_{P(x)}$ to check whether $P(x) = \lock$. Note that we cannot use certified everlasting FHE \cite{myC:BarKhu23} in a black-box manner since $\CCObf$ is a classical algorithm that cannot obfuscate a quantum circuit, in particular, the decryption algorithm of the FHE. Instead, we use BB84 states along with classical FHE as follows. The obfuscated circuit $\tlqP$ consists of $\widetilde{\dec} \seteq \CCObf(\dec(\sk, \cdot), \lock, 1)$ and  $\{\ket{z_i}_{\theta_i}, \ct_i\}_{i\in [\ell_P]}$ where $\ct_i \seteq \enc(\theta_i \concat \wt{b}_i)$, $z_i,\theta_i \gets \bit^{\secp}$ for  $i\in [\ell_P]$, $\wt{b}_i\seteq b_i \oplus \bigoplus_{j: \theta_{i,j} = 0} z_{i, j}$ and $b_i$ is the $i$-th bit of the binary string of length $\ell_P$ representing the circuit $P$. The verification key is $\vk = (\{z_i, \theta_i\}_{i\in [\ell_P]})$. To evaluate the obfuscated circuit with an input $x$, we first coherently compute an evaluated FHE ciphertext $\ket{\ct_{U_x(P)}}$ where $U_x$ is a circuit that on input $(\{z_i, \theta_i, \wt{b}_i\}_{i\in [\ell_P]})$ first recovers $b_i$, the bits representing $P$, and then outputs $P(x)$. Then, we coherently evaluate the obfuscated circuit $\widetilde{\dec}$ with input $\ket{\ct_{U_x(P)}}$ and check that the measured outcome is $1$ to decide $P(x) = \lock$. The deletion and verification algorithm works similarly as in the certified everlasting PKE scheme described in \cref{sec:tech_overview_CRFE}. That is, we use the concrete certified everlasting secure FHE scheme by Bartusek and Khurana in a non-black-box way.\par 
However, the above scheme cannot guarantee certified everlasting security. The reason is that the classical compute-and-compare obfuscation cannot hide the lock value from an unbounded adversary. More precisely, the unbounded adversary is given a target circuit and an auxiliary input and can easily distinguish between the obfuscated circuit $\widetilde{\dec}\gets \CCObf(1^\secp,\dec,\lock,1)$ and the corresponding simulated circuit $\widetilde{\dec}\gets \sfCC.\Sim(1^\secp,\pp_{\dec},1^1)$ if the auxiliary input and $\lock$ are correlated, where $\pp_{\dec}$ consists of parameters of $\dec$ (input and output length and circuit size).

We solve this problem by masking the obfuscated circuit that encodes $\lock$ using the XOR function in combination with the BB84 states. In particular, we sample ``dummy'' lock value $R \gets \bit^{\secp}$ and set the obfuscated circuit $\qL_C$ as $(\widetilde{\dec} \seteq \CCObf(\dec(\sk, \cdot), R, 1), \{\ket{z_i}_{\theta_i}, \ct_i\}_{i\in [\ell]})$ where $\ell = \ell_P + \ell_{\tlI}$ and $\{\ct_i\}_{i \in [\ell]}$ encrypts the binary string representing the circuits $(P \concat \tlI)$ where $\tlI \seteq \CCObf(I, \lock, R)$. We denote $I$ by the identity circuit that is $I(x) = x$ for all $x$. The evaluation algorithm works as before except that the circuit $U_x$ on input $(\{z_i, \theta_i, \wt{b}_i\}_{i\in [\ell]})$ first reconstructs $(P \concat \tlI)$ and then outputs the result obtained in the evaluation of $\tlI$ with input $P(x)$. Hence, checking $P(x)=\lock$ is deferred until evaluating $\tlI$, which is hidden due to the certified everlasting security of FHE. The correctness follows from the fact that $U_x$ returns $R$ if $P(x) = \lock$ and evaluation of $\widetilde{\dec}$ outputs $1$ if $U_x(P\concat \tlI) = R$.\par 

The simulated circuit $\tlqP$ consists of $\widetilde{\dec} = \CCObf(\dec(\sk, \cdot), R, 1)$ and $ \{\ket{z_i}_{\theta_i}, \ct_i\}_{i\in [\ell]}$ where $\ct_i \seteq \enc(\theta_i\concat \wt{b}_i)$ and $\wt{b}_i\seteq 0 \oplus \bigoplus_{j: \theta_{i,j} = 0} z_{i, j}$ for  $i\in [\ell]$. Note that, $\tlqP$ does not contain any information about $P$ and $\lock$. We rely on the certified everlasting lemma of \cite{myC:BarKhu23} to show that the real obfuscated circuit is indistinguishable from the simulated circuit for any unbounded adversary who produces a valid certificate of deletion. Although an unbounded adversary can recover $\sk$ from $\widetilde{\dec}$, $\sk$ is useless for distinguishing after the deletion. Since the lemma only allows us to flip one bit at a time, we use a sequence of $\ell$ hybrid experiments. In the $i$-th hybrid, we change the bit $b_i$ from $1$ to $0$. If we can show that $\enc(\theta_i \concat \wt{b}_i)$ is computationally indistinguishable from $\enc(\bm{0}\concat \wt{b}_i)$ and then it is possible to apply the certified everlasting lemma to flip the bit $b_i$ without noticing the unbounded adversary. To establish the computational indistinguishability, we first replace the circuit $\tlI\gets \CCObf(I,\lock,R)$ with the simulated one $\tlI\gets\sfCC.\Sim(1^\secp,\pp_{I},1^{\abs{R}})$ and then change the circuit $\widetilde{\dec}\gets \CCObf(\dec,\lock,1)$ to the corresponding simulated circuit $\widetilde{\dec}\gets \sfCC.\Sim(1^\secp,\pp_{\dec},1^1)$ depending on the security of the underlying compute-and-compare obfuscation scheme. Since the FHE secret key $\sk$ is no longer required to simulate the adversary's view, we can change $\enc(\theta_i\concat \wt{b}_i)$ to $\enc(\bm{0}\concat \wt{b}_i)$ using the IND-CPA security of FHE. Hence, $b_i$ can be set to $0$ by employing the certified everlasting lemma.

\paragraph{Compute-and-compare obfuscation with a message.} Next, we discuss extending the above construction into a certified everlasting compute-and-compare obfuscation scheme that obfuscates a circuit $P$ along with $\lock$ and a message $m = m_1 \ldots m_n \in \bit^n$. Our idea is to encrypt the message using FHE in combination with the BB84 states and recover the message bits during evaluation depending on the outcome of the obfuscated circuit $\widetilde{\dec}$. The obfuscated circuit $\tlqP$ now additionally includes $\{\ket{z_{\ell+k}}_{\theta_{\ell+k}}, \ct_{\ell+k}\}_{k\in [n]}$ where $z_{\ell+k},\theta_{\ell+k} \gets \bit^{\secp}$ , $\ct_{\ell+k} \seteq \enc(\theta_{\ell+k}\concat \wt{b}_{\ell+k})$ and $\wt{b}_{\ell+k}\seteq m_k \oplus \bigoplus_{j: \theta_{\ell+k,j} = 0} z_{\ell+k, j}$ for  $k\in [n]$. The evaluation procedure works as before except the $U_x$ on input $((\{z_i, \theta_i, \wt{b}_i\}_{i\in [\ell]}), (z_{\ell+k}, \theta_{\ell+k}, \wt{b}_{\ell+k}))$ first reconstructs $(P \concat \tlI)$ from $\{z_i, \theta_i, \wt{b}_i\}_{i\in [\ell]}$ and $m_k$ from $(z_{\ell+k}, \theta_{\ell+k}, \wt{b}_{\ell+k})$, and then outputs $m_k \cdot \tlI(P(x))$. We can similarly define the deletion and verification algorithms as before. The scheme correctly recovers $m$ in a bit-by-bit manner. Let us consider $P(x) = \lock$ and $m_k = 1$. Then, by the definition of $U_x$ and the correctness of compute-and-compare obfuscation, we have $m_k \cdot \tlI(P(x)) = R$. Consequently, $\widetilde{\dec}$ evaluates to $1$ for an input $\ct_{m_k \cdot \tlI(P(x))}$. If the result of the evaluation is not $1$, then we set $m_k \seteq 0$. We prove the certified everlasting security of the scheme using the same idea as discussed for the compute-and-compare obfuscation scheme without a message. The only difference is that we additionally delete the information of $m$ using the IND-CPA security of FHE and certified everlasting lemma of \cite{myC:BarKhu23} after we erase the information about $P$ and $\lock$. 
\ifnum\llncs=0 The formal construction and its security analysis are provided in \cref{sec:ccobf}.
\fi

\paragraph{Certified everlasting predicate encryption.} Goyal, Koppula and Waters \cite{FOCS:GoyKopWat17} and Wichs and Zirdelis \cite{FOCS:WicZir17} showed a generic construction of PE\footnote{It satisfies one-sided attribute-hiding security, meaning that the attribute and message are both hidden to a user who does not have a secret key for successful decryption.} from compute-and-compare obfuscation and ABE. The construction works as follows. The setup and key generation algorithms are the same as the underlying ABE. Let $\enc$ and $\dec$ be the encryption and decryption algorithms of ABE. To encrypt a message $m$ with attribute $x$, the encryption algorithm samples a random lock $R \in \bit^{\ell}$ and computes $\ct \seteq \enc(x, R)$ and $\widetilde{\dec} \seteq \CCObf(\dec(\cdot, \ct), R, m)$. The ciphertext is the obfuscated circuit $\widetilde{\dec}$. Given a secret key $\sk_P$ for a policy $P$, a user simply evaluates $\widetilde{\dec}$ with input $\sk_P$ to recover the message $m$. Note that, if $P(x) = 1$ then by the correctness of ABE, $\dec(\sk_P, \ct) = R$ and hence $\widetilde{\dec}(\sk_P)$ returns $m$.\par 

One might hope that replacing the compute-and-compare obfuscation and ABE with their certified everlasting counterparts in the classical construction yields a certified everlasting PE. This would not work since our certified everlasting compute-and-compare obfuscation cannot obfuscate a quantum decryption circuit $\qdec(\cdot, \qct)$ of the certified everlasting ABE. However, we need to erase the information about $R$ from the ABE ciphertext $\ct$ in order to apply the certified everlasting security of the compute-and-compare obfuscation. In other words, after a valid certificate of deletion is produced, an unbounded adversary should not be able to distinguish between $\enc(x, R)$ and $\enc(x, \bm{0})$. A classical ABE alone can not guarantee such indistinguishability. We solve this problem by using a classical ABE coupled with BB84 states and a certified everlasting compute-and-compare obfuscation in the above construction. In particular, we first sample $z_i, \theta_i \gets \bit^{\ell}$, set $\wt{r}_i\seteq r_i \oplus \bigoplus_{j: \theta_{i,j} = 0} z_{i, j}$ and then compute $\ct \seteq \enc(x, (\theta_1, \ldots, \theta_{\ell}, \wt{r}_i, \ldots, \wt{r}_{\ell}))$ where $r_i$ denotes the $i$-th bit of $R$. The ciphertext consists of $\widetilde{\qdec} \seteq \qCCObf(\dec(\cdot, \ct), R, m)$ and $\{\ket{z_i}_{\theta_i}\}_{i \in [\ell]}$. The verification key associated with the ciphertext includes $\{z_i, \theta_i\}_{i \in [\ell]}$ and a verification key $\vk_{\dec}$ corresponding to the circuit $\dec$. The deletion and verification algorithms can be defined in a natural way. That is, we use the concrete certified everlasting secure ABE scheme by Bartusek and Khurana in a non-black-box way.\par
Suppose an adversary queries secret keys $\sk_P$ such that $P(x) = 0$ and becomes unbounded after delivering a valid certificate of deletion of the ciphertext. Our idea is to use the security of ABE and the certified everlasting lemma of \cite{myC:BarKhu23} to delete the information of $R$. Then, we utilize the certified everlasting security of compute-and-compare obfuscation for replacing $\widetilde{\qdec}$ with a simulated circuit that does not contain any information about $m, x$. 
\ifnum\llncs=0 The formal security analysis can be found in \cref{sec:PE_from_LObf_ABE}. 
\fi


\ifnum\submission=1
\subsection{More on Related Works}\label{sec:related_work}
\else
\subsection{More on Related Works}\label{sec:related_work}
\fi

\paragraph{Ciphertext certified deletion.}
Unruh~\cite{JACM:Unruh15} introduced the concept of revocable quantum time-released encryption.
In this primitive, a receiver possessing quantum encrypted data can obtain its plaintext after a predetermined time $T$.
The sender can revoke the quantum encrypted data before time $T$. 
If the revocation succeeds, the receiver cannot obtain the plaintext information even if its computing power becomes unbounded.

Broadbent and Islam~\cite{TCC:BroIsl20} constructed one-time SKE with certified deletion.
It is standard one-time SKE except that once the receiver issues a valid classical certificate, the receiver cannot obtain the plaintext information even if the receiver later becomes a computationally \emph{unbounded} adversary. (See also \cite{ARXIV:KunTan20}.)

Hiroka, Morimae, Nishimaki, and Yamakawa~\cite{AC:HMNY21} constructed reusable SKE, PKE, and ABE with certified deletion.
These reusable SKE, PKE, and ABE with certified deletion are standard reusable SKE, PKE, and ABE with additional properties, respectively. Once the receiver issues a valid classical certificate, the receiver cannot obtain the plaintext information even if it obtains some secret information (e.g., the master secret key of ABE). In these primitives, the security holds against computationally bounded adversaries, unlike in this work.
Poremba~\cite{myITCS:Poremba23} achieved FHE with certified deletion. In addition, certificates for deletion are publicly verifiable in his construction. The security holds against computationally bounded adversaries, unlike in this work.
However, the security of the construction relies on a strong conjecture that a particular hash function is “strong Gaussian-collapsing”.

Hiroka, Morimae, Nishimaki, and Yamakawa~\cite{C:HMNY22} constructed commitments with statistical binding and certified everlasting hiding. From it, they also constructed a certified everlasting zero-knowledge proof system for 
$\class{QMA}$ based on the zero-knowledge protocol of~\cite{FOCS:BroGri20}.


\paragraph{Key certified deletion.}
Kitagawa and Nishimaki~\cite{myAC:KitNis22} introduced the notion of FE with secure key leasing, where functional decryption keys are quantum states and we can generate certificates for deleting the keys. This can be seen as certified deletion of keys and the dual of certified deletion of ciphertexts. They achieved bounded collusion-resistant secret-key FE with secure key leasing for $\Ppoly$ from standard SKE.

\paragraph{Secure software leasing.}
Ananth and La Place introduced the notion of secure software leasing and achieved it for a sub-class of evasive functions from public-key quantum money (need IO and OWFs) and the LWE assumption~\cite{EC:AnaLaP21}.
Secure software leasing encode classical program into quantum program and has an explicit returning process. After a lessor verifies that a returned quantum program is valid, a lessee cannot run the leased program anymore.
Later, several secure software leasing schemes for a sub-class of evasive functions or cryptographic functionalities (or its variant) with various properties (such as classical communication, without assumptions) were presented~\cite{EPRINT:ColMajPor20,TCC:BJLPS21,TCC:KitNisYam21,C:ALLZZ21}. None of them are certified everlasting secure.

\paragraph{Compute-and-compare obfuscation, PE, and FE.}
There are tremendous amount of previous works on standard FE and PE for general circuits and standard compute-and-compare obfuscation. We focus on strongly related works. No previous work consider certified everlasting secure FE, PE, and compute-and-compare obfuscation.

Gorbunov, Vaikuntanathan, and Wee~\cite{C:GorVaiWee12} constructed bounded collusion-resistant adaptively secure PKFE for $\Ppoly$ from standard PKE (and either the DDH or LWE assumption).
Later, Ananth and Vaikuntanathan improved ciphertext size and assumptions. They constructed adaptively secure bounded collusion-resistant PKFE for $\Ppoly$ with optimal ciphertext size from standard PKE.
Garg, Gentry, Halevi, Raykova, and Sahai~\cite{SIAMCOMP:GGHRSW16} constructed selectively secure collusion-resistant PKFE for $\Ppoly$ from IO and OWFs.
Waters~\cite{C:Waters15} constructed adaptively secure PKFE collusion-resistant for $\Ppoly$ from IO and OWFs.
Ananth, Brakerski, Segev, and Vaikuntanathan~\cite{C:ABSV15} presented a transformation from selectively secure collusion-resistant FE for $\Ppoly$ to adaptively secure collusion-resistant FE for $\Ppoly$.
Jain, Lin, and Sahai constructed IO for $\Ppoly$ from well-founded assumptions~\cite{STOC:JaiLinSah21,EC:JaiLinSah22}. However, their constructions are not post-quantum secure.\footnote{There are a few candidate constructions of post-quantum secure IO~\cite{TCC:BGMZ18,TCC:CHVW19,EC:AgrPel20}.}

Gorbunov, Vaikuntanathan, and Wee~\cite{C:GorVaiWee15} constructed PE for $\Ppoly$ from the LWE assumption.
Goyal, Koppula, and Waters~\cite{FOCS:GoyKopWat17} and Wichs and Zirdelis~\cite{FOCS:WicZir17} presented the notion of compute-and-compare obfuscation (or lockable obfuscation) and achieved it from the LWE assumption. These two works also presented a general transformation from ABE to PE using compute-and-compare obfuscation. Kluczniak~\cite{myPKC:Kluczniak22} constructed compute-and-compare obfuscation from circular \emph{insecure} FHE. However, all known instantiations of circular insecure FHE rely on compute-and-compare obfuscation.

\ifnum\llncs=0
\paragraph{Organization.}
In~\cref{sec:prel}, we define the notation and preliminaries that we require in this work.
In~\cref{sec:fecd_everlasting}, we define the notion of certified everlasting secure collusion-resistant FE and provide a construction.
In~\cref{sec:bounded_FE}, we define the notion of certified everlasting secure bounded collusion-resistant FE and provide constructions.
In~\cref{sec:ccobf}, we define the notion of certified everlasting secure compute-and-compare obfuscation and provide a construction.
In~\cref{sec:pecd}, we define the notion of certified everlasting secure PE and provide a construction.

In~\cref{sec:adaptive_PKFE_public_slot}, we provide a construction of adaptively secure public-slot FE, which is a building block of the construction in~\cref{sec:fecd_everlasting}.
In~\cref{sec:cd_ske_pke}, we define the notion of certified everlasting secure SKE and PKE and provide constructions, which are building blocks of the constructions in~\cref{sec:bounded_FE,sec:RNCE,sec:garbling_scheme}.
In~\cref{sec:RNCE}, we define the notion of certified everlasting secure RNCE and provide a construction, which is a building block of the construction in~\cref{sec:bounded_FE}.
In~\cref{sec:garbling_scheme}, we define the notion of certified everlasting secure garbling schemes and provide a construction, which is a building block of the construction in~\cref{sec:bounded_FE}.
\else
\fi

\ifnum\submission=1
\subsection{Organization}
Although we have more space in this extended abstract version, if we put any self-contained part of the technical sections, it breaks the page limitations. We can arrange (non-self-contained) compressed contents in the main body. However, we need to put some auxiliary contents in a supplemental material in that case.
Hence, we omit all the technical details in this version. Please see the full version~\cite{myEprint:HKMNPY23} for the technical details.
\else
\fi

\else

\subsection{Organization}
If we put any self-contained part of the technical sections, it breaks the page limitations. Hence, we omit all the technical details in this version. Please see the full version~\cite{myEprint:HKMNPY23} for the technical details.
In \cref{sec:tech_overview_CRFE}, we provide a technical overview for constructing certified everlasting secure collusion-resistant FE.
In \cref{sec:tech_overview_bounded_FE}, we provide a technical overview for constructing certified everlasting secure bounded collusion-resistant FE.
In \cref{sec:tech_overview_cc_obfuscation}, we provide a technical overview for constructing compute-and-compare obfuscation.

\section{Technical Overview}\label{sec:technical_overview}

\fi

\ifnum\submission=1
\else

\section{Preliminaries}\label{sec:prel}
\subsection{Notations}\label{sec:notations}
Here we introduce basic notations we will use in this paper.

In this paper, standard math or sans serif font stands for classical algorithms (e.g., $C$ or $\algo{Gen}$) and classical variables (e.g., $x$ or $\keys{pk}$).
Calligraphic font stands for quantum algorithms (e.g., $\qalgo{Gen}$) and calligraphic font and/or the bracket notation for (mixed) quantum states (e.g., $\qstate{q}$ or $\ket{\psi}$).

Let $x\leftarrow X$ denote selecting an element $x$ from a finite set $X$ uniformly at random, and $y\leftarrow A(x)$ denote assigning to $y$ the output of a quantum or probabilistic or deterministic algorithm $A$ on an input $x$.
When we explicitly show that $A$ uses randomness $r$, we write $y\leftarrow A(x;r)$.
When $D$ is a distribution, $x\leftarrow D$ denotes sampling an element $x$ from $D$.
$y\seteq z$ denotes that $y$ is set, defined, or substituted by $z$.
Let $[n]\seteq \{1,\dots,n\}$. Let $\secp$ be a security parameter.
By $[N]_p$ we denote the set of all size-$p$ subsets of $\{1,2\cdots, N\}$.
For classical strings $x$ and $y$, $x\|y$ denotes the concatenation of $x$ and $y$.
For a bit string $s\in\{0,1\}^n$, $s_i$ and $s[i]$ denotes the $i$-th bit of $s$.
QPT stands for quantum polynomial time.
PPT stands for (classical) probabilistic polynomial time.
A function $f: \N \ra \R$ is a negligible function if for any constant $c$, there exists $\secp_0 \in \N$ such that for any $\secp>\secp_0$, $f(\secp) < \secp^{-c}$. We write $f(\secp) \leq \negl(\secp)$ to denote $f(\secp)$ being a negligible function. 


\subsection{Quantum Computations}\label{sec:quantum_computation}
We assume familiarity with the basics of quantum computation and use standard notations. 
Let $\cQ$ be the state space of a single qubit.
$I$ is the two-dimensional identity operator.
$X$ and $Z$ are the Pauli $X$ and $Z$ operators, respectively.
For an operator $A$ acting on a single qubit and a bit string $x\in\bit^n$, we write $A^x$ as $A^{x_1}\otimes A^{x_2}\otimes \cdots A^{x_n}$.
The trace distance between two states $\rho$ and $\sigma$ is given by 
$\frac{1}{2}\norm{\rho-\sigma}_{\tr}$, where $\norm{A}_{\tr}\seteq \tr \sqrt{{\it A}^{\dagger}{\it A}}$ is the trace norm. 
If $\frac{1}{2}\norm{\rho-\sigma}_{\tr}\leq \epsilon$, we say that $\rho$ and $\sigma$ are $\epsilon$-close.
If $\epsilon \le \negl(\lambda)$, then we say that $\rho$ and $\sigma$ are statistically indistinguishable.

\paragraph{Quantum Random Oracle.}
We use the quantum random oracle model (QROM)~\cite{AC:BDFLSZ11} to construct SKE and PKE with certified everlasting deletion in \cref{sec:const_ske_rom,sec:const_pke_rom}, respectively.
In the QROM, a uniformly random function with a certain domain and range is chosen at the beginning, and quantum access to this function is given to all parties including an adversary.
Zhandry showed that quantum access to random functions can be efficiently simulatable by using so-called compressed random oracle technique~\cite{C:Zhandry19}.

We review the one-way to hiding lemma~\cite{JACM:Unruh15,C:AmbHamUnr19}, which is useful when analyzing schemes in the QROM. The following form of the lemma is based on~\cite{C:AmbHamUnr19}.

\begin{lemma}[One-Way to Hiding Lemma \cite{C:AmbHamUnr19}]\label{lem:one-way_to_hiding}
Let $S\subseteq \mathcal{X}$ be a random subset of $\mathcal{X}$. Let $G,H:\mathcal{X}\rightarrow\mathcal{Y}$ be random functions satisfying $\forall x\notin S$ $[G(x)=H(x)]$. Let $z$ be a random classical bit string. 
($S,G,H,z$ may have an arbitrary joint distribution.)
Let $\qA$ be an oracle-aided quantum algorithm that makes at most $q$ quantum queries.
Let $\cB$ be an algorithm that on input $z$ chooses $i\leftarrow[q]$, runs $\qA^{H}(z)$, measures $\qA$'s $i$-th query, and outputs the measurement outcome.
Then we have
$
    \abs{\Pr[\qA^G(z)=1]-\Pr[\qA^H(z)=1]}\leq2q\sqrt{\Pr[\cB^H(z)\in S]}.
$
\end{lemma}

\paragraph{Quantum Teleportation.}
We use quantum teleportation to prove that our construction of the FE scheme in \cref{sec:const_fe_adapt} satisfies adaptive security.
\begin{lemma}[Quantum Teleportation]\label{lem:quantum_teleportation}
Suppose that we have $N$ Bell pairs between registers $A$ and $B$, i.e.,$\frac{1}{\sqrt{2^N}}\sum_{s\in\bit^N}\allowbreak \ket{s}_A \otimes \ket{s}_B$, and let $\rho$ be an arbitrary $N$-qubit quantum state in register $C$. Suppose that we measure $j$-th qubits of $C$ and $A$ in the Bell basis and let $(x_j,z_j)\in\bit\times\bit$ be the measurement outcome for all $j\in[N]$.
Let $x\seteq x_1||x_2||\cdots ||x_N$ and $z\seteq z_1||z_2||\cdots ||z_N$.
Then $(x,z)$ is uniformly distributed  over $\bit^N\times \bit^N$.
Moreover, conditioned on the measurement outcome $(x,z)$, the resulting state in $B$ is $X^xZ^z\rho Z^zX^x$.
\end{lemma}

\paragraph{CSS code.}
We explain basics of CSS codes. CSS codes are used only in the constructions of SKE and PKE with certified everlasting deletion (\cref{sec:const_ske_wo_rom} and \cref{sec:const_pke_wo_rom}),
and therefore readers who are not interested in these constructions can skip this paragraph.
A CSS code with parameters $q,k_1,k_2, t$ consists of two classical linear binary codes. One is a $[q,k_1]$ code $C_1$
\footnote{A $[q,k]$ code is a code consisting of $2^k$ codewords, each of length $q$.
That is, a $k$-dimensional subspace of $\bit^q={\rm GF}(2)^q$.
}
and the other is a $[q,k_2]$ code. 
Both $C_1$ and $C_2^\bot$ can correct up to $t$ errors, 
and they satisfy $C_2\subseteq C_1$.
We require that the parity check matrices of $C_1,C_2$ are computable in polynomial time, and that error correction can be performed in polynomial time. 
Given two binary codes $C\subseteq D$, let $D/C\seteq\{x \mbox{ mod } C:x\in D\}$.
Here, mod $C$ is a linear polynomial-time operation on $\bit^q$ with the following three properties.
First, $x$ mod $C=x'$ mod $C$ if and only if $x-x'\in C$ for any $x,x'\in\bit^q$.
Second, for any binary code $D$ such that $C\subseteq D$,
$x$ mod $C\in D$ for any $x\in D$.
Third,  ($x$ mod $C$) mod $C$= $x$ mod $C$ for any $x\in\bit^q$.

\subsection{Cryptographic Tools}\label{sec:crypt_tool}
In this section, we review the cryptographic tools used in this paper.

\begin{lemma}[Difference Lemma~\cite{EPRINT:Shoup04}]\label{lem:defference}
Let $A,B,F$ be events defined in some probability distribution, and suppose 
$\Pr[A\wedge \overline{F}]=\Pr[B\wedge \overline{F}]$.
Then $\abs{\Pr[A]-\Pr[B]}\leq \Pr[F]$.
\end{lemma}

\paragraph{Pseudorandom generators.}
\begin{definition}[Pseudorandom Generator]\label{def:prg}
A pseudorandom generator (PRG) $\PRG: \zo{\secp} \ra \zo{\secp + \ell(\secp)}$ with stretch $\ell(\secp)$ ($\ell$ is some polynomial function) is a polynomial-time computable function that satisfies the following.
For any QPT adversary $\qA$, it holds that
\[\abs{\Pr[\qA(\PRG(s))=1 \mid s \chosen \cU_\secp]- \Pr[\qA(r)\mid r \chosen \cU_{\secp+\ell(\secp)}]}\le \negl(\secp),\]
where $\cU_{m}$ denotes the uniform distribution over $\zo{m}$.
\end{definition}

\begin{theorem}[\cite{SIAMCOMP:HILL99}]\label{thm:owf_prg}
If there exists a OWF, there exists a PRG.
\end{theorem}

\paragraph{Pseudorandom Functions.}

\begin{definition}[Pseudorandom Function]\label{def:prf}
Let $\{\algo{F}_{K}: \bin^{\ell_1} \ra \allowbreak \bin^{\ell_2} \mid K \in \bin^\secp\}$ be a family of polynomially computable functions, where $\ell_1$ and $\ell_2$ are some polynomials of $\secp$.
We say that $\prf$ is a pseudorandom function (PRF) family if, for any QPT distinguisher $\qA$, there exists $\negl(\cdot)$ such that it holds that
\begin{align}
\abs{\Pr[\qA^{\algo{F}_{K}(\cdot)}(1^\secp) = 1 \mid K \chosen \bin^{\secp}]
-\Pr[\qA^{\algo{R}(\cdot)}(1^\secp) = 1 \mid \algo{R} \chosen \cU]
} \le\negl(\secp),
\end{align}
where $\cU$ is the set of all functions from $\bin^{\ell_1}$ to $\bin^{\ell_2}$.
\end{definition}

\begin{theorem}[\cite{JACM:GolGolMic86}]\label{thm:prf-owf} If one-way functions exist, then for all efficiently computable functions $n(\lambda)$ and $m(\lambda)$, there exists a PRF that maps $n(\lambda)$ bits to $m(\lambda)$ bits.
\end{theorem}

\paragraph{Secret Key Encryption (SKE).}
\begin{definition}[Secret Key Encryption (Syntax)]\label{def:ske}
Let $\lambda$ be a security parameter and let $p$, $q$, $r$ and $s$ be some polynomials.
A secret key encryption scheme is a tuple of algorithms
$\Sigma=(\keygen,\Enc,\Dec)$ with plaintext space $\Ms\seteq\{0,1\}^n$, ciphertext space $\Cs\seteq \bit^{p(\lambda)}$,
and secret key space $\mathcal{SK}\seteq\{0,1\}^{q(\lambda)}$.
\begin{description}
    \item[$\keygen (1^\secp) \ra \sk$:]
    The key generation algorithm takes the security parameter $1^\secp$ as input and outputs a secret key $\sk \in \mathcal{SK}$.
    \item[$\Enc(\sk,m) \ra \ct$:]
    The encryption algorithm takes $\sk$ and a plaintext $m\in\Ms$ as input, and outputs a ciphertext $\ct\in \Cs$.
    \item[$\Dec(\sk,\ct) \ra m^\prime~or~\bot$:]
    The decryption algorithm takes $\sk$ and $\ct$ as input, and outputs a plaintext $m^\prime \in \Ms$ or $\bot$.
\end{description}
\end{definition}

We require that a SKE scheme satisfies correctness defined below.
\begin{definition}[Correctness for SKE]\label{def:correctness_ske}
There are two types of correctness, namely,
decryption correctness and special correctness.
\paragraph{Decryption Correctness:}
There exists a negligible function $\negl$ such that for any $\secp\in \N$ and $m\in\Ms$, 
\begin{align}
\Pr\left[
\Dec(\sk,\ct)\neq m
\ \middle |
\begin{array}{ll}
\sk\lrun \keygen(1^\secp)\\
\ct \lrun \Enc(\sk,m)
\end{array}
\right] 
\leq\negl(\secp).
\end{align}

\paragraph{Special Correctness:}
There exists a negligible function $\negl$ such that for any $\secp\in \N$ and $m\in\Ms$, 
\begin{align}
\Pr\left[
\Dec(\sk_2,\ct)\neq \bot
\ \middle |
\begin{array}{ll}
\sk_2,\sk_1\lrun \keygen(1^\secp)\\
\ct \lrun \Enc(\sk_1,m)
\end{array}
\right] 
\leq\negl(\secp).
\end{align}
\end{definition}
\begin{remark}
In the original definition of SKE schemes, only decryption correctness is required.
In this paper, however, we additionally require special correctness as Lindell and Pinkas~\cite{JC:LinPin09}.
This is because we need special correctness for the construction of garbling in \cref{sec:const_garbling}.
In fact, special correctness can be easily satisfied as well as shown by Lindell and Pinkas~\cite{JC:LinPin09}.
\end{remark}

As security of SKE schemes, we consider OW-CPA security or IND-CPA security defined below.
\begin{definition}[OW-CPA Security for SKE]\label{def:OW-CPA_security_ske}
Let $\ell$ be a polynomial of the security parameter $\secp$.
Let $\Sigma=(\keygen,\Enc,\Dec)$ be a SKE scheme.
We consider the following security experiment $\expb{\Sigma,\qA}{ow}{cpa}(\secp)$ against a QPT adversary $\qA$.
\begin{enumerate}
    \item The challenger computes $\sk \la \keygen(1^\secp)$.
    \item $\qA$ sends an encryption query $m$ to the challenger.
    The challenger computes $\ct\la\Enc(\sk,m)$ and returns $\ct$ to $\qA$.
    $\qA$ can repeat this process polynomially many times.
    \item The challenger samples $(m^1,\cdots, m^\ell)\la \Ms^\ell$, computes $\ct^i \la \Enc(\sk,m^i)$ for all $i\in[\ell]$ and sends $\{\ct^i\}_{i\in[\ell]}$ to $\qA$.
    \item $\qA$ sends an encryption query $m$ to the challenger.
    The challenger computes $\ct\la\Enc(\sk,m)$ and returns $\ct$ to $\qA$.
    $\qA$ can repeat this process polynomially many times.
    \item $\qA$ outputs $m'$.
    \item The output of the experiment is $1$ if $m'=m^i$ for some $i\in[\ell]$. Otherwise, the output of the experiment is $0$.
\end{enumerate}
We say that the $\Sigma$ is OW-CPA secure if, for any QPT $\qA$, it holds that
\begin{align}
\advb{\Sigma,\qA}{ow}{cpa}(\secp)
\seteq \Pr[ \expb{\Sigma,\qA}{ow}{cpa}(\secp)=1]\leq \negl(\secp).
\end{align}
Note that we assume $1/|\Ms|$ is negligible.
\end{definition}

\begin{definition}[IND-CPA Security for SKE]\label{def:IND-CPA_security_ske}
Let $\Sigma=(\keygen,\Enc,\Dec)$ be a SKE scheme.
We consider the following security experiment $\expb{\Sigma,\qA}{ind}{cpa}(\secp,b)$ against a QPT adversary $\qA$.
\begin{enumerate}
    \item The challenger computes $\sk \la \keygen(1^\secp)$.
    \item $\qA$ sends an encryption query $m$ to the challenger.
    The challenger computes $\ct\la\Enc(\sk,m)$ and returns $\ct$ to $\qA$.
    $\qA$ can repeat this process polynomially many times.
    \item $\qA$ sends $(m_0,m_1)\in\cM^2$ to the challenger.
    \item The challenger computes $\ct \la \Enc(\sk,m_b)$ and sends $\ct$ to $\qA$.
    \item $\qA$ sends an encryption query $m$ to the challenger.
    The challenger computes $\ct\la\Enc(\sk,m)$ and returns $\ct$ to $\qA$.
    $\qA$ can repeat this process polynomially many times.
    \item $\qA$ outputs $b'\in \bit$. This is the output of the experiment.
\end{enumerate}
We say that $\Sigma$ is IND-CPA secure if, for any QPT $\qA$, it holds that
\begin{align}
\advb{\Sigma,\qA}{ind}{cpa}(\secp)
\seteq \abs{\Pr[ \expb{\Sigma,\qA}{ind}{cpa}(\secp, 0)=1] - \Pr[ \expb{\Sigma,\qA}{ind}{cpa}(\secp, 1)=1] }\leq \negl(\secp).
\end{align}
\end{definition}
It is well-known that IND-CPA security implies OW-CPA security. 
A SKE scheme exists if there exists a pseudorandom function.

\begin{definition}[Ciphertext Pseudorandomness for SKE]\label{def:ske_pseudorandomct}
Let $\Sigma=(\keygen,\Enc,\Dec)$ be a SKE scheme whose ciphertext space is $\zo{\ell}$.
We consider the following security experiment $\expb{\Sigma,\qA}{ct}{pr}(\secp,b)$ against a QPT adversary $\qA$.
\begin{enumerate}
    \item The challenger computes $\sk \la \keygen(1^\secp)$.
    \item $\qA$ sends an encryption query $m_i$ to the challenger.
    If $b=0$, the challenger computes $\ct_i\la\Enc(\sk,m_i)$ and returns $\ct_i$ to $\qA$.
    If $b=1$, the challenger chooses $\ct_i \chosen \zo{\ell}$ and returns $\ct_i$ to $\qA$.
    $\qA$ can repeat this process polynomially many times.
    \item $\qA$ outputs $b'\in \bit$. This is the output of the experiment.
\end{enumerate}
We say that $\Sigma$ is ciphertext pseudorandom if, for any QPT $\qA$, it holds that
\begin{align}
\advb{\Sigma,\qA}{ct}{pr}(\secp)
\seteq \abs{\Pr[ \expb{\Sigma,\qA}{ct}{cr}(\secp, 0)=1] - \Pr[ \expb{\Sigma,\qA}{ct}{pr}(\secp, 1)=1] }\leq \negl(\secp).
\end{align}

\end{definition}

\begin{theorem}\label{thm:pseudorandom_ske}
If OWFs exist, there exists an SKE scheme that is ciphertext pseudorandom.
\end{theorem}

\paragraph{Public Key Encryption (PKE).}
\begin{definition}[Public Key Encryption (Syntax)]\label{def:pke}
Let $\lambda$ be a security parameter and let $p$, $q$ and $r$ be some polynomials.
A public key encryption scheme is a tuple of algorithms $\Sigma=(\keygen,\Enc,\Dec)$ with 
plaintext space $\Ms\seteq\{0,1\}^n$, ciphertext space $\Cs\seteq \bit^{p(\lambda)}$, public key space $\mathcal{PK}\seteq\{0,1\}^{q(\lambda)}$ 
and secret key space $\mathcal{SK}\seteq \bit^{r(\lambda)}$.
\begin{description}
    \item[$\keygen(1^\secp)\ra (\pk,\sk)$:] The key generation algorithm takes as input the security parameter $1^\secp$ and outputs a public key $\pk\in\mathcal{PK}$ and a secret key $\sk\in\mathcal{SK}$.
    \item[$\Enc(\pk,m) \ra \ct$:] The encryption algorithm takes as input $\pk$ and a plaintext $m \in \Ms$, and outputs a ciphertext $\ct\in\Cs$.
    \item[$\Dec(\sk,\ct) \ra m^\prime \mbox{ or } \bot$:] The decryption algorithm takes as input $\sk$ and $\ct$, and outputs a plaintext $m^\prime$ or $\bot$.
\end{description}
\end{definition}

We require that a PKE scheme satisfies decryption correctness defined below.
\begin{definition}[Decryption Correctness for PKE]\label{def:correctness_pke}
There exists a negligible function $\negl$ such that for any $\secp\in \N$, $m\in\Ms$,
\begin{align}
\Pr\left[
\Dec(\sk,\ct)\ne m
\ \middle |
\begin{array}{ll}
(\pk,\sk)\lrun \keygen(1^\secp)\\
\ct \lrun \Enc(\pk,m)
\end{array}
\right] 
\le\negl(\secp).
\end{align}
\end{definition}

As security, we consider OW-CPA security or IND-CPA security defined below.
\begin{definition}[OW-CPA Security for PKE]\label{def:OW-CPA_security_pke}
Let $\ell$ be a polynomial of the security parameter $\secp$.
Let $\Sigma=(\keygen,\Enc,\Dec)$ be a PKE scheme.
We consider the following security experiment $\expb{\Sigma,\qA}{ow}{cpa}(\secp)$ against a QPT adversary $\qA$.
\begin{enumerate}
    \item The challenger computes $(\pk,\sk) \la \keygen(1^\secp)$.
    \item The challenger samples $(m^1,\cdots, m^\ell)\la \Ms^\ell$, computes $\ct^i \la \Enc(\pk,m^i)$ for all $i\in[\ell]$ and sends $\{\ct^i\}_{i\in[\ell]}$ to $\qA$.
    \item $\qA$ outputs $m'$.
    \item The output of the experiment is $1$ if $m'=m^i$ for some $i\in[\ell]$. Otherwise, the output of the experiment is $0$.
\end{enumerate}
We say that $\Sigma$ is OW-CPA secure if, for any QPT $\qA$, it holds that
\begin{align}
\advb{\Sigma,\qA}{ow}{cpa}(\secp)
\seteq \Pr[ \expb{\Sigma,\qA}{ow}{cpa}(\secp)=1]\leq \negl(\secp).
\end{align}
Note that we assume $1/|\Ms|$ is negligible.
\end{definition}

\begin{definition}[IND-CPA Security for PKE]\label{def:IND-CPA_pke}
Let $\Sigma=(\keygen,\Enc,\Dec)$ be a PKE scheme.
We consider the following security experiment $\expb{\Sigma,\qA}{ind}{cpa}(\secp,b)$ against a QPT adversary $\qA$.

\begin{enumerate}
    \item The challenger generates $(\pk,\sk)\lrun \keygen(1^{\secp})$, and sends $\pk$ to $\qA$.
    \item $\qA$ sends $(m_0,m_1)\in\cM^2$ to the challenger.
    \item The challenger computes $\ct \lrun \Enc(\pk,m_b)$, and sends $\ct$ to $\qA$.
    \item $\qA$ outputs $b'\in\bit$. This is the output of the experiment.
\end{enumerate}
We say that $\Sigma$ is IND-CPA secure if, for any QPT $\qA$, it holds that
\begin{align}
\advb{\Sigma,\qA}{ind}{cpa}(\secp) \seteq \abs{\Pr[\expb{\Sigma,\qA}{ind}{cpa}(\secp,0)=1]  - \Pr[\expb{\Sigma,\qA}{ind}{cpa}(\secp,1)=1]} \leq \negl(\secp).
\end{align}
\end{definition}

It is well known that IND-CPA security implies OW-CPA security.
There are many IND-CPA secure PKE schemes against QPT adversaries under standard cryptographic assumptions.
A famous one is Regev PKE scheme, which is IND-CPA secure if the learning with errors (LWE) assumption holds against QPT adversaries~\cite{JACM:Regev09}.
See ~\cite{JACM:Regev09,STOC:GenPeiVai08} for the LWE assumption and constructions of post-quantum PKE.

\paragraph{Encryption with Certified Deletion.}
Broadbent and Islam~\cite{TCC:BroIsl20} introduced the notion of encryption with certified deletion.

\begin{definition}[One-Time SKE with Certified Deletion (Syntax)~\cite{TCC:BroIsl20,AC:HMNY21}]\label{def:sk_cert_del}
Let $\lambda$ be a security parameter and let $p$, $q$ and $r$ be some polynomials.
A one-time secret key encryption scheme with certified deletion is a tuple of algorithms $\Sigma=(\keygen,\qEnc,\qDec,\qDelete,\Vrfy)$ with 
plaintext space $\Ms\seteq \{0,1\}^n$, ciphertext space $\Cs\seteq \cQ^{\otimes p(\lambda)}$, key space $\Ks\seteq\{0,1\}^{q(\lambda)}$ and deletion certificate space $\mathcal{D}\seteq \{0,1\}^{r(\lambda)}$.
\begin{description}
    \item[$\keygen (1^\secp) \ra \sk$:] The key generation algorithm takes as input the security parameter $1^\secp$, and outputs a secret key $\sk \in \Ks$.
    \item[$\qEnc(\sk,m) \ra \qct$:] The encryption algorithm takes as input $\sk$ and a plaintext $m\in\Ms$, and outputs a ciphertext $\qct\in \Cs$.
    \item[$\qDec(\sk,\qct) \ra m^\prime~or~\bot$:] The decryption algorithm takes as input $\sk$ and $\qct$, and outputs a plaintext $m^\prime \in \Ms$ or $\bot$.
    \item[$\qDelete(\qct) \ra \cert$:] The deletion algorithm takes as input $\qct$, and outputs a certification $\cert\in\mathcal{D}$.
    \item[$\Vrfy(\sk,\cert)\ra \top$ or $\bot$:] The verification algorithm takes $\sk$ and $\cert$ as input, and outputs $\top$ or $\bot$.
\end{description}
\end{definition}

We require that a one-time SKE scheme with certified deletion satisfies correctness defined below.
\begin{definition}[Correctness for One-Time SKE with Certified Deletion]\label{def:correctness_sk_cert_del}
There are two types of correctness, namely,
decryption correctness and verification correctness.
\paragraph{Decryption Correctness:} There exists a negligible function $\negl$ such that for any $\secp\in \N$ and $m\in\Ms$, 
\begin{align}
\Pr\left[
m'\neq m
\ \middle |
\begin{array}{ll}
\sk\lrun \keygen(1^\secp)\\
\qct \lrun \qEnc(\sk,m)\\
 m'\la\qDec(\sk,\qct)
\end{array}
\right] 
\leq\negl(\secp).
\end{align}

\paragraph{Verification Correctness:} There exists a negligible function $\negl$ such that for any $\secp\in \N$ and $m\in\Ms$, 
\begin{align}
\Pr\left[
\Vrfy(\sk,\cert)=\bot
\ \middle |
\begin{array}{ll}
\sk\lrun \keygen(1^\secp)\\
\qct \lrun \qEnc(\sk,m)\\
\cert \lrun \qDelete(\qct)
\end{array}
\right] 
\leq
\negl(\secp).
\end{align}

\end{definition}

We additionally require verification correctness with QOTP in this work.
This is because we need it for the construction of FE in \cref{sec:const_fe_adapt}.
This notion means that even if we encrypt a ciphertext with quantum one-time pad (QOTP), we can run the original deletion algorithm $\qDelete$ and recover a valid certificate by using the QOTP key.
In fact, the construction of ~\cite{TCC:BroIsl20} satisfies verification correctness with QOTP as well.

\begin{definition}[Verification Correctness with QOTP]\label{def:ske_ver_correctness_QOTP}
There exists a negligible function $\negl$ and a PPT algorithm $\Modify$ such that for any $\secp\in \N$ and $m\in\Ms$, 
\begin{align}
\Pr\left[
\Vrfy(\sk,\cert^*)=\bot
\ \middle |
\begin{array}{ll}
\sk\lrun \keygen(1^\secp)\\
\qct \lrun \qEnc(\sk,m)\\
a,b\la\bit^{p(\lambda)}\\
\wtl{\cert} \lrun \qDelete(Z^bX^a\qct X^aZ^b)\\
\cert^*\lrun \Modify(a,b,\wtl{\cert})  
\end{array}
\right] 
\leq
\negl(\secp).
\end{align}
\end{definition}

We require that a one-time SKE with certified deletion satisfies certified deletion security defined below.
\begin{definition}[Certified Deletion Security for One-Time SKE with Certified Deletion]\label{def:security_sk_cert_del}
Let $\Sigma=(\keygen,\qEnc,\qDec,\allowbreak \qDelete,\Vrfy)$ be a one-time SKE scheme with certified deletion.
We consider the following security experiment $\expc{\Sigma,\qA}{otsk}{cert}{del}(\secp,b)$ against an unbounded adversary $\qA$.

\begin{enumerate}
    \item The challenger computes $\sk \la \keygen(1^\secp)$.
    \item $\qA$ sends $(m_0,m_1)\in\cM^2$ to the challenger.
    \item The challenger computes $\qct \la \qEnc(\sk,m_b)$ and sends $\qct$ to $\qA$.
    \item $\qA$ sends $\cert$ to the challenger.
    \item The challenger computes $\Vrfy(\sk,\cert)$. If the output is $\bot$, the challenger sends $\bot$ to $\qA$.
    If the output is $\top$, the challenger sends $\sk$ to $\qA$. 
    \item $\qA$ outputs $b'\in \bit$. This is the output of the experiment.
\end{enumerate}
We say that $\Sigma$ is OT-CD secure if, for any unbounded $\qA$, it holds that
\begin{align}
\advc{\Sigma,\qA}{otsk}{cert}{del}(\secp)
\seteq \abs{\Pr[ \expc{\Sigma,\qA}{otsk}{cert}{del}(\secp, 0)=1] - \Pr[ \expc{\Sigma,\qA}{otsk}{cert}{del}(\secp, 1)=1] }\leq \negl(\secp).
\end{align}
\end{definition}
\begin{theorem}[\cite{TCC:BroIsl20}]
There exists one-time SKE with certified deletion that satisfies~\cref{def:sk_cert_del,def:correctness_sk_cert_del,def:security_sk_cert_del,def:ske_ver_correctness_QOTP} exists unconditionally.
\end{theorem}

\paragraph{Attribute-Based Encryption.}

\begin{definition}[KP-ABE (Syntax)]\label{def:kp-abe} A key-policy ABE (KP-ABE) scheme is a tuple of PPT algorithms $(\setup, \keygen,\allowbreak \enc, \dec)$ with a class of predicates $\mcl{P}_n$ (represented as class of boolean circuits with $n$ input bits), and a message space $\mcl{M}$.
\begin{description}
\item[$\setup(1^{\secp})$:] The setup algorithm takes as input a security parameter $1^{\secp}$ and outputs a public key $\pk$ and a master secret key $\msk$.

\item[$\keygen(\msk, P)$:] The key generation algorithm takes as input the master secret key $\msk$ and a predicate $P \in \mcl{P}_n$, and outputs a secret key $\sk_P$ corresponding to the predicate $P$.

\item[$\enc(\pk, a, m)$:] The encryption algorithm takes as input a public key $\pk$, an attribute $a \in \{0, 1\}^n$ and a message $m\in \mcl{M}$, and outputs a ciphertext $\ct$.

\item[$\dec(\sk_P, \ct)$:] The decryption algorithm takes as input a secret key $\sk_P$ and a ciphertext $\ct$, and outputs a classical message $m'$ or $\bot$.
\end{description}
\end{definition}

\begin{definition}[Correctness of KP-ABE] The correctness of KP-ABE for a class of predicates $\mcl{P}_n$ and a message space $\mcl{M}$ is defined as follows.
\begin{description}
	
\item[Decryption correctness:] For any $\secp \in \N, P \in \calP_n,a\in \zo{n}, m  \in \Ms$ such that $P(a)=1$,
		\begin{align}
		\mathrm{Pr}\left[ \dec(\sk_P, \qct) \ne m \ \middle | \begin{array}{l} (\pk, \msk) \leftarrow \setup(1^{\secp}) \\
			\sk_P \leftarrow \keygen(\msk, P) \\
			\ct \leftarrow \enc(\pk,a, m)  \end{array} \right] \le\negl(\secp).
			\end{align}
	\end{description}
\end{definition}

\begin{definition}[IND-CPA Security of KP-ABE] Let $\Sigma = (\setup, \keygen, \enc, \dec)$ be a KP-ABE scheme for a class predicates $\mcl{P}_n$ and message space $\mcl{M}$. We consider the following experiment $\expb{\Sigma,\qA}{ada}{ind}(\secp, b)$.
	\begin{enumerate}
		\item The challenger computes $(\pk, \msk) \leftarrow \setup(1^{\secp})$ and sends $\pk$ to $\qA$.
		\item $\qA$ sends a predicate $P_i \in \mcl{P}_n$, called key query, to the challenger and it returns $\sk_{P_i} \leftarrow \keygen(\msk, P)$. $\qA$ can send unbounded polynomially many key queries. Let $q$ be the total number of key queries.
		\item $\qA$ sends a challenge message pair $(m_0, m_1)$ and an attribute $a \in \{0, 1\}^n$ such that $\abs{m_0} = \abs{m_1}$.
		\item The challenger computes $\ct_b \leftarrow \enc(\pk, a, m_b)$. It sends $\ct_b$ to $\qA$.
		\item $\qA$ can send key queries again.
		\item $\qA$ outputs its guess $b' \in \zo{}$. If $P_i(a)=0$ for all $i\in [q]$, the experiments outputs $b^\prime$.
	\end{enumerate}
	We say that the $\Sigma$ is adaptively secure if for any QPT adversary $\qA$, it holds that 
	\[\advb{\Sigma,\qA}{ada}{ind}(\secp) \seteq \left|\Pr[\expb{\Sigma,\qA}{ada}{ind}(\secp, 0) = 1] - \Pr[\expb{\Sigma,\qA}{ada}{ind}(\secp, 1) = 1]\right|\leq \negl(\secp).\]
	We can define similar experiment $\expb{\Sigma,\qA}{sel}{ind}(\secp, b)$ where $\qA$ is restricted to submit the challenge attribute $a \in \{0, 1\}^n$ before it receives $\pk$ from the challenger. We say that the $\Sigma$ is selectively secure if for any QPT adversary $\qA$, it holds that 
		\[\advb{\Sigma,\qA}{sel}{ind}(\secp) \seteq \left|\Pr[\expb{\Sigma,\qA}{sel}{ind}(\secp, 0) = 1] - \Pr[\expb{\Sigma,\qA}{sel}{ind}(\secp, 1) = 1]\right|\leq \negl(\secp).\]
\end{definition}

\begin{theorem}[\cite{JACM:GorVaiWee15,EC:BGGHNS14}]\label{thm:ABE_from_LWE}
If the LWE assumption holds, there exists selectively secure KP-ABE for all boolean circuits. In addition, if the LWE assumption holds against sub-exponential time algorithms, there exists adaptively secure KP-ABE for all boolean circuits.
\end{theorem}

\paragraph{Classical Fully Homomorphic Encryption.}

\begin{definition}[Leveled Fully Homomorphic Encryption]\label{def:fhe} A leveled FHE is a tuple of PPT algorithms $(\keygen, \enc, \allowbreak \Eval, \dec)$ with a class of circuits $\cC=\setbk{\cC_d}_{d\in\bbN}$, where $\cC_d$ contains all Boolean circuits of depth up to $d$.
\begin{description}
\item[$\keygen(1^{\secp},1^d)$:] The key generation algorithm takes as input the security parameters $1^{\secp}$ and $1^d$ and outputs a public key $\pk$ and a secret key $\sk$.

\item[$\enc(\pk, x)$:]	The encryption algorithm takes as input a public key $\pk$ and a message $x \in \{0, 1\}$, and outputs a ciphertext $\ct$.

\item[$\Eval(\pk, C, \ct_1, \dots, \ct_n)$:] The evaluation algorithm takes as input a public key $\pk$, a circuit $C \in \cC$, ciphertexts $\ct_1, \dots, \ct_n$ where $n$ denotes the input length of the circuit $C$, and outputs a ciphertext $\ct_C$.

\item[$\dec(\sk, \ct)$:] The decryption algorithm takes as input a secret key $\sk$ and a ciphertext $\ct$, and outputs a message $x'$ or $\bot$.
\end{description}
\end{definition}	

\begin{definition}[Compactness]
 A classical FHE is compact if there exists a fixed polynomial bound $B(\cdot)$ so that, for all $\secp \in \bbN$, any $C\in\cC$, and plaintext $x \in \zo{n}$, it holds that
    \begin{align}
        \Pr\left[ \abs{\ct_C} \le B(\secp) \ \middle | \begin{array}{l} (\pk, \sk) \leftarrow \keygen(1^{\secp}) \\
            \ct_i \leftarrow \enc(\pk, x_i)\\
            \ct_C \leftarrow \Eval(\pk, C, \ct_1, \dots, \ct_{n}) 
        \end{array} \right] = 1.
    \end{align}
\end{definition}

\begin{definition}[Correctness] An FHE scheme is said to be correct for $\cC$ if for any $\secp, n \in \mbb{N}, C \in \cC, \bm{x} = (x_1, \dots, x_n) \in \{0, 1\}^n$,
	\begin{align}
		\Pr\left[ \dec(\sk, \ct_C) \ne C(\bm{x})\ \middle | \begin{array}{l} (\pk, \sk) \leftarrow \keygen(1^{\secp}) \\
			\ct_i \leftarrow \enc(\pk, x_i)\\
			\ct_C \leftarrow \Eval(\pk, C, \ct_1, \dots, \ct_n) 
		\end{array} \right] \le \negl(\secp).
	\end{align}
\end{definition}

\begin{definition}[Security of FHE] An FHE scheme $\Sigma = (\setup, \keygen, \enc, \dec)$ with a class of circuits $\cC$ is said to be IND-CPA secure if for any QPT adversary $\qA$, any $\secp, n \in \mbb{N}$, the following holds:
\begin{align}
\Pr\left[
 \qA(1^{\secp}, \pk, \ct)=1
\ \middle |
\begin{array}{ll}
(\pk, \sk) \leftarrow \keygen(1^{\secp}),\\
\ct \leftarrow \enc(\pk, 0)
\end{array}
\right] -
\Pr\left[
 \qA(1^{\secp}, \pk, \ct)=1
\ \middle |
\begin{array}{ll}
(\pk, \sk) \leftarrow \keygen(1^{\secp}),\\
\ct \leftarrow \enc(\pk, 1)
\end{array}
\right] 
=\negl(\secp).
\end{align}

\end{definition}

\begin{theorem}[\cite{SIAMCOMP:BraVai14,C:GenSahWat13}]
If the LWE assumption holds, there exists leveled FHE.
\end{theorem}

\paragraph{Indistinguishability Obfuscation}

\begin{definition}[Indistinguishability Obfuscator~\cite{JACM:BGIRSVY12}]\label{def:io}
A PPT algorithm $\iO$ is a secure IO for a classical circuit class $\{\cC_\secp\}_{\secp \in \bbN}$ if it satisfies the following two conditions.

\begin{description}
\item[Functionality:] For any security parameter $\secp \in \bbN$, circuit $C \in \cC_\secp$, and input $x$, we have that
\begin{align}
\Pr[C^\prime(x)=C(x) \mid C^\prime \lrun \iO(C)] = 1\enspace.
\end{align}

\item[Indistinguishability:] For any PPT $\Sampler$ and QPT distinguisher $\qD$, the following holds:

If $\Pr[\forall x\ C_0(x)=C_1(x) \land \abs{C_0}=\abs{C_1} \mid (C_0,C_1,\aux)\gets \Sampler(1^\secp)]> 1 - \negl(\secp)$, then we have
 \begin{align}
\adva{\iO,\qD}{io}(\secp) &\seteq \left|
 \Pr\left[\qD(\iO(C_0),\aux) = 1 \mid (C_0,C_1,\aux)\gets \Sampler(1^\secp) \right] \right.\\
&~~~~~~~ \left. - \Pr\left[\qD(\iO(C_1),\aux)= 1\mid (C_0,C_1,\aux)\gets \Sampler(1^\secp)  \right] \right|
  \leq \negl(\secp).
 \end{align}
\end{description}
\end{definition}
There are a few candidates of secure IO for polynomial-size classical circuits against quantum adversaries~\cite{TCC:BGMZ18,TCC:CHVW19,EC:AgrPel20}.

\paragraph{Obfuscation for compute-and-compare programs.}

\begin{definition}[Compute-and-Compare Circuits]
A compute-and-compare circuit $\cnc{P}{\lock,m}$ is of the form
\[
\cnc{P}{\lock,m}(x) = \left\{
\begin{array}{ll}
m&(P(x)=\lock)\\
0&(\text{otherwise})~,
\end{array}
\right.
\]
where $P$ is a circuit, $\lock$ is a string called lock value, and $m$ is a message.
\end{definition}

We assume that a program $P$ has an associated set of parameters $\pp_P$ (input size, output size, circuit size) which we do not need to hide.
\begin{definition}[Compute-and-Compare Obfuscation]\label{def:CCObf}
A PPT algorithm $\CCObf$ is a secure obfuscator for the family of distributions $D=\{D_\param\}_\param$ if the following holds:
\begin{description}
\item[Functionality Preserving:] There exists a negligible function $\negl$ such that for all program $P$, all lock value $\lock$, and all message $m$, it holds that
\begin{align}
\Pr[\forall x, \tlP(x)=\cnc{P}{\lock,m}(x) \mid \tlP\la\CCObf(1^\secp,P,\lock,m)] \ge 1-\negl(\secp).
\end{align}
\item[Distributional Indistinguishability:] There exists an efficient simulator $\Sim$ such that for all $\qD$, $\param$ and message $m$, we have
\begin{align}
\abs{\Pr[\qD(\CCObf(1^\secp,P,\lock,m),\qaux)=1] -\Pr[\qD(\Sim(1^\secp,\pp_P,1^\abs{m}),\qaux)=1]} \le \negl(\secp),
\end{align}
where $(P,\lock,\qaux)\la D_\param$.
\end{description}
\end{definition}
\begin{theorem}[\cite{FOCS:GoyKopWat17,FOCS:WicZir17}]
If the LWE assumption holds, there exists compute-and-compare obfuscation for all families of distributions $D=\{D_\param\}$, where each $D_\param$ outputs uniformly random lock value $\lock$ independent of $P$ and $\qaux$.
\end{theorem}

\fi
\ifnum\submission=1
\else

\newcommand{\pke}{\mathsf{pke}}
\newcommand{\coset}{\mathsf{coset}}
\renewcommand{\IP}[2]{\langle #1,#2\rangle}
\renewcommand{\sig}{\mathsf{sig}}
\newcommand{\shO}{\mathsf{shO}}
\newcommand{\CDSKE}{\mathsf{CDSKE}}
\newcommand{\CDABE}{\mathsf{CDABE}}
\newcommand{\cdabe}{\mathsf{cdabe}}
\newcommand{\cdske}{\mathsf{cdske}}
\newcommand{\tbectlen}{\ell}
\newcommand{\CDTBE}{\mathsf{CDTBE}}
\newcommand{\cdtbe}{\mathsf{cdtbe}}
\newcommand{\qAux}{\qalgo{Aux}}
\newcommand{\qInv}{\qalgo{Inv}}
\newcommand{\qSim}{\qalgo{Sim}}
\newcommand{\CDFE}{\mathsf{CDFE}}
\newcommand{\msglen}{n}
\newcommand{\qC}{\qalgo{C}}
\newcommand{\qSamp}{\qalgo{Samp}}
\renewcommand{\tag}{\mathsf{tag}}
\newcommand{\mat}[1]{\boldsymbol{#1}}
\newcommand{\mapaux}{\Phi_{\mathsf{aux}}}
\newcommand{\mapct}{\Phi_{\mathsf{ct}}}
\newcommand{\qinit}{\qalgo{init}}
\newcommand{\ctlen}{\ell_{\mathsf{ct}}}
\newcommand{\tbe}{\mathsf{tbe}}
\newcommand{\Select}{\mathsf{Sel}}

\section{Collusion-Resistant Functional Encryption with Certified Everlasting Deletion}\label{sec:fecd_everlasting}

In this section, we present the definitions of FE with certified everlasting deletion and a collusion-resistant construction.
\subsection{Definitions}\label{sec:def_FE_CED}
First, we introduce the syntax and security definitions of FE with certified everlasting deletion.
\begin{definition}[Functional Encryption with Certified Everlasting Deletion]
    A functional encryption with certified everlasting deletion for a class $\mathcal{F}$ of functions is a tuple of QPT algorithms $(\setup, \keygen, \qenc, \qdec, \qdel, \vrfy)$ with plaintext space $\Ms$, ciphertext space $\Cs$, master public key space $\mathcal{MPK}$, master secret key space $\mathcal{MSK}$, and secret key space $\mathcal{SK}$,
that works as follows.
\begin{description}
    \item[$\Setup(1^\secp)\ra(\MPK,\MSK)$:] The setup algorithm takes the security parameter as input, and outputs a master public key $\MPK\in\mathcal{MPK}$ and a master secret key $\MSK\in\mathcal{MSK}$.
    \item[$\keygen(\MSK,f)$:]The key generation algorithm takes $\MSK$ and $f\in\mathcal{F}$ as input, and outputs a secret key $\sk_f\in\mathcal{SK}$.
    \item[$\qenc(\MPK,m)\ra(\qct, \vk)$:]The encryption algorithm takes $\MPK$ and $m\in\Ms$ as input, and outputs a ciphertext $\qct\in\mathcal{C}$ and a verification key $\vk$.
    \item[$\qdec(\sk_f,\qct)\ra y$ or $\bot$:] The decryption algorithm takes $\sk_f$ and $\qct$ as input, and outputs $y$ or $\bot$.
    \item[$\qdel(\qct)\ra\cert$:] The deletion algorithm takes the ciphertext $\qct$ as input, and outputs a classical certificate $\cert$.
    \item[$\vrfy(\vk,\cert)\ra \top$ or $\bot$:]The verification algorithm takes $\vk$ and $\cert$ as input, and outputs $\top$ or $\bot$.
\end{description}
\end{definition}

\begin{definition}[Correctness of Functional Encryption with Certified Everlasting Deletion]\label{def:FE_CED_correctness}
    The correctness of FE with certified everlasting deletion for a class of functions $\mathcal{F}$ and plaintext space $\Ms$ is defined as follows.
\begin{description}
\item[Evaluation Correctness:] For any $\secp\in\N$, $m\in\Ms$, and $f\in\mathcal{F}$,
\begin{align}
\Pr\left[
y\ne f(m)
\ \middle |
\begin{array}{ll}
(\MPK,\MSK)\la\Setup(1^\secp)\\
\sk_f\la\keygen(\MSK,f)\\
(\qct, \vk )\lrun \qenc(\MPK,m)\\
y\la\qdec(\sk_f,\qct)
\end{array}
\right] 
\le\negl(\secp).
\end{align}

\item[Verification Correctness:] For any $\secp\in\N$, $m\in\Ms$, and $f\in\mathcal{F}$,
\begin{align}
\Pr\left[
\vrfy(\vk,\cert)\neq\top
\ \middle |
\begin{array}{ll}
(\MPK,\MSK)\la\Setup(1^\secp)\\
(\qct, \vk)\lrun \qenc(\MPK,m)\\
\cert\la\qdel(\qct)
\end{array}
\right] 
\le\negl(\secp).
\end{align}
\end{description}
\end{definition}

\begin{remark}
In FE, we should be able to run $\qDec$ algorithm for many different functions $f$ on the same ciphertext $\ct$.
One might think that the quantum $\qct$ is destroyed by $\qDec$ algorithm, and it can be used only once.
However, it is easy to see that $\qDec$ algorithm can be always modified so that it does not disturb the quantum state $\qct$
by using the gentle measurement lemma~\cite{Winter99} thanks to the evaluation correctness.
\end{remark}

\paragraph{Security notions.}
We define an indistinguishability-based security notion in the collusion-resistant setting in this section.
We extend the certified everlasting security notion of PKE by Bartusek and Khurana~\cite{myC:BarKhu23} to the FE setting to obtain our indistinguishability-based security notion.
\begin{definition}[Certified Everlasting Indistinguishable-Security of FE]\label{def:FE_CED_IND_sec}
Let $\Sigma=(\Setup,\keygen,\qenc,\qdec,\qdel,\vrfy)$ be a functional encryption with certified everlasting deletion for a class $\mathcal{F}$ of functions, plaintext space $\Ms$. We consider two experiments $\EV\expb{\Sigma,~\qA}{ada}{ind}(\secp,b)$ and $\C\expb{\Sigma,~\qA}{ada}{ind}(\lambda,b)$ played between a challenger and a non-uniform QPT adversary $\qA = \{\qA_{\secp}, \ket{\psi}_{\secp}\}_{\secp \in \mbb{N}}$. The experiments are defined as follows:
\begin{enumerate}
    \item The challenger computes $(\MPK,\MSK)\la\setup(1^\secp)$ and sends $\MPK$ to $\qA_\secp(\ket{\psi}_\secp)$.
    \item $\qA_\secp$ is allowed to make arbitrarily many key queries.
    For the $\ell$-th key query, the challenger receives $f_\ell\in\mathcal{F}$, 
    computes $\sk_{f_\ell}\la\keygen(\MSK,f_\ell)$,
    and sends $\sk_{f_\ell}$ to $\qA_\secp$.
    \item $\qA_\secp$ sends $(m_0,m_1)\in\Ms^2$ to the challenger.
    \item The challenger computes $(\qct_b,\vk_b)\la\qenc(\MPK,m_b)$, and sends $\ct_b$ to $\qA_\secp$.
    \item $\qA_\secp$ is allowed to make arbitrarily many key queries.
    For the $\ell$-th key query, the challenger receives $f_\ell\in\mathcal{F}$, 
    computes $\sk_{f_\ell}\la\keygen(\MSK,f_\ell)$,
    and sends $\sk_{f_\ell}$ to $\qA_\secp$.
    \item $\qA_\secp$ sends a certificate of deletion $\cert$ and its internal state $\rho$ to the challenger.
    \item The challenger computes $\vrfy(\vk_b, \cert)$. If the outcome is $\top$ and $f_\ell(m_0)=f_\ell(m_1)$ holds for all key queries $f_\ell$, the experiment $\EV\expb{\Sigma,~ \qA}{ada}{ind}(\secp, b)$ outputs $\rho$; otherwise $\EV\expb{\Sigma,~ \qA}{ada}{ind}(\secp, b)$ outputs $\bot$.
    \item The challenger sends the outcome of $\Vrfy(\vk_b,\cert)$ to $\qA_\secp$.
    \item Again, $\qA_\secp$ is allowed to make arbitrarily many key queries. For the $\ell$-th key query, the challenger receives $f_\ell\in\mathcal{F}$, 
    computes $\sk_{f_\ell}\la\keygen(\MSK,f_\ell)$,
    and sends $\sk_{f_\ell}$ to $\qA_\secp$.
    \item $\qA_\secp$ outputs its guess $b'$. If $f_\ell(m_0)=f_\ell(m_1)$ holds for all key queries $f_\ell$, the experiment $\C\expb{\Sigma,~ \qA}{ada}{ind}(\secp, b)$ outputs $b'$; otherwise $\C\expb{\Sigma,~ \qA}{ada}{ind}(\secp, b)$ outputs $\bot$.
\end{enumerate}
We say that the $\Sigma$ is adaptively certified everlasting indistinguishable-secure if for any non-uniform QPT adversary $\qA=\{\qA_\secp,\ket{\psi}_\secp\}_{\secp\in\N}$, it holds that
\ifnum\submission=1
\begin{align}
	\TD(\EV\expb{\Sigma,~ \qA}{ada}{ind}(\secp, 0), \EV\expb{\Sigma,~ \qA}{ada}{ind}(\secp, 1)) \leq \negl(\secp), \textrm{ and}\\
	\left|\Pr[\C\expb{\Sigma,~ \qA}{ada}{ind}(\secp, 0) = 1] - \Pr[\C\expb{\Sigma,~ \qA}{ada}{ind}(\secp, 1) = 1]\right| \leq \negl(\secp).
\end{align}
\else
\begin{align}
    \TD(\EV\expb{\Sigma,~ \qA}{ada}{ind}(\secp, 0), \EV\expb{\Sigma,~ \qA}{ada}{ind}(\secp, 1)) \leq \negl(\secp),
\end{align}
and
\begin{align}
    \left|\Pr[\C\expb{\Sigma,~ \qA}{ada}{ind}(\secp, 0) = 1] - \Pr[\C\expb{\Sigma,~ \qA}{ada}{ind}(\secp, 1) = 1]\right| \leq \negl(\secp).
\end{align}
\fi
\end{definition}

\subsection{Tools}\label{sec:tools_for_FE_CED}
We introduce a few tools for FE with certified everlasting deletion in this section.
\paragraph{(Interactive) certified everlasting lemma.}
First, we recall the certified everlasting lemma by Bartusek and Khurana~\cite{myC:BarKhu23}.
\begin{lemma}[Certified Everlasting Lemma~{\cite[ePrint Ver. 20221122:050839]{myC:BarKhu23}}]\label{lem:ce}
    Let $\{\mcl{Z}_{\secp}(\theta)\}_{\secp \in \mbb{N}, \theta \in \zo{\secp}}$ be a family of distributions over either classical bit strings or quantum states, and let $\mathfrak{A}$ be any class of adversaries such that for any $\{\qA_{
        \secp}\}_{\secp 
        \in \mbb{N}} \in \mathfrak{A}$, it holds that   
    \[
    \bigg| \Pr[\qA_{\secp}(\mcl{Z}_{\secp}(\theta)) = 1] - \Pr[\qA_{\secp}(\mcl{Z}_{\secp}(0^\secp)) = 1] \bigg| \le \negl(\secp).
    \]
    For any $\{\qA_{\secp}\}_{\secp \in \mbb{N}} \in \mathfrak{A}$, consider the following distribution $\{\wt{\mcl{Z}}_{\secp}^{\qA_{\secp}}(b)\}_{\secp \in \mbb{N}, b \in \{0, 1\}}$ over quantum states, obtained by running $\qA_{\secp}$ as follows:
    \begin{itemize}
        \item Sample $z, \theta \leftarrow \{0, 1\}^{\secp}$ and initialize $\qA_{\secp}$ with $(\ket{z}_{\theta}, b \oplus \bigoplus_{i: \theta_i = 0} z_i, \mcl{Z}_{\secp}(\theta))$.
        \item $\qA_{\secp}$'s output is parsed as bit string $z^\prime \in \{0, 1\}^{\secp}$ and a residual quantum state $\rho$.
        \item If $z_i = z_i^\prime$ for all $i$ such that $\theta_i = 1$ then output $\rho$, and otherwise output a special symbol $\bot$.
    \end{itemize}   
    Then, it holds that 
    \[
    \TD(\wt{\mcl{Z}}_{\secp}^{\qA_{\secp}}(0), \wt{\mcl{Z}}_{\secp}^{\qA_{\secp}}(1)) \le \negl(\secp).
    \]
\end{lemma} 
\begin{remark}
Although the description of the lemma above is slightly different from the original version (see the first item of~\cref{rem:difference_from_BK}), it is essentially the same as what Bartusek and Khurana~\cite{myC:BarKhu23} proved. In addition, even if we put $b \oplus \bigoplus_{i: \theta_i = 0} z_i$ in $\mcl{Z}_{\secp}$ with $\theta$, the lemma holds.
\end{remark}

We can generalize~\cref{lem:ce} to a variant in the interactive game setting as follows.
\begin{lemma}[Interactive Certified Everlasting Lemma~\cite{myC:BarKhu23}]\label{lem:ce_int}
For interactive QPT algorithms $\qA$ and $\qChal$, $\theta\in \bin^\secp$, and $\beta\in \bit$, let $\expt{\qA,\qChal}{}(\secp,\theta,\beta)$ 
be an experiment that works as follows:
\begin{itemize}
    \item $\qA$ takes $1^\secp$ as input and $\qChal$ takes $(1^\secp,\theta,\beta)$ as input.  
    \item $\qA$ and $\qChal$ interact with each other through a quantum channel. 
    \item $\qA$ outputs a bit $b'$, which is treated as the output of the  experiment.
\end{itemize}
For interactive QPT algorithms $\qA'$  and $\qChal$ and $b\in\bin$, let $\tildeexpt{\qA',\qChal}{}(\secp,b)$ 
be an experiment that works as follows:
	\begin{itemize}
		\item Sample $z, \theta \leftarrow \{0, 1\}^{\secp}$. 
		  \item $\qA'$ takes $(1^\secp, \ket{z}_{\theta})$ as input and $\qChal$ takes $(1^\secp,\theta,b \oplus \bigoplus_{i: \theta_i = 0} z_i)$ as input.  
		\item $\qA'$ and $\qChal$ interact with each other through a quantum channel.
		\item $\qA'$ outputs a string $z' \in \{0, 1\}^{\secp}$ and a quantum state $\rho$.
		\item If $z_i = z_i'$ for all $i$ such that $\theta_i = 1$ then the experiment outputs $\rho$, and otherwise it outputs a special symbol $\bot$.
	\end{itemize}

For a QPT algorithm $\qChal$, if for any QPT algorithm $\qA$, $\theta\in\bin^\secp$, and $\beta\in\bit$, it holds that
	\[
	\bigg| \Pr[\expt{\qA,\qChal}{}(\secp,\theta,\beta) = 1] - \Pr[\expt{\qA,\qChal}{}(\secp,0^\secp,\beta) = 1] \bigg| \le \negl(\secp),
	\]
then for any QPT algorithm $\qA'$, it holds that
	\[
	\TD(\tildeexpt{\qA',\qChal}{}(\secp,0),\tildeexpt{\qA',\qChal}{}(\secp,1)) \le \negl(\secp).
	\]
\end{lemma}	
\begin{remark}\label{rem:difference_from_BK}
There are the following three differences from the certified everlasting lemma of \cite{myC:BarKhu23} besides notational adaptations.
\begin{enumerate}
     \item The challenger can use $\theta$ in an arbitrary manner whereas they require the challenger to use $\theta$ in a bit-by-bit manner.\footnote{In their notation, the challenger corresponds to $\mathcal{Z}(\theta)$.} 
    \item We consider an interactive setting whereas they consider a non-interactive setting.
    \item The challenger also takes $b \oplus \bigoplus_{i: \theta_i = 0} z_i$ as part of its input. 
\end{enumerate}
Indeed, we believe that the above variant is implicitly used in the security proof of their certified everlasting secure ABE in \cite{myC:BarKhu23}. 
We observe that the above variant can be proven in essentially the same way as their original proof. \takashi{We can write a proof sketch if we have a time.}
\end{remark}

\paragraph{Public-Slot functional encryption.}
We introduce a new primitive which we call public-slot functional encryption. In this primitive, a decryption key is associated with two-input function where the first and second inputs are referred to as secret and public inputs, respectively. Given a ciphertext of a message $m$ and a decryption key for a two-input function $f$, one can compute $f(m,\pub)$ for any public input $\pub$. In the security experiment, we require that a pair of challenge messages $(m_0,m_1)$ must satisfy $f(m_0,\pub)=f(m_1,\pub)$ for all key queries $f$ and public inputs $\pub$ to prevent trivial attacks. 
A formal definition is given below. 
\begin{definition}[Public-Slot FE (Syntax)]\label{def:fe_public_slot}
A public-slot functional encryption scheme for a class $\mathcal{F}$ of functions 
is a tuple of PPT algorithms
$\Sigma=(\Setup,\keygen,\Enc,\Dec)$ 
with plaintext space $\Ms$, ciphertext space $\Cs$, master public key space $\mathcal{MPK}$, master secret key space $\mathcal{MSK}$, secret key space $\mathcal{SK}$, and public input space $\mathcal{P}$, 
that work as follows. \takashi{I'm specifying spaces for everything following other definitions in HNMY22. But this may be redundant.}
\begin{description}
    \item[$\Setup(1^\secp)\ra(\MPK,\MSK)$:]
    The setup algorithm takes the security parameter $1^\secp$ as input, and outputs a master public key $\MPK\in\mathcal{MPK}$
    and a master secret key $\MSK\in\mathcal{MSK}$.
    \item[$\keygen(\MSK,f)\ra\sk_f$:]
    The key generation algorithm takes $\MSK$ and $f\in\mathcal{F}$ as input,
    and outputs a secret key $\sk_f\in\mathcal{SK}$.
    \item[$\Enc(\MPK,m)\ra \CT$:]
    The encryption algorithm takes $\MPK$ and $m\in\Ms$ as input,
     and outputs  a ciphertext $\CT\in \Cs$.
    \item[$\Dec(\sk_f,\CT,\pub)\ra y$ $\mbox{\bf  or }\bot$:]
    The decryption algorithm takes $\sk_f$, $\CT$, and a public input $\pub\in \mathcal{P}$ as input, and outputs $y$ or $\bot$.
\end{description}
\end{definition}

We require that an FE with certified everlasting deletion scheme satisfies correctness defined below.
\begin{definition}[Correctness of Public-Slot FE]\label{def:correctness_fe_public_slot}
There exists a negligible function $\negl$ such that for any $\secp\in \N$, $m\in\Ms$, $f\in\mathcal{F}$, and $\pub\in \mathcal{P}$, 
\begin{align}
\Pr\left[
y\ne f(m,\pub)
\ \middle |
\begin{array}{ll}
(\MPK,\MSK)\la\Setup(1^\secp)\\
\sk_f\la\keygen(\MSK,f)\\
\CT\lrun \Enc(\MPK,m)\\
y\la\Dec(\sk_f,\CT,\pub)
\end{array}
\right] 
\le\negl(\secp).
\end{align}
\end{definition}

\begin{definition}[Security of Public-Slot FE]\label{def:security_fe_public_slot}
Let $\Sigma=(\Setup,\keygen,\Enc,\Dec)$ be a public-slot FE scheme.
We consider the following security experiment $\expb{\Sigma,\qA}{ada}{ind}(\secp,b)$ against a QPT adversary $\qA$. 
\begin{enumerate}
    \item The challenger runs $(\MPK,\MSK)\la \Setup(1^\secp)$ and sends $\MPK$ to $\qA$.
    \item $\qA$ is allowed to make arbitrarily many key queries.
    For the $\ell$-th key query, the challenger receives $f_\ell\in\mathcal{F}$, 
    computes $\sk_{f_\ell}\la\keygen(\MSK,f_\ell)$,
    and sends $\sk_{f_\ell}$ to $\qA$.
    \item $\qA$ sends $(m_0,m_1)\in\Ms^2$ to the challenger.
    \item 
    The challenger computes $\CT\la \Enc(\MPK,m_b)$ and sends $\CT$ to $\qA$.
    \item Again, $\qA$ is allowed to make arbitrarily many key queries.
    \item $\qA$ outputs $b'\in\bit$. 
    If $f_\ell(m_0,\pub)=f_\ell(m_1,\pub)$ holds for all key queries $f_\ell$ and public inputs $\pub\in \mathcal{P}$, the experiment outputs $b'$. Otherwise, it outputs $\bot$. \takashi{The definition is a little bit tricky since the condition is not efficiently checkable and thus we cannot simply say that the adversary must comply with this restriction. I believe the current definition is fine even though this is not falsifiable.}
\end{enumerate}
We say that $\Sigma$ is adaptively indistinguishable-secure if for any $QPT$ adversary $\qA$ it holds that
\begin{align}
\advb{\Sigma,\qA}{ada}{ind}(\secp) \seteq \abs{\Pr[\expb{\Sigma,\qA}{ada}{ind}(\secp,0)=1]  - \Pr[\expb{\Sigma,\qA}{ada}{ind}(\secp,1)=1]} \leq \negl(\secp).
\end{align}
\end{definition}

It is easy to construct selectively single-ciphertext public-slot SKFE from multi-input FE by Goldwasser et al., which can be instantiated with IO~\cite{EC:GGGJKL14}.
We convert it to adaptively indistinguishable-secure public-slot PKFE via a few transformation.
We prove the following theorem in~\cref{sec:adaptive_PKFE_pub_slot}.
\begin{theorem}\label{thm:adaptive_PKFE_pub_slot_from_IO}
If there exist IO and OWFs, there exists an adaptively indistinguishable-secure public-slot PKFE for $\Ppoly$.
\end{theorem}

\subsection{Collusion-Resistant Construction}\label{sec:CRFE_CED_const}

\paragraph{Ingredients.} We need the following building blocks.
\begin{itemize}
    \item Public-slot FE $\FE = \FE.(\setup,\keygen, \enc, \dec)$ for all polynomial-size circuits. 
    \item PRG $\PRG:\bit^\secp\ra\bit^{2\secp}$.
\end{itemize}

\paragraph{Scheme description.} Our FE with certified everlasting deletion scheme $\CED= \CED.(\Setup,\keygen,\qEnc,\qDec,\qDelete,\Vrfy)$ is described below.
\begin{description}
\item[$\CED.\setup(1^{\secp})$:]$\vspace{0.01cm}$
\begin{enumerate}
    \item Generate $(\fe.\MPK,\fe.\MSK) \leftarrow \FE.\setup(1^{\secp})$.
    \item Output $\MPK \seteq \fe.\MPK$ and $\MSK \seteq \fe.\MSK$.
\end{enumerate}

\item[$\CED.\keygen(\msk, f)$:]$\vspace{0.01cm}$
\begin{enumerate}
    \item Parse $\MSK =\fe.\MSK$.
    \item Generate $\fe.\sk_{g[f]}\gets \FE.\keygen(\fe.\MSK,g[f])$ where $g[f]$ is a function described in \cref{fig:g}. 
    \item Output $\sk_f=\fe.\sk_{g[f]}$.
\end{enumerate}

\item[$\CED.\qenc(\MPK, m)$:]$\vspace{0.01cm}$
\begin{enumerate}
    \item Parse $\MPK = \fe.\MPK$.
    \item Generate $z_{i}, \theta_{i} \leftarrow \{0, 1\}^{\secp}$ for every $i\in[2\msglen+1]$.
    \item Generate $u_{i,j,b}\la\bit^\secp$ and compute $v_{i,j,b}\la\PRG(u_{i,j,b})$ for every $i\in[2\msglen+1]$, $j\in[\secp]$, and $b\in\bit$. 
    Set $U\seteq (u_{i,j,b})_{i\in[2\msglen+1],j\in[\secp],b\in\bit}$ and $V\seteq (v_{i,j,b})_{i\in[2\msglen+1],j\in[\secp],b\in\bit}$. 
    \item 
    Generate a state 
    \begin{align}
\ket{\psi_{i,j}}\seteq \begin{cases}
    \ket{z_{i,j}}\ket{u_{i,j,z_{i,j}}} &\text{~if~} \theta_{i,j}=0\\ 
    \ket{0}\ket{u_{i,j,0}}+(-1)^{z_{i,j}} \ket{1}\ket{u_{i,j,1}} &\text{~if~} \theta_{i,j}=1
    \end{cases} 
    \end{align}
    where $\theta_{i,j}$ (resp. $z_{i,j}$) is the $j$-th bit of $\theta_i$ (resp. $z_i$)
    for every $i\in [2\msglen+1]$ and $j\in [\secp]$. 
    \item 
    Generate 
    \begin{align}
        \beta_i\seteq \begin{cases}
        m_i\oplus \bigoplus_{j: \theta_{i,j} = 0} z_{i,j} &\text{if~}i\in[\msglen]\\
        0\oplus \bigoplus_{j: \theta_{i,j} = 0} z_{i,j} &\text{if~}i\in[\msglen+1, 2\msglen+1]
        \end{cases}.
    \end{align}
\item Generate $\fe.\ct\la\FE.\enc(\fe.\MPK,V\|\theta_1\|\ldots\|\theta_{2\msglen+1}\|\beta_1\|\ldots \|\beta_{2\msglen+1})$.
    \item Output $\qct = (\fe.\ct, \bigotimes_{i\in[2\msglen+1],j\in[\secp]}\ket{\psi_{i,j}})$ and $\vk = (U,(z_i,\theta_i)_{i\in[2\msglen+1]})$.
\end{enumerate}

\item[$\CED.\qdec(\sk_f, \qct)$:]$\vspace{0.01cm}$
\begin{enumerate}
\item Parse $\sk_f\la\fe.\sk_{g[f]}$ and $\qct = (\fe.\ct, \bigotimes_{i\in[2\msglen+1],j\in[\secp]}\ket{\psi_{i,j}})$.
    \item Coherently apply $\FE.\dec(\fe.\sk_{g[f]},\fe.\ct,\cdot)$ on $\bigotimes_{i\in[2\msglen+1],j\in[\secp]}\ket{\psi_{i,j}}$
 and measure the outcome $y$.
    \item Output $y$.
\end{enumerate}

\item[$\CED.\qdel(\qct)$:]$\vspace{0.01cm}$
\begin{enumerate}
    \item Parse $\qct = (\fe.\ct, \bigotimes_{i\in[2\msglen+1],j\in[\secp]}\ket{\psi_{i,j}})$. 
    \item Measure $\ket{\psi_{i,j}}$ in the Hadamard basis to get $c_{i,j}\|d_{i,j}\in\bit^{\secp+1}$ for every $i\in[2\msglen+1]$ and $j\in [\secp]$.
    \item Output $\cert=(c_{i,j},d_{i,j})_{i\in[2\msglen+1],j\in[\secp]}$.
\end{enumerate}

\item[$\CED.\vrfy(\vk,\cert)$:]$\vspace{0.01cm}$
\begin{enumerate}
    \item Parse $\vk = (U,(z_i,\theta_i)_{i\in [2\msglen+1]})$ and $\cert=(c_{i,j},d_{i,j})_{i\in[2\msglen+1],j\in[\secp]}$.
    \item Check if $z_{i,j}=c_{i,j} \oplus d_{i,j}\cdot(u_{i,j,0}\oplus u_{i,j,1})$ holds for every $i\in[2\msglen+1]$ and $j\in[\secp]$ such that $\theta_{i,j}=1$, where $z_{i,j}$  is the $j$-th bit of $z_i$. If so, output $\top$ and otherwise output $\bot$.
\end{enumerate}

\begin{figure}
    \begin{framed}
    \begin{center}
    \underline{$g[f]$}
    \end{center}
    \textbf{Secret Input:}~ $V,\theta_1,\ldots,\theta_{2\msglen+1},\beta_1,\ldots,\beta_{2\msglen+1}$\\
     \textbf{Public Input:}~ $(b_{i,j},u_{i,j})_{i\in [2\msglen+1],j\in [\secp]}$
    \begin{enumerate}
    \item Parse $V=(v_{i,j,b})_{i\in[2\msglen+1],j\in[\secp],b\in\bit}$. 
    \item Check if $\PRG(u_{i,j})=v_{i,j,b_{i,j}}$ holds for every $i\in[2\msglen+1]$ and $j\in[\secp]$. If so, go to the next step and otherwise output $\bot$.
    \item Compute  $m_i\seteq \beta_i\oplus \bigoplus_{j: \theta_{i,j} = 0} b_{i,j}$ for every $i\in[2\msglen+1]$.
    \item Output $f(m_1\|\cdots\| m_\msglen)$ if $m_{2\msglen+1}=0$ and output $f(m_{\msglen+1}\|\cdots\| m_{2\msglen})$ otherwise.
    \end{enumerate}
    \end{framed}
    \caption{The description of the function $g[f]$}
    \label{fig:g}
\end{figure}
\end{description}

\begin{theorem}\label{thm:FE_CED_adaptive}
If $\FE$ is adaptively indistinguishable-secure public-slot FE for $\Ppoly$ and $\PRG$ is a secure PRG, $\CED$ is adaptively certified everlasting indistinguishable-secure FE for $\Ppoly$.
\end{theorem}

\paragraph{Decryption Correctness.}
Let $\qct = (\fe.\ct, \bigotimes_{i\in[2\msglen+1],j\in[\secp]}\ket{\psi_{i,j}})$ be an honestly generated ciphertext for a message $m$ and $\sk_f=\fe.\sk_{g[f]}$ be an honestly generated decryption key for a function $f$. Then we have 
$\fe.\ct\in \FE.\enc(\fe.\MPK,V\|\theta_1\|\ldots\|\theta_{2\msglen+1}\|\beta_1\|\ldots \|\beta_{2\msglen+1})$ and $\ket{\psi_{i,j}}=\ket{z_{i,j}}\ket{u_{i,j,z_{i,j}}}$ for all $i\in [2\msglen+1]$ and $j\in[\secp]$ such that $\theta_{i,j}=0$. Since we have $\PRG(u_{i,j,b})=v_{i,j,b}$ for all $i\in [2\msglen+1]$, $j\in[\secp]$, and $b\in \bit$,  and $m_{2\msglen+1}=\beta_{2\msglen+1}\oplus \bigoplus_{j: \theta_{2\msglen+1,j} = 0} z_{i,j}=0$,  if we coherently run $g[f](V\|\theta_1\|\ldots\|\theta_{2\msglen+1}\|\beta_1\|\ldots \|\beta_{2\msglen+1},\cdot)$ on $\bigotimes_{i\in[2\msglen+1],j\in[\secp]}\ket{\psi_{i,j}}$
and measure the output, then the resulting outcome is $f(\beta_{1}\oplus \bigoplus_{j: \theta_{1,j} = 0}z_{1,j}\|\ldots\|\beta_{\msglen}\oplus \bigoplus_{j: \theta_{\msglen,j}=0}z_{\msglen,j})=f(m)$.
Then the decryption correctness follows from that of $\FE$. 

\paragraph{Verification Correctness.}
Let $\qct = (\fe.\ct, \bigotimes_{i\in[2\msglen+1],j\in[\secp]}\ket{\psi_{i,j}})$ be an honestly generated ciphertext and $\vk = (U,(z_i,\theta_i)_{i\in[2\msglen+1]})$ be the corresponding verification key. 
For all $i\in [2\msglen+1]$ and $j\in [\secp]$ such that $\theta_{i,j}=1$,  
since we have $\ket{\psi_{i,j}}=\ket{0}\ket{u_{i,j,0}}+(-1)^{z_{i,j}} \ket{1}\ket{u_{i,j,1}}$, if we measure it in the Hadamard basis, then the outcome $c_i\|d_i$ satisfies  $z_{i,j}=c_{i,j} \oplus d_{i,j}\cdot(u_{i,j,0}\oplus u_{i,j,1})$.
This immediately implies the verification correctness. 

\paragraph{Security for $\C\expb{\CED,~ \qA}{ada}{ind}(\secp, b)$.}
We omit the proof in the main body. See~\cref{appsec:omitted_proof_crfecd}.

\paragraph{Security for $\EV\expb{\CED,~ \qA}{ada}{ind}(\secp, b)$.}
Let $\qA$ be an adversary against the adaptive certified everlasting indistinguishable-security.
We consider the following sequence of hybrids.

\begin{description}
\item[$\hyb_0$:]
This is the original everlasting adaptive security experiment where the challenge bit is set to be $0$. Specifically, it works as follows:
\begin{enumerate}
    \item The challenger generates $(\fe.\MPK,\fe.\MSK) \leftarrow \FE.\setup(1^{\secp})$, sets $\MPK \seteq \fe.\MPK$ and $\MSK \seteq \fe.\MSK$, and sends $\MPK$ to $\qA$.
    \item $\qA$ can make arbitrarily many key queries at any point of the experiment. When it makes a key query $f$, the challenger generates  $\fe.\sk_{g[f]}\gets \FE.\keygen(\fe.\MSK,g[f])$ and returns $\sk_f=\fe.\sk_{g[f]}$ to $\qA$. 
    \item $\qA$ sends $(m^{(0)},m^{(1)})$ to the challenger.\footnote{We use $(m^{(0)},m^{(1)})$ instead of $(m_0,m_1)$ to denote a pair of challenge messages to avoid a notational collision.} It must satisfy $f(m^{(0)})=f(m^{(1)})$ for all key queries $f$ that are made before or after sending $(m^{(0)},m^{(1)})$. 
    \item The challenger generates $(\qct,\vk)\gets \qenc(\MPK,m^{(0)})$. Specifically,
    \begin{enumerate}
    \item Generate $z_{i}, \theta_{i} \leftarrow \{0, 1\}^{\secp}$ for every $i\in[2\msglen+1]$.
    \item Generate $u_{i,j,b}\la\bit^\secp$ and compute $v_{i,j,b}\la\PRG(u_{i,j,b})$ for every $i\in[2\msglen+1]$, $j\in[\secp]$ and $b\in\bit$ and set $U=(u_{i,j,b})_{i\in[2\msglen+1],j\in[\secp],b\in\bit}$ and $V\seteq (v_{i,j,b})_{i\in[2\msglen+1],j\in[\secp],b\in\bit}$. 
    \item 
    Generate a state 
    \begin{align}
        \ket{\psi_{i,j}}\seteq \begin{cases}
    \ket{z_{i,j}}\ket{u_{i,j,z_{i,j}}} &\text{~if~} \theta_{i,j}=0\\ 
    \ket{0}\ket{u_{i,j,0}}+(-1)^{z_{i,j}} \ket{1}\ket{u_{i,j,1}} &\text{~if~} \theta_{i,j}=1
    \end{cases}
    \end{align}
    where $\theta_{i,j}$ (resp. $z_{i,j}$) is the $j$-th bit of $\theta_i$ (resp. $z_i$)
    for every $i\in [2\msglen+1]$ and $j\in [\secp]$. 
    \item  \label{step:generate_beta_ever}
    Generate
      \begin{align}
        \beta_i\seteq \begin{cases}
        m_i^{(0)}\oplus \bigoplus_{j: \theta_{i,j} = 0} z_{i,j} &\text{if~}i\in[\msglen]\\
        0\oplus \bigoplus_{j: \theta_{i,j} = 0} z_{i,j} &\text{if~}i\in[\msglen+1, 2\msglen+1]
        \end{cases}.
    \end{align}
\item Generate $\fe.\ct\la\FE.\enc(\fe.\MPK,V\|\theta_1\|\ldots\|\theta_{2\msglen+1}\|\beta_1\|\ldots \|\beta_{2\msglen+1})$.
    \item Set $\qct = (\fe.\ct, \bigotimes_{i\in[2\msglen+1],j\in[\secp]}\ket{\psi_{i,j}})$ and $\vk = (U,(z_i,\theta_i)_{i\in[2\msglen+1]})$.
    \end{enumerate}
    The challenger sends $\qct$ to $\qA$.
\item $\qA$ sends $\cert=(c_{i,j},d_{i,j})_{i\in [2\msglen+1],j\in[\secp]}$ and its internal state $\rho$ to the challenger.
\item The challenger checks if $z_{i,j}=c_{i,j} \oplus d_{i,j}\cdot(u_{i,j,0}\oplus u_{i,j,1})$ holds for every $i\in[2\msglen+1]$ and $j\in[\secp]$ such that $\theta_{i,j}=1$. If it does not hold, the challenger outputs $\bot$ as a final output of the experiment. Otherwise, go to the next step. \takashi{In the current definition in HNMY, we have to send $\MSK$ to $\qA_2$. However, I think this is redundant as pointed out by BK. I'm omitting this because this make the proof simpler. }
\item The experiment outputs $\rho$ as a final output.\ryo{I modified the final output to fit the security definition (\cref{def:FE_CED_IND_sec}.)}
\end{enumerate}

\item[$\hyb_{1,k}$:] For $k=0,1,\ldots,\msglen$, this hybrid is identical to $\hyb_0$ except that 
the way of generating $\beta_i$ is modified as follows:
   \begin{align}
        \beta_i\seteq \begin{cases}
        m_i^{(0)}\oplus \bigoplus_{j: \theta_{i,j} = 0} z_{i,j} &\text{if~}i\in[\msglen]\\
        m_{i-n}^{(1)}\oplus \bigoplus_{j: \theta_{i,j} = 0} z_{i,j} &\text{if~}i\in[\msglen+1,\msglen+k]\\
        0\oplus \bigoplus_{j: \theta_{i,j} = 0} z_{i,j} &\text{if~}i\in[\msglen+k+1, 2\msglen+1]
        \end{cases}.
    \end{align}
 Remark that $\hyb_{1,0}$ is identical to $\hyb_0$. 
 
 \item[$\hyb_{2}$:] This hybrid is identical to $\hyb_{1,\msglen}$ except that $\beta_{2\msglen+1}$ is flipped. That is, $\beta_i$ is generated as follows.
   \begin{align}
        \beta_i\seteq \begin{cases}
        m_i^{(0)}\oplus \bigoplus_{j: \theta_{i,j} = 0} z_{i,j} &\text{if~}i\in[\msglen]\\
        m_{i-n}^{(1)}\oplus \bigoplus_{j: \theta_{i,j} = 0} z_{i,j} &\text{if~}i\in[\msglen+1,2\msglen]\\
        1\oplus \bigoplus_{j: \theta_{i,j} = 0} z_{i,j} &\text{if~}i=2\msglen+1
        \end{cases}.
    \end{align}

\item[$\hyb_{3,k}$:] For $k=0,1,\ldots,\msglen$, this hybrid is identical to $\hyb_2$ except that the way of generating $\beta_i$ is modified as follows:
   \begin{align}
        \beta_i\seteq \begin{cases}
        m_i^{(1)}\oplus \bigoplus_{j: \theta_{i,j} = 0} z_{i,j} &\text{if~}i\in[k]\\
        m_i^{(0)}\oplus \bigoplus_{j: \theta_{i,j} = 0} z_{i,j} &\text{if~}i\in[k+1,\msglen]\\
        m_{i-n}^{(1)}\oplus \bigoplus_{j: \theta_{i,j} = 0} z_{i,j} &\text{if~}i\in[\msglen+1,2\msglen]\\
        1\oplus \bigoplus_{j: \theta_{i,j} = 0} z_{i,j} &\text{if~}i\in[\msglen+k+1, 2\msglen+1]
        \end{cases}.
    \end{align}
    Remark that $\hyb_{3,0}$ is identical to $\hyb_2$. 
    
\item[$\hyb_{4}$:] This hybrid is identical to $\hyb_{3,\msglen}$ except that $\beta_{2\msglen+1}$ is flipped. That is, $\beta_i$ is generated as follows.
   \begin{align}
        \beta_i\seteq \begin{cases}
        m_i^{(1)}\oplus \bigoplus_{j: \theta_{i,j} = 0} z_{i,j} &\text{if~}i\in[\msglen]\\
        m_{i-n}^{(1)}\oplus \bigoplus_{j: \theta_{i,j} = 0} z_{i,j} &\text{if~}i\in[\msglen+1,2\msglen]\\
        0\oplus \bigoplus_{j: \theta_{i,j} = 0} z_{i,j} &\text{if~}i=2\msglen+1
        \end{cases}.
    \end{align}    
    
\item[$\hyb_{5,k}$:] For $k=0,1,\ldots,\msglen$, this hybrid is identical to $\hyb_4$ except that the way of generating $\beta_i$ is modified as follows:
   \begin{align}
        \beta_i\seteq \begin{cases}
        m_i^{(1)}\oplus \bigoplus_{j: \theta_{i,j} = 0} z_{i,j} &\text{if~}i\in[\msglen]\\
        0\oplus \bigoplus_{j: \theta_{i,j} = 0} z_{i,j} &\text{if~}i\in[\msglen+1,\msglen+k]\\
        m_{i-n}^{(0)}\oplus \bigoplus_{j: \theta_{i,j} = 0} z_{i,j} &\text{if~}i\in[\msglen+k+1,2\msglen]\\
        0\oplus \bigoplus_{j: \theta_{i,j} = 0} z_{i,j} &\text{if~}i=2\msglen+1
        \end{cases}.
    \end{align} 
    Remark that $\hyb_{5,0}$ is identical to $\hyb_4$. 
\end{description}

Remark that $\hyb_{5,\msglen}$ is exactly the everlasting adaptive experiment where the challenge bit is set to be $1$. Thus, we have to prove 
\begin{align} \label{eq:everlasting_adaptive_conclusion}
    \TD(\hyb_0, \hyb_{5,\msglen})\le \negl(\secp).
\end{align}
We prove this by the following lemmata.  

\begin{lemma}\label{lem:hybrid_one_k}
If $\FE$ is adaptively secure and
$\PRG$ is a secure PRG, 
for any $k\in [\msglen]$, 
\begin{align}
        \TD(\hyb_{1,k-1}, \hyb_{1,k})\le \negl(\secp).
\end{align}
\end{lemma}
\begin{lemma}\label{lem:hybrid_two}
If $\FE$ is adaptively secure and
$\PRG$ is a secure PRG, 
\begin{align}
        \TD(\hyb_{1,\msglen}, \hyb_{2})\le \negl(\secp).
\end{align}
\end{lemma}
\begin{lemma}\label{lem:hybrid_three_k}
If $\FE$ is adaptively secure and
$\PRG$ is a secure PRG,  
for any $k\in [\msglen]$, 
\begin{align}
        \TD(\hyb_{3,k-1}, \hyb_{3,k})\le \negl(\secp).
\end{align}
\end{lemma}
\begin{lemma}\label{lem:hybrid_four}
If $\FE$ is adaptively secure and
$\PRG$ is a secure PRG, 
\begin{align}
        \TD(\hyb_{3,\msglen}, \hyb_{4})\le \negl(\secp).
\end{align}
\end{lemma}
\begin{lemma}\label{lem:hybrid_five_k}
If $\FE$ is adaptively secure and
$\PRG$ is a secure PRG, 
for any $k\in [\msglen]$, 
\begin{align}
        \TD(\hyb_{5,k-1}, \hyb_{5,k})\le \negl(\secp).
\end{align}
\end{lemma}
Noting that $\hyb_{1,0}$, $\hyb_{3,0}$, and $\hyb_{5,0}$ are identical to $\hyb_0$, $\hyb_2$, and $\hyb_4$, respectively, \cref{lem:hybrid_one_k,lem:hybrid_two,lem:hybrid_three_k,lem:hybrid_four,lem:hybrid_five_k} imply  \cref{eq:everlasting_adaptive_conclusion}. 

What is left is to prove these lemmata. Actually, the proofs of these lemmata are very similar. 
We give a full proof of \cref{lem:hybrid_one_k} below.
After that, we also explain how to modify it to prove \cref{lem:hybrid_two}.
The proofs of \cref{lem:hybrid_three_k,lem:hybrid_five_k} are almost identical to that of \cref{lem:hybrid_one_k} and the proof of \cref{lem:hybrid_four} is almost identical to that of \cref{lem:hybrid_two}, and thus we omit them. 

\paragraph{Proof of~\cref{lem:hybrid_one_k}.}
First, we prove~\cref{lem:hybrid_one_k} below.
\begin{proof}
For applying \cref{lem:ce_int}, we consider the following experiment $\expt{\qB,\qChal}{1,k}(\secp,\theta,\beta)$ between a QPT adversary $\qB$ and a challenger $\qChal$ for $\theta\in \bit^\secp$ and $\beta\in \bit$ as follows:
\begin{description}
\item[$\expt{\qB,\qChal}{1,k}(\secp,\theta,\beta)$:] In this experiment, $\qB$ and $\qChal$ play the roles of $\qA$ and the challenger of $\hyb_{1,k}$ with the differences that $\qChal$ sets $\beta_{\msglen+k}\seteq \beta$ and $\theta_{\msglen+k}\seteq \theta$, $\qC$ does not generate $\ket{\psi_{\msglen+k,j}}$ for $j\in [\secp]$ and thus not send it to $\qB$, $\qChal$ additionally sends $\{u_{\msglen+k,j,b}\}_{j\in[\secp],b\in \bit}$ to $\qB$, 
and $\qB$ finally outputs a bit $b'$ instead of a certificate. Specifically, it works as follows: 
\begin{enumerate}
    \item $\qChal$ generates $(\fe.\MPK,\fe.\MSK) \leftarrow \FE.\setup(1^{\secp})$, sets $\MPK \seteq \fe.\MPK$ and $\MSK \seteq \fe.\MSK$, and sends $\MPK$ to $\qB$.
    \item $\qB$ can make arbitrarily many key queries at any point of the experiment. When it makes a key query $f$, $\qChal$ generates  $\fe.\sk_{g[f]}\gets \FE.\keygen(\fe.\MSK,g[f])$ and returns $\sk_f=\fe.\sk_{g[f]}$ to $\qB$. 
    \item $\qB$ sends $(m^{(0)},m^{(1)})$ to $\qChal$.  It must satisfy $f(m^{(0)})=f(m^{(1)})$ for all key queries $f$ that are made before or after sending $(m^{(0)},m^{(1)})$. 
    \item $\qChal$ does the following:
    \begin{enumerate}
    \item Generate $z_{i}, \theta_{i} \leftarrow \{0, 1\}^{\secp}$ for every $i\in[2\msglen+1]\setminus \{\msglen+k\}$ and set $\theta_{\msglen+k}\seteq \theta$.
    \item Generate $u_{i,j,b}\la\bit^\secp$ and compute $v_{i,j,b}\la\PRG(u_{i,j,b})$ for every $i\in[2\msglen+1]$, $j\in[\secp]$ and $b\in\bit$ and set $U=(u_{i,j,b})_{i\in[2\msglen+1],j\in[\secp],b\in\bit}$ and $V\seteq (v_{i,j,b})_{i\in[2\msglen+1],j\in[\secp],b\in\bit}$. 
    \item 
    Generate a state 
    \begin{align}
    \ket{\psi_{i,j}}\seteq \begin{cases}
    \ket{z_{i,j}}\ket{u_{i,j,z_{i,j}}} &\text{~if~} \theta_{i,j}=0\\ 
    \ket{0}\ket{u_{i,j,0}}+(-1)^{z_{i,j}} \ket{1}\ket{u_{i,j,1}} &\text{~if~} \theta_{i,j}=1
    \end{cases}
    \end{align}
    where $\theta_{i,j}$ (resp. $z_{i,j}$) is the $j$-th bit of $\theta_i$ (resp. $z_i$)
    for every $i\in [2\msglen+1]\setminus \{\msglen+k\}$ and $j\in [\secp]$. 
    \item  
    Generate
      \begin{align}
        \beta_i\seteq \begin{cases}
        m_i^{(0)}\oplus \bigoplus_{j: \theta_{i,j} = 0} z_{i,j} &\text{if~}i\in[\msglen]\\
        m_{i-n}^{(1)}\oplus \bigoplus_{j: \theta_{i,j} = 0} z_{i,j} &\text{if~}i\in[\msglen+1,\msglen+k-1]\\
        \beta &\text{if~}i=\msglen+k\\
        0\oplus \bigoplus_{j: \theta_{i,j} = 0} z_{i,j} &\text{if~}i\in[\msglen+k+1, 2\msglen+1]
        \end{cases}.
    \end{align}
\item Generate $\fe.\ct\la\FE.\enc(\fe.\MPK,V\|\theta_1\|\ldots\|\theta_{2\msglen+1}\|\beta_1\|\ldots \|\beta_{2\msglen+1})$.
    \end{enumerate}
    $\qChal$ sends $(\fe.\ct, \bigotimes_{i\in[2\msglen+1]\setminus \{\msglen+k\},j\in[\secp]}\ket{\psi_{i,j}},\{u_{\msglen+k,j,b}\}_{j\in[\secp],b\in \bit})$ to $\qB$.
\item $\qB$ outputs a bit $b'$ as a final output of the experiment. 
\end{enumerate}
\end{description}

We prove the following lemma.
\begin{lemma}\label{lem:comp_ind_one_k}
For any QPT $\qB$, 
\begin{align}
\bigg| \mrm{Pr}[\expt{\qB,\qChal}{1,k}(\secp,\theta,\beta) = 1] - \mrm{Pr}[\expt{\qB,\qChal}{1,k}(\secp,0^n,\beta) = 1] \bigg| \le \negl(\secp).
\end{align}
\end{lemma}
Before proving \cref{lem:comp_ind_one_k}, we complete the proof of \cref{lem:hybrid_one_k} assuming that \cref{lem:comp_ind_one_k} is true. 
By \cref{lem:ce_int,lem:comp_ind_one_k}, for any QPT adversary $\qB'$, we have 
	\begin{align}\label{eq:tilde_bound_one_k}
	\TD(\tildeexpt{\qB',\qChal}{1,k}(\secp,0),\tildeexpt{\qB',\qChal}{1,k}(\secp,1)) \le \negl(\secp).
	\end{align}
where $\tildeexpt{\qB',\qChal}{1,k}(\secp,b)$ is an experiment that works as follows:
\begin{description}
\item[$\tildeexpt{\qB',\qChal}{1.k}(\secp,b)$:]~
\begin{enumerate}
\item Sample $z, \theta \leftarrow \{0, 1\}^{\secp}$. 
\item $\qB'$ takes $(1^\secp, \ket{z}_{\theta})$ as input. 
\item $\qB'$ interacts with $\qChal$ as in $\expt{\qB,\qChal}{1,k}(\secp,\theta,b\oplus \bigoplus_{j: \theta_{i,j} = 0} z_{i,j})$ where $\qB'$ plays the role of $\qB$. 
\item $\qB'$ outputs a string $z' \in \{0, 1\}^{\secp}$ and a quantum state $\rho$.
\item If $z_j = z_j'$ for all $j\in[\secp]$ such that $\theta_j = 1$ then the experiment outputs $\rho$, and otherwise it outputs a special symbol $\bot$.
\end{enumerate}
\end{description}
Note that the only difference between $\hyb_{1,k-1}$ and $\hyb_{1,k}$ is that $\beta_{\msglen+k}$ is set to be $0\oplus \bigoplus_{j: \theta_{i,j} = 0} z_{i,j}$ in $\hyb_{1,k-1}$ and $m_{k}^{(1)}\oplus \bigoplus_{j: \theta_{i,j} = 0} z_{i,j}$ in $\hyb_{1,k}$. If $m_{k}^{(1)}=0$, then there is no difference. Thus, we assume that $m_{k}^{(1)}=1$. 
Then we construct $\qB'$ that distinguishes $\tildeexpt{\qB',\qChal}{1.k}(\secp,0)$ and $\tildeexpt{\qB',\qChal}{1.k}(\secp,1)$ using $\qA$ that distinguishes $\hyb_{1,k-1}$ and $\hyb_{1,k}$ as follows. 
\begin{description}
\item[$\qB'(1^\secp,\ket{z}_\theta)$:]~
\begin{enumerate}
\item $\qB'$ plays the role of $\qA$ in $\hyb_{1,k}$ where the external challenger $\qChal$ of $\tildeexpt{\qB',\qChal}{1,k}(\secp,b)$ is used to simulate the challenger of $\hyb_{1,k}$. $\qChal$ provides everything that should be sent to $\qA$ except for $\ket{\psi_{\msglen+k,j}}$ for $j\in [\secp]$. $\qB'$ generates $\ket{\psi_{\msglen+k,j}}$ by applying the map $\ket{b}\ra\ket{b}\ket{u_{\msglen+k,j,b}}$ on the $j$-th qubit of $\ket{z}_\theta$ and uses it as part of $\qct$ sent to $\qA$. Note that this is possible since $(u_{\msglen+k,j,b})_{j\in[\secp],b\in\bit}$ is provided from $\qC$. 
\item Suppose that $\qA$ returns a certificate $(c_{i,j},d_{i,j})_{i\in [2\msglen+1],j\in[\secp]}$.  $\qB'$ sets $z'_{\msglen+k,j}=c_{\msglen+k,j} \oplus d_{\msglen+k,j}\cdot(u_{\msglen+k,j,0}\oplus u_{\msglen+k,j,1})$ for $j\in [\secp]$. Again, note that this is possible since $(u_{\msglen+k,j,b})_{j\in[\secp],b\in\bit}$ is provided from $\qC$. 
\item Output $z'\seteq z'_{\msglen+k,1}\concat\ldots\concat z'_{\msglen+k,\secp}$ and $\qA$'s internal state $\rho$.
\end{enumerate}
\end{description}
We can see that $\qB'$ perfectly simulates $\hyb_{1,k-1}$ if $b=0$ and $\hyb_{1,k}$ if $b=1$. (Recall that we are assuming $m_{k}^{(1)}=1$.) 
Moreover, we have $z_j = z_j'$ for all $j\in[\secp]$ (which is the condition to not output $\bot$ in $\tildeexpt{\qB',\qChal}{1.k}(\secp,b)$) whenever  $z_{i,j}=c_{i,j} \oplus d_{i,j}\cdot(u_{i,j,0}\oplus u_{i,j,1})$ holds for every $i\in[2\msglen+1]$ and $j\in[\secp]$ such that $\theta_{i,j}=1$ (which is the condition to not output $\bot$ in $\hyb_{1,k-1}$ and $\hyb_{1,k}$). 
Therefore, we must have 
\begin{align}
\TD(\hyb_{1,k-1}, \hyb_{1,k})
    \le \TD(\tildeexpt{\qB',\qChal}{1,k}(\secp,0),\tildeexpt{\qB',\qChal}{1,k}(\secp,1)).
\end{align}
\takashi{More explanation needed here?}
Combined with \cref{eq:tilde_bound_one_k}, this completes the proof of \cref{lem:hybrid_one_k}.
\end{proof}

Now, we are left to prove \cref{lem:comp_ind_one_k}. 
\begin{proof}[Proof of \cref{lem:comp_ind_one_k}]
We further consider the following sequence of hybrids:
\begin{description}
\item[$\expt{\qB,\qChal}{1,k,a}(\secp,\theta,\beta)$:] 
This is identical to $\expt{\qB,\qChal}{1,k}(\secp,\theta,\beta)$ except that $v_{i,j,1\oplus z_{i,j}}$ is uniformly chosen from $\bit^{2\secp}$ instead of being set to be $\PRG(u_{i,j,1\oplus z_{i,j}})$ for all $i\in [2\msglen+1]\setminus \{\msglen+k\}$ and $j\in [\secp]$ such that $\theta_{i,j}=0$. 
\item[$\expt{\qB,\qChal}{1,k,b}(\secp,\theta,\beta)$:] 
This is identical to $\expt{\qB,\qChal}{1,k,a}(\secp,\theta,\beta)$ except that $\theta_{\msglen+k}=\theta$ is replaced with $0^n$. Note that $\theta_{\msglen+k}$ only appears in the encrypted message for $\fe.\ct$ in $\expt{\qB,\qChal}{1,k,a}(\secp,\theta,\beta)$.
\end{description}
\begin{proposition}\label{prop_one_k_a}
If $\PRG$ is a secure PRG, 
\begin{align}
\bigg| \mrm{Pr}[\expt{\qB,\qChal}{1,k}(\secp,\theta,\beta) = 1] - \mrm{Pr}[\expt{\qB,\qChal}{1,k,a}(\secp,\theta,\beta) = 1] \bigg| \le \negl(\secp).
\end{align}
\end{proposition}
\begin{proof}
Noting that $u_{i,j,1\oplus z_{i,j}}$ for $i\in[2\msglen+1]\setminus \{\msglen+k\}$ and $j\in [\secp]$ such that $\theta_{i,j}=0$ is used only for generating $v_{i,j,1\oplus z_{i,j}}$ in $\expt{\qB,\qChal}{1,k}(\secp,\theta,\beta)$, \cref{prop_one_k_a} directly follows from the security of $\PRG$.
\end{proof}
\begin{proposition}\label{prop_one_k_b}
If 
$\FE$ is adaptively secure, 
\begin{align}
\bigg| \mrm{Pr}[\expt{\qB,\qChal}{1,k,a}(\secp,\theta,\beta) = 1] - \mrm{Pr}[\expt{\qB,\qChal}{1,k,b}(\secp,\theta,\beta) = 1] \bigg| \le \negl(\secp).
\end{align}
\end{proposition}
\begin{proof}
For each $i\in[2\msglen+1]\setminus \{\msglen+k\}$ and $j\in [\secp]$ such that $\theta_{i,j}=0$, there is no $u$ such that $\PRG(u)=v_{i,j,1\oplus z_{i,j}}$ except for probability $2^{-\secp}$. Let $\mathsf{Good}$ be the event that the above holds for all $i\in[2\msglen+1]\setminus \{\msglen+k\}$ and $j\in [\secp]$ such that $\theta_{i,j}=0$. We have $\Pr[\mathsf{Good}]\ge 1-2\msglen \secp 2^{-\secp}=1-\negl(\secp)$. 
We prove that whenever $\mathsf{Good}$ occurs, we have 
\begin{align} \label{eq:equivalence_function}
&g[f]((V,\theta_1,\ldots,\theta_{2\msglen+1},\beta_1,\ldots,\beta_{2\msglen+1}),(b_{i,j},u_{i,j})_{i\in [2\msglen+1],j\in [\secp]})\\
=
&g[f]((V,\theta_1,\ldots,\theta_{\msglen+k-1},0^n,\theta_{\msglen+k+1,\ldots}\theta_{2\msglen+1},\beta_1,\ldots,\beta_{2\msglen+1}),(b_{i,j},u_{i,j})_{i\in [2\msglen+1],j\in [\secp]})
\end{align}
for all key queries $f$ and $(b_{i,j},u_{i,j})_{i\in [2\msglen+1],j\in [\secp]}$. If this is proven, \cref{prop_one_k_b} directly follows from the adaptive security of $\FE$. 

Below, we prove \cref{eq:equivalence_function}. 
We consider the following two cases.
\begin{itemize}
    \item If $\PRG(u_{i,j})=v_{i,j,b_{i,j}}$ holds for every $i\in[2\msglen+1]$ and $j\in[\secp]$, then by the assumption that $\mathsf{Good}$ occurs, we have $b_{i,j}=z_{i,j}$ for all $i\in[2\msglen+1]\setminus \{\msglen+k\}$ and $j\in [\secp]$ such that $\theta_{i,j}=0$. Then we have $\beta_i\oplus \bigoplus_{j: \theta_{i,j} = 0} b_{i,j}=m^{(0)}_i$ for $i\in [\msglen]$ and  $\beta_{2\msglen+1}\oplus \bigoplus_{j: \theta_{2\msglen+1,j} = 0} b_{2\msglen+1,j}=0$. Then both sides of \cref{eq:equivalence_function} are equal to $f(m^{(0)})$. 
    \item Otherwise, both sides of \cref{eq:equivalence_function} are equal to $\bot$.
\end{itemize}
In either case, \cref{eq:equivalence_function} holds. This completes the proof of \cref{prop_one_k_b}. 
\end{proof}
\begin{proposition}\label{prop_one_k_c}
If $\PRG$ is a secure PRG,  
\begin{align}
\bigg| \mrm{Pr}[\expt{\qB,\qChal}{1,k,b}(\secp,\theta,\beta) = 1] - \mrm{Pr}[\expt{\qB,\qChal}{1,k}(\secp,0^n,\beta) = 1] \bigg| \le \negl(\secp).
\end{align}
\end{proposition}
\begin{proof}
Noting that $u_{i,j,1\oplus z_{i,j}}$  for $i\in[2\msglen+1]\setminus \{\msglen+k\}$ and $j\in [\secp]$ such that $\theta_{i,j}=0$ is used only for generating $v_{i,j,1\oplus z_{i,j}}$ in $\expt{\qB,\qChal}{1,k}(\secp,0^n,\beta)$, \cref{prop_one_k_c} directly follows from the security of $\PRG$.
\end{proof}
\cref{lem:comp_ind_one_k} follows from the above propositions.
\end{proof}

\paragraph{Proof of~\cref{lem:hybrid_two}.}
Next, we prove~\cref{lem:hybrid_two}.
\begin{proof}
Since the proof of \cref{lem:hybrid_two} is very similar to that of \cref{lem:hybrid_one_k}, we only explain the difference. First, we define an experiment $\expt{\qB,\qChal}{2}(\secp,\theta,\beta)$ that is similar to $\expt{\qB,\qChal}{1,k}(\secp,\theta,\beta)$ except that $\msglen+k$ is replaced with $2\msglen+1$. Then by almost the same argument as that in the proof of \cref{lem:hybrid_one_k} using \cref{lem:ce_int}, we only have to prove 
\begin{align}
\bigg| \mrm{Pr}[\expt{\qB,\qChal}{2}(\secp,\theta,\beta) = 1] - \mrm{Pr}[\expt{\qB,\qChal}{2}(\secp,0^n,\beta) = 1] \bigg| \le \negl(\secp)
\end{align}
for all QPT $\qB$. 
Its proof is also similar to that of \cref{lem:comp_ind_one_k}. 
We define $\expt{\qB,\qChal}{2,a}(\secp,\theta,\beta)$ and $\expt{\qB,\qChal}{2,b}(\secp,\theta,\beta)$ similarly to $\expt{\qB,\qChal}{1,k,a}(\secp,\theta,\beta)$ and $\expt{\qB,\qChal}{1,k,b}(\secp,\theta,\beta)$ except that $\msglen+k$ is replaced with $2\msglen+1$.  
Then the computational indistinguishability between 
 $\expt{\qB,\qChal}{2}(\secp,\theta,\beta)$ and  $\expt{\qB,\qChal}{2,a}(\secp,\theta,\beta)$ and between  $\expt{\qB,\qChal}{2,b}(\secp,\theta,\beta)$ and  $\expt{\qB,\qChal}{2}(\secp,0^n,\beta)$ immediately follow from the security of $\PRG$. 
 We argue  the computational indistinguishability between  $\expt{\qB,\qChal}{2,a}(\secp,\theta,\beta)$ and  $\expt{\qB,\qChal}{2,b}(\secp,\theta,\beta)$ based on  the security of $\FE$ as follows. 

 Let $\mathsf{Good}$ be the event that there is no $u$ such that $\PRG(u)=v_{i,j,1\oplus z_{i,j}}$ for all $i\in[2\msglen]$ and $j\in [\secp]$ such that $\theta_{i,j}=0$.  
 We have $\Pr[\mathsf{Good}]\ge 1-\negl(\secp)$. 
 Similarly to the proof of  \cref{prop_one_k_b}, 
it suffices to prove that whenever $\mathsf{Good}$ occurs, we have 
\begin{align} \label{eq:equivalence_function_again}
&g[f]((V,\theta_1,\ldots,\theta_{2\msglen+1},\beta_1,\ldots,\beta_{2\msglen+1}),(b_{i,j},u_{i,j})_{i\in [2\msglen+1],j\in [\secp]})\\
=
&g[f]((V,\theta_1,\ldots,\theta_{2\msglen},0^\secp,\beta_1,\ldots,\beta_{2\msglen+1}),(b_{i,j},u_{i,j})_{i\in [2\msglen+1],j\in [\secp]})
\end{align}
for all key queries $f$ and $(b_{i,j},u_{i,j})_{i\in [2\msglen+1],j\in [\secp]}$.  
Below, we prove \cref{eq:equivalence_function_again}. 
We consider the following two cases.
\begin{itemize}
    \item If $\PRG(u_{i,j})=v_{i,j,b_{i,j}}$ holds for every $i\in[2\msglen+1]$ and $j\in[\secp]$, then by the assumption that $\mathsf{Good}$ occurs, we have $b_{i,j}=z_{i,j}$ for all $i\in[2\msglen]$ and $j\in [\secp]$ such that $\theta_{i,j}=0$. Then we have 
    \begin{align}
        \beta_i\oplus \bigoplus_{j: \theta_{i,j} = 0} b_{i,j}=
        \begin{cases}
            m^{(0)}_i   &\text{~if~} i\in [\msglen]\\
            m^{(1)}_{i-\msglen} &\text{~if~} i\in [\msglen+1,2\msglen]
        \end{cases}.
    \end{align}
Then the LHS of \cref{eq:equivalence_function_again} is equal to $f(m^{(\gamma)})$ where $\gamma=\beta_{2\msglen+1}\oplus \bigoplus_{j: \theta_{i,j} = 0} b_{2\msglen+1,j}$
and the RHS of \cref{eq:equivalence_function_again} is equal to $f(m^{(\gamma')})$ where $\gamma'=\beta_{2\msglen+1}\oplus \bigoplus_{j\in [\secp]} b_{2\msglen+1,j}$.
By the restriction on $\qB$, we have  $f(m^{(0)})=f(m^{(1)})$. Therefore, both sides of \cref{eq:equivalence_function_again} are equal to  $f(m^{(0)})=f(m^{(1)})$.
    \item Otherwise, both sides of \cref{eq:equivalence_function_again} are equal to $\bot$.
\end{itemize}
In either case, \cref{eq:equivalence_function_again} holds. This completes the proof of \cref{lem:hybrid_two}.
\end{proof}

\section{Bounded Collusion-Resistant Functional Encryption with Certified Everlasting Deletion}\label{sec:bounded_FE}

\subsection{Definitions}\label{sec:bounded_FE_definitions}

We also require verification correctness with QOTP for $q$-bounded certified everlasting simulation-secure FE because we need it for the construction of certified everlasting secure FE in \cref{sec:const_fe_adapt}.

\begin{definition}[Verification Correctness with QOTP]\label{def:ver_correctness_QOTP_cert_ever_func}
There exists a negligible function $\negl$ and a PPT algorithm $\Modify$ such that for any $\secp\in \N$ and $m\in\Ms$,
\begin{align}
\Pr\left[
\Vrfy(\vk,\cert^*)=\bot
\ \middle |
\begin{array}{ll}
(\MPK,\MSK)\la\Setup(1^\secp,1^q)\\
(\vk,\qct) \lrun \qEnc(\MPK,m)\\
a,b\la\bit^{p(\lambda)}\\
\wtl{\cert} \lrun \qDelete(Z^bX^a\ct X^aZ^b)\\
\cert^*\lrun \Modify(a,b,\wtl{\cert})  
\end{array}
\right] 
\leq
\negl(\secp).
\end{align}
\end{definition}

Another is an adaptively simulation-based security notion in the bounded collusion-resistant setting. The other is a non-adaptively simulation-based security notion in the bounded collusion-resistant setting.
We consider only bounded collusion-resistance in the simulation-based definitions because achieving simulation-based security is impossible in the collusion-resistant setting~\cite{C:AGVW13}.

Our simulation-based security notion is a natural extension of that in the classical FE setting~\cite{C:GorVaiWee12}.
Note that the setup algorithm additionally takes $1^q$ as input in the bounded collusion-resistant setting where $q$ is the total number of key queries.
\begin{definition}[$q$-Bounded Certified Everlasting Simulation-Security]\label{def:FE_CED_SIM_sec}
Let $q$ be a polynomial of $\lambda$.
Let $\Sigma=(\Setup,\keygen,\qEnc,\qDec,\qDelete,\Vrfy)$ be a $q$-bounded FE with certified everlasting deletion scheme.
We consider the following security experiment $\expd{\Sigma,\qA}{cert}{ever}{ada}{sim}(\secp,b)$ against a QPT adversary $\qA_1$ and an unbounded adversary $\qA_2$.
Let $\qSim_1$, $\qSim_2$, and $\qSim_3$ be a QPT algorithm.
\begin{enumerate}
    \item The challenger runs $(\MPK,\MSK)\la \Setup(1^\secp,1^q)$ and sends $\MPK$ to $\qA_1$.
    \item $\qA_1$ is allowed to make arbitrary key queries.
    For the $\ell$-th key query, the challenger receives $f_\ell\in\mathcal{F}$, 
    computes $\sk_{f_\ell}\la\keygen(\MSK,f_\ell)$
    and sends $\sk_{f_\ell}$ to $\qA_1$.
    Let $q_{\sfpre}$ be the number of times that $\qA_1$ makes key queries in this phase.
    Let $\mathcal{V}\seteq \{y_i\seteq f_i(m),f_i,\sk_{f_i}\}_{i\in[q_{\sfpre}]}$.
    \item $\qA_1$ chooses $m\in\Ms$ and sends $m$ to the challenger.
    \item The experiment works as follows:
    \begin{itemize}
        \item If $b=0$, the challenger computes $(\vk,\qct)\la \qEnc(\MPK,m)$, and sends $\qct$ to $\qA_1$.
        \item If $b=1$, the challenger computes $(\qct,\state_{q_{\sfpre}})\la \qSim_1(\MPK,\mathcal{V},1^{|m|})$, and sends $\qct$ to $\qA_1$, where $\state_{q_{\sfpre}}$ is a quantum state.
    \end{itemize}
    \item $\qA_1$ is allowed to make arbitrary key queries at most $q-q_{\sfpre}$ times.
    For the $\ell$-th key query, the challenger works as follows.
    \begin{itemize}
    \item If $b=0$, the challenger receives $f_\ell\in\mathcal{F}$, 
    computes $\sk_{f_\ell}\la\keygen(\MSK,f_\ell)$,
    and sends $\sk_{f_\ell}$ to $\qA_1$.
    \item If $b=1$, the challenger receives $f_\ell\in\mathcal{F}$, computes $(\sk_{f_{\ell}},\state_{\ell})\la\qSim_2(\MSK,f_\ell,f_\ell(m),\state_{\ell-1})$, and sends $\sk_{f_\ell}$ to $\qA_1$,
    where $\state_{\ell}$ is a quantum state.
    \end{itemize}
    \item If $b=1$, the challenger runs $\vk\la\qSim_3(\state_{q})$.
    \item At some point, $\qA_1$ sends $\cert$ to the challenger and its internal state to $\qA_2$.
    \item The challenger computes $\Vrfy(\vk,\cert)$. If the output is $\top$, then the challenger outputs $\top$, and sends $\MSK$ to $\qA_2$.
    Otherwise, the challenger outputs $\bot$, and sends $\bot$ to $\qA_2$. 
    \item $\qA_2$ outputs $b'\in\bit$. If the challenger outputs $\top$, the output of the experiment is $b'$. Otherwise, the output of the experiment is $\bot$.
\end{enumerate}
We say that $\Sigma$ is $q$-bounded adaptively certified everlasting simulation-secure if there exists a $QPT$ simulator $\qSim=(\qSim_1,\qSim_2,\qSim_3)$ such that for any $QPT$ adversary $\qA_1$ and any unbounded adversary $\qA_2$ it holds that
\begin{align}
\advd{\Sigma,\qA}{cert}{ever}{ada}{sim}(\secp) \seteq \abs{\Pr[\expd{\Sigma,\qA}{cert}{ever}{ada}{sim}(\secp,0)=1]  - \Pr[\expd{\Sigma,\qA}{cert}{ever}{ada}{sim}(\secp,1)=1]} \leq \negl(\secp).
\end{align}
\end{definition}

\begin{remark}
Note that~\cref{def:ever_non_ada_security_cert_ever_func_sim,def:FE_CED_SIM_sec} were presented before the work by Bartusek and Khurana~\cite{myC:BarKhu23} appeared. Although we can define simulation-based definitions based on the definitions by Bartusek and Khurana, we leave our original simulation-based definitions as a concurrent and independent work. We also note that the challenger can omit sending $\MSK$ to $\qA_2$ in~\cref{def:FE_CED_SIM_sec,def:ever_non_ada_security_cert_ever_func_sim} in the public-key setting based on the results by Bartusek and Khurana~\cite[Claim A.3 and A.4]{myC:BarKhu23}.
\end{remark}

We can consider a non-adaptive variant of the definition above.
\begin{definition}[$q$-Bounded Non-Adaptive Certified Everlasting Simulation-Security]\label{def:ever_non_ada_security_cert_ever_func_sim}
Let $q$ be a polynomial of $\lambda$.
Let $\Sigma=(\Setup,\keygen,\qEnc,\qDec,\qDelete,\Vrfy)$ be a $q$-bounded FE with certified everlasting deletion scheme.
We consider the following security experiment $\expd{\Sigma,\qA}{cert}{ever}{noada}{sim}(\secp,b)$ against a QPT adversary $\qA_1$ and an unbounded adversary $\qA_2$. Let $\qSim$ be a QPT algorithm.
\begin{enumerate}
    \item The challenger runs $(\MPK,\MSK)\la \Setup(1^\secp)$ and sends $\MPK$ to $\qA_1$.
    \item $\qA_1$ is allowed to make arbitrary key queries.
    For the $\ell$-th key query, the challenger receives $f_\ell\in\mathcal{F}$, 
    computes $\sk_{f_\ell}\la\keygen(\MSK,f_\ell)$
    and sends $\sk_{f_\ell}$ to $\qA_1$.
    Let $q$ be the total number of times that $\qA_1$ makes key queries.
    Let $\mathcal{V}\seteq \{y_i\seteq f_i(m),f_i,\sk_{f_i}\}_{i\in[q]}$.
    \item $\qA_1$ chooses $m\in\Ms$ and sends $m$ to the challenger.
    \item The experiment works as follows:
    \begin{itemize}
        \item If $b=0$, the challenger computes $(\vk,\qct)\la \qEnc(\MPK,m)$, and sends $\qct$ to $\qA_1$.
        \item If $b=1$, the challenger computes $(\vk,\qct)\la \qSim(\MPK,\mathcal{V},1^{|m|})$, and sends $\qct$ to $\qA_1$.
    \end{itemize}
    \item At some point, $\qA_1$ sends $\cert$ to the challenger and its internal state to $\qA_2$.
    \item The challenger computes $\Vrfy(\vk,\cert)$.
    If the output is $\top$, then the challenger outputs $\top$, and sends $\MSK$ to $\qA_2$.
    Otherwise, the challenger outputs $\bot$, and sends $\bot$ to $\qA_2$.
    \item $\qA_2$ outputs $b'\in\bit$. If the challenger outputs $\top$, the output of the experiment is $b'$. Otherwise, the output of the experiment is $\bot$.
\end{enumerate}
We say that $\Sigma$ is $q$-bounded non-adaptive certified everlasting simulation-secure if there exists a $QPT$ simulator $\qSim$ such that for any $QPT$ adversary $\qA_1$ and any unbounded adversary $\qA_2$ it holds that
\begin{align}
\advd{\Sigma,\qA}{cert}{ever}{noada}{sim}(\secp) \seteq \abs{\Pr[\expd{\Sigma,\qA}{cert}{ever}{noada}{sim}(\secp,0)=1]  - \Pr[\expd{\Sigma,\qA}{cert}{ever}{noada}{sim}(\secp,1)=1]} 
\leq \negl(\secp).
\end{align}

\end{definition}

We need to consider standard simulation-security notions for FE with certified everlasting deletion.
We note that the following two security definitions are simulation-based ones defined in \cite{C:GorVaiWee12}.

\begin{definition}[$q$-Bounded Non-Adaptive Simulation-Security for FE with Certified Everlasting Deletion~\cite{C:GorVaiWee12}]\label{def:non_ada_security_cert_ever_func_sim}
Let $q$ be a polynomial of $\lambda$.
Let $\Sigma=(\Setup,\keygen,\qEnc,\qDec,\qDelete,\Vrfy)$ be a $q$-bounded FE with certified everlasting deletion scheme.
We consider the following security experiment $\expc{\Sigma,\qA}{non}{ada}{sim}(\secp,b)$ against a QPT adversary $\qA$.
Let $\qSim$ be a QPT algorithm.
\begin{enumerate}
    \item The challenger runs $(\MPK,\MSK)\la \Setup(1^\secp,1^q)$ and sends $\MPK$ to $\qA$.
    \item $\qA$ is allowed to make arbitrary key queries.
    For the $\ell$-th key query, the challenger receives $f_\ell\in\mathcal{F}$, 
    computes $\sk_{f_\ell}\la\keygen(\MSK,f_\ell)$,
    and sends $\sk_{f_\ell}$ to $\qA$.
    Let $q$ be the total number of times that $\qA$ makes key queries.
    Let $\mathcal{V}\seteq \{y_i\seteq f_i(m),f_i,\sk_{f_i}\}_{i\in[q]}$.
    \item $\qA$ chooses $m\in\Ms$ and sends $m$ to the challenger.
    \item The experiment works as follows:
    \begin{itemize}
        \item If $b=0$, the challenger computes $(\vk,\qct)\la \qEnc(\MPK,m)$, and sends $\qct$ to $\qA$.
        \item If $b=1$, the challenger computes $\qct\la \qSim(\MPK,\mathcal{V},1^{|m|})$, and sends $\qct$ to $\qA$.
    \end{itemize}
    \item $\qA$ outputs $b'\in\bit$. The output of the experiment is $b'$.
\end{enumerate}
We say that $\Sigma$ is $q$-bounded non-adaptive secure if there exists a $QPT$ simulator $\qSim$ such that for any $QPT$ adversary $\qA$ it holds that
\begin{align}
\advc{\Sigma,\qA}{non}{ada}{sim}(\secp) \seteq \abs{\Pr[\expc{\Sigma,\qA}{non}{ada}{sim}(\secp,0)=1]  - \Pr[\expc{\Sigma,\qA}{non}{ada}{sim}(\secp,1)=1]} \leq \negl(\secp).
\end{align}
\end{definition}

\begin{definition}[$q$-Bounded Adaptive Simulation-Security for FE with Certified Everlasting Deletion~\cite{C:GorVaiWee12}]\label{def:ada_security_cert_ever_func_sim}
Let $q$ be a polynomial of $\lambda$.
Let $\Sigma=(\Setup,\keygen,\qEnc,\qDec,\qDelete,\Vrfy)$ be a $q$-bounded FE with certified everlasting deletion scheme.
We consider the following security experiment $\expb{\Sigma,\qA}{ada}{sim}(\secp,b)$ against a QPT adversary $\qA$.
Let $\qSim_1$ and $\qSim_2$ be a QPT algorithm.
\begin{enumerate}
    \item The challenger runs $(\MPK,\MSK)\la \Setup(1^\secp,1^q)$ and sends $\MPK$ to $\qA$.
    \item $\qA$ is allowed to make arbitrary key queries.
    For the $\ell$-th key query, the challenger receives $f_\ell\in\mathcal{F}$, 
    computes $\sk_{f_\ell}\la\keygen(\MSK,f_\ell)$,
    and sends $\sk_{f_\ell}$ to $\qA$.
    Let $q_{\sfpre}$ be the number of times that $\qA$ makes key queries in this phase.
    Let $\mathcal{V}\seteq \{y_i\seteq f_i(m),f_i,\sk_{f_i}\}_{i\in[q_{\sfpre}]}$.
    \item $\qA$ chooses $m\in\Ms$ and sends $m$ to the challenger.
    \item The experiment works as follows:
    \begin{itemize}
        \item If $b=0$, the challenger computes $(\vk,\qct)\la \qEnc(\MPK,m)$, and sends $\qct$ to $\qA$.
        \item If $b=1$, the challenger computes $(\qct,\state_{q_{\sfpre}})\la \qSim_1(\MPK,\mathcal{V},1^{|m|})$, and sends $\qct$ to $\qA$, where $\state_{q_{\sfpre}}$ is a quantum state.
    \end{itemize}
    \item $\qA$ is allowed to make arbitrary key queries at most $(q-q_{\sfpre})$ times.
    For the $\ell$-th key query, the challenger works as follows:
    \begin{itemize}
        \item If $b=0$, the challenger receives $f_\ell\in\mathcal{F}$, 
    computes $\sk_{f_\ell}\la\keygen(\MSK,f_\ell)$,
    and sends $\sk_{f_\ell}$ to $\qA$.
    \item If $b=1$, the challenger receives $f_\ell\in\mathcal{F}$, computes $(\sk_{f_{\ell}},\state_{\ell})\la\qSim_2(\MSK,f_{\ell},f_\ell(m),\state_{\ell-1})$, and sends $\sk_{f_\ell}$ to $\qA$.
    \end{itemize}
    \item $\qA$ outputs $b'\in\bit$. The output of the experiment is $b'$.
\end{enumerate}
We say that $\Sigma$ is $q$-bounded adaptive simulation-secure if there exists a $QPT$ simulator $\qSim=(\qSim_1,\qSim_2)$ such that for any $QPT$ adversary $\qA$ it holds that
\begin{align}
\advb{\Sigma,\qA}{ada}{sim}(\secp) \seteq \abs{\Pr[\expb{\Sigma,\qA}{ada}{sim}(\secp,0)=1]  - \Pr[\expb{\Sigma,\qA}{ada}{sim}(\secp,1)=1]} \leq \negl(\secp).
\end{align}
\end{definition}

\subsection{\texorpdfstring{$1$}{1}-Bounded Construction with Non-Adaptive Security}\label{sec:const_func_non_adapt}
To achieve $q$-bounded adaptive certified everlasting simulation-secure FE in~\cref{sec:const_multi_fe}, we prepare building blocks in this section and~\cref{sec:const_fe_adapt}.
In this section, we construct a $1$-bounded non-adaptive certified everlasting simulation-secure FE scheme from a certified everlasting secure garbling scheme~(\cref{def:cert_ever_garbled}) and a certified everlasting secure PKE scheme~(\cref{def:cert_ever_pke}). See~\cref{sec:const_garbling,sec:const_pke_rom,sec:const_pke_wo_rom} for how to achieve these building blocks. Regarding PKE, we can also use the construction by Bartusek and Khurana~\cite{myC:BarKhu23}.

\paragraph{Our $1$-bounded non-adaptive certified everlasting secure FE scheme.}
This construction is essentially the same as the $1$-bound FE by Sahai and Seyalioglu~\cite{CCS:SahSey10}. We use a universal circuit $U(\cdot,x)$ in which a plaintext $x$ is hard-wired. The universal circuit takes a function $f$ as input and outputs $f(x)$.
Let $s\seteq |f|$.
We construct a $1$-bounded non-adaptive certified everlasting secure FE scheme $\Sigma_{\mathsf{cefe}}=(\Setup,\keygen,\qEnc,\qDec,\qDelete,\Vrfy)$
from a certified everlasting secure garbling scheme $\Sigma_{\mathsf{cegc}}=\mathsf{GC}.(\Setup,\qGarble,\qEval,\qDelete,\Vrfy)$~(\cref{def:cert_ever_garbled})
and a certified everlasting secure PKE scheme $\Sigma_{\mathsf{cepk}}=\PKE.(\keygen,\qEnc,\qDec,\qDelete,\Vrfy)$~(\cref{def:cert_ever_pke}).

\begin{description}
\item[$\Setup(1^\secp)$:] $ $
\begin{itemize}
    \item Generate $(\pke.\pk_{i,\alpha},\pke.\sk_{i,\alpha})\la\PKE.\keygen(1^\secp)$ for every $i\in [s]$ and $\alpha\in\bit$.
    \item Output $\MPK\seteq \{\pke.\pk_{i,\alpha}\}_{i\in[s],\alpha\in\{0,1\}}$ and $\MSK\seteq \{\pke.\sk_{i,\alpha}\}_{i\in[s],\alpha\in\bit}$.
\end{itemize}
\item[$\keygen(\MSK,f)$:] $ $
\begin{itemize}
    \item Parse $\MSK=\{\pke.\sk_{i,\alpha}\}_{i\in[s],\alpha\in\bit}$ and $f=(f_1,\cdots,f_s)$.
    \item Output $\sk_f\seteq (f,\{\pke.\sk_{i,f[i]}\}_{i\in[s]})$.
\end{itemize}
\item[$\qEnc(\MPK,m)$:]$ $
\begin{itemize}
    \item Parse $\MPK=\{\pke.\pk_{i,\alpha}\}_{i\in[s],\alpha\in\bit}$.
    \item Compute $\{L_{i,\alpha}\}_{i\in[s],\alpha\in\bit}\la\mathsf{GC}.\Setup(1^\secp)$.
    \item Compute $(\qugate,\mathsf{gc}.\vk)\la\mathsf{GC}.\qGarble(1^\secp,U(\cdot,m),\{L_{i,\alpha}\}_{i\in[s],\alpha\in\bit})$.
    \item For every $i\in[s]$ and $\alpha\in\bit$, compute $(\pke.\vk_{i,\alpha},\pke.\qct_{i,\alpha})\la\PKE.\qEnc(\pke.\pk_{i,\alpha},L_{i,\alpha})$.
    \item Output $\vk\seteq (\mathsf{gc}.\vk,\{\pke.\vk_{i,\alpha}\}_{i\in[s],\alpha\in\bit})$
    and $\qct\seteq (\qugate,\{\pke.\qct_{i,\alpha}\}_{i\in[s],\alpha\in\bit})$.
\end{itemize}
\item[$\qDec(\sk_f,\qct)$:] $ $
\begin{itemize}
    \item Parse $\sk_f=(f,\{\pke.\sk_i\}_{i\in[s]})$ and $\qct=(\qugate,\{\pke.\qct_{i,\alpha}\}_{i\in[s],\alpha\in\bit})$.
    \item For every $i\in [s]$, compute $L_i,\la\PKE.\qDec(\pke.\sk_{i},\pke.\qct_{i,f[i]})$.
    \item Compute $y\la \mathsf{GC}.\qEval(\qugate,\{L_i\}_{i\in[s]})$.
    \item Output $y$.
\end{itemize}
\item[$\qDelete(\qct)$:] $ $
\begin{itemize}
    \item Parse $\qct=(\qugate,\{\pke.\qct_{i,\alpha}\}_{i\in[s],\alpha\in\bit})$.
    \item Compute $\mathsf{gc}.\cert\la\mathsf{GC}.\qDelete(\qugate)$.
    \item For every $i\in[s]$ and $\alpha\in\bit$, compute $\pke.\cert_{i,\alpha}\la\PKE.\qDelete(\pke.\ct_{i,\alpha})$.
    \item Output $\cert\seteq (\mathsf{gc}.\cert,\{\pke.\cert_{i,\alpha}\}_{i\in[s],\alpha\in\bit})$.
\end{itemize}
\item[$\Vrfy(\vk,\cert)$:] $ $
\begin{itemize}
    \item Parse $\vk=(\mathsf{gc}.\vk,\{\pke.\vk_{i,\alpha}\}_{i\in[s],\alpha\in\bit})$ and $\cert=(\mathsf{gc}.\cert,\{\pke.\cert_{i,\alpha}\}_{i\in[s],\alpha\in\bit})$.
    \item Output $\top$ if $\top\la\mathsf{GC}.\Vrfy(\mathsf{gc}.\vk,\mathsf{gc}.\cert)$ and $\top\la\PKE.\Vrfy(\pke.\vk_{i,\alpha},\pke.\cert_{i,\alpha})$ for every $i\in[s]$ and $\alpha\in\bit$.
    Otherwise, output $\bot$.
\end{itemize}
\end{description}

\paragraph{Correctness:}
Correctness easily follows from that of $\Sigma_{\mathsf{cegc}}$ and $\Sigma_{\mathsf{cepk}}$.

\paragraph{Security:}
The following two theorems hold.
\begin{theorem}\label{thm:func_non_ad_seucirty}
If $\Sigma_{\mathsf{cegc}}$ satisfies the selective security~(\cref{def:sel_sec_ever_garb}) and $\Sigma_{\mathsf{cepk}}$ satisfies the IND-CPA security~(\cref{def:IND-CPA_security_cert_ever_pke}),
$\Sigma_{\mathsf{cefe}}$ satisfies the $1$-bounded non-adaptive simulation-security~(\cref{def:non_ada_security_cert_ever_func_sim}).
\end{theorem}
Its proof is similar to that of \cref{thm:func_ever_non_ad_security}, and therefore we omit it.
\begin{theorem}\label{thm:func_ever_non_ad_security}
If $\Sigma_{\mathsf{cegc}}$ satisfies the selective certified everlasting security~(\cref{def:ever_sec_ever_garb}) and $\Sigma_{\mathsf{cepk}}$ satisfies the certified everlasting IND-CPA security~(\cref{def:cert_ever_security_cert_ever_pke}),
$\Sigma_{\mathsf{cefe}}$ satisfies the $1$-bounded non-adaptive certified everlasting simulation-security~(\cref{def:ever_non_ada_security_cert_ever_func_sim}).
\end{theorem}

\begin{proof}[Proof of \cref{thm:func_ever_non_ad_security}]
Let us describe how the simulator $\qSim$ works.
\begin{description}
\item[$\qSim(\MPK,\mathcal{V},1^{|m|})$:]$ $
\begin{enumerate}
    \item Parse $\MPK=\{\pke.\pk_{i,\alpha}\}_{i\in[s],\alpha\in\bit}$ and $\mathcal{V}=\{f(m),f,(f,\{\pke.\sk_{i,f[i]}\}_{i\in[s]})\} \,\, or \,\,\emptyset$.
    \item If $\mathcal{V}=\emptyset$, generate $f\la\bit^s$.
    \item Generate $\{L_{i,\alpha}\}_{i\in[s],\alpha\in\bit}\la\mathsf{GC}.\Setup(1^\secp)$ and $L_{i,f[i]\oplus1}^*\la\mathcal{L}$ for every $i\in[s]$.
    \item Compute $(\qugate,\mathsf{gc}.\vk)\la\mathsf{GC}.\qSim(1^\secp,1^{|f|},U(f,m),\{L_{i,f[i]}\}_{i\in[s]})$.
    \item Compute $(\pke.\vk_{i,f[i]},\pke.\qct_{i,f[i]})\la \PKE.\qEnc(\pke.\pk_{i,f[i]},L_{i,f[i]})$ and $(\pke.\vk_{i,f[i]\oplus 1},\pke.\qct_{i,f[i]\oplus 1})\la \PKE.\qEnc(\pke.\pk_{i,f[i]\oplus 1},L^*_{i,f[i]\oplus 1})$ for every $i\in[s]$.
    \item Output $\vk\seteq (\mathsf{gc}.\vk,\{\pke.\vk_{i,\alpha}\}_{i\in[s],\alpha\in\bit})$ and $\qct\seteq (\qugate,\{\pke.\qct_{i,\alpha}\}_{i\in[s],\alpha\in\bit})$.
\end{enumerate}
\end{description}
Let us define the sequence of hybrids as follows.

\begin{description}
\item [$\hyb_{0}$:] This is identical to $\expd{\Sigma_{\mathsf{cefe}},\qA}{cert}{ever}{non}{adapt}(\secp,0)$.
\begin{enumerate}
    \item The challenger generates $(\pke.\pk_{i,\alpha},\pke.\sk_{i,\alpha})\la\PKE.\keygen(1^\secp)$ for every $i\in [s]$ and $\alpha\in\bit$, and sends
    $\{\pke.\pk_{i,\alpha}\}_{i\in[s],\alpha\in\bit}$ to $\qA_1$.
    \item\label{step:function_query_single_fe}
    $\qA_1$ is allowed to call a key query at most one time.
    If a key query is called,
    the challenger receives an function $f$ from $\qA_1$, and sends $(f,\{\pke.\sk_{i,f[i]}\}_{i\in[s]})$ to $\qA_1$.
    \item $\qA_1$ chooses $m\in\Ms$, and sends $m$ to the challenger.
    \item \label{step:encryption_single_non_ad_fe}
    The challenger computes $\{L_{i,\alpha}\}_{i\in[s],\alpha\in\bit}\la\mathsf{GC}.\Setup(1^\secp)$, $(\qugate,\mathsf{gc}.\vk)\la\mathsf{GC}.\qGarble(1^\secp,U(\cdot,m),\allowbreak \{L_{i,\alpha}\}_{i\in[s],\alpha\in\bit})$, and
    $(\pke.\vk_{i,\alpha},\pke.\qct_{i,\alpha})\la\PKE.\qEnc(\pke.\pk_{i,\alpha},L_{i,\alpha})$ for every $i\in[s]$ and $\alpha\in\bit$,
    and sends $(\qugate,\{\pke.\qct_{i,\alpha}\}_{i\in[s],\alpha\in\bit})$ to $\qA_1$.
    \item $\qA_1$ sends $(\mathsf{gc}.\cert,\{\pke.\cert_{i,\alpha}\}_{i\in[s],\alpha\in\bit})$ to the challenger,
    and sends its internal state to $\qA_2$.
    \item If $\top\la \mathsf{GC}.\Vrfy(\mathsf{gc}.\vk,\mathsf{gc}.\cert)$, and $\top\la\PKE.\Vrfy(\pke.\vk_{i,\alpha},\pke.\cert_{i,\alpha})$ for every $i\in[s]$ and $\alpha\in\bit$, the challenger outputs $\top$, and  sends $\{\pke.\sk_{i,\alpha}\}_{i\in[s],\alpha\in\bit}$ to $\qA_2$.
    Otherwise, the challenger outputs $\bot$, and sends $\bot$ to $\qA_2$.
    \item $\qA_2$ outputs $b'$. If the challenger outputs $\top$, the output of the experiment is $b'$.
    Otherwise, the output of the experiment is $\bot$.
\end{enumerate}
\item[$\hyb_{1}$:]This is identical to $\hyb_{0}$ except for the following four points.
First, the challenger generates $f\in\bit^s$ if a key query is not called in step~\ref{step:function_query_single_fe}.
Second, the challenger randomly generates $L^*_{i,f[i]\oplus 1}\la \mathcal{L}$ for every $i\in[s]$
and $\{L_{i,\alpha}\}_{i\in[s],\alpha\in\bit}\la\mathsf{GC}.\Setup(1^\secp)$ in step~\ref{step:function_query_single_fe} regardless of whether a key query is called or not.
Third, the challenger does not compute $\{L_{i,\alpha}\}_{i\in[s],\alpha\in\bit}\la\mathsf{GC}.\Setup(1^\secp)$ in step~\ref{step:encryption_single_non_ad_fe}.
Fourth, 
the challenger computes 
$(\pke.\vk_{i,f[i]\oplus 1},\pke.\qct_{i,f[i]\oplus 1})\la\PKE.\qEnc(\pke.\pk_{i,f[i]\oplus 1},L^*_{i,f[i]\oplus 1})$ for every $i\in[s]$ instead of computing 
$(\pke.\vk_{i,f[i]\oplus 1},\allowbreak\pke.\qct_{i,f[i]\oplus 1})\la\PKE.\qEnc(\pke.\pk_{i,f[i]\oplus 1},L_{i,f[i]\oplus 1})$ for every $i\in[s]$.

\item[$\hyb_{2}$:]This is identical to $\hyb_{1}$ except for the following point.
The challenger computes $(\qugate,\mathsf{gc}.\vk)\la\GC.\qSim(1^\secp,1^{|f|},\allowbreak U(f,m), \{L_{i,f[i]}\}_{i\in[s]})$
instead of computing $(\qugate,\mathsf{gc}.\vk)\la\mathsf{GC}.\qGarble(1^\secp,U(\cdot,m),\{L_{i,\alpha}\}_{i\in[s],\alpha\in\bit})$.
\end{description}
From the definition of $\expd{\Sigma_{\mathsf{cefe}},\qA}{cert}{ever}{non}{adapt}(\secp, b)$ and $\qSim$,
it is clear that $\Pr[\hyb_{0}=1]=\Pr[\expd{\Sigma_{\mathsf{cefe}},\qA}{cert}{ever}{non}{adapt}(\secp,0)=1]$ and
$\Pr[\hyb_{2}=1]=\Pr[\expd{\Sigma_{\mathsf{cefe}},\qA}{cert}{ever}{non}{adapt}(\secp,1)=1]$.
Therefore, \cref{thm:func_ever_non_ad_security} easily follows from the following \cref{prop:Exp_Hyb_1_non_ad_fe,prop:Hyb_1_Hyb_2_non_ad_fe}. (whose proof is given later.)
\end{proof}

\begin{proposition}\label{prop:Exp_Hyb_1_non_ad_fe}
If $\Sigma_{\mathsf{cepk}}$ satisfies the certified everlasting IND-CPA security,
\begin{align}
    \abs{\Pr[\hyb_{0}=1]-\Pr[\hyb_{1}=1]}\leq\negl(\lambda).
\end{align}
\end{proposition}
\begin{proposition}\label{prop:Hyb_1_Hyb_2_non_ad_fe}
If $\Sigma_{\mathsf{cegc}}$ satisfies the certified everlasting selective security,
\begin{align}
    \abs{\Pr[\hyb_{1}=1]-\Pr[\hyb_{2}=1]}\leq\negl(\lambda).
\end{align}
\end{proposition}

\begin{proof}[Proof of \cref{prop:Exp_Hyb_1_non_ad_fe}]

For the proof, we use \cref{lem:cut_and_choose_pke} whose statement and proof is given in \cref{sec:const_rnce_classic}.
We assume that $\abs{\Pr[\hyb_{0}=1]-\Pr[\hyb_{1}=1]}$ is non-negligible, and construct an adversary $\qB$ that breaks the security experiment of  $\expc{\Sigma_{\mathsf{cepk}},\qB}{multi}{cert}{ever}(\secp,b)$ defined in \cref{lem:cut_and_choose_pke}.
This contradicts the certified everlasting IND-CPA of $\Sigma_{\mathsf{cepk}}$ from \cref{lem:cut_and_choose_pke}.
Let us describe how $\qB$ works below.
\begin{enumerate}
    \item $\qB$ receives $\{\pke.\pk_{i,\alpha}\}_{i\in[s],\alpha\in\bit}$ from the challenger of $\expc{\Sigma_{\mathsf{cepk}},\qB}{multi}{cert}{ever}(\secp,b)$,
    and sends $\{\pke.\pk_{i,\alpha}\}_{i\in[s],\alpha\in\bit}$ to $\qA_1$.
    \item $\qA_1$ is allowed to call a key query at most one time.
    If a key query is called,
    $\qB$ receives an function $f$ from $\qA_1$, generates $L^*_{i,f[i]\oplus 1}\la \mathcal{L}$ for every $i\in[s]$ and  $\{L_{i,\alpha}\}_{i\in[s],\alpha\in\bit}\la\mathsf{GC}.\Setup(1^\secp)$.
    If a key query is not called, $\qB$ generates $f\la\bit^s$, $L^*_{i,f[i]\oplus 1}\la \mathcal{L}$ for every $i\in[s]$ and  $\{L_{i,\alpha}\}_{i\in[s],\alpha\in\bit}\la\mathsf{GC}.\Setup(1^\secp)$. 
    \item $\qB$ sends $(f,L_{1,f[1]\oplus 1},L_{2,f[2]\oplus 1},\cdots ,L_{s,f[s]\oplus 1},L^*_{1,f[1]\oplus 1},L^*_{2,f[2]\oplus 1},\cdots, L^{*}_{s,f[s]\oplus 1})$ to the challenger of $\expc{\Sigma_{\mathsf{cepk}},\qB}{multi}{cert}{ever}(\secp,b)$.
    \item $\qB$ receives $(\{\pke.\sk_{i,f[i]}\}_{i\in[s]},\{\pke.\qct_{i,f[i]\oplus 1}\}_{i\in[s]})$ from the challenger.
    If a key query is called, $\qB$ sends $(f,\{\pke.\sk_{i,f[i]}\}_{i\in[s]})$ to $\qA_1$.
    \item $\qA_1$ chooses $m\in\Ms$, and sends $m$ to $\qB$.
    \item $\qB$ computes
    $(\qugate,\mathsf{gc}.\vk)\la\mathsf{GC}.\qGarble(1^\secp,U(\cdot,m),\{L_{i,\alpha}\}_{i\in[s],\alpha\in\bit})$ and
    $(\pke.\vk_{i,f[i]},\pke.\qct_{i,f[i]})\la\PKE.\qEnc(\allowbreak \pke.\pk_{i,f[i]},\allowbreak L_{i,f[i]})$ for every $i\in[s]$, and sends $(\qugate,\{\pke.\qct_{i,\alpha}\}_{i\in[s],\alpha\in\bit})$ to $\qA_1$.
    \item $\qA_1$ sends $(\mathsf{gc}.\cert,\{\pke.\cert_{i,\alpha}\}_{i\in[s],\alpha\in\bit})$ to $\qB$, and sends its internal state to $\qA_2$.
    \item $\qB$ sends $\{\pke.\cert_{i,f[i]\oplus1}\}_{i\in[s]}$ to the challenger, and receives $\{\pke.\sk_{i,f[i]\oplus 1}\}_{i\in[s]}$ or $\bot$ from the challenger. If $\qB$ receives $\bot$ from the challenger, it outputs $\bot$ and aborts. 
    \item $\qB$ sends $\{\pke.\sk_{i,\alpha}\}_{i\in[s],\alpha\in\bit}$ to $\qA_2$.
    \item $\qA_2$ outputs $b'$.
    \item $\qB$ computes $\mathsf{GC}.\Vrfy$ for $\gc.\cert$ and $\PKE.\Vrfy$ for all $\{\pke.\cert_{i,f[i]}\}_{i\in[s]}$, and outputs $b'$ if all results are $\top$.
    Otherwise, $\qB$ outputs $\bot$. 
\end{enumerate}
It is clear that $\Pr[1\la\qB \mid b=0]=\Pr[\hyb_{0}=1]$ and 
$\Pr[1\la\qB \mid b=1]=\Pr[\hyb_{1}=1]$.
By assumption, $\allowbreak \abs{\Pr[\hyb_{0}=1]- \Pr[\hyb_{1}=1]}$ is non-negligible, and therefore $\abs{\Pr[1\la\qB \mid b=0]-\Pr[1\la\qB \mid b=1]}$ is non-negligible, which contradicts the certified everlasting IND-CPA security of $\Sigma_{\mathsf{cepk}}$ from \cref{lem:cut_and_choose_pke}.
\end{proof}

\begin{proof}[Proof of \cref{prop:Hyb_1_Hyb_2_non_ad_fe}]
We assume that $\abs{\Pr[\hyb_{1}=1]-\Pr[\hyb_{2}=1]}$ is non-negligible, and construct an adversary $\qB$ that breaks the selective certified everlasting security of $\Sigma_{\mathsf{cegc}}$.
Let us describe how $\qB$ works below.
\begin{enumerate}
    \item $\qB$ generates $(\pke.\pk_{i,\alpha},\pke.\sk_{i,\alpha})\la\PKE.\keygen(1^\secp)$ for every $i\in [s]$ and $\alpha\in\bit$, and sends
    $\{\pke.\pk_{i,\alpha}\}_{i\in[s],\alpha\in\bit}$ to $\qA_1$.
    \item $\qA_1$ is allowed to call a key query at most one time.
    If a key query is called, $\qB$ receives an function $f$ from $\qA_1$, generates $L^*_{i,f[i]\oplus 1}\la \mathcal{L}$ for every $i\in[s]$, and sends $(f,\{\pke.\sk_{i,f[i]}\}_{i\in[s]})$ to $\qA_1$.
    If a key query is not called, $\qB$ generates $f\la\bit^s$ and $L^*_{i,f[i]\oplus 1}\la \mathcal{L}$ for every $i\in[s]$.
    \item $\qA_1$ chooses $m\in\Ms$, and sends $m$ to $\qB$.
    \item $\qB$ sends a circuit $U(\cdot,m)$ and an input $f\in\bit^s$ to the challenger of $\expd{\qB,\Sigma_{\mathsf{cegc}}}{cert}{ever}{sel}{gbl}(1^\secp,b)$.
    \item The challenger computes $\{L_{i,\alpha}\}_{i\in[s],\alpha\in\bit}\la\mathsf{GC}.\Setup(1^\secp)$ and does the following:
    \begin{itemize}
        \item If $b=0$, the challenger computes $(\qugate,\mathsf{gc}.\vk)\la\mathsf{GC}.\qGarble(1^\secp,U(\cdot,m),\{L_{i,\alpha}\}_{i\in[s],\alpha\in\bit})$,
        and sends $(\qugate,\{L_{i,f[i]}\}_{i\in[s]})$ to $\qB$.
        \item If $b=1$, the challenger computes $(\qugate,\mathsf{gc}.\vk)\la\GC.\qSim(1^\secp,1^{|f|},U(f,m),\{L_{i,f[i]}\}_{i\in[s]})$,
        and sends $(\qugate,\{L_{i,f[i]}\}_{i\in[s]})$ to $\qB$.
    \end{itemize}
    \item $\qB$ computes $(\pke.\vk_{i,f[i]},\pke.\qct_{i,f[i]})\la\PKE.\qEnc(\pke.\pk_{i,f[i]},L_{i,f[i]})$ and \allowbreak$(\pke.\vk_{i,f[i]\oplus 1},\pke.\qct_{i,f[i]\oplus 1})\la\PKE.\qEnc(\pke.\pk_{i,f[i]\oplus 1},L^*_{i,f[i]\oplus 1}) $ for every $i\in[s]$.
    \item $\qB$ sends $(\qugate,\{\pke.\qct_{i,\alpha}\}_{i\in[s],\alpha\in\bit})$ to $\qA_1$.
    \item $\qA_1$ sends $(\mathsf{gc}.\cert,\{\pke.\cert_{i,\alpha}\}_{i\in[s],\alpha\in\bit})$ to the challenger, and sends its internal state to $\qA_2$.
    \item $\qB$ sends $\mathsf{gc}.\cert$ to the challenger, and receives $\top$ or $\bot$ from the challenger.
    If $\qB$ receives $\bot$ from the challenger, it outputs $\bot$ and aborts.
    \item $\qB$ sends $\{\pke.\sk_{i,\alpha}\}_{i\in[s],\alpha\in\bit}$ to $\qA_2$.
    \item $\qA_2$ outputs $b'$.
    \item $\qB$ computes $\PKE.\Vrfy$ for all $\pke.\cert_{i,\alpha}$, and outputs $b'$ if all results are $\top$.
    Otherwise, $\qB$ outputs $\bot$.
\end{enumerate}
It is clear that $\Pr[1\la\qB \mid b=0]=\Pr[\hyb_{1}=1]$ and $\Pr[1\la\qB \mid b=1]=\Pr[\hyb_{2}=1]$.
By assumption, $\abs{\Pr[\hyb_{1}=1]-\Pr[\hyb_{2}=1]}$ is non-negligible, and therefore $\abs{\Pr[1\la\qB \mid b=0]-\Pr[1\la\qB \mid b=1]}$ is non-negligible,
which contradicts the selective certified everlasting security of $\Sigma_{\mathsf{cegc}}$. 
\end{proof}

\subsection{\texorpdfstring{$1$}{1}-Bounded Construction with Adaptive Security}\label{sec:const_fe_adapt}
In this section, we convert the non-adaptive scheme constructed in the previous subsection to the adaptive one by using
a certified everlasting secure RNC scheme~(\cref{def:cert_ever_rnce_classic}).
See~\cref{sec:const_rnce_classic} for how to achieve this building block.

\paragraph{Our $1$-bounded adaptive certified everlasting secure FE scheme.}
We construct a $1$-bounded adaptive certified everlasting secure FE scheme $\Sigma_{\mathsf{cefe}}=(\Setup,\keygen,\qEnc,\qDec,\qDelete,\Vrfy)$ 
from a $1$-bounded non-adaptive certified everlasting secure FE scheme $\Sigma_{\mathsf{nad}}=\mathsf{NAD}.(\Setup,\keygen,\qEnc,\qDec,\qDelete,\Vrfy)$,
where the ciphertext space is $\Cs\seteq\cQ^{\otimes n}$,
and a certified everlasting secure RNCE scheme $\Sigma_{\mathsf{cence}}=\NCE.(\Setup,\keygen,\qEnc,\qDec,\qFake,\Reveal,\allowbreak \qDelete,\Vrfy)$ (\cref{def:cert_ever_rnce_classic}).
Let $\NAD.\Modify$ be a QPT algorithm such that
\begin{align}
\Pr\left[
\NAD.\Vrfy(\nad.\vk,\nad.\cert^*)\neq\top
\ \middle |
\begin{array}{ll}
(\nad.\MPK,\nad.\MSK)\lrun \NAD.\Setup(1^\secp)\\
(\nad.\vk,\nad.\qct) \lrun \NAD.\qEnc(\nad.\MPK,m)\\
a,c\la\bit^{n}\\
\nad.\wtl{\cert} \lrun \NAD.\qDelete(Z^cX^a\nad.\qct X^aZ^c)\\
\nad.\cert^*\lrun \NAD.\Modify(a,c,\nad.\wtl{\cert}) 
\end{array}
\right] 
\leq
\negl(\secp).
\end{align}
for any $m$.

Our construction is as follows.
\begin{description}
\item[$\Setup(1^\secp)$:]$ $
\begin{itemize}
    \item Run $(\mathsf{nad}.\MPK,\mathsf{nad}.\MSK)\la\mathsf{NAD}.\Setup(1^\secp)$.
    \item Run $(\nce.\pk,\nce.\MSK)\la\NCE.\Setup(1^{\secp})$.
    \item Output $\MPK\seteq(\mathsf{nad}.\MPK,\nce.\pk)$ and $\MSK\seteq (\nad.\MSK,\nce.\MSK)$.
\end{itemize}
\item[$\keygen(\MSK,f)$:]$ $
\begin{itemize}
    \item Parse $\MSK=(\nad.\MSK,\nce.\MSK)$.
    \item Compute $\mathsf{nad}.\sk_f\la\mathsf{NAD}.\keygen(\mathsf{nad}.\MSK,f)$.
    \item Compute $\nce.\sk\la\NCE.\keygen(\nce.\MSK)$.
    \item Output $\sk_f\seteq (\nad.\sk_f,\nce.\sk)$.
\end{itemize}
\item[$\qEnc(\MPK,m)$:]$ $
\begin{itemize}
    \item Parse $\MPK=(\nad.\MPK,\nce.\pk)$.
    \item Compute $(\mathsf{nad}.\vk,\mathsf{nad}.\qct)\la\mathsf{NAD}.\qEnc(\mathsf{nad}.\MPK,m)$.
    \item Generate $a,c\la\bit^n$. Let $\Psi\seteq Z^cX^a\mathsf{nad}.\qct X^aZ^c$.
    \item Compute $(\nce.\vk,\nce.\qct)\la\NCE.\qEnc(\nce.\pk,(a,c))$.
    \item Output $\vk\seteq(\mathsf{nad}.\vk,\nce.\vk,a,c)$ and $\qct\seteq (\Psi,\nce.\qct)$.
\end{itemize}
\item[$\qDec(\sk_f,\qct)$:]$ $
\begin{itemize}
    \item Parse $\sk_f=(\nad.\sk_f,\nce.\sk)$ and $\qct=(\Psi,\nce.\qct)$.
    \item Compute $(a',c')\la\NCE.\qDec(\nce.\sk,\nce.\qct)$.
    \item Compute $\mathsf{nad}.\qct'\seteq X^{a'}Z^{c'}\Psi Z^{c'}X^{a'}$.
    \item Compute $y\la\mathsf{NAD}.\qDec(\mathsf{nad}.\sk_f,\mathsf{nad}.\qct')$.
    \item Output $y$.
\end{itemize}
\item[$\qDelete(\qct)$:]$ $
\begin{itemize}
    \item Parse $\qct=(\Psi,\nce.\qct)$.
    \item Compute $\nad.\wtl{\cert}\la \NAD.\qDelete(\Psi)$.
    \item Compute $\nce.\cert\la\NCE.\qDelete(\nce.\qct)$.
    \item Output $\cert\seteq(\nad.\wtl{\cert},\nce.\cert)$.
\end{itemize}
\item[$\Vrfy(\vk,\cert)$:]$ $
\begin{itemize}
    \item Parse $\vk=(\mathsf{nad}.\vk,\nce.\vk,a,c)$ and $\cert=(\nad.\wtl{\cert},\nce.\cert)$.
    \item Compute $\mathsf{nad}.\cert^*\la\NAD.\Modify (a,c,\nad.\wtl{\cert})$.
    \item Output $\top$ if $\top\la\NCE.\Vrfy(\nce.\vk,\nce.\cert)$ and $\top\la\mathsf{NAD}.\Vrfy(\mathsf{nad}.\vk,\mathsf{nad}.\cert^*)$.
    Otherwise, output $\bot$.
\end{itemize}
\end{description}
\paragraph{Correctness:}
Correctness easily follows from that of $\Sigma_{\mathsf{nad}}$ and $\Sigma_{\mathsf{cence}}$.

\paragraph{Security:}
The following two theorems hold.
\begin{theorem}\label{thm:comp_security_single_ad}
If $\Sigma_{\mathsf{nad}}$ satisfies the $1$-bounded non-adaptive simulation-security~(\cref{def:non_ada_security_cert_ever_func_sim}) and $\Sigma_{\mathsf{cence}}$ satisfies the RNC security~(\cref{def:rec_nc_security_classic}), 
$\Sigma_{\mathsf{cefe}}$ satisfies the $1$-bounded adaptive simulation-security~(\cref{def:ada_security_cert_ever_func_sim}).
\end{theorem}
Its proof is similar to that of \cref{thm:ever_security_single_ad}, and therefore we omit it.

\begin{theorem}\label{thm:ever_security_single_ad}
If $\Sigma_{\mathsf{nad}}$ satisfies the $1$-bounded non-adaptive certified everlasting simulation-security(~\cref{def:ever_non_ada_security_cert_ever_func_sim}) and $\Sigma_{\mathsf{cence}}$ satisfies the certified everlasting RNC security~(\cref{def:cert_ever_rec_nc_security_classic}), 
$\Sigma_{\mathsf{cefe}}$ satisfies the $1$-bounded adaptive certified everlasting simulation-security~(\cref{def:FE_CED_SIM_sec}).
\end{theorem}

\begin{proof}[Proof of \cref{thm:ever_security_single_ad}]
For a given $2n$-qubit, let $A$ be the $n$-qubit of the first half of the $2n$-qubit, and let $B$ be the $n$-qubit of the second half of the $2n$-qubit.
Let $\NAD.\qSim$ be the simulating algorithm of the ciphertext $\nad.\qct$ .
Let us describe how the simulator $\qSim=(\qSim_1,\qSim_2,\qSim_3)$ works below.
\begin{description}
\item[$\qSim_1(\MPK,\mathcal{V},1^{|m|})$:]$ $
\begin{enumerate}
\item Parse $\MPK=(\nad.\MPK,\nce.\pk)$ and $\mathcal{V}=(f,f(m),(\nad.\sk_f,\nce.\sk))\,\, or\,\, \emptyset$.
\footnote{
If an adversary calls a key query before the adversary receives a challenge ciphertext,
then $\mathcal{V}=(f,f(m),(\nad.\sk_f,\nce.\sk))$.
Otherwise, $\mathcal{V}=\emptyset$.
}
\item $\qSim_1$ does the following:
\begin{itemize}
\item If $\mathcal{V}=\emptyset$, generate $\ket{\widetilde{0^n0^n}}$ and $(\nce.\vk,\widetilde{\nce.\qct},\nce.\aux)\la\NCE.\qFake(\nce.\pk)$.
Let $\Psi_A\seteq\Tr_B(\ket{\widetilde{0^n0^n}}\bra{\widetilde{0^n0^n}})$ and
$\Psi_B\seteq\Tr_A(\ket{\widetilde{0^n0^n}}\bra{\widetilde{0^n0^n}})$.
Output $\qct\seteq (\Psi_A,\widetilde{\nce.\qct})$ and $\state\seteq (\nce.\aux,\nce.\pk,\nad.\MPK,\Psi_B,1^{|m|},\nce.\vk,0)$.
\item If $\mathcal{V}=(f,f(m),(\nad.\sk_f,\nce.\sk))$, generate $a,c\la\bit^n$, $(\nce.\vk,\nce.\qct)\la\NCE.\qEnc(\nce.\pk,(a,c))$, $(\nad.\vk,\nad.\qct)\la\NAD.\qSim(\nad.\MPK,(f,f(m),\nad.\sk_f),1^{|m|})$ and $\Psi\seteq Z^cX^a\nad.\qct X^a Z^c$.
Output $\qct\seteq (\Psi,\nce.\qct)$ and $\state\seteq(\nad.\vk,\nce.\vk,a,c,1)$.
\end{itemize}

\end{enumerate}
\item[$\qSim_2(\MSK,f,f(m),\state)$:]$ $
\begin{enumerate}
\item Parse $\MSK\seteq (\nad.\MSK,\nce.\MSK)$ and $\state =(\nce.\aux,\nce.\pk,\nad.\MPK,\Psi_B,1^{|m|},\nce.\vk,0)$.
\item Compute $\nad.\sk_f\la\NAD.\keygen(\nad.\MSK,f)$.
\item Compute $(\nad.\vk,\nad.\qct)\la\NAD.\qSim(\nad.\MPK,(f,f(m),\nad.\sk_f),1^{|m|})$. 
Measure the $i$-th qubit of $\mathsf{nad}.\qct$ and $\Psi_B$ in the Bell basis and let $(x_i,z_i)$ be the measurement outcome for all $i\in[N]$.
\item Compute $\widetilde{\nce.\sk}\la\NCE.\Reveal(\nce.\pk,\nce.\MSK,\nce.\aux,(x,z))$.
\item Output $\sk_f\seteq (\nad.\sk_f,\widetilde{\nce.\sk})$ and $\state'\seteq(\nad.\vk,\nce.\vk,x,z,1)$.
\end{enumerate}
\item[$\qSim_3(\state^*)$:]$ $
\begin{enumerate}
\item Parse  $\state^* = (\nad.\vk,\nce.\vk,x^*,z^*,1)$ or $\state^* =(\nce.\aux,\nce.\pk,\nad.\MPK,\Psi_B,1^{|m|},\nce.\vk,0)$.
\item $\qSim_3$ does the following:
\begin{itemize}
    \item If the final bit of $\state^*$ is $0$,
    compute $(\nad.\vk,\nad.\qct)\la\NAD.\qSim(\nad.\MPK,\emptyset,1^{|m|})$. 
    Measure the $i$-th qubit of $\mathsf{nad}.\qct$ and $\Psi_B$ in the Bell basis and let $(x_i,z_i)$ be the measurement outcome for all $i\in[N]$.
     Output $\vk\seteq (\nad.\vk,\nce.\vk,x,z)$.
     \item If the final bit of $\state^*$ is $1$, output $\vk\seteq(\nad.\vk,\nce.\vk,x^*,z^*)$.
\end{itemize}
\end{enumerate}
Let us define the sequence of hybrids as follows.
\item[$\hyb_{0}$:]This is identical to $\expd{\Sigma_{\mathsf{cefe},\qA}}{cert}{ever}{ada}{sim}(0)$.
\begin{enumerate}
    \item The challenger generates $(\mathsf{nad}.\MPK,\mathsf{nad}.\MSK)\la\mathsf{NAD}.\Setup(1^\secp)$ and $(\nce.\pk,\nce.\MSK)\la\NCE.\Setup(1^{\secp})$, and sends $(\nad.\MPK,\nce.\pk)$ to $\qA_1$.
    \item\label{step:key_query_adapt_fe_1} $\qA_1$ is allowed to make an arbitrary key query at most one time.
    For a key query, the challenger receives $f\in\mathcal{F}$,  
    computes $\mathsf{nad}.\sk_{f}\la\mathsf{NAD}.\keygen(\mathsf{nad}.\MSK,f)$
    and $\nce.\sk\la\NCE.\keygen(\nce.\MSK)$,
    and sends $(\nad.\sk_f,\nce.\sk)$ to $\qA_1$.
    \item $\qA_1$ chooses $m\in\Ms$, and sends $m$ to the challenger.
    \item\label{step:encryption_single_fe} The challenger generates $a,c\la\bit^n$,
    computes $(\mathsf{nad}.\vk,\mathsf{nad}.\qct)\la \mathsf{NAD}.\qEnc(\mathsf{nad}.\MPK,m)$,
    $\Psi\seteq Z^{c}X^{a}\mathsf{nad}.\qct X^{a}Z^{c}$
    and $(\nce.\vk,\nce.\qct)\la\NCE.\qEnc(\nce.\pk,(a,c))$,
    and sends $(\Psi,\nce.\qct)$ to $\qA_1$.
    \item\label{step:key_query_adapt_fe_2}
    If a key query is not called in step~\ref{step:key_query_adapt_fe_1},
    $\qA_1$ is allowed to make an arbitrary key query at most one time.
    For a key query, the challenger receives $f\in\mathcal{F}$,  
    computes $\mathsf{nad}.\sk_{f}\la\mathsf{NAD}.\keygen(\mathsf{nad}.\MSK,f)$
    and $\nce.\sk\la\NCE.\keygen(\nce.\MSK)$,
    and sends $(\nad.\sk_f,\nce.\sk)$ to $\qA_1$.
    \item $\qA_1$ sends $(\nad.\cert,\nce.\cert)$ to the challenger and its internal state to $\qA_2$.
    \item\label{step:vrfy_fe} The challenger computes $\mathsf{nad}.\cert^*\la\NAD.\Modify(a,c,\nad.\cert)$.
        The challenger computes $\NCE.\Vrfy(\allowbreak \nce.\vk,\nce.\cert)$ and $\mathsf{NAD}.\Vrfy(\mathsf{nad}.\vk,\mathsf{nad}.\cert^*)$.
        If the results are $\top$, the challenger outputs $\top$ and sends $(\nad.\MSK,\nce.\MSK)$ to $\qA_2$.
        Otherwise, the challenger outputs $\bot$ and sends $\bot$ to $\qA_2$.
    \item $\qA_2$ outputs $b'$. The output of the experiment is $b'$ if the challenger outputs $\top$.
    Otherwise, the output of the experiment is $\bot$.
\end{enumerate}
\item[$\hyb_{1}$:]$ $
This is different from $\hyb_{0}$ in the following second points.
First, when a key query is not called in step~\ref{step:key_query_adapt_fe_1}, the challenger computes $(\nce.\vk,\widetilde{\nce.\qct},\nce.\aux)\la\NCE.\qFake(\nce.\pk)$ and sends $(\Psi,\widetilde{\nce.\qct})$ to $\qA_1$
instead of computing $(\nce.\vk,\nce.\qct)\la\NCE.\qEnc(\nce.\pk,(a,c))$
and sending $(\Psi,\nce.\qct)$ to $\qA_1$.
Second, in step~\ref{step:key_query_adapt_fe_2}, the challenger computes $\widetilde{\nce.\sk}\la\NCE.\Reveal(\nce.\pk,\nce.\MSK,\nce.\aux,(a,c))$ and sends $(\nad.\sk_f,\nce.\sk)$ to $\qA_1$ instead of computing $\nce.\sk\la\NCE.\keygen(\nce.\MSK)$ and sending $(\nad.\sk_f,\nce.\sk)$ to $\qA_1$.

\item[$\hyb_{2}$:]$ $
This is different from $\hyb_{1}$ in the following three points.
First, when a key query is not called in step \ref{step:key_query_adapt_fe_1}, the challenger generates $\ket{\widetilde{0^n0^n}}$ instead of generating $a,c\la\bit^n$ and $\Psi=Z^{c}X^{a}\mathsf{nad}.\qct X^{a}Z^{c}$.
Let $\Psi_A\seteq \Tr_B(\ket{\widetilde{0^n0^n}}\bra{\widetilde{0^n0^n}})$ and $\Psi_B\seteq \Tr_A(\ket{\widetilde{0^n0^n}}\bra{\widetilde{0^n0^n}})$.
Second, when a key query is not called in step \ref{step:key_query_adapt_fe_1}, the challenger sends $(\Psi_A,\widetilde{\nce.\qct})$ to $\qA_1$ instead of sending $(\Psi,\widetilde{\nce.\qct})$ to $\qA_1$ and then that measures the $i$-th qubit of $\nad.\qct$ and $\Psi_B$ in the Bell basis for all $i\in[n]$.
Let $(x_i,z_i)$ be the measurement outcome for all $i\in[n]$.
Third, the challenger computes $\widetilde{\nce.\sk}\la\NCE.\Reveal(\nce.\pk,\nce.\MSK,\nce.\aux,(x,z))$ instead of computing $\widetilde{\nce.\sk}\la\NCE.\Reveal(\nce.\pk,\nce.\MSK,\nce.\aux,(a,c))$ in step~\ref{step:key_query_adapt_fe_2} and computes $\mathsf{nad}.\cert^*\la\NAD.\Modify(x,z,\nad.\cert)$ instead of computing
$\mathsf{nad}.\cert^*\la\NAD.\Modify(a,c,\nad.\cert)$ in step~\ref{step:vrfy_fe}.

\item[$\hyb_{3}$:]$ $
This is different from $\hyb_{2}$ in the following three points.
First, when a key query is not called in step~\ref{step:key_query_adapt_fe_1}, the challenger does not generate $(\mathsf{nad}.\vk,\mathsf{nad}.\qct)\la \mathsf{NAD}.\qEnc(\mathsf{nad}.\MPK,m)$ and measure the $i$-th qubit of $\nad.\qct$ and $\Psi_B$ in the Bell basis in step~\ref{step:encryption_single_fe}.
Second, if a key query is called in step~\ref{step:key_query_adapt_fe_2}, the challenger computes $(\mathsf{nad}.\vk,\mathsf{nad}.\qct)\la \mathsf{NAD}.\qEnc(\mathsf{nad}.\MPK,m)$
and measures the $i$-th qubit of $\nad.\qct$ and $\Psi_B$ in the Bell basis for all $i\in[n]$ after it computes $\mathsf{nad}.\sk_{f}\la\mathsf{NAD}.\keygen(\mathsf{nad}.\MSK,f)$.
Third, if a key query is not called throughout the experiment,
the challenger computes $(\mathsf{nad}.\vk,\mathsf{nad}.\qct)\la \mathsf{NAD}.\qEnc(\mathsf{nad}.\MPK,m)$, measures the $i$-th qubit of $\nad.\qct$ and $\Psi_B$ in the Bell basis after step~\ref{step:key_query_adapt_fe_2}.

\item[$\hyb_{4}$:]$ $
This is identical to $\hyb_{3}$ except that the challenger computes $(\mathsf{nad}.\vk,\mathsf{nad}.\qct)\la \mathsf{NAD}.\qSim(\mathsf{nad}.\MPK,\mathcal{V},1^{|m|})$ instead of computing $(\mathsf{nad}.\vk,\mathsf{nad}.\qct)\la \mathsf{NAD}.\qEnc(\mathsf{nad}.\MPK,m)$, where $\mathcal{V}=(f,f(m),\nad.\sk_f)$ if a key query is called and $\mathcal{V}=\emptyset$ if a key query is not called.

\end{description}
From the definition of $\expd{\Sigma_{\mathsf{cefe}},\qA}{cert}{ever}{ada}{sim}(\secp, b)$ and $\qSim=(\qSim_1,\qSim_2,\qSim_3)$,
it is clear that $\Pr[\hyb_{0}=1]=\Pr[\expd{\Sigma_{\mathsf{cefe}},\qA}{cert}{ever}{ada}{sim}(\secp,0)=1]$ and
$\Pr[\hyb_{4}=1]=\Pr[\expd{\Sigma_{\mathsf{cefe}},\qA}{cert}{ever}{ada}{sim}(\secp,1)=1]$.
Therefore, \cref{thm:ever_security_single_ad} easily follows from \cref{prop:exp_hyb_1_cefe_si_ad,prop:hyb_1_hyb_2_cefe_si_ad,prop:hyb_2_hyb_3_cefe_si_ad,prop:hyb_3_hyb_4_cefe_si_ad}.
(Whose proof is given later.)
\end{proof}

\begin{proposition}\label{prop:exp_hyb_1_cefe_si_ad}
If $\Sigma_{\mathsf{cence}}$ is certified everlasting RNC secure, 
it holds that
\begin{align}
\abs{\Pr[\hyb_{0}=1]-\Pr[\hyb_{1}=1]}\leq\negl(\secp).
\end{align}
\end{proposition}

\begin{proposition}\label{prop:hyb_1_hyb_2_cefe_si_ad}
\begin{align}
\Pr[\hyb_{1}=1]=\Pr[\hyb_{2}=1].
\end{align}
\end{proposition}

\begin{proposition}\label{prop:hyb_2_hyb_3_cefe_si_ad}
\begin{align}
    \Pr[\hyb_{2}=1]=\Pr[\hyb_{3}=1].
\end{align}
\end{proposition}

\begin{proposition}\label{prop:hyb_3_hyb_4_cefe_si_ad}
If $\Sigma_{\mathsf{nad}}$ is $1$-bounded non-adaptive certified everlasting simulation-secure,
it holds that
\begin{align}
    \abs{\Pr[\hyb_{3}=1]-\Pr[\hyb_{4}=1]}\leq\negl(\secp).
\end{align}
\end{proposition}

\begin{proof}[Proof of \cref{prop:exp_hyb_1_cefe_si_ad}]
When an adversary makes key queries in step \ref{step:key_query_adapt_fe_1}, it is clear that $\Pr[\hyb_{0}=1]=\Pr[\hyb_{1}=1]$.
Hence, we consider the case where the adversary does not make a key query in step $\ref{step:key_query_adapt_fe_1}$ below.

We assume that $\abs{\Pr[\hyb_{0}=1]-\Pr[\hyb_{1}=1]}$ is non-negligible, and construct an adversary $\qB$ that breaks the certified everlasting RNC security of $\Sigma_{\mathsf{cence}}$.
Let us describe how $\qB$ works below.
\begin{enumerate}
    \item $\qB$ receives $\nce.\pk$ from the challenger of $\expd{\Sigma_{\mathsf{cence}},\qB}{cert}{ever}{rec}{nc}(\secp,b)$,
    generates $(\mathsf{nad}.\MPK,\mathsf{nad}.\MSK)\la \mathsf{NAD}.\keygen(1^\secp)$,
    and sends $(\nad.\MPK,\nce.\pk)$ to $\qA_1$.
    \item $\qB$ receives a message $m\in\Ms$, computes $(\mathsf{nad}.\vk,\mathsf{nad}.\qct)\la\mathsf{NAD}.\qEnc(\mathsf{nad}.\MPK,m)$, generates $a,c\la\bit^n$, computes $\Psi\seteq Z^cX^a\mathsf{nad}.\qct X^aZ^c$, sends $(a,c)$ to the challenger, receives $(\nce.\qct^*,\nce.\sk^*)$ from the challenger, and sends $(\Psi,\nce.\qct^*)$ to $\qA_1$.
    \item $\qA_1$ is allowed to send a key query at most one time.
    For a key query, $\qB$ receives an function $f$, generates $\mathsf{nad}.\sk_f\la\mathsf{NAD}.\keygen(\mathsf{nad}.\MSK,f)$,
    and sends $(\nad.\sk_f,\nce.\sk^*)$ to $\qA_1$.
    \item $\qA_1$ sends $(\nad.\cert,\nce.\cert)$ to $\qB$ and its internal state to $\qA_2$.
    \item $\qB$ sends $\nce.\cert$ to the challenger, and receives $\nce.\MSK$ or $\bot$ from the challenger.
    $\qB$ computes $\mathsf{nad}.\cert^*\la\NAD.\Modify(a,c,\nad.\cert)$ and $\mathsf{NAD}.\Vrfy(\mathsf{nad}.\vk,\mathsf{nad}.\cert^*)$.
    If the result is $\top$ and $\qB$ receives $\nce.\MSK$ from the challenger,
    $\qB$ sends $(\nad.\MSK,\nce.\MSK)$ to $\qA_2$.
    Otherwise, $\qB$ outputs $\bot$, sends $\bot$ to $\qA_2$, and aborts.
    \item $\qA_2$ outputs $b'$.
    \item $\qB$ outputs $b'$. 
\end{enumerate}
It is clear that $\Pr[1\la\qB \mid b=0]=\Pr[\hyb_{0}=1]$ and $\Pr[1\la\qB \mid b=1]=\Pr[\hyb_{1}=1]$.
By assumption, $\abs{\Pr[\hyb_{0}=1]-\Pr[\hyb_{1}=1]}$ is non-negligible, and therefore $\allowbreak\abs{\Pr[1\la\qB \mid b=0]-\allowbreak\Pr[1\la\qB \mid b=1]}$ is non-negligible,
which contradicts the certified everlasting RNC security of $\Sigma_{\mathsf{cence}}$. 
\end{proof}

\begin{proof}[Proof of \cref{prop:hyb_1_hyb_2_cefe_si_ad}]
We clarify the difference between $\hyb_{1}$ and $\hyb_{2}$.
First, in $\hyb_{2}$, the challenger uses $(x,z)$ instead of using $(a,c)$ as in $\hyb_{1}$.
Second, in $\hyb_{2}$, the challenger sends $\Psi_A$ to $\qA_1$ instead of sending $Z^{c}X^{a}\mathsf{nad}.\qct X^{a}Z^{c}$ to $\qA_1$ as in $\hyb_{1}$.
Hence,
it is sufficient to prove that $x$ and $z$ are uniformly randomly distributed and
$\Psi_A$ is identical to $Z^{z}X^{x}\mathsf{nad}.\qct X^{x}Z^{z}$.
These two things are obvious from \cref{lem:quantum_teleportation}.
\end{proof}

\begin{proof}[Proof of \cref{prop:hyb_2_hyb_3_cefe_si_ad}]
The difference between $\hyb_{2}$ and $\hyb_{3}$ is only the order of operating the algorithm $\NAD.\qEnc$ and the Bell measurement on $\nad.\qct$ and $\Psi_B$.
Therefore, it is clear that the probability distribution of the ciphertext and the decryption key given to the adversary in $\hyb_{2}$ is identical to that the ciphertext and the decryption key given to the adversary in $\hyb_{3}$.
\end{proof}

\begin{proof}[Proof of \cref{prop:hyb_3_hyb_4_cefe_si_ad}]
We assume that $\abs{\Pr[\hyb_{3}=1]-\Pr[\hyb_{4}=1]}$ is non-negligible, and construct an adversary $\qB$ that breaks the $1$-bounded non-adaptive certified everlasting simulation-security of $\Sigma_{\mathsf{nad}}$.
Let us describe how $\qB$ works below.
\begin{enumerate}
    \item $\qB$ receives $\mathsf{nad}.\MPK$ from the challenger of $\expd{\Sigma_{\mathsf{nad}},\qB}{cert}{ever}{noada}{sim}(\secp,b)$, 
    generates $(\nce.\pk,\nce.\MSK)\la\NCE.\Setup(1^\secp)$, and sends $(\nad.\MPK,\nce.\pk)$ to $\qA_1$.
    \item\label{step:key_query_adapt_fe_hyb_3_hyb_4}  $\qA_1$ is allowed to call a key query at most one time.
    For a key query, $\qB$ receives $f$ from $\qA_1$, sends $f$ to the challenger as a key query, receives $\mathsf{nad}.\sk_f$ from the challenger, computes $\nce.\sk\la\NCE.\keygen(\nce.\MSK)$, and sends $(\nad.\sk_f,\nce.\sk)$ to $\qA_1$.
    \item $\qA_1$ chooses $m\in\Ms$ and sends $m$ to $\qB$.
    \item $\qB$ does the following.
    \begin{itemize}
        \item If a key query is called in step~\ref{step:key_query_adapt_fe_hyb_3_hyb_4}, $\qB$ sends a challenge query $m$ to the challenger, receives $\nad.\qct$ from the challenger, generates $a,c\la\bit^n$,
        $\Psi\seteq Z^{c}X^{a}\mathsf{nad}.\qct X^{a}Z^{c}$ and $(\nce.\vk,\nce.\qct)\la\NCE.\qEnc(\nce.\pk,(a,c))$,
        and sends $(\Psi,\nce.\qct)$ to $\qA_1$.
        \item If a key query is not called in step~\ref{step:key_query_adapt_fe_hyb_3_hyb_4}, $\qB$ generates $\ket{\widetilde{0^n0^n}}$.
        Let $\Psi_A\seteq \Tr_B(\ket{\widetilde{0^n0^n}}\bra{\widetilde{0^n0^n}})$ and $\Psi_B\seteq \Tr_A(\ket{\widetilde{0^n0^n}}\bra{\widetilde{0^n0^n}})$.
        $\qB$ computes $(\nce.\vk,\widetilde{\nce.\qct},\nce.\aux)\la\NCE.\qFake(\nce.\pk)$ and sends $(\Psi_A,\widetilde{\nce.\qct})$ to $\qA_1$.
    \end{itemize}
    \item
    If a key query is not called in step~\ref{step:key_query_adapt_fe_hyb_3_hyb_4}, 
    $\qA_1$ is allowed to make a key query at most one time.
    If $\qB$ receives an function $f$ as key query, $\qB$ sends $f$ to the challenger as key query, and receives $\mathsf{nad}.\sk_f$ from the challenger.
    $\qB$ sends a challenge query $m$ to the challenger, receives $\nad.\qct$, measures the $i$-th qubit of $\nad.\qct$ and $\Psi_B$ in the Bell basis, and let $(x_i,z_i)$ be the measurement outcome for all $i\in[n]$. 
    $\qB$ computes $\widetilde{\nce.\sk}\la\NCE.\Reveal(\nce.\pk,\nce.\MSK,\nce.\aux,(x,z))$ and sends $(\nad.\sk_f,\widetilde{\nce.\sk})$ to $\qA_1$.
    \item If $\qB$ does not receive a key query throughout the experiment,
    $\qB$ sends a challenge query $m$ to the challenger, receives $\mathsf{nad}.\qct$, and measures the $i$-th qubit of $\mathsf{nad}.\qct$ and $\Psi_B$ in the Bell basis and let $(x_i,z_i)$ be the measurement outcome for all $i\in[n]$.
    \item $\qA_1$ sends $(\nad.\cert,\nce.\cert)$ to $\qB$ and its internal state to $\qA_2$.
    \item $\qB$ computes $\mathsf{nad}.\cert^*\la\NAD.\Modify(x^*,z^*,\nad.\cert)$, where $(x^*,z^*)=(a,c)$ if a key query is called in step~\ref{step:key_query_adapt_fe_hyb_3_hyb_4} and $(x^*,z^*)=(x,z)$ if a key query is not called in step~\ref{step:key_query_adapt_fe_hyb_3_hyb_4}.
        $\qB$ sends $\mathsf{nad}.\cert$ to the challenger, and receives $\mathsf{nad}.\MSK$ or $\bot$ from the challenger.
        $\qB$ computes $\NCE.\Vrfy(\nce.\vk,\nce.\cert)$.
        If the result is $\top$ and $\qB$ receives $\mathsf{nad}.\MSK$ from the challenger, $\qB$ sends $(\nad.\MSK,\nce.\MSK)$ to $\qA_2$.
        Otherwise, $\qB$ outputs $\bot$, sends $\bot$ to $\qA_2$, and aborts.
    \item $\qA_2$ outputs $b'$.
    \item $\qB$ outputs $b'$. 
\end{enumerate}

It is clear that $\Pr[1\la\qB \mid b=0]=\Pr[\hyb_{3}=1]$ and $\Pr[1\la\qB \mid b=1]=\Pr[\hyb_{4}=1]$.
By assumption, $\abs{\Pr[\hyb_{3}=1]-\Pr[\hyb_{4}=1]}$ is non-negligible, and therefore $\abs{\Pr[1\la\qB \mid b=0]-\Pr[1\la\qB \mid b=1]}$ is non-negligible, which contradicts the $1$-bounded non-adaptive certified everlasting simulation-security of $\Sigma_{\mathsf{nad}}$. 
\end{proof}

\subsection{\texorpdfstring{$q$}{q}-Bounded Construction with Adaptive Security for \texorpdfstring{$\NCone$}{NC1} circuits}\label{sec:const_multi_fe}
In this section, we construct a $q$-bounded FE with certified everlasting deletion scheme for all $\NCone$ circuits from $1$-bounded certified everlasting secure FE constructed in the previous subsection
and Shamir's secret sharing~(\cite{cacm.Shamir79}).
Our construction is similar to that of standard FE for all $\NCone$ circuits in \cite{C:GorVaiWee12}
except that we use $1$-bounded certified everlasting secure FE instead of standard $1$-bounded FE.


\paragraph{Our $q$-bounded adaptive certified everlasting secure FE scheme for $\NCone$ circuits.}
We consider the polynomial representation of circuits $C$ in $\NCone$.
The input message space is $\Ms\seteq\mathbb{F}^{\ell}$, and for each $\NCone$ circuit $C$, $C(\cdot)$ is an $\ell$-variate polynomial over $\mathbb{F}$ of total degree at most $D$.
Let $q=q(\lambda)$ be a polynomial of $\lambda$. Our scheme is associated with additional parameters $S=S(\lambda)$, $N=N(\lambda)$, $t=t(\lambda)$ and $v=v(\lambda)$
that satisfy
\begin{align}
    t(\lambda)=\Theta(q^2\lambda), N(\lambda)=\Theta(D^2q^2t), v(\lambda)=\Theta(\lambda), S(\lambda)=\Theta(vq^2).
\end{align}

Let us define a family $\mathcal{G}\seteq\{G_{C,\Delta}\}_{C\in\NCone,\Delta\subseteq[S]}$, where
\begin{align}
G_{C,\Delta}(x,Z_1,Z_2,\cdots,Z_S)\seteq C(x)+\sum_{i\in\Delta}Z_i    
\end{align}
is a function and $Z_1,\cdots,Z_S\in\mathbb{F}$.

We construct a $q$-bounded certified everlasting secure FE scheme for all $\NCone$ circuits  $\Sigma_{\mathsf{cefe}}=(\Setup,\keygen,\qEnc,\allowbreak \qDec,\qDelete,\allowbreak \Vrfy)$
from a $1$-bounded certified everlasting secure FE scheme $\Sigma_{\mathsf{one}}=\mathsf{ONE}.(\Setup,\keygen,\qEnc,\qDec,\qDelete,\allowbreak \Vrfy)$.

\begin{description}
\item[$\Setup(1^\lambda)$:]$ $
\begin{itemize}
    \item For $i\in[N]$, generate $(\mathsf{one}.\MPK_i,\mathsf{one}.\MSK_i)\la\mathsf{ONE}.\Setup(1^\secp)$.
    \item Output $\MPK\seteq\{\mathsf{one}.\MPK_i\}_{i\in[N]}$ and $\MSK\seteq \{\mathsf{one}.\MSK_i\}_{i\in[N]}$.
\end{itemize}
\item[$\keygen(\MSK,C)$:]$ $
\begin{itemize}
    \item Parse $\MSK= \{\mathsf{one}.\MSK_i\}_{i\in[N]}$.
    \item Chooses a uniformly random set $\Gamma\subseteq[N]$ of size $tD+1$.
    \item Chooses a uniformly random set $\Delta\subseteq[S]$ of size $v$.
    \item For $i\in\Gamma$, compute $\mathsf{one}.\sk_{C,\Delta,i}\la\mathsf{ONE}.\keygen(\one.\MSK_i,G_{C,\Delta})$.
    \item Output $\sk_{C}\seteq (\Gamma,\Delta,\{\one.\sk_{C,\Delta,i}\}_{i\in\Gamma})$.
\end{itemize}
\item[$\qEnc(\MPK,x)$:]$ $
\begin{itemize}
    \item Parse $\MPK=\{\mathsf{one}.\MPK_i\}_{i\in[N]}$.
    \item For $i\in[\ell]$, pick a random degree $t$ polynomial $\mu_i(\cdot)$ whose constant term is $x[i]$.
    \item For $i\in[S]$, pick a random degree $Dt$ polynomial $\xi_i(\cdot)$ whose constant term is $0$.
    \item For $i\in[N]$, compute $(\one.\vk_i,\one.\qct_i)\la\ONE.\qEnc(\one.\MPK_i,(\mu_1(i),\cdots,\mu_{\ell}(i),\allowbreak\xi_1(i),\cdots,\xi_S(i) ))$.
    \item Output $\vk=\{\one.\vk_i\}_{i\in[N]}$ and $\qct\seteq\{\one.\qct_i\}_{i\in[N]}$.
    \end{itemize}
\item[$\qDec(\sk_C,\qct)$:]$ $
\begin{itemize}
    \item Parse $\sk_{C}= (\Gamma,\Delta,\{\one.\sk_{C,\Delta,i}\}_{i\in\Gamma})$ and $\qct=\{\one.\qct_i\}_{i\in[N]}$.
    \item For $i\in \Gamma$, compute $\eta(i)\la\ONE.\qDec(\one.\sk_{C,\Delta,i},\one.\qct_i)$.
    \item Output $\eta(0)$.
\end{itemize}
\item[$\qDelete(\qct)$:]$ $
\begin{itemize}
    \item Parse $\qct=\{\one.\qct_i\}_{i\in[N]} $.
    \item For $i\in[N]$, compute $\one.\cert_i\la\ONE.\qDelete(\one.\qct_i)$.
    \item Output $\cert\seteq \{\one.\cert_i\}_{i\in[N]}$.
\end{itemize}
\item[$\Vrfy(\vk,\cert)$:]$ $
\begin{itemize}
    \item Parse $\vk=\{\one.\vk_i\}_{i\in[N]}$ and $\cert= \{\one.\cert_i\}_{i\in[N]}$.
    \item For $i\in[N]$, compute $\top/\bot\la\ONE.\Vrfy(\one.\vk_i,\one.\cert_i)$.
    If all results are $\top$, output $\top$. Otherwise, output $\bot$.
\end{itemize}
\end{description}

\paragraph{Correctness:}
Verification correctness easily follows from verification correctness of $\Sigma_{\one}$.
Let us show evaluation correctness.
By decryption correctness of $\Sigma_{\one}$, for all $i\in\Gamma$ we have 
\begin{align}
\eta(i)&=G_{C,\Delta} (\mu_1(i),\cdots,\mu_\ell(i),\xi_1(i),\cdots,\xi_S(i))\\
&=C(\mu_1(i),\cdots,\mu_\ell(i))+\Sigma_{a\in\Delta}\xi_a(i).
\end{align}
Since $|\Gamma|\geq Dt+1$, this means that $\eta$ is equal to the degree $Dt$ polynomial 
\begin{align}
    \eta(\cdot)=C(\mu_1(\cdot),\cdots,\mu_\ell(\cdot) )+\Sigma_{a\in\Delta}\xi_a(\cdot)
\end{align}
Hence $\eta(0)=C(x_1,\cdots,x_\ell)=C(x)$, which means that our construction satisfies evaluation correctness.

\paragraph{Security:}
The following two theorems hold.
\begin{theorem}\label{thm:compt_bounded_fe}
If $\Sigma_{\one}$ satisfies the $1$-bounded adaptive simulation-security, $\Sigma_{\mathsf{cefe}}$ satisfies the $q$-bounded adaptive simulation-security.
\end{theorem}
Its proof is similar to that of \cref{thm:ever_bounded_fe}, and therefore we omit it.

\begin{theorem}\label{thm:ever_bounded_fe}
If $\Sigma_{\one}$ satisfies the $1$-bounded adaptive certified everlasting simulation-security, $\Sigma_{\mathsf{cefe}}$ the $q$-bounded adaptive certified everlasting simulation-security.
\end{theorem}

\begin{proof}[Proof of \cref{thm:ever_bounded_fe}]
Let us denote the simulating algorithm of $\Sigma_{\one}$ as $\ONE.\qSim=\ONE.(\qSim_1,\qSim_2,\qSim_3)$.
Let us describe how the simulator $\qSim=(\qSim_1,\qSim_2,\qSim_3)$ works below.
\begin{description}
\item[$\qSim_1(\MPK,\mathcal{V},1^{|x|})$:] Let $q^*$ be the number of times that $\qA_1$ has made key queries before it sends a challenge query. 
\begin{enumerate}
    \item Parse $\MPK\seteq \{\one.\MPK_{i}\}_{i\in[N]}$ and $\mathcal{V}\seteq \{C_j,C_j(x),(\Gamma_j,\Delta_j,\{\one.\sk_{C_j,\Delta_j,i}\}_{i\in[\Gamma_j]})\}_{j\in[q^*]}$.
    \item Generate a uniformly random set $\Gamma_i\subseteq[N]$ of size $Dt+1$ and a uniformly random set $\Delta_i\subseteq[S]$ of size $v$ for all $i\in\{q^*+1,\cdots, q\}$.
    Let $\Delta_0\seteq \emptyset$. Let $\mathcal{L}\seteq \bigcup_{i\neq i'}(\Gamma_i\cap \Gamma_{i'})$.
    $\qSim_1$ aborts if $\left|\mathcal{L}\right|> t$ or
    there exists some $i\in[q]$ such that $\Delta_i \setminus (\bigcup_{j\neq i}\Delta_j)=\emptyset$.
    \item $\qSim_1$ uniformly and independently samples $\ell$ random degree $t$ polynomials $\mu_1,\cdots, \mu_\ell$ whose constant terms are all $0$. 
    \item $\qSim_1$ samples the polynomials $\xi_1,\cdots, \xi_S$ as follows for $j\in[q]$:
    \begin{itemize}
        \item fix $a^*\in\Delta_j \setminus (\Delta_0\cup\cdots\cup\Delta_{j-1})$;
        \item for all $a\in(\Delta_j \setminus (\Delta_0\cup\cdots\cup \Delta_{j-1})) \setminus \{a^*\}$, set $\xi_a$ to be a uniformly random degree $Dt$ polynomial whose constant term is $0$;
        \item if $j\leq q^*$, pick a random degree $Dt$ polynomial $\eta_j(\cdot)$ whose constant term is $C_j(x)$; if $j>q^*$, pick random values for $\eta_j(i)$ for all $i\in\mathcal{L}$;
        \item the evaluation of $\xi_{a^*}$ on the points in $\mathcal{L}$ is defined by the relation:
        \begin{align}
            \eta_j(\cdot)=C_j(\mu_1(\cdot),\cdots, \mu_\ell(\cdot))+\sum_{a\in\Delta_j}\xi_a(\cdot).
        \end{align}
        \item Finally, for all $a\notin (\Delta_1\cup\cdots \cup \Delta_q)$, set $\xi_a$ to be a uniformly random degree $Dt$ polynomial whose constant term is $0$.
    \end{itemize}
    \item For each $i\in\mathcal{L}$, $\qSim_1$ computes 
    \begin{align}
        (\one.\vk_i,\one.\qct_i)\la\ONE.\qEnc(\one.\MPK_i,(\mu_1(i),\cdots,\mu_\ell(i),\xi_1(i),\cdots, \xi_S(i) )).
    \end{align}
    \item For each $i\notin \mathcal{L}$, $\qSim_1$ does the following:
    \begin{itemize}
        \item If $i\in \Gamma_j$ for some $j\in[q^*]$
        \footnote
        { Note that $j$ is uniquely determined since $i\notin \mathcal{L}$ .
        },
        computes
        \begin{align}
        (\one.\qct_i,\one.\state_i)\la\ONE.\qSim_1(\one.\MPK_i,(G_{C_j,\Delta_j,i},\eta_j(i),\one.\sk_{C_j,\Delta_j,i}),1^{|m|}).
        \end{align}
        \item If $i\notin\Gamma_j$ for all $j\in[q^*]$, 
        computes 
        \begin{align}
            (\one.\qct_i,\one.\state_i)\la\ONE.\qSim_1(\one.\MPK_i,\emptyset,1^{|m|}).
        \end{align}
    \end{itemize}
    \item Output $\qct\seteq \{\one.\qct_i\}_{i\in[N]}$ and $\state\seteq (\{\Gamma_i\}_{i\in[q]},\{\Delta_i\}_{i\in[q]},\{\eta_j(i)\}_{j\in\{q^*+1,\cdots,q\},i\in\mathcal{L}},\allowbreak\{\one.\state_i\}_{i\in[N] \setminus \mathcal{L}},\allowbreak \{\one.\vk_i\}_{i\in\mathcal{L}})$.
\end{enumerate}
\item[$\qSim_2(\MSK,C_j,C_j(x),\state)$:]The simulator simulates the $j$-th key query for $j>q^*$.
\begin{enumerate}
    \item Parse $\MSK\seteq\{\one.\MSK_i\}_{i\in[N]}$ and $\state_{j-1} \seteq  (\{\Gamma_i\}_{i\in[q]},\{\Delta_i\}_{i\in[q]},\{\eta_s(i)\}_{s\in\{q^*+1,\cdots,q\},i\in\mathcal{L}},\allowbreak\{\one.\state_i\}_{i\in[N] \setminus \mathcal{L}},\allowbreak \{\one.\vk_i\}_{i\in\mathcal{L}})$.
    \item For each $i\in\Gamma_j\cap \mathcal{L}$, generate $\one.\sk_{C_j,\Delta_j,i}\la\ONE.\keygen(\one.\MSK_i,G_{C_j,\Delta_j})$.
    \item For each $i\in\Gamma_j \setminus \mathcal{L}$, 
    generate a random degree $Dt$ polynomial $\eta_j(\cdot)$ whose constant term is $C_j(x)$ and subject to the constraints on the values in $\mathcal{L}$ chosen earlier,
    and generate 
    \begin{align}
    (\one.\sk_{C_j,\Delta_j,i},\one.\state_i^*)\la\ONE.\qSim_2(\one.\MSK_i,\eta_j(i),G_{C_j,\Delta_j},\one.\state_i).
    \end{align}
    For simplicity, let us denote $\one.\state_i^*$ as $\one.\state_i$ for $i\in\Gamma_j \setminus \mathcal{L}$.
    \item Output $\sk_{C_j}\seteq(\Gamma_j,\Delta_j,\{\one.\sk_{C_j,\Delta_j,i}\}_{i\in\Gamma_j})$ and
    $\state_j\seteq (\{\Gamma_i\}_{i\in[q]},\{\Delta_i\}_{i\in[q]},\allowbreak\{\eta_j(i)\}_{j\in\{q^*+1,\cdots,q\},i\in\mathcal{L}},\allowbreak \{\one.\state_i\}_{i\in[N]\setminus \mathcal{L}},\{\one.\vk_i\}_{i\in\mathcal{L}})$.
\end{enumerate}
\item[$\qSim_3(\state^*)$:]The simulator simulates a verification key.
\begin{enumerate}
    \item Parse $\state^*\seteq (\{\Gamma_i\}_{i\in[q]},\{\Delta_i\}_{i\in[q]},\{\eta_j(i)\}_{j\in\{q^*+1,\cdots,q\},i\in\mathcal{L}}, \{\one.\state_i\}_{i\in[N] \setminus \mathcal{L}},\{\one.\vk_i\}_{i\in\mathcal{L}})$.
    \item For each $i\in [N] \setminus \mathcal{L}$, compute $\one.\vk_i\la\ONE.\qSim_3(\one.\state_i)$.
    \item Output $\vk\seteq\{\one.\vk_i\}_{i\in[N]}$.
\end{enumerate}
Let us define the sequence of hybrids as follows.

\item [$\hyb_{0}$:] This is identical to $\expd{\Sigma_{\mathsf{cefe}},\qA}{cert}{ever}{ada}{sim}(\secp, 0)$.
\begin{enumerate}
    \item\label{step:keygen_bounded_fe} The challenger generates $(\one.\MPK_i,\one.\MSK_i)\la\ONE.\Setup(1^\secp)$ for $i\in[N]$.
    \item $\qA_1$ is allowed to call key queries at most $q$ times.
    For the $j$-th key query, the challenger receives an function $C_{j}$ from $\qA_1$, generates a uniformly random set $\Gamma_{j}\in[N]$ of size $Dt+1$ and $\Delta_{j}\in [S]$ of size $v$.
    For $i\in \Gamma_{j}$, the challenger generates $\mathsf{one}.\sk_{C_{j},\Delta_{j},i}\la\mathsf{ONE}.\keygen(\one.\MSK_i,G_{C_{j},\Delta_{j}})$, and sends $(\Gamma_{j},\Delta_{j},\{\one.\sk_{C_{j},\Delta_{j},i}\}_{i\in\Gamma_{j}})$ to $\qA_1$.
    Let $q^*$ be the number of times that $\qA_1$ has called key queries in this step.
    \item $\qA_1$ chooses $x\in\Ms$ and sends $x$ to the challenger.
    \item\label{step:enc_bounded_fe} The challenger generates a random degree $t$ polynomial $\mu_i(\cdot)$ whose constant term is $x[i]$ for $i\in[\ell]$
    and a random degree $Dt$ polynomial $\xi_i(\cdot)$ whose constant term is $0$.
    For $i\in[N]$, the challenger computes $(\one.\vk_i,\one.\qct_i)\la\ONE.\qEnc(\one.\MPK_i,(\mu_1(i),\cdots,\mu_{\ell}(i),\xi_1(i),\cdots,\xi_S(i)))$,
    and sends $\{\one.\qct_i\}_{i\in[N]}$ to $\qA_1$.
    \item\label{step:key_bounded_fe_2} $\qA_1$ is allowed to call a key query at most $q-q^*$ times.
    For the $j$-th key query, the challenger receives an function $C_{j}$ from $\qA_1$, generates a uniformly random set $\Gamma_{j}\in[N]$ of size $Dt+1$ and $\Delta_{j}\in [S]$ of size $v$.
    For $i\in \Gamma_{j}$, the challenger generates $\mathsf{one}.\sk_{C_{j},\Delta_{j},i}\la\mathsf{ONE}.\keygen(\one.\MSK_i,G_{C_{j},\Delta_{j}})$, and sends $(\Gamma_{j},\Delta_{j},\{\one.\sk_{C_{j},\Delta_{j},i}\}_{i\in\Gamma_{j}})$ to $\qA_1$.
    \item $\qA_1$ sends $\{\one.\cert_i\}_{i\in[N]}$ to the challenger and its internal state to $\qA_2$.
    \item If $\top\la \ONE.\Vrfy(\one.\vk_i,\one.\cert_i)$ for all $i\in[N]$, the challenger outputs $\top$ and sends $\{\one.\MSK_i\}_{i\in[N]}$ to $\qA_2$.
    Otherwise, the challenger outputs $\bot$ and sends $\bot$ to $\qA_2$.
    \item $\qA_2$ outputs $b$.
    \item The experiment outputs $b$ if the challenger outputs $\top$.
    Otherwise, the experiment outputs $\bot$.
\end{enumerate}
\item[$\hyb_{1}$:] This is identical to $\hyb_{0}$ except for the following three points.
First, the challenger generates uniformly random set $\Gamma_i\in[N]$ of size $Dt+1$ and $\Delta_{i}\in [S]$ of size $v$ for $i\in\{q^*+1,\cdots,q\}$ in step \ref{step:enc_bounded_fe} instead of generating them when a key query is called.
Second, if $\left|\mathcal{L}\right|> t$, the challenger aborts and the experiment outputs $\bot$.
Third, if there exists some $i\in[q]$ such that $\Delta_i \setminus (\bigcup_{j\neq i}\Delta_j)=\emptyset$, the challenger aborts and the experiment outputs $\bot$.

\item[$\hyb_{2}$:]This is identical to $\hyb_{1}$ except that the challenger samples $\xi_1,\cdots,\xi_S, \eta_1,\cdots,\eta_q$ as in the simulator $\qSim_1$ described above.

\item[$\hyb_{3}$:]
This is identical to $\hyb_{2}$ except that the challenger generates $\{\one.\qct_i\}_{i\in[N] \setminus \{\mathcal{L}\}}$,
$\{\one.\sk_{C_j,\Delta_j,i}\}_{i\in\Gamma_j}$ for $j\in\{q^*+1,\cdots ,q'\}$,
and $\vk\seteq\{\one.\vk_i\}_{i\in[N] \setminus \{\mathcal{L}\}}$ as in the simulator $\qSim=(\qSim_1,\qSim_2,\qSim_3)$ described above,
where $q'$ is the number of key queries that the adversary makes in total.

\item[$\hyb_{4}$:]
This is identical to $\hyb_{3}$ except that the challenger generates $\mu_1,\cdots, \mu_\ell$ as in the simulator $\qSim_1$ described above.
\end{description}

From the definition of $\expd{\Sigma_{\mathsf{cefe}},\qA}{cert}{ever}{ada}{sim}(\secp, b)$ and $\qSim=(\qSim_1,\qSim_2,\qSim_3)$,
it is clear that $\Pr[\hyb_{0}=1]=\Pr[\expd{\Sigma_{\mathsf{cefe}},\qA}{cert}{ever}{ada}{sim}(\secp,0)=1]$ and
$\Pr[\hyb_{4}=1]=\Pr[\expd{\Sigma_{\mathsf{cefe}},\qA}{cert}{ever}{ada}{sim}(\secp,1)=1]$.
Therefore, \cref{thm:ever_bounded_fe} easily follows from \cref{prop:hyb_0_hyb_1_bounded_fe,prop:hyb_1_hyb_2_bounded_fe,prop:hyb_2_hyb_3_bounded_fe,prop:hyb_3_hyb_4_bounded_fe} (whose proofs are given later).
\end{proof}

\begin{proposition}\label{prop:hyb_0_hyb_1_bounded_fe}
 $   \abs{\Pr[\hyb_{0}=1]-\Pr[\hyb_{1}=1]}\leq \negl(\lambda).$
\end{proposition}

\begin{proposition}\label{prop:hyb_1_hyb_2_bounded_fe}
    $\Pr[\hyb_{1}=1]=\Pr[\hyb_{2}=1].$
\end{proposition}

\begin{proposition}\label{prop:hyb_2_hyb_3_bounded_fe}
If $\Sigma_{\one}$ is $1$-bounded adaptive certified everlasting simulation-secure,
\begin{align}
    \abs{\Pr[\hyb_{2}=1]-\Pr[\hyb_{3}=1]}\leq \negl(\lambda).
\end{align}
\end{proposition}

\begin{proposition}\label{prop:hyb_3_hyb_4_bounded_fe}
    $\Pr[\hyb_{3}=1]=\Pr[\hyb_{4}=1].$
\end{proposition}

\begin{proof}[Proof of \cref{prop:hyb_0_hyb_1_bounded_fe}]
Let $\hybij{0}{\prime}$ be the experiment identical to $\hyb_{0}$ except that the challenger generates a set $\Gamma_i\in[N]$ and $\Delta_{i}\in [S]$ for $i\in\{q^*+1,\cdots ,q\}$
in step~\ref{step:enc_bounded_fe}.
It is clear that $\Pr[\hyb_{0}=1]=\Pr[\hybij{0}{\prime}=1]$. 

Let $\hybij{0}{*}$ be the experiment identical to $\hybij{0}{\prime}$ except that it outputs $\bot$ if $|\mathcal{L}|>t$.
It is clear that $\Pr[\hybij{0}{\prime}=1\wedge (|\mathcal{L}|\leq t)]=\Pr[\hybij{0}{*}=1\wedge (|\mathcal{L}|\leq t)]$.
Hence, it holds that
\begin{align}
    \abs{\Pr[\hybij{0}{\prime}=1]-\Pr[\hybij{0}{*}=1]}\leq \Pr[|\mathcal{L}|>t]
\end{align}
from \cref{lem:defference}.

Let $\mathsf{Collide}$ be the event that
there exists some $i\in[q]$ such that $\Delta_i \setminus (\bigcup_{j\neq i}\Delta_j)=\emptyset$.
$\hybij{0}{*}$ is identical to $\hyb_{1}$ when $\mathsf{Collide}$ does not occur.
Hence, it is clear that $\Pr[\hybij{0}{*}=1\wedge\overline{\mathsf{Collide}}]=\Pr[\hyb_{1}=1\wedge \overline{\mathsf{Collide}}]$.
Therefore, it holds that 
\begin{align}
    \abs{\Pr[\hybij{0}{*}=1]-\Pr[\hyb_{1}=1]}\leq \Pr[\mathsf{Collide}]
\end{align}
from \cref{lem:defference}.

From the discussion above, we have 
\begin{align}
    \abs{\Pr[\hyb_{0}=1]-\Pr[\hyb_{1}=1]}\leq \Pr[|\mathcal{L}|>t]+\Pr[\mathsf{Collide}].
\end{align}
The following \cref{lem:small_pair,lem:coverfree} shows that 
$\Pr[|\mathcal{L}|>t]\leq 2^{-\Omega(\lambda)}$
and $\Pr[\mathsf{Collide}]\leq q2^{-\Omega(\lambda)}$, which completes the proof.
\end{proof}

\begin{lemma}[\cite{C:GorVaiWee12}]\label{lem:small_pair}
Let $\Gamma_1,\cdots ,\Gamma_q\subseteq [N]$ be randomly chosen subsets of size $tD+1$.
Let $t=\Theta(q^2\lambda)$ and $N=\Theta(D^2q^2 t)$.
Then, 
\begin{align}
    \Pr\left[\left|\bigcup_{i\neq i'}(\Gamma_i\cap \Gamma_i)\right|>t\right]\leq2^{-\Omega(\lambda)}
\end{align}
where the probability is over the random choice of the subsets $\Gamma_1,\cdots,\Gamma_q$.
\end{lemma}
\begin{lemma}[\cite{C:GorVaiWee12}]\label{lem:coverfree}
Let $\Delta_1,\cdots,\Delta_q\subseteq[S]$ be randomly chosen subsets of size $v$.
Let $v(\lambda)=\Theta(\lambda)$ and $S(\lambda)=\Theta(vq^2)$.
Let $\mathsf{Collide}$ be the event that there exists some $i\in[q]$ such that $\Delta_i\backslash(\bigcup_{j\neq i}\Delta_j)=\emptyset$.
Then, we have
\begin{align}
    \Pr\left[\mathsf{Collide}\right]\leq q2^{-\Omega(\lambda)}
\end{align}
where the probability is over the random choice of subsets $\Delta_1,\cdots ,\Delta_q$.
\end{lemma}

\begin{proof}[Proof of \cref{prop:hyb_1_hyb_2_bounded_fe}]
In the encryption in $\hyb_{1}$, $\xi_{a^*}$ is chosen at random and $\eta_j(\cdot)$ is defined by the relation.
$\qSim$ essentially chooses $\eta_j(\cdot)$ at random which defines $\xi_{a^*}$.
It is easy to see that reversing the order of how the polynomials are chosen produces the same distribution.
\end{proof}

\begin{proof}[Proof of \cref{prop:hyb_2_hyb_3_bounded_fe}]
To prove the proposition, let us define a hybrid experiment $\hyb_{2}{s}$ for each $s\in[N]$ as follows.
\begin{description}
\item[$\hybij{2}{s}$:] This is identical to $\hyb_{2}$ except for the following three points.
First, the challenger generates $\{\one.\qct_i\}_{i\in[s]\setminus\mathcal{L}}$ as in the simulator $\qSim_1$.
Second, the challenger generates $\{\one.\sk_{C_j,\Delta_j,i}\}_{i\in\Gamma_j\cap [s]}$ for $j\in\{q^*+1,\cdots ,q'\}$ as in the simulator $\qSim_2$,
where $q'$ is the number of key queries that the adversary makes in total.
Third, the challenger generates $\{\one.\vk_i\}_{i\in[s]\setminus\mathcal{L}}$ as in the simulator $\qSim_3$.
\end{description}
Let us denote $\hyb_{2}$ as $\hybij{2}{0}$.
It is clear that $\Pr[\hybij{2}{N}=1]=\Pr[\hyb_{3}=1]$.
Furthermore, we can show that
\begin{align}
    \abs{\Pr[\hybij{2}{s-1}=1]-\Pr[\hybij{2}{s}=1]}\leq \negl(\lambda)
\end{align}
for $s\in[N]$. (Its proof is given later.)
From these facts, we obtain \cref{prop:hyb_2_hyb_3_bounded_fe}.

Let us show the remaining one.
In the case $s\in\mathcal{L}$, it is clear that $\hybij{2}{s-1}$ is identical to $\hybij{2}{s}$.
Hence, we consider the case $s\notin\mathcal{L}$.
To show the inequality above, let us assume that $\abs{\Pr[\hybij{2}{s-1}=1]-\Pr[\hybij{2}{s}=1]}$ is non-negligible.
Then, we can construct an adversary $\qB$ that can break the $1$-bounded adaptive certified everlasting simulation-security of $\Sigma_{\one}$ as follows.

\begin{enumerate}
    \item $\qB$ receives $\one.\MPK$ from the challenger of $\expd{\Sigma_{\mathsf{one}},\qA}{cert}{ever}{ada}{sim}(\secp,b)$.
    $\qB$ sets $\one.\MPK_s\seteq \one.\MPK$.
    \item $\qB$ generates $(\one.\MPK_i,\one.\MSK_i)\la\ONE.\Setup(1^\secp)$ for all $i\in[N] \setminus s$, and sends $\{\one.\MPK_i\}_{i\in[N]}$ to $\qA_1$.
    \item $\qA_1$ is allowed to call key queries at most $q$ times.
    For the $j$-th key query, $\qB$ receives an function $C_{j}$ from $\qA_1$, generates a uniformly random set $\Gamma_{j}\in[N]$ of size $Dt+1$ and $\Delta_{j}\in [S]$ of size $v$.
    For $i\in \Gamma_{j} \setminus s$, $\qB$ generates $\mathsf{one}.\sk_{C_{j},\Delta_{j},i}\la\mathsf{ONE}.\keygen(\one.\MSK_i,G_{C_{j},\Delta_{j}})$.
    If $s\in\Gamma_j$, $\qB$ sends $G_{C_j,\Delta_j}$ to the challenger,
    receives $\one.\sk_{C_j,\Delta_j,s}$ from the challenger,
    and sends $(\Gamma_{j},\Delta_{j},\{\one.\sk_{C_{j},\Delta_{j},i}\}_{i\in\Gamma_{j}})$ to $\qA_1$.
    Let $q^*$ be the number of times that $\qA_1$ has called key queries in this step.
    \item $\qA_1$ chooses $x\in\Ms$, and sends $x$ to $\qB$.
    \item\label{step:enc_bounded_fe_hyb_2} $\qB$ generates uniformly random set $\Gamma_i\in[N]$ of size $Dt+1$ and $\Delta_{i}\in [S]$ of size $v$ for $i\in\{q^*+1,\cdots,q\}$.
    $\qB$ generates a random degree $t$ polynomial $\mu_i(\cdot)$ whose constant term is $x[i]$ for $i\in[\ell]$, and $\xi_1,\cdots,\xi_S, \eta_1,\cdots,\eta_q$ as in the simulator $\qSim_1$.
    For $i\in[s-1] \setminus \mathcal{L}$, $\qB$ generates $\one.\qct_i$ as in the simulator $\qSim_1$.
    For $i\in\{s+1,\cdots N\}\cup \mathcal{L}$, $\qB$ generates $(\one.\vk_i,\one.\qct_i)\la \ONE.\qEnc(\one.\MPK_i, (\mu_1(i),\cdots,\mu_\ell(i),\xi_1(i),\cdots, \xi_S(i)))$.
    $\qB$ sends $\mu_1(s),\cdots,\mu_\ell(s),\allowbreak\xi_1(s),\cdots ,\xi_S(s)$ to the challenger, and receives $\one.\qct_s$ from the challenger. 
    $\qB$ sends $\{\one.\qct_i\}_{i\in[N]}$ to $\qA_1$.
    \item $\qA_1$ is allowed to call key queries at most $q-q^*$ times.
    For the $j$-th key query, $\qB$ receives an function $C_j$ from $\qA_1$.
    For $i\in\Gamma_j \setminus [s]$, $\qB$ generates $\one.\sk_{C_j,\Delta_j,i}\la\ONE.\keygen(\one.\MSK_i,G_{C_j,\Delta_j})$.
    For $i\in \Gamma_j\wedge[s-1]$, $\qB$ generates $\one.\sk_{C_j,\Delta_j,i}$ as in the simulator $\qSim_2$.
    If $s\in\Gamma_j$, $\qB$ sends $G_{C_j,\Delta_j}$ to the challenger, and receives $\one.\sk_{C_j,\Delta_j,s}$ from the challenger.
    $\qB$ sends $(\Gamma_j,\Delta_j,\{\one.\sk_{C_j,\Delta_j,i}\}_{i\in \Gamma_j})$ to $\qA_1$.
    \item For $i\in[s-1] \setminus \mathcal{L}$, $\qB$ generates $\one.\vk_i$ as in the simulator $\qSim_3$
    \footnote{
    For $i\in\{s+1,\cdots N\}\cup \mathcal{L}$, $\qB$ generated $\one.\vk_i$ in step~\ref{step:enc_bounded_fe_hyb_2}.
    }.
    \item $\qA_1$ sends $\{\one.\cert_i\}_{i\in[N]}$ to $\qB$ and its internal state to $\qA_2$.
    \item $\qB$ sends $\one.\cert_s$ to the challenger, and receives $\one.\MSK_s$ or $\bot$ from the challenger.
    $\qB$ computes $\ONE.\Vrfy(\one.\vk_i,\allowbreak \one.\cert_i)$ for all $i\in[N] \setminus s$.
    If the results are $\top$ and $\qB$ receives $\one.\MSK_s$ from the challenger, $\qB$ sends $\{\one.\MSK_i\}_{i\in[N]}$ to $\qA_2$.
    Otherwise, $\qB$ aborts.
    \item $\qA_2$ outputs $b'$.
    \item $\qB$ outputs $b'$.
\end{enumerate}
It is clear that $\Pr[1\la\qB \mid b=0]=\Pr[\hybij{2}{s-1}=1]$ and $\Pr[1\la\qB\mid b=1]=\Pr[\hybij{2}{s}=1]$.
By assumption, $\abs{\Pr[\hybij{2}{s-1}=1]-\Pr[\hybij{2}{s}=1]}$ is non-negligible, and therefore $\abs{\Pr[1\la\qB \mid b=0]-\Pr[1\la\qB \mid b=1]}$ is non-negligible, which contradicts the $1$-bounded adaptive certified everlasting simulation-security of $\Sigma_{\one}$.
\end{proof}

\begin{proof}[Proof of \cref{prop:hyb_3_hyb_4_bounded_fe}]
In $\hyb_{3}$, the polynomials $\mu_1,\cdots ,\mu_\ell$ are chosen with constant terms $x_1,\cdots ,x_\ell$, respectively.
In $\hyb_{4}$, these polynomials are now chosen with $0$ constant terms.
This only affects the distribution of $\mu_1,\cdots, \mu_\ell$ themselves and polynomials $\xi_1,\cdots,\xi_S$.
Moreover, only the evaluations of these polynomials on the points in $\mathcal{L}$ affect the outputs of the experiments.
Now observe that:
\begin{itemize}
    \item The distribution of the values $\{\mu_1(i),\cdots, \mu_\ell(i)\}_{i\in\mathcal{L}}$ are identical to both $\hyb_{3}$ and $\hyb_{4}$.
    This is because in both experiments, we choose these polynomials to be random degree $t$ polynomials (with different constraints in the constant term),
    so their evaluation on the points in $\mathcal{L}$ are identically distributed, since $|\mathcal{L}|\leq t$.
    \item The values $\{\xi_1(i),\cdots ,\xi_S(i)\}_{i\in\mathcal{L}}$ depend only on the values $\{\mu_1(i),\cdots,\mu_{\ell}(i)\}_{i\in\mathcal{L}}$.
\end{itemize}
\cref{prop:hyb_3_hyb_4_bounded_fe} follows from these observations.
\end{proof}

\subsection{Discussion on \texorpdfstring{$q$}{q}-Bounded Consturction for All Circuits}\label{sec:discussion_bounded_FE_all_circuits}
We discuss technical hurdles to achieve certified everlasting secure bounded collusion-resistant FE for $\Ppoly$ from standard PKE.

Gorbunov, Vaikuntanathan, and Wee~\cite{C:GorVaiWee12} presented a conversion from FE for $\NCone$ to FE for $\Ppoly$ by using randomized encoding or FHE. However, we cannot directly apply their techniques in the certified everlasting setting.
When we use randomized encoding, we use a functional decryption key for circuit $G_{f}$ that takes $m$ as an input and outputs a randomized encoding $\wtl{f(m)}$.\footnote{For simplicity, we ignore how to set randomness for randomized encoding here since it is not an essential issue.}
That is, we can obtain $\wtl{f(m)}$ (and $f(m)$ via a decoding algorithm) from the functional decryption key and ciphertext of $m$ since randomized encoding is computable in a constant-depth circuit~\cite{CC:AppIshKus06}.

The first problem is that even if we use \emph{certified everlasting secure} FE for $\NCone$, information about $m$ remains in $\wtl{f(m)}$ since the decryption result does not directly provide $f(m)$. More specifically, adversaries can keep $\wtl{f(m)}$ (this is classical information) before deletion and an unbounded adversary could recover $m$ from $\wtl{f(m)}$ even after $\Enc(m)$ was deleted.

The second problem is that we cannot use certified everlasting secure randomized encoding to solve the first problem since we use FE for \emph{classical} circuits here. In certified everlasting secure randomized encoding, $\wtl{f(m)}$ must be quantum state, which cannot be supported by FE for classical circuits. We do not have certified everlasting secure FE that supports quantum circuits computing quantum state. Moreover, we do not know how to achieve certified everlasting secure randomized encoding. Thus, the approach using randomized encoding does not work.

The approach using FHE also has problems. In this approach, we consider a functional decryption key for circuit $G_f$ that takes an FHE ciphertext $\fhe.\ct$ and an FHE decryption key $\fhe.\sk$ and outputs $\fhe.\ct$ and $\FHE.\Dec(\fhe.\sk,\FHE.\Eval(f,\fhe.\ct))$. Here, we must output $\fhe.\ct$ as the public part because we use FE for $\NCone$ and need to apply $f\in \Ppoly$ by $\FHE.\Eval$ in the public part (though the FHE decryption part is in $\NCone$).\footnote{See~\cite{C:GorVaiWee12} for the detail.} That is, the FHE part must be also certified everlasting secure.

First, we cannot use certified everlasting secure FHE in a black-box way. We need to encrypt an FHE ciphertext by FE for $\NCone$ in this approach. However, if FHE is certified everlasting secure, a ciphertext is quantum state, which is not supported by our certified everlasting secure FE for $\NCone$.

Second, even if we use certified everlasting secure FHE in a non-black-box way like our compute-and-compare obfuscation construction in~\cref{sec:CCO_CED_const} (by separating the classical FHE part from the BB84 state), the approach does not work due to the following reason. To achieve certified everlasting security, $\fhe.\ct$ is an encryption of $m \xor \bigoplus_{i}\theta_i$ where $\theta$ is a basis choice as in~\cref{sec:CCO_CED_const}. To unmask $\bigoplus_i \theta_i$, we need to coherently apply $f$ to $\fhe.\ct$ and BB84 state as the certified everlasting secure FHE by Bartusek and Khurana~\cite{myC:BarKhu23}. However, we cannot execute the coherent evaluation in the FE decryption mechanism (cannot take BB84 state as input). Hence, we obtain $f(m \xor \bigoplus_i \theta_i)$ and the correctness does not hold. Thus, the approach using FHE does not work too.

Another plausible (but failed) approach is using the framework by Ananth and Vaikuntanathan~\cite{TCC:AnaVai19}. They constructed bounded collusion-resistant FE for $\Ppoly$ \emph{without the bootstrapping method} by Gorbunov et al.~\cite{C:GorVaiWee12}. However, their construction heavily relies on a secure multi-party computation protocol based on \emph{PRG}. It is hard to define certified everlasting security for PRG because there is nothing to delete. Thus, it is unclear how to use their framework in the certified everlasting setting.

Therefore, previous approaches for converting FE for $\NCone$ to FE for $\Ppoly$ do not work in the certified everlasting setting.

\section{Compute-and-Compare Obfuscation with Certified Everlasting Deletion}\label{sec:ccobf}

\subsection{Definition}\label{sec:def_CCObf_CED}
In this section, we introduce the notion of compute-and-compare obfuscation with certified everlasting security.

\begin{definition}[Compute-and-Compare Obfuscation with Certified Everlasting Deletion (Syntax)]\label{def:CCObf_CED_syntax}
A compute-and-compare obfuscation with certified everlasting deletion is a tuple of algorithms $(\qCCObf,\qdel,\vrfy)$ for the family of distributions $D=\{D_\param\}_\param$ and message space $\cM$.
\begin{description}
	\item[$\qCCObf(1^{\secp}, P, \lock, m)$:] The obfuscation algorithm takes as input a security parameter $1^{\secp}$, a circuit $P$, a lock string $\lock \in \{0, 1\}^{p(\secp)}$ and a message $m \in \cM$, and outputs an obfuscated circuit $\tlqP$ and a verification key $\vk$.

	\item[$\qdel(\tlqP) \rightarrow \cert$:] The deletion algorithm takes as input an obfuscated circuit $\tlqP$ and outputs a classical certificate $\cert$.
	
	\item[$\vrfy(\vk, \cert) \rightarrow \top \textbf{ or }\bot$:] The verification algorithm takes as input the verification key $\vk$ and a certificate $\cert$, and outputs $\top$ or $\bot$.
\end{description}
\end{definition}

\begin{definition}[Correctness of Compute-and-Compare Obfuscation with Certified Everlasting Deletion]\label{def:CCObf_CED_correctness}
The correctness of compute-and-compare obfuscation with certified everlasting deletion for the family of distributions $D=\setbk{D_\param}_{\param}$ and message space $\cM$ is defined as follows.
\begin{description}
\item[Functionality Preserving:] There exists a negligible function $\negl$ such that for all circuit $P$, all lock value $\lock$, and all message $\msg \in \cM$, it holds that
\begin{align}
\Pr[\forall x, \tlqP(x)=\cnc{P}{\lock,\msg}(x) \mid \tlqP\gets \qCCObf(1^\secp,P,\lock,\msg)]\ge1-\negl(\secp).
\end{align}

\item[Verification Correctness:] There exists a negligible function $\negl$ such that for all circuit $P$, all lock value $\lock$, and all message $\msg\in\cM$, it holds that
	\begin{align}
		\Pr\left[ \vrfy(\vk, \cert) \ne \top \ \middle| \begin{array}{l} (\tlqP, \vk)\leftarrow \qCCObf(1^{\secp}, P, \lock, \msg)\\
			\cert \leftarrow \qdel(\tlqP)\\
		 \end{array} \right] \le \negl(\secp).
	\end{align}
\end{description}
\end{definition}

\begin{definition}[Certified Everlasting Security of Compute-and-Compare Obfuscation]\label{def:CCObf_CED_security}
Let $\Sigma_\CCO=(\qCCObf,\qdel, \vrfy)$ be a compute-and-compare obfuscation with certified everlasting deletion for the family of distributions $D=\setbk{D_\param}_{\param}$ and a message space $\cM$. We consider experiments $\EV\expb{\Sigma_{\CCO},~ \qA}{sim}{ccobf}(\secp, b)$ and $\C\expb{\Sigma_{\CCO},~ \qA}{sim}{ccobf}(\secp, b)$ played between a challenger and a non-uniform QPT adversary $\qA = \{\qA_{\secp}, \ket{\psi}_{\secp}\}_{\secp \in \mbb{N}}$. Let $\qsim$ be a QPT algorithm. The experiments are defined as follows:
	\begin{enumerate}
		\item $\qA_{\secp}(\ket{\psi}_{\secp})$ submits a message $\msg \in \cM$ to the challenger.
		\item The challenger chooses $(P,\lock,\qaux) \chosen D_\param$.
		\item The challenger computes $(\tlqP^{(0)}, \vk^{(0)}) \leftarrow \qCCObf(1^{\secp}, P, \lock, \msg)$ or $(\tlqP^{(1)}, \vk^{(1)}) \leftarrow \qsim(1^{\secp}, \pp_P, 1^{|\msg|})$ and sends $(\tlqP^{(b)},\qaux)$ to $\qA_{\secp}$ according to the bit $b$. Recall that a program $P$ has an associated set of parameters $\pp_P$ (input size, output size, circuit size) which we do not need to hide.
		\item $\qA_{\secp}$ submits a certificate of deletion $\cert$ and its internal state $\rho$ to the challenger.
		\item The challenger computes $\vrfy(\vk^{(b)}, \cert)$. If the outcome is $\top$, the experiment $\EV\expb{\Sigma_{\CCO},~ \qA}{sim}{ccobf}(\secp, b)$ outputs $\rho$; otherwise if the outcome is $\bot$ then $\EV\expb{\Sigma_{\CCO},~ \qA}{sim}{ccobf}(\secp, b)$ outputs $\bot$ and ends.
		\item The challenger sends the outcome of $\vrfy(\vk^{(b)}, \cert)$ to $\qA_{\secp}$.
		\item $\qA_{\secp}$ outputs its guess $b' \in \{0, 1\}$ which is the output of the experiment $\C\expb{\Sigma_{\CCO},~ \qA}{sim}{ccobf}(\secp, b)$.
	\end{enumerate}
	We say that the $\Sigma_{\CCO}$ is certified everlasting  secure if for any non-uniform QPT adversary $\qA = \{\qA_{\secp}, \ket{\psi}_{\secp}\}_{\secp \in \mbb{N}}$, it holds that 
	\begin{align}
		\TD(\EV\expb{\Sigma_{\CCO},~ \qA}{sim}{ccobf}(\secp, 0), \EV\expb{\Sigma_{\CCO},~ \qA}{sim}{ccobf}(\secp, 1)) \le \negl(\secp),
	\end{align}	
and 	
	\begin{align}
		\left|\Pr[\C\expb{\Sigma_{\CCO},~ \qA}{sim}{ccobf}(\secp, 0) = 1] - \Pr[\C\expb{\Sigma_{\CCO},~ \qA}{sim}{ccobf}(\secp, 1) = 1]\right| \le \negl(\secp).
	\end{align}
\end{definition}


\subsection{Construction}\label{sec:CCO_CED_const}
In this section, we construct a compute-and-compare obfuscation with certified everlasting deletion from classical compute-and-compare obfuscation and FHE.

\paragraph{Ingredients.} We use the following building blocks.
\begin{enumerate}
	\item $\Sigma_{\msf{fhe}} = \FHE.(\keygen, \enc, \Eval, \dec)$ be a classical FHE scheme.
	\item $\Sigma_{\CCO} = \CCObf$ be a classical compute-and-compare obfuscation scheme. 
\end{enumerate}

\paragraph{Certified everlasting compute-and-compare obfuscation for multi-bit message.} We construct $\Sigma_{\CECCO} = (\qCCObf, \allowbreak \qdel, \vrfy)$ for the family of distribution $D=\setbk{D_\param}_\param$ and message space $\cM$. We let the message space $\cM \seteq  \{0, 1\}^n $.

\begin{description}
\item[$\qCCObf(1^{\secp}, P, \lock, \msg)$:]$\vspace{0.01cm}$
\begin{enumerate}
	\item Sample $R \leftarrow \{0, 1\}^{\secp}$.
	\item Sample $(\fpk, \fsk) \leftarrow \FHE.\keygen(1^{\secp})$.
	\item Compute $\tlfDec \leftarrow \CCObf(1^{\secp}, \fDec, R, 1)$ where $\fDec(\cdot) = \FHE.\dec(\fsk, \cdot)$.
	\item Compute $\tlI \leftarrow \CCObf(1^{\secp}, I, \lock, R)$ where $I(X) = X$ for every $X$.
	\item Represent $(P \concat \tlI) = (b_1, \ldots, b_{\ell})  \in \{0, 1\}^{\ell}$.
	\item Sample $\bm{\theta}_{i}, \bm{z}_{i} \leftarrow \{0, 1\}^{\secp}$ for all $i \in [\ell]$.
	\item Set $\wt{b}_{i} \seteq b_{i} \oplus \bigoplus_{j: \theta_{i, j} = 0}z_{i, j}$ for all $i \in [\ell]$.
	\item Denote $\msg = (\msg_1, \dots, \msg_n) \in \{0, 1\}^n$.
	\item Sample $\bm{\theta}_{\ell+k}, \bm{z}_{\ell+k} \leftarrow \{0, 1\}^{\secp}$ for all $k \in [n]$.
	\item Set $\wt{b}_{\ell+k} \seteq \msg_{k} \oplus \bigoplus_{j: \theta_{\ell+k, j} = 0}z_{\ell+k, j}$ for all $k \in [n]$.
	\item Compute $\fct_{i} \leftarrow \FHE.\enc(\fpk, (\bm{\theta}_{i}, \wt{b}_{i}))$ for all $i \in [\ell+n]$.
	\item Output $\tlqP \seteq (\tlfDec, \{(\ket{\bm{z}_{i}}_{\bm{\theta}_{i}}, \fct_{i})\}_{ i\in [\ell+n]}, \fpk)$ and $\vk\seteq (\{\bm{z}_{i}, \bm{\theta}_{i}\}_{i \in [\ell+n]})$. 
\end{enumerate}

\item[How to evaluate $\tlqP(x)$:]$\vspace{0.01cm}$
\begin{enumerate}
	\item Parse $\tlqP \seteq (\tlfDec, \{(\ket{\bm{z}_{i}}_{\bm{\theta}_{i}}, \fct_{i})\}_{ i\in [\ell+n]}, \fpk)$.
	\item  Define a circuit $\wh{U}_{x}$ as in Figure \ref{fig:lobf2}.
	\item  To compute an evaluated ciphertext for $\FHE.\Eval(\fpk, \wh{U}_{x}, \cdot)$, apply $\wh{U}_{x}$ homomorphically in superposition with the input $(\{(\ket{\bm{z}_i}_{\bm{\theta}_i}, \fct_i)\}_{i\in [\ell]}, (\ket{\bm{z}_{\ell+k}}_{\bm{\theta}_{\ell+k}}, \fct_{\ell+k}))$ and obtain a ciphertext $\ket{\fct_{\ell+k, P}}$ for each $k \in [n]$.
	\item Apply $\tlfDec(\cdot)$ in superposition with the input $\ket{\fct_{\ell+k, P}}$ and measure the output register in the standard basis to get a classical outcome $\beta_k$ for each $k \in [n]$.
	\item Set $\msg_k = 1$ if $\beta_k = 1$, else set $\msg_k = 0$, for each $k \in [n]$. 
	\item Output $\msg = (\msg_1, \dots, \msg_n)$.
\end{enumerate}

\begin{figure}
	\begin{framed}
		   \begin{center}
			   \underline{Circuit $\wh{U}_{x}$}
			   \end{center}

		\textbf{Hardwire:}~ $x$\newline
		\textbf{Input:}~ $(\{(\bm{z}_{i}, \bm{\theta}_{i}, \wt{b}_i)\}_{i \in [\ell]}, (\bm{z}_{\ell+k}, \bm{\theta}_{\ell+k}, \wt{b}_{\ell+k}))$ 
		\begin{enumerate}
			\item Compute $b_i \seteq \wt{b}_i \oplus \bigoplus_{j: \theta_{i, j} = 0}z_{i, j}$ for all $i \in [\ell]$.
			\item Reconstruct $(C \concat \tlI)$ from $(b_1, \ldots, b_{\ell})$.
			\item Compute $m_{k} \seteq \wt{b}_{\ell+k} \oplus \bigoplus_{j: \theta_{\ell+k, j} = 0}z_{\ell+k, j}$.
			\item Output $m_{k} \cdot \tlI(C(x))$
		\end{enumerate}
	\end{framed}
	\caption{The description of the circuit $\wh{U}_{x}$}
	\label{fig:lobf2}
\end{figure}

\item[$\qdel(\tlqP)$:]$\vspace{0.01cm}$
\begin{enumerate}
	\item Parse $\tlqP \seteq (\tlfDec, \{(\ket{\bm{z}_{i}}_{\bm{\theta}_{i}}, \fct_{i})\}_{ i\in [\ell+n]}, \fpk)$.
	\item Measure $\ket{\bm{z}_i}_{\bm{\theta}_i}$ in the Hadamard basis for all $i \in [\ell+n]$, and obtain $(\bm{z}_1', \ldots, \bm{z}_{\ell+n}')$.
	\item Output $\cert \seteq (\bm{z}_1', \ldots, \bm{z}_{\ell+n}')$.
\end{enumerate}

\item[$\vrfy(\vk, \cert)$:]$\vspace{0.01cm}$
\begin{enumerate}
	\item Parse $\vk = (\{(\bm{z}_i, \bm{\theta}_i)\}_{i \in [\ell+n]})$ and $\cert = (\bm{z}_1', \ldots, \bm{z}_{\ell+n}')$.
	\item If $\left((z_{i, j} = z_{i, j}') \wedge (\theta_{i, j} = 1)\right)$ holds for all $i \in [\ell+n]$ and $j \in [\lambda]$, then output $\top$; otherwise output $\bot$. 
\end{enumerate}
\end{description}

\paragraph{Security.}
We use \cref{lem:ce} by Bartusek and Khurana~\cite{myC:BarKhu23} to prove the security of our construction.

\begin{theorem}\label{thm:LObf_CED}
	If $\Sigma_{\CCO}$ is a secure compute-and-compare obfuscation and $\Sigma_{\tsf{fhe}}$ is an IND-CPA secure fully homomorphic encryption then $\Sigma_{\CECCO}$ is a certified everlasting secure compute-and-compare obfuscation scheme for the family of distribution $D=\setbk{D_\param}_\param$.
\end{theorem}	
\begin{proof}[Proof of~\cref{thm:LObf_CED}]
	We first describe the simulator of $\qCCObf$, denoted as $\qsim$, before we go to the formal security analysis. Let $\CCO.\sm$ be the simulator the classical compute-and-compare obfuscation employed as a building block in the above construction. For $(P,\lock,\qaux)\chosen D_\param$, the algorithm $\qsim$ works as follows:
	\begin{description}
	\item[$\qsim(1^{\secp}, \pp_P, 1^n)$:]$ $
	\begin{enumerate}
		\item Sample $R \leftarrow \{0, 1\}^{\secp}$.
		\item Sample $(\fpk, \fsk) \leftarrow \FHE.\keygen(1^{\secp})$.
		\item Compute $\tlfDec \leftarrow \CCObf(1^{\secp}, \fDec, R, 1)$.
		\item Sample $\bm{\theta}_i, \bm{z}_i \leftarrow \{0, 1\}^{\secp}$ for all $i \in [\ell+n]$.
		\item Set $\wt{b}_i \seteq 0 \oplus \bigoplus_{j: \theta_{i, j} = 0}z_{i, j}$ for all $i \in [\ell+n]$.
		\item Compute $\fct_{i} \leftarrow \FHE.\enc(\fpk, (\bm{\theta}_{i}, \wt{b}_{i}))$ for all $i \in [\ell+n]$.
		\item Output $\tlqP \seteq (\tlfDec, \{(\ket{\bm{z}_{i}}_{\bm{\theta}_{i}}, \fct_{i})\}_{ i\in [\ell+n]}, \fpk)$ and $\vk\seteq (\{\bm{z}_{i}, \bm{\theta}_{i}\}_{i \in [\ell+n]})$. 
	\end{enumerate}	
	\end{description}
	\noindent Note that $\qsim$ does not need information about $(P,\lock,\qaux)$ except $\pp_P$. We show that 
	\begin{equation}
		\TD(\EV\expb{\Sigma_{\CECCO},~ \qA}{sim}{ccobf}(\secp, 0), \EV\expb{\Sigma_{\CECCO},~ \qA}{sim}{ccobf}(\secp, 1))
		\le \negl(\secp).
	\end{equation} 
	using Lemma \ref{lem:ce} and the post-quantum security of $\Sigma_{\CCO}$ and $\Sigma_{\msf{fhe}}$. We use the following sequence of hybrids to prove this. 
	\begin{description}
		\item[$\hyb_0$:] This is the same as $\EV\expb{\Sigma_{\CECCO},~ \qA}{sim}{ccobf}(\secp, 0)$. Let $\tlqP^{(0)} \seteq (\tlfDec, \{(\ket{\bm{z}_{i}}_{\bm{\theta}_{i}}, \fct_{i})\}_{ i\in [\ell+n]}, \fpk)$ be the compute-and-compare obfuscated circuit computed using the honest $\qCCObf$ algorithm.
		
		\item[$\hyb_1$:] This hybrid works as follows:
		 \begin{enumerate}
		 	\item $\qA$ submits a message $\msg \in \zo{n}$ to the challenger.
		 	\item The challenger chooses $(P,\lock,\qaux ) \chosen D_\param$.
		 	\item The challenger computes the obfuscated circuit as follows:
		 	\begin{enumerate}
		 		\item Sample $(\fpk, \fsk) \leftarrow \FHE.\keygen(1^{\secp})$ and $R \leftarrow \{0, 1\}^{\secp}$.
		 		\item Compute $\tlfDec \leftarrow \CCObf(1^{\secp}, \fDec, R, 1)$.
		 		\item Sample $\bm{\theta}_i, \bm{z}_i \leftarrow \{0, 1\}^{\secp}$ for all $i \in [\ell+n]$.
		 		\item Set $\wt{b}_{i} \seteq 0 \oplus \bigoplus_{j: \theta_{i,j} = 0} z_{i,j}$ for $i \in [\ell]$.
		 		\item Set $\wt{b}_{\ell+k} \seteq \msg_k \oplus \bigoplus_{j: \theta_{\ell+k, j} = 0}z_{\ell+k, j}$ for all $k \in [n]$.
		 		\item Compute $\fct_{i} \leftarrow \FHE.\enc(\fpk, (\bm{\theta}_{i}, \wt{b}_{i}))$ for all $i \in [\ell+n]$.
		 		\item Set $\tlqP \seteq (\tlfDec, \{(\ket{\bm{z}_{i}}_{\bm{\theta}_{i}}, \fct_{i})\}_{ i\in [\ell+n]}, \fpk)$.
		 	\end{enumerate}
		 	The challenge sends $\tlqP$ to $\qA$.
		 	\item $\qA$ sends a deletion certificate $\cert \seteq (\bm{z}_1', \ldots, \bm{z}_{\ell+n}')$ and its internal state $\rho$ to the challenger. 
		 	\item The challenger checks if $\left((z_{i, j} = z_{i, j}') \wedge (\theta_{i, j} = 1)\right)$ holds for all $i \in [\ell+n]$ and $j \in [\lambda]$. If the check fails, the experiment halts and returns $\bot$; otherwise, go to the next step.
		 	\item The experiment outputs $\rho$ as a final output.
		 \end{enumerate}	
		Note that, the FHE ciphertexts $\{\fct_{i}\}_{i \in [\ell]}$ contain no information about the lock value $\lock$, the random string $R$ and the circuit $P$. To prove the indistinguishability between $\hyb_0$ and $\hyb_1$, we consider a sequence of intermediate hybrids $\hyb_{1, k}$ for $k \in [0, \ell]$ where $\hyb_{1, 0}$ is identical to $\hyb_0$ and the only difference between $\hyb_{1, k-1}$ and $\hyb_{1, k}$ is that $\fct_k$ is an encryption of $(\bm{\theta}_k, b_{k} \oplus \bigoplus_{j: \theta_{k, j} = 0}z_{k, j})$ where $b_k$ in $\hyb_{1, k-1}$ is the same as $b_k$ in $\hyb_0$ and $b_k$ in $\hyb_{1, k}$ is set to zero for $k \in [\ell]$.\par
		
		Now, we consider a sequence of experiments $\expt{\qB,\qChal}{1, k}(\secp,\bm{\theta},\beta)$ for $k \in [\ell]$ between a QPT adversary $\qB$ and a challenger $\qChal$ for $\bm{\theta}  \in \{0, 1\}^{\secp}$ and $\beta \in \{0, 1\}$. The experiment $\expt{\qB,\qChal}{1, k}(\secp,\bm{\theta},\beta)$ is basically the same as $\hyb_{1, k}$ where we take $\bm{\theta}_k = \bm{\theta}$, $\wt{b}_k = \beta$ and $\qB$ plays the role of $\qA$, $\qChal$ plays the role of the challenger. In particular, it works as follows:
		
		\begin{enumerate}
			\item[$\expt{\qB,\qChal}{1, k}(\secp,\bm{\theta},\beta)$:] 
			\item $\qB$ submits a message $\msg \in \cM$ to the challenger.
			\item The challenger chooses $(P,\lock,\qaux) \chosen D_\param$.
			\item The challenger computes the obfuscated circuit as follows:
			\begin{enumerate}
				\item Sample $(\fpk, \fsk) \leftarrow \FHE.\keygen(1^{\secp})$ and $R \leftarrow \{0, 1\}^{\secp}$.
				\item Compute $\tlfDec \leftarrow \CCObf(1^{\secp}, \fDec, R, 1)$.
				\item Compute $\tlI \leftarrow \CCObf(1^{\secp}, I, \lock, R)$.
				\item Represent $(P \concat \tlI) = (b_1, \ldots, b_{\ell}) \in \{0, 1\}^{\ell}$.
				\item Sample $\bm{\theta}_i, \bm{z}_i \leftarrow \{0, 1\}^{\secp}$ for all $i \in [\ell+n] \setminus \{k\}$.
				\item Set $\wt{b}_{i}$ for $i \in [\ell]$ as follows:
				\begin{align}
					\widetilde{b}_{i} \seteq \begin{cases}
						0 \oplus \bigoplus_{j: \theta_{i,j} = 0} z_{i,j}
						&\text{if~}i\in [1, k-1]\\
						\beta & \text{if~}i=k\\
						b_{i} \oplus \bigoplus_{j: \theta_{i, j} = 0}z_{i, j} 
						&\text{if~}i\in[k+1, \ell]
					\end{cases}.
				\end{align}
				\item Set $\wt{b}_{\ell+k} \seteq \msg_k \oplus \bigoplus_{j: \theta_{\ell+k, j} = 0}z_{\ell+k, j}$ for all $k \in [n]$.
				\item Compute $\fct_{i} \leftarrow \FHE.\enc(\fpk, (\bm{\theta}_{i}, \wt{b}_{i}))$ for all $i \in [\ell+n]$ where $\bm{\theta}_k \seteq \bm{\theta}$.
			\end{enumerate}
			The challenge sends $(\tlfDec, \{(\ket{\bm{z}_{i}}_{\bm{\theta}_{i}}, \fct_{i})\}_{ i\in [\ell+n]\setminus \{k\}}, \fct_k, \fpk)$ to $\qB$.
			\item $\qB$ outputs a bit $b'$ as the final output of the experiment.
		\end{enumerate}	
		Let us define $\mcl{Z}_{\secp}^k(\bm{\theta}) = \expt{\qB,\qChal}{1, k}(\secp,\bm{\theta},\beta)$. We first show that 
		\begin{align}
			\left|\Pr[\mcl{Z}_{\secp}^k(\bm{\theta})=1]-\Pr[\mcl{Z}_{\secp}^k(\bm{0}_{\secp})=1]\right|\le \negl(\secp).\label{eq7}
		\end{align}
		
		\begin{description}
			\item[$\mcl{Z}^{k, 1}_{\secp}$:] This is exactly the same as $\mcl{Z}_{\secp}^k(\bm{\theta})$ except the challenger uses the bits of $\tlI \leftarrow \sfCC.\sm(1^{\lambda}, \pp_{I},1^{\abs{R}})$ instead of $\tlI \leftarrow \CCObf(1^{\secp}, I, \lock, R)$ to set $\wt{b}_i$ for all $i \in [k+1, \ell]$. The indistinguishability between the distributions $\mcl{Z}_{\secp}^k(\bm{\theta})$ and $\mcl{Z}^{k, 1}_{\secp}$ follows from the post-quantum security of the classical compute-and-compare obfuscation scheme $\Sigma_\CCO$.
			\item[$\mcl{Z}^{k, 2}_{\secp}$:] This is exactly the same as $\mcl{Z}_{\secp}^{k, 1}$ except the challenger replaces $\tlfDec \leftarrow \CCObf(1^{\secp}, \fDec, R, 1)$ with the simulated obfuscated circuit $\tlfDec \leftarrow \sfCC.\sm(1^{\secp}, \pp_{\fDec},1^1)$. The indistinguishability between the distributions $\mcl{Z}_{\secp}^{k, 1}$ and $\mcl{Z}^{ k, 2}_{\secp}$ follows from the post-quantum security of the classical compute-and-compare scheme $\Sigma_\CCO$.
			\item[$\mcl{Z}^{k, 3}_{\secp}$:] This is exactly the same as $\mcl{Z}_{\secp}^{k, 2}$ except the challenger computes $\fct_{k} \leftarrow \FHE.\enc(\fpk, (\bm{0}_{\secp}, \wt{b}_{k}))$ instead of encrypting $(\bm{\theta}, \wt{b}_k)$. The indistinguishability between the distributions $\mcl{Z}_{\secp}^{k, 2}$ and $\mcl{Z}^{ k, 3}_{\secp}$ follows from the post-quantum security of $\Sigma_{\msf{fhe}}$.
			\item[$\mcl{Z}^{k, 4}_{\secp}$:] This is exactly the same as $\mcl{Z}_{\secp}^{k, 3}$ except the challenger replaces $\tlfDec \leftarrow \sfCC.\sm(1^{\secp}, \pp_{\fDec},1^1)$ with the real obfuscated circuit $\tlfDec \leftarrow \CCObf(1^{\secp}, \fDec, R, 1)$. The indistinguishability between the distributions $\mcl{Z}_{\secp}^{k, 3}$ and $\mcl{Z}^{ k, 4}_{\secp}$ follows from the post-quantum security of the classical compute-and-compare obfuscation scheme $\Sigma_\CCO$.
			\item[$\mcl{Z}^{k, 5}_{\secp}$:] This is exactly the same as $\mcl{Z}_{\secp}^{ k, 4}$ except the challenger uses the bits of $\tlI \leftarrow \CCObf(1^{\secp}, I, \lock, R)$ instead of $\tlI \leftarrow \sfCC.\sm(1^{\lambda}, \pp_{I},1^1)$ to set $\wt{b}_i$ for all $i \in [k+1, \ell]$. The indistinguishability between the distributions $\mcl{Z}_{\secp}^{k, 4}$ and $\mcl{Z}^{ k, 5}_{\secp}$ follows from the post-quantum security of the classical compute-and-compare obfuscation scheme $\Sigma_\CCO$.
		\end{description}	
		Observe that, the distributions $\mcl{Z}^{ k, 5}_{\secp}$ and $\mcl{Z}_{\secp}^k(\bm{0}_{\secp})$ are identical. Hence,~\cref{eq7} holds for all $k \in [\ell]$. Therefore, by Lemma \ref{lem:ce}, for any (unbounded) adversary $\qB'$ we have
		
		\begin{align}
			\TD(\tildeexpt{\qB',\qChal}{1, k}(\secp,0),\tildeexpt{\qB',\qChal}{1, k}(\secp,1)) \le \negl(\secp)\label{eql:td1}
		\end{align}
		where the experiment $\wt{\mcl{Z}}_{\secp}^k(b) = \tildeexpt{\qB',\qChal}{1, k}(\secp,b)$ works as follows:
		\begin{enumerate}
			\item[$\tildeexpt{\qB',\qChal}{1, k}(\secp,b):$] 
			\item Sample $\bm{z}, \bm{\theta} \leftarrow \{0, 1\}^{\secp}$. 
			\item $\qB'$ receives $(1^\secp, \ket{\bm{z}}_{\bm{\theta}})$ as input. 
			\item $\qB'$ interacts with $\qChal$ as in $\expt{\qB,\qChal}{1, k}(\secp,\bm{\theta}, b\oplus \bigoplus_{j: \theta_{i,j} = 0} z_{i,j})$ where $\qB'$ plays the role of $\qB$.
			\item $\qB'$ outputs a string $\bm{z}' \in \{0, 1\}^{ \secp}$ and a quantum state $\rho$.
			\item If $z_j = z_j'$ for all $j\in[\secp]$ such that $\theta_j = 1$ then the experiment outputs $\rho$, and otherwise it outputs a special symbol $\bot$.
		\end{enumerate}	
		Note that the only difference between $\hyb_{1,k-1}$ and $\hyb_{1,k}$ is that $\wt{b}_{k}$ is set to be $b_k \oplus \bigoplus_{j: \theta_{i,j} = 0} z_{i,j}$ in $\hyb_{1,k-1}$ and $0 \oplus \bigoplus_{j: \theta_{i,j} = 0} z_{i,j}$ in $\hyb_{1,k}$. Let us assume $b_k = 1$, since otherwise $\hyb_{1,k-1}$ and $\hyb_{1,k}$ are identical. We construct $\qB'$ that distinguishes $\tildeexpt{\qB',\qChal}{1, k}(\secp,0)$ and $\tildeexpt{\qB',\qChal}{1, k}(\secp,1)$ if $\qA$ distinguishes between the hybrids $\hyb_{1,k-1}$ and $\hyb_{1,k}$. 
		\begin{enumerate}
			\item[$\qB'(1^{\secp}, \ket{\bm{z}}_{\bm{\theta}})$:]
			\item $\qB'$ plays the role of $\qA$ in $\hyb_{1,k}$ where the external challenger $\qChal$ of $\tildeexpt{\qB',\qChal}{1,k}(\secp,b)$ is used to simulate the challenger of $\hyb_{1,k}$. $\qChal$ sends the obfuscated circuit to $\qA$.
			\item Suppose $\qA$ sends a certificate $\cert = (\bm{z}_1', \ldots, \bm{z}_{\ell}')$ to the challenger where $\bm{z}_i' = (z_{i, j}')_{j \in [\secp]}$ for all $i \in [\ell]$. Then, $\qB'$ sets $\bm{z}' \seteq \bm{z}_k'$.
			\item Outputs $\bm{z}'$ and the internal state $\rho$ of $\qA$.
		\end{enumerate}	
		We observe that $\qB'$ perfectly simulates $\hyb_{1, k-1}$ if $b = 1$ and $\hyb_{1, k}$ if $b = 0$ (since we are assuming $b_k = 1$). Therefore, we can write
		\begin{align}
			\TD(\hyb_{1,k-1}, \hyb_{1,k})
			\le \TD(\wt{\mcl{Z}}_{\secp}^k(0), \wt{\mcl{Z}}_{\secp}^k(1)).\label{eql:hyb1}
		\end{align}
		Combining~\cref{eql:td1,eql:hyb1}, we have
		\begin{align}
			\TD(\hyb_{1,k-1}, \hyb_{1,k})
			\le \negl(\secp).\label{eql:hyb01}
		\end{align}
		Recall that $\hyb_{1, 0} \equiv \hyb_0$ and $\hyb_{1, \ell} \equiv \hyb_1$. Therefore, combining the advantages of $\qA$ in the sequence of intermediate hybrids as obtained in Equation \ref{eql:hyb01}, we have
 		\begin{align}
			\TD(\hyb_{0}, \hyb_{1})
			\le \negl(\secp).
		\end{align}
		
		\item[$\hyb_2$:]   This is exactly the same as $\hyb_1$ except the fact that instead of encrypting the challenge message $\msg \in \{0, 1\}^n$ the FHE ciphertexts $\{\fct_{\ell+k}\}_{k\in [n]}$ are encrypted to the message $\bm{0}_n$. More precisely, the challenger samples $\bm{\theta}_{\ell+k}, \bm{z}_{\ell+k} \leftarrow \{0, 1\}^{\secp}$ and sets $\wt{b}_{\ell+k} \seteq 0 \oplus \bigoplus_{j: \theta_{\ell+k, j} = 0}z_{\ell+k, j}$ for all $k \in [n]$ instead of setting $\wt{b}_{\ell+k} \seteq \msg_k \oplus \bigoplus_{j: \theta_{\ell+k, j} = 0}z_{\ell+k, j}$. Finally, it obtains $\fct_{\ell+k} \leftarrow \FHE.\enc(\fpk, (\bm{\theta}_{\ell+k}, \wt{b}_{\ell+k}))$ for all $k \in [n]$ where the encrypted bits $\{\wt{b}_{\ell+k}\}_{k \in [n]}$ contain no information about the message $m$. Since the FHE master secret key $\fsk$ is not required to simulate the hybrids, the indistinguishability between $\hyb_1$ and $\hyb_2$ is guaranteed by the post-quantum semantic security of FHE. We can follow a similar argument as in the previous hybrid and show that
		\begin{align}
			\TD(\hyb_{1}, \hyb_{2})
			\le \negl(\secp).
		\end{align}
\end{description}
We observe that $\hyb_2$ is equivalent to $\EV\expb{\Sigma_{\CECCO},~ \qA}{sim}{ccobf}(\secp, 1)$. Therefore, by combing the advantages of $\qA$ in the consecutive hybrids and applying the triangular inequality, we have
\begin{align}
		\TD(\EV\expb{\Sigma_{\CECCO},~ \qA}{sim}{ccobf}(\secp, 0), \EV\expb{\Sigma_{\CECCO},~ \qA}{sim}{ccobf}(\secp, 1))
	\le \negl(\secp).
\end{align} 
Finally, it is easy to show the computational indistinguishability, i.e.,
\begin{align}
\left|\Pr[\C\expb{\Sigma_{\CECCO},~ \qA}{sim}{ccobf}(\secp, 0) = 1] - \Pr[\C\expb{\Sigma_{\CECCO},~ \qA}{sim}{ccobf}(\secp, 1) = 1]\right| \le \negl(\secp)
\end{align}
using the security of FHE and the post-quantum security of $\CCO$.
We skip the formal description as it follows from the similar sequence of hybrids that we used to establish~\cref{eq7} except that $\fct_k$ is changed from encryption of $(\bm{\theta}, \wt{b}_k)$ to $(\bm{\theta}, 0 \oplus \bigoplus_{j: \theta_{k, j} = 0}z_{k, j})$ (instead of changing it from $(\bm{\theta}, \wt{b}_k)$ to $(\bm{0}_{\secp}, \wt{b}_k)$ in $\mcl{Z}_{\secp}^{k, 3}$). This completes the proof.
\end{proof}


\section{Predicate Encryption with Certified Everlastng Deletion}\label{sec:pecd}

\subsection{Definition}\label{sec:def_PE_CED}
We describe the notion of PE with certified everlasting deletion which generates a quantum ciphertext that can be deleted when required and the deletion is verified using a classical certificate of deletion.

\begin{definition}[PE with Ceritifed Everlasting Deletion (Syntax)]\label{def:PE_CED_syntax}
A certified everlasting PE is tuple of QPT algorithms $(\setup, \keygen, \qenc,\allowbreak \qdec, \qdel, \vrfy)$ with a class predicates  $\calP$, a class of attributes $\mcl{X}$ and a message space $\Ms$.
\begin{description}
\item[$\setup(1^{\secp}) \rightarrow (\pk, \msk)$:] The parameter setup algorithm takes as input the security parameter $1^{\secp}$ and outputs a public key $\pk$ and a master secret key $\msk$.  

\item[$\keygen(\msk, \msf{P})$:] The key generation algorithm takes as input the master secret key $\msk$ and a predicate $\msf{P} \in \calP$, and outputs a secret key $\sk_{\msf{P}}$ corresponding to the predicate $\msf{P}$.

\item[$\qenc(\pk, x, m) \rightarrow (\qct, \vk)$:] The encryption algorithm takes as input the public key $\pk$, an attribute $x \in \mcl{X}$ and a message $m \in \Ms$, and outputs a quantum ciphertext $\qct$ and a classical verification key $\vk$.

\item[$\qdec(\sk_{\msf{P}}, \qct) \rightarrow m' \textbf{ or } \bot$:] The decryption algorithm takes as input a secret key $\sk_{\tsf{P}}$ and a quantum ciphertext $\qct$, and outputs a classical plaintext $m'$ or $\bot$.

\item[$\qdel(\qct) \rightarrow \cert$:] The deletion algorithm takes as input the ciphertext $\qct$ and outputs a classical certificate $\cert$.

\item[$\vrfy(\vk, \cert) \rightarrow \top \textbf{ or }\bot$:] The verification algorithm takes as input the verification key $\vk$ and a certificate $\cert$, and outputs $\top$ or $\bot$.
\end{description}
\end{definition}

\begin{definition}[Correctness of PE with Certified Everlasting Deletion]\label{def:PE_CED_correctness}
The correctness of PE with certified deletion for a class of predicates $\calP$ is defined as follows.
\begin{description}
\item[Decryption correctness:] For any $\secp \in \N, \msf{P} \in \calP, x \in \mcl{X}, m \in \Ms$ such that $\msf{P}(x) = 1$, 
\begin{align}
    \Pr\left[ \qdec(\sk_{\msf{P}}, \qct) \ne m \ \middle| \begin{array}{l} (\pk, \msk) \leftarrow \setup(1^{\secp}) \\
    \sk_{\msf{P}} \leftarrow \keygen(\msk, \msf{P}) \\
    (\qct, \vk) \leftarrow \qenc(\pk, m)  \end{array} \right] \le \negl(\secp).
\end{align}

\item[Verification correctness:] For any $\secp \in \N, x \in \mcl{X}, m \in \Ms$, 
\begin{align}
    \Pr\left[ \vrfy(\vk, \cert) \ne \top \ \middle| \begin{array}{l} (\pk, \msk) \leftarrow \setup(1^{\secp}) \\ (\qct, \vk) \leftarrow \qenc(\pp, x, m) \\ \cert \leftarrow \qdel(\qct) \end{array} \right] \le\negl(\secp).
\end{align}
\end{description}
\end{definition}

\begin{definition}[Certified Everlasting Security of PE]\label{def:PE_CED_security}
Let $\Sigma = (\setup, \keygen, \qenc, \qdec, \qdel, \vrfy)$ be a PE with certified everlasting deletion for a class of predicates $\mcl{P}$, a class of attributes $\mcl{X}$ and a message space $\Ms$. We consider two experiments $\EV\expb{\Sigma,~ \qA}{ada}{ind}(\secp, b)$ and $\C\expb{\Sigma,~ \qA}{ada}{ind}(\secp, b)$ played between a challenger and and a non-uniform QPT adversary $\qA = \{\qA_{\secp}, \ket{\psi}_{\secp}\}_{\secp \in \mbb{N}}$. The experiments are defined as follows:
\begin{enumerate}
    \item The challenger computes $(\pk, \msk) \leftarrow \setup(1^{\secp})$ and sends $\pk$ to $\qA_{\secp}(\ket{\psi}_{\secp})$.
    \item $\qA_{\secp}$ sends $\msf{P} \in \mcl{P}$ to the challenger and receives $\sk_{\tsf{P}} \leftarrow \keygen(\msk, \msf{P})$ from the challenger.
    \item $\qA_{\secp}$ sends a pair of challenge attributes $(x_0, x_1)$ and a pair of challenge messages $(m_0, m_1)$ satisfying the fact that $\msf{P}(x_0) = \msf{P}(x_1) = 0$ for all $\msf{P}$ queried so far in the key query phase.
    \item The challenger computes $(\qct_b, \vk_b) \leftarrow \qenc(\pk, x_b, m_b)$ and sends $\qct_b$ to $\qA_{\secp}$.
    \item $\qA_{\secp}$ can make further key queries with $\tsf{P}$ satisfying $\msf{P}(x_0) = \msf{P}(x_1) = 0$.
    \item $\qA_{\secp}$ sends a certificate of deletion $\cert$ and its internal state $\rho$ to the challenger.
    \item The challenger computes $\vrfy(\vk_b, \cert)$. If the outcome is $\top$, the experiment $\EV\expb{\Sigma,~ \qA}{ada}{ind}(\secp, b)$ outputs $\rho$; otherwise if the outcome is $\bot$ then $\EV\expb{\Sigma,~ \qA}{ada}{ind}(\secp, b)$ output $\bot$ and ends.
    \item The challenger sends the outcome of $\vrfy(\vk^{(b)}, \cert)$ to $\qA_{\secp}$. 
    \item Again, $\qA_{\secp}$ can make key queries with polynomial number of policies $\tsf{P}$ satisfying $\msf{P}(x_0) = \msf{P}(x_1) = 0$.
    \item $\qA_{\secp}$ outputs its guess $b' \in \{0, 1\}$ which is the output of the experiment $\C\expb{\Sigma,~ \qA}{ada}{ind}(\secp, b)$.
\end{enumerate}
We say that the $\Sigma$ is adaptively certified everlasting secure if for any non-uniform QPT adversary $\qA = \{\qA_{\secp}, \ket{\psi}_{\secp}\}_{\secp \in \mbb{N}}$, it holds that 
\begin{align}
	\TD(\EV\expb{\Sigma,~ \qA}{ada}{ind}(\secp, 0), \EV\expb{\Sigma,~ \qA}{ada}{ind}(\secp, 1)) \le \negl(\secp),
\end{align}	
and 	
\begin{align}
	\left|\Pr[\C\expb{\Sigma,~ \qA}{ada}{ind}(\secp, 0) = 1] - \Pr[\C\expb{\Sigma,~ \qA}{ada}{ind}(\secp, 1) = 1]\right| \le \negl(\secp).
\end{align}
\end{definition}

We can define similar experiment $\EV\expb{\Sigma,~ \qA}{sel}{ind}(\secp, b)$ and $\C\expb{\Sigma,~ \qA}{ada}{ind}(\secp, b)$ where $\qA_{\secp}$ is restricted to submit the challenge attributes $x_0, x_1$ before it receives $\pk$ from the challenger. We say that the $\Sigma$ is selectively certified everlasting secure if for any non-uniform QPT adversary $\qA = \{\qA_{\secp}, \ket{\psi}_{\secp}\}_{\secp \in \mbb{N}}$, it holds that 
\begin{align}
	\TD(\EV\expb{\Sigma,~ \qA}{sel}{ind}(\secp, 0), \EV\expb{\Sigma,~ \qA}{sel}{ind}(\secp, 1)) \le \negl(\secp),
\end{align}	
and 	
\begin{align}
	\left|\Pr[\C\expb{\Sigma,~ \qA}{sel}{ind}(\secp, 0) = 1] - \Pr[\C\expb{\Sigma,~ \qA}{sel}{ind}(\secp, 1) = 1]\right| \le \negl(\secp).
\end{align}

\subsection{Construction}\label{sec:PE_from_LObf_ABE}

In this section, we construct a PE with certified everlasting deletion from a compute-and-compare obfuscation wiht certified everlasting deletion introduced in~\cref{sec:ccobf} and a classical ABE.

\paragraph{Ingredients.} We use the following building blocks.
\begin{enumerate}
    \item $\Sigma_{\tsf{abe}} = \ABE.(\setup, \keygen, \enc, \dec)$ be a classical ABE scheme for a class of predicates $\mcl{P}$ and message space $\mcl{M}_{\tsf{abe}} = \{0, 1\}^{\secp+1}$.
 \item $\Sigma_{\CECCO} = \CCO.(\qObf, \qdel, \vrfy)$ be a compute-and-compare obfuscation with certified everlasting deletion for a message space $\mcl{M}_{\tsf{pe}}$ and the family of distributions $D=\setbk{D_{\apk,x,\setbk{\bm{\theta}_i}_i,\setbk{\bm{z}_i}_i}}_{{\apk,x,\setbk{\bm{\theta}_i}_i,\setbk{\bm{z}_i}_i}}$.
\end{enumerate}

Let $D=\setbk{D_{\apk,x,\setbk{\bm{\theta}_i}_i,\setbk{\bm{z}_i}_i}}_{{\apk,x,\setbk{\bm{\theta}_i}_i,\setbk{\bm{z}_i}_i}}$ be a family of distributions where $D_{\apk,x,\setbk{\bm{\theta}_i}_i,\setbk{\bm{z}_i}_i}$ outputs $(\msf{aDec},\lock,\qaux)$ generated as follows.
\begin{itemize}
\item Generate $\act_i\gets \ABE.\enc(\apk,x,(\bm{\theta}_i, \bigoplus_{j: \theta_{i, j} = 0}z_{i, j}))$ for all $i\in[\ell]$.
\item Construct $\msf{aDec}$ described in~\cref{fig:pe1}.
\item Choose $\lock \chosen \zo{\ell} = \cK$.
\item Output $(\msf{aDec},\lock,\qaux \seteq \bot)$.
\end{itemize}

\paragraph{PE with certified everlasting deletion.} We construct $\Sigma_{\tsf{pe-ce}} = (\setup, \keygen, \qenc, \qdec, \qdel, \vrfy)$ for a class of predicates $\mcl{P}$, a class of attributes $\mcl{X}$ and a message space $\mcl{M}_{\tsf{pe}}$.
\paragraph{$\setup(1^{\lambda})$}$\vspace{0.1cm}$
\begin{enumerate}
	\item Sample $(\apk, \amsk) \leftarrow \ABE.\setup(1^{\secp})$.
	\item Output $(\pk \seteq \apk, \msk\seteq \amsk)$.
\end{enumerate}	

\paragraph{$\keygen(\msk, P)$}$\vspace{0.1cm}$
\begin{enumerate}
	\item Parse $\msk = \amsk$.
	\item Compute $\sk_P \leftarrow \ABE.\keygen(\amsk, P)$.
	\item Output $\sk_P$.
\end{enumerate}	

\paragraph{$\qenc(\apk, x, m)$}$\vspace{0.1cm}$
\begin{enumerate}
	\item Parse $\pk =\apk$.
	\item Sample $R \leftarrow \{0, 1\}^{\ell}$ and denote $R = (r_1, \ldots, r_{\ell})$.
	\item Sample $\bm{\theta}_{i}, \bm{z}_{i} \leftarrow \{0, 1\}^{\secp}$ for all $i \in [\ell]$.
	\item Set $\wt{r}_{i} \seteq r_{i} \oplus \bigoplus_{j: \theta_{i, j} = 0}z_{i, j}$ for all $i \in [\ell]$.
	\item Compute $\act_{i} \leftarrow \ABE.\enc(\apk, x,  (\bm{\theta}_i, \wt{r}_i))$ for all $i \in [\ell]$.
	\item $(\tlaDec, \vk_{\msf{aDec}}) \leftarrow \CCO.\qobf(1^{\secp}, \msf{aDec}, R, m)$ where $\msf{aDec}$ is defined in Figure \ref{fig:pe1}.
	\item Output $\qct \seteq (\tlaDec, \{\ket{\bm{z}_{i}}_{\bm{\theta}_i}\}_{i \in [\ell]})$ and $\vk \seteq (\{(\bm{z}_i, \bm{\theta}_i)\}_{i \in [\ell]}, \vk_{\msf{aDec}})$.
\end{enumerate}

\begin{figure}
	\begin{framed}
		\textbf{Hardwire:}~ $\{\act_i\}_{i \in [\ell]}$\newline
		\textbf{Input:}~ $(\{\bm{z}_{i}\}_{i \in [\ell]}, \sk_P)$ 
		\begin{enumerate}
			\item Compute $(\bm{\theta}_i, \wt{r}_i) \leftarrow \ABE.\dec(\sk_P, \act_i)$ for all $i \in [\ell]$.
			\item Compute $r_{i} =  \wt{r}_{i}  \oplus \bigoplus_{j: \theta_{i, j} = 0}z_{i, j}$ for all $i \in [\ell]$.
			\item Set $R \seteq (r_1, \ldots, r_{\ell})$
			\item Output $R$
		\end{enumerate}
	\end{framed}
	\caption{The description of the circuit $\msf{aDec}$}
	\label{fig:pe1}
\end{figure}

\paragraph{$\qdec(\sk_P, \qct)$}$\vspace{0.1cm}$
\begin{enumerate}
	\item Parse $\qct =(\tlaDec, \{\ket{\bm{z}_{i}}_{\bm{\theta}_i}\}_{i \in [\ell]})$.
	\item Apply $\tlaDec$ in superposition to the input $(\{\ket{\bm{z}_{i}}_{\bm{\theta}_i}\}_{i \in [\ell]}, \sk_P)$ and measure the output register to obtain $m'$.
	\item Output $m'$.
\end{enumerate}

\paragraph{$\qdel(\qct)$}$\vspace{0.1cm}$
\begin{enumerate}
		\item Parse $\qct =(\tlaDec, \{\ket{\bm{z}_{i}}_{\bm{\theta}_i}\}_{i \in [\ell]})$.
	\item Measure $\ket{\bm{z}_i}_{\bm{\theta}_i}$ in the Hadamard basis for all $i \in [\ell]$ and obtain $\bm{z}' \seteq (\bm{z}_1', \ldots, \bm{z}_{\ell}')$.
		\item Compute $\cert_{\msf{aDec}} \leftarrow \CCO.\qdel(\tlaDec)$.
	\item Output $\cert \seteq (\bm{z}', \cert_{\msf{aDec}})$.
\end{enumerate}	

\paragraph{$\vrfy(\vk, \cert)$}$\vspace{0.1cm}$
\begin{enumerate}
	\item Parse $\vk \seteq (\{(\bm{z}_i, \bm{\theta}_i)\}_{i \in [\ell]}, \vk_{\msf{aDec}})$ and $\cert \seteq (\bm{z}', \cert_{\msf{aDec}})$.
	\item If $(z_{i, j} = z_{i, j}') \wedge (\theta_{i, j} = 1)$ holds for all $i \in [\ell]$ and $j \in [\lambda]$ and $\CCO.\vrfy(\vk_{\msf{aDec}}, \cert_{\msf{aDec}}) = \top$, then output $\top$; otherwise output $\bot$.
\end{enumerate}

\begin{theorem}
	If $\Sigma_{\CECCO}$ is a certified everlasting secure compute-and-compare obfuscation for a message space $\mcl{M}_{\tsf{pe}}$ and the family of distributions $D=\setbk{D_{\apk,x,\setbk{\bm{\theta}_i}_i,\setbk{\bm{z}_i}_i}}_{{\apk,x,\setbk{\bm{\theta}_i}_i,\setbk{\bm{z}_i}_i}}$ and $\Sigma_{\tsf{abe}}$ is an adaptively (resp. selectively) secure ABE for a class of predicates $\mcl{P}$, then $\Sigma_{\tsf{pe-ce}}$ is an adaptively (resp. selectively) certified everlasting secure predicate encryption scheme for the class of predicates $\mcl{P}$, message space $\mcl{M}_{\tsf{pe}}$.
\end{theorem}
We focus on the case of adaptive security.
\begin{proof}

To prove the theorem we consider an adversary $\qA$ against the certified everlasting security of $\Sigma_{\tsf{pe-ce}}$. We consider the following sequence of hybrids.

\begin{description}
	\item[$\hyb_0:$] This is the original certified everlasting security experiment where the challenge bit is set to 0 ($\EV\expb{\Sigma_{\tsf{pe-ce}},~ \qA}{ada}{ind}(\secp, 0)$). More precisely, it works as follows:
	\begin{enumerate}
		\item The challenger computes $(\apk, \amsk) \leftarrow \ABE.\setup(1^{\secp})$, sets $\pk \seteq \apk$, and sends $\pk$ to $\qA$.
		\item The adversary $\qA$ sends any polynomial number of secret key queries for $P \in \mcl{P}$ at any point of the experiment. The challenger generates $\sk_P \leftarrow \ABE.\keygen(\amsk, P)$ and sends $\sk_P$ to $\qA$.
		\item $\qA$ sends a pair of challenge attributes $(x_0, x_1)$ and a pair of challenge messages $(m_0, m_1)$ satisfying the fact that $P(x_0) = P(x_1) = 0$ for all $P$ queried so far in the key query phase.
		\item The challenger computes the challenge ciphertext as follows:
		\begin{enumerate}
			\item Sample $R^* = (r_1, \ldots, r_{\ell}) \leftarrow \{0, 1\}^{\ell}$.
			\item Sample $\bm{\theta}_{i}, \bm{z}_{i} \leftarrow \{0, 1\}^{\secp}$ for all $i \in [\ell]$ where $\bm{\theta}_i = (\theta_{i, j})_{j \in [\secp]}$ and $\bm{z}_i = (z_{i, j})_{j \in [\secp]}$.
			\item Set $\wt{r}_{i} \seteq r_{i} \oplus \bigoplus_{j: \theta_{i, j} = 0}z_{i, j}$ for all $i \in [\ell]$.
			\item Compute $\act_{i} \leftarrow \ABE.\enc(\apk, x_0,  (\bm{\theta}_i, \wt{r}_i))$ for all $i \in [\ell]$.
		\item $(\tlaDec, \vk_{\msf{aDec}}) \leftarrow \CCO.\qobf(1^{\secp}, \msf{aDec}, R^*, m_0)$ where $\msf{aDec}$ is defined in Figure \ref{fig:pe1}.
			\item Set $\qct^* \seteq (\tlaDec, \{\ket{\bm{z}_{i}}_{\bm{\theta}_i}\}_{i \in [\ell]})$ and $\vk \seteq (\{(\bm{z}_i, \bm{\theta}_i)\}_{i \in [\ell]}, \vk_{\msf{aDec}})$.
		\end{enumerate}	
	The challenger sends $\qct^*$ to $\qA$.
	\item $\qA$ sends a certificate $\cert = (\bm{z}' = (\bm{z}_1', \ldots, \bm{z}_{\ell}'), \cert_{\msf{aDec}})$ and its internal state $\rho$ to the challenger where $\bm{z}_i' = (z_{i, j}')_{j \in [\secp]}$ for all $i \in [\ell]$.
\item The challenger checks if $(z_{i, j} = z_{i, j}') \wedge (\theta_{i, j} = 1)$ holds for all $i \in [\ell]$ and $j \in [\lambda]$ and $\CCO.\vrfy(\vk_{\msf{aDec}}, \cert_{\msf{aDec}}) = \top$. If it does not hold, the challenger outputs $\bot$ as the final output of the experiment. Otherwise, go to the next step.
	\item The experiment outputs $\rho$ as a final output.
	\end{enumerate}	
	\item[$\hyb_1:$] This hybrid proceeds exactly similar to $\tsf{Hybd}_0$ except that the ABE ciphertexts $\act_i$ is now replaced with encryption of zero string. In particular, the hardwired values of $\msf{aDec}$ are computed as $\act_i \leftarrow \ABE.\enc(\apk, x_0, (\bm{\theta}_i, 0 \oplus \bigoplus_{j: \theta_{i, j} = 0}z_{i, j}))$ for all $i \in [\ell]$. \par
	To prove the indistinguishability between $\hyb_0$ and $\hyb_1$, we consider a sequence of intermediate hybrids $\hyb_{1, k}$ for $k \in [\ell]$ where we take $\hyb_{1, 0}$ is identical to $\hyb_0$ and the only difference between $\hyb_{1, k-1}$ and $\hyb_{1, k}$ is that $\act_k$ is an encryption of $(\bm{\theta}_k, r_{k} \oplus \bigoplus_{j: \theta_{k, j} = 0}z_{k, j})$ in $\hyb_{1, k-1}$ whereas it is an encryption of $(\bm{\theta}_k, 0 \oplus \bigoplus_{j: \theta_{k, j} = 0}z_{k, j})$ in $\hyb_{1, k}$.\par

	Now, we consider a sequence of experiments $\expt{\qB,\qChal}{1, k}(\secp,\bm{\theta},\beta)$ for $k \in [\ell]$ between a QPT adversary $\qB$ and a challenger $\qChal$ for $\bm{\theta}  \in \{0, 1\}^{\secp}$ and $\beta \in \{0, 1\}$. The experiment $\expt{\qB,\qChal}{1, 0}(\secp,\bm{\theta},\beta)$ is basically the same as $\hyb_0$ where $\qB$ plays the role of $\qA$ and $\qChal$ plays the role of the challenger. In particular, it works as follows:
	 \begin{enumerate}
	 	\item[$\expt{\qB,\qChal}{1, k}(\secp,\bm{\theta},\beta):$]
	 	\item $\qC$ computes $(\apk, \amsk) \leftarrow \ABE.\setup(1^{\secp})$, sets $\pk \seteq \apk$, and sends $\pk$ to $\qB$.
	 	\item $\qB$ sends any polynomial number of secret key queries for $P \in \mcl{P}$ at any point of the experiment and $\qChal$ generates $\sk_P \leftarrow \ABE.\keygen(\amsk, P)$ and sends $\sk_P$ to $\qB$.
	 	\item $\qB$ sends a pair of challenge attributes $(x_0, x_1)$ and a pair of challenge messages $(m_0, m_1)$ satisfying the fact that $P(x_0) = P(x_1) = 0$ for all $P$ queried so far in the key query phase.
	 	\item $\qChal$ computes the challenge ciphertext as follows:
	 	\begin{enumerate}
	 		\item Sample $R^* = (r_1, \ldots, r_{\ell}) \leftarrow \{0, 1\}^{\ell}$.
	 		\item Sample $\bm{z}_{i}, \bm{\theta}_i \leftarrow \{0, 1\}^{\secp}$ for all $i \in [\ell]\setminus \{k\}$ where $\bm{\theta}_i = (\theta_{i, j})_{j \in [\secp]}$ and $\bm{z}_i = (z_{i, j})_{j \in [\secp]}$.
	 		\item Set $\wt{r}_{i}$ as follows:
	 		\begin{align}
	 			 			\widetilde{r}_{i} \seteq \begin{cases}
	 				0\oplus \bigoplus_{j: \theta_{i,j} = 0} z_{i,j}
	 				&\text{if~}i\in [1, k-1]\\
	 				\beta & \text{if~}i=k\\
	 				r_{i} \oplus \bigoplus_{j: \theta_{i, j} = 0}z_{i, j} 
	 				 &\text{if~}i\in[k+1, \ell]
	 			\end{cases}.
	 		\end{align}
	 		\item Compute $\act_{i} $ as follows:
	 		\begin{align}
	 			\act_{i} \leftarrow \begin{cases}
	 				\ABE.\enc(\apk, x_0,  (\bm{\theta}, \wt{r}_k)) &\text{if~}i = k\\
	 				 \ABE.\enc(\apk, x_0,  (\bm{\theta}_i, \wt{r}_i)) &\text{if~}i\in [\ell]\setminus \{k\}\\
	 			\end{cases}.
	 		\end{align}
	 		\item $(\tlaDec, \vk_{\msf{aDec}}) \leftarrow \CCO.\qobf(1^{\secp}, \msf{aDec}, R^*, m_0)$ where $\msf{aDec}$ is defined in Figure \ref{fig:pe1}.
	 	\end{enumerate}	
	 	The challenger sends $(\tlaDec, \{\ket{\bm{z}_{i}}_{\bm{\theta}_i}\}_{i \in [\ell]\setminus\{k\}})$ to $\qB$.
	 	\item $\qB$ outputs a bit $b'$ as the final output of the experiment.
	 	\end{enumerate}
	Since all the secret keys $\sk_P$ corresponding to predicates $P$ queried by the adversary satisfy the condition that $P(x_0) = 0$, the semantic security of ABE ensures that
	$$
	\begin{array}{c}
		\left|\Pr[\expt{\qB,\qChal}{1, k}(\secp,\bm{\theta},\beta)=1]-\Pr[\expt{\qB,\qChal}{1, k}(\secp,\bm{0}_{\secp},\beta)=1]\right|\le \negl(\secp).
	\end{array} $$
	Therefore, by Lemma \ref{lem:ce_int}, for any QPT (unbounded) adversary $\qB'$, we have
	
\begin{equation}
	\TD(\tildeexpt{\qB',\qChal}{1, k}(\secp,0),\tildeexpt{\qB',\qChal}{1, k}(\secp,1)) \le \negl(\secp)\label{eq:td1}
\end{equation}

where the experiment $\tildeexpt{\qB',\qChal}{1, k}(\secp,b)$ works as follows:
\begin{enumerate}
	\item[$\tildeexpt{\qB',\qChal}{1, k}(\secp,b):$] 
	\item Sample $\bm{z}, \bm{\theta} \leftarrow \{0, 1\}^{\secp}$. 
	\item $\qB'$ takes $(1^\secp, \ket{\bm{z}}_{\bm{\theta}})$ as input. 
	\item $\qB'$ interacts with $\qChal$ as in $\expt{\qB,\qChal}{1, k}(\secp,\bm{\theta}, b\oplus \bigoplus_{j: \theta_{i,j} = 0} z_{i,j})$ where $\qB'$ plays the role of $\qB$. 
	\item $\qB'$ outputs a string $\bm{z}' \in \{0, 1\}^{ \secp}$ and a quantum state $\rho$.
	\item If $z_j = z_j'$ for all $j\in[\secp]$ such that $\theta_j = 1$ then the experiment outputs $\rho$, and otherwise it outputs a special symbol $\bot$.
\end{enumerate}	
Note that the only difference between $\hyb_{1,k-1}$ and $\hyb_{1,k}$ is that $\wt{r}_{k}$ is set to be $r_k \oplus \bigoplus_{j: \theta_{i,j} = 0} z_{i,j}$ in $\hyb_{1,k-1}$ and $0 \oplus \bigoplus_{j: \theta_{i,j} = 0} z_{i,j}$ in $\hyb_{1,k}$. Let us assume $r_k = 1$, since if $r_k$ is $0$ then $\hyb_{1,k-1}$ and $\hyb_{1,k}$ are identical. We construct $\qB'$ that distinguishes $\tildeexpt{\qB',\qChal}{1, k}(\secp,0)$ and $\tildeexpt{\qB',\qChal}{1, k}(\secp,1)$ if $\qA$ distinguishes between the hybrids $\hyb_{1,k-1}$ and $\hyb_{1,k}$.
\begin{enumerate}
	\item[$\qB'(1^{\secp}, \ket{\bm{z}}_{\bm{\theta}}):$]
	\item $\qB'$ plays the role of $\qA$ in $\hyb_{1,k}$ where the external challenger $\qChal$ of $\tildeexpt{\qB',\qChal}{1,k}(\secp,b)$ is used to simulate the challenger of $\hyb_{1,k}$. $\qChal$ provides everything that should be sent to $\qA$ (as in $\hyb_0$).
	\item Suppose $\qA$ sends a certificate $\cert = ((\bm{z}_1', \ldots, \bm{z}_{\ell}'), \cert_{\msf{aDec}})$ to the challenger where $\bm{z}_i' = (z_{i, j}')_{j \in [\secp]}$ for all $i \in [\ell]$. Then, $\qB'$ sets $\bm{z}' = \bm{z}_k'$.
	\item Outputs $\bm{z}'$ and the internal state $\rho$ of $\qA$ which it sends to $\qA_2$.
\end{enumerate}	
We observe that $\qB'$ perfectly simulates $\hyb_{1, k}$ if $b = 0$ and $\hyb_{1, k-1}$ if $b = 1$ (since we are assuming $r_k = 1$). Therefore, we can write
\begin{equation}
\TD(\hyb_{1,k-1}, \hyb_{1,k})
\le \TD(\tildeexpt{\qB',\qChal}{1,k}(\secp,0),\tildeexpt{\qB',\qChal}{1,k}(\secp,1)).\label{eq:hyb1}
\end{equation}
Combining Equations \ref{eq:td1} and \ref{eq:hyb1}, we have
\begin{equation}
\TD(\hyb_{1,k-1}, \hyb_{1,k})
\le \negl(\secp).\label{eq:hyb01}
\end{equation}
Recall that $\hyb_{1, 0} \equiv \hyb_0$ and $\hyb_{1, \ell} \equiv \hyb_1$. Therefore, combining the advantages of $\qA$ in the sequence of intermediate hybrids as obtained in Equation \ref{eq:hyb01}, we have
\begin{equation}
		\TD(\hyb_{0}, \hyb_{1})
	\le \negl(\secp).
\end{equation}

\item[$\hyb_2:$] This hybrid proceeds exactly similar to $\hyb_1$ except that the obfuscated circuit is now replaced with a simulated version of it. In particular, $\tlaDec \leftarrow \CCO.\qobf(1^\secp,\msf{aDec},R^\ast,m_0)$ is replaced with the circuit $\tlaDec \leftarrow \CCO.\qsim(1^{\secp}, \pp_{\msf{aDec}}, 1^{|m_b|})$. In particular, the hybrid works as follows:

\begin{enumerate}
	\item The challenger computes $(\apk, \amsk) \leftarrow \ABE.\setup(1^{\secp})$, sets $\pk \seteq \apk$, and sends $\pk$ to $\qA$.
	\item The adversary $\qA$ sends any polynomial number of secret key queries for $P \in \mcl{P}$ at any point of the experiment. The challenger generates $\sk_P \leftarrow \ABE.\keygen(\amsk, P)$ and sends $\sk_P$ to $\qA$.
	\item $\qA$ sends a pair of challenge attributes $(x_0, x_1)$ and a pair of challenge messages $(m_0, m_1)$ satisfying the fact that $P(x_0) = P(x_1) = 0$ for all $P$ queried so far in the key query phase.
	\item The challenger computes the challenge ciphertext as follows:
	\begin{enumerate}
		\item Sample $\bm{\theta}_{i}, \bm{z}_{i} \leftarrow \{0, 1\}^{\secp}$ for all $i \in [\ell]$ where $\bm{\theta}_i = (\theta_{i, j})_{j \in [\secp]}$ and $\bm{z}_i = (z_{i, j})_{j \in [\secp]}$.
\item $\tlaDec \leftarrow \CCO.\qsim(1^{\secp}, \pp_{\msf{aDec}}, 1^{|m_b|})$ where $\msf{aDec}$ is defined in Figure \ref{fig:pe1}. (Note that, we do not need to compute ABE ciphertexts since we only require the lengths of a ABE ciphertext in order to calculate $\pp_{\msf{aDec}}$.)
		\item Set $\qct^* \seteq (\tlaDec, \{\ket{\bm{z}_{i}}_{\bm{\theta}_i}\}_{i \in [\ell]})$ and $\vk \seteq (\{(\bm{z}_i, \bm{\theta}_i)\}_{i \in [\ell]}, \vk_{\msf{aDec}})$.
	\end{enumerate}	
	The challenger sends $\qct^*$ to $\qA$.
	\item $\qA$ sends a certificate $\cert = (\bm{z}' = (\bm{z}_1', \ldots, \bm{z}_{\ell}'), \cert_{\msf{aDec}})$ and its internal state $\rho$ to the challenger where $\bm{z}_i' = (z_{i, j}')_{j \in [\secp]}$ for all $i \in [\ell]$.
\item The challenger checks if $(z_{i, j} = z_{i, j}') \wedge (\theta_{i, j} = 1)$ holds for all $i \in [\ell]$ and $j \in [\lambda]$ and $\CCO.\vrfy(\vk_{\msf{aDec}},\allowbreak \cert_{\msf{aDec}}) = \top$. If it does not hold, the challenger outputs $\bot$ as the final output of the experiment. Otherwise, go to the next step.
    \item The experiment outputs $\rho$ as a final output.
\end{enumerate}
Since the information of lock string $R^*$ is not used in generating the ABE ciphertexts $\act_i$, the certified everlasting security of compute-and-compare obfuscation guarantees that $\hyb_1$ and $\hyb_2$ are indistinguishable to $\qA$. In other words, we have
\begin{equation}
		\TD(\hyb_{1}, \hyb_{2})
	\le \negl(\secp).
\end{equation}

\item[$\hyb_3:$] This hybrid proceeds exactly similar to $\hyb_2$ except that the simulated circuit is now replaced with a honestly obfuscated version of it. In particular, the obfuscated circuit is computed as $\tlaDec \leftarrow \CCO.\qobf(1^{\secp}, \msf{aDec}, R^*, m_1)$ where the circuit $\msf{aDec}$ is defined using the hardwired values $\act_i \leftarrow \ABE.\enc(\apk, x_1, (\bm{\theta}_i, 0 \oplus \bigoplus_{j: \theta_{i, j} = 0}z_{i, j}))$. In particular, the hybrid works as follows:

\begin{enumerate}
	\item The challenger computes $(\apk, \amsk) \leftarrow \ABE.\setup(1^{\secp})$, sets $\pk \seteq \apk$, and sends $\pk$ to $\qA$.
	\item The adversary $\qA$ sends any polynomial number of secret key queries for $P \in \mcl{P}$ at any point of the experiment. The challenger generates $\sk_P \leftarrow \ABE.\keygen(\amsk, P)$ and sends $\sk_P$ to $\qA$.
	\item $\qA$ sends a pair of challenge attributes $(x_0, x_1)$ and a pair of challenge messages $(m_0, m_1)$ satisfying the fact that $P(x_0) = P(x_1) = 0$ for all $P$ queried so far in the key query phase.
	\item The challenger computes the challenge ciphertext as follows:
	\begin{enumerate}
		\item Sample $R^* = (r_1, \ldots, r_{\ell}) \leftarrow \{0, 1\}^{\ell}$.
		\item Sample $\bm{\theta}_{i}, \bm{z}_{i} \leftarrow \{0, 1\}^{\secp}$ for all $i \in [\ell]$ where $\bm{\theta}_i = (\theta_{i, j})_{j \in [\secp]}$ and $\bm{z}_i = (z_{i, j})_{j \in [\secp]}$.
		\item Set $\wt{r}_{i} \seteq 0 \oplus \bigoplus_{j: \theta_{i, j} = 0}z_{i, j}$ for all $i \in [\ell]$.
		\item Compute $\act_{i} \leftarrow \ABE.\enc(\apk, x_1,  (\bm{\theta}_i, \wt{r}_i))$ for all $i \in [\ell]$.
\item $\tlaDec \leftarrow \CCO.\qobf(1^{\secp}, \msf{aDec}, R^*, m_1)$ where $\msf{aDec}$ is defined in Figure \ref{fig:pe1}.
		\item Set $\qct^* \seteq (\tlaDec, \{\ket{\bm{z}_{i}}_{\bm{\theta}_i}\}_{i \in [\ell]})$ and $\vk \seteq (\{(\bm{z}_i, \bm{\theta}_i)\}_{i \in [\ell]}, \vk_{\msf{aDec}})$.
	\end{enumerate}	
	The challenger sends $\qct^*$ to $\qA$.
	\item $\qA$ sends a certificate $\cert = (\bm{z}' = (\bm{z}_1', \ldots, \bm{z}_{\ell}'), \cert_{\msf{aDec}})$ and its internal state $\rho$ to the challenger where $\bm{z}_i' = (z_{i, j}')_{j \in [\secp]}$ for all $i \in [\ell]$.
	\item The challenger checks if $(z_{i, j} = z_{i, j}') \wedge (\theta_{i, j} = 1)$ holds for all $i \in [\ell]$ and $j \in [\lambda]$ and $\CCO.\vrfy(\vk_{\msf{aDec}}, \allowbreak \cert_{\msf{aDec}}) = \top$. If it does not hold, the challenger outputs $\bot$ as the final output of the experiment. Otherwise, go to the next step.
		\item The experiment outputs $\rho$ as a final output.
\end{enumerate}
By similar argument as in the previous hybrid, the hybrids $\hyb_2$ and $\hyb_3$ are indistinguishable by the certified everlasting security of compute-and-compare obfuscation. In other words, we have
\begin{equation}
		\TD(\hyb_{2}, \hyb_{3})
	\le \negl(\secp).
\end{equation}

\item[$\hyb_4:$] This hybrid proceeds exactly similar to $\hyb_3$ except that the ABE ciphertexts $\act_i$ is now replaced with encryption of $(\bm{\theta}_i, \wt{r}_i)$ where $\wt{r}_{i} \seteq r_i \oplus \bigoplus_{j: \theta_{i, j} = 0}z_{i, j}$ for all $i \in [\ell]$. In particular, the hybrid works as follows:

\begin{enumerate}
	\item The challenger computes $(\apk, \amsk) \leftarrow \ABE.\setup(1^{\secp})$, sets $\pk \seteq \apk$, and sends $\pk$ to $\qA$.
	\item The adversary $\qA$ sends any polynomial number of secret key queries for $P \in \mcl{P}$ at any point of the experiment. The challenger generates $\sk_P \leftarrow \ABE.\keygen(\amsk, P)$ and sends $\sk_P$ to $\qA$.
	\item $\qA$ sends a pair of challenge attributes $(x_0, x_1)$ and a pair of challenge messages $(m_0, m_1)$ satisfying the fact that $P(x_0) = P(x_1) = 0$ for all $P$ queried so far in the key query phase.
	\item The challenger computes the challenge ciphertext as follows:
	\begin{enumerate}
		\item Sample $R^* = (r_1, \ldots, r_{\ell}) \leftarrow \{0, 1\}^{\ell}$.
		\item Sample $\bm{\theta}_{i}, \bm{z}_{i} \leftarrow \{0, 1\}^{\secp}$ for all $i \in [\ell]$ where $\bm{\theta}_i = (\theta_{i, j})_{j \in [\secp]}$ and $\bm{z}_i = (z_{i, j})_{j \in [\secp]}$.
		\item Set $\wt{r}_{i} \seteq r_i \oplus \bigoplus_{j: \theta_{i, j} = 0}z_{i, j}$ for all $i \in [\ell]$.
		\item Compute $\act_{i} \leftarrow \ABE.\enc(\apk, x_1,  (\bm{\theta}_i, \wt{r}_i))$ for all $i \in [\ell]$.
\item $\tlaDec \leftarrow \CCO.\qobf(1^{\secp}, \msf{aDec}, R^*, m_1)$ where $\msf{aDec}$ is defined in Figure \ref{fig:pe1}.
		\item Set $\qct^* \seteq (\tlaDec, \{\ket{\bm{z}_{i}}_{\bm{\theta}_i}\}_{i \in [\ell]})$ and $\vk \seteq (\{(\bm{z}_i, \bm{\theta}_i)\}_{i \in [\ell]}, \vk_{\msf{aDec}})$.
	\end{enumerate}	
	The challenger sends $\qct^*$ to $\qA$.
	\item $\qA$ sends a certificate $\cert = (\bm{z}' = (\bm{z}_1', \ldots, \bm{z}_{\ell}'), \cert_{\msf{aDec}})$ and its internal state $\rho$ to the challenger where $\bm{z}_i' = (z_{i, j}')_{j \in [\secp]}$ for all $i \in [\ell]$.
\item The challenger checks if $(z_{i, j} = z_{i, j}') \wedge (\theta_{i, j} = 1)$ holds for all $i \in [\ell]$ and $j \in [\lambda]$ and $\CCO.\vrfy(\vk_{\msf{aDec}},\allowbreak \cert_{\msf{aDec}}) = \top$. If it does not hold, the challenger outputs $\bot$ as the final output of the experiment. Otherwise, go to the next step.
		\item The experiment outputs $\rho$ as a final output.
\end{enumerate}
Since all the secret keys $\sk_P$ corresponding to predicates $P$ queried by the adversary satisy the condition that $P(x_1) = 0$, we can depend on the semantic security of ABE and show that the hybrids $\hyb_3$ and $\hyb_4$ are indistinguishable from $\qA$'s point of view using the similar argument that we used while establishing the indistinguishability between the hybrids $\hyb_0$ and $\hyb_1$. In other words, we have
\begin{equation}
		\TD(\hyb_{3}, \hyb_{4})
	\le \negl(\secp).
\end{equation}

\end{description}	
Finally, we note that $\hyb_4$ is the original certified everlasting experiment of $\Sigma_{\tsf{pe-ce}}$ where the challenge bit is set to 1. Therefore, combing the advantages of $\qA$ in the consecutive hybrids and applying the triangular inequality, we have
\begin{equation}
		\TD(\hyb_{0}, \hyb_{4})
	\le \negl(\secp).
\end{equation}

Finally, it is easy to show the computational indistinguishability, i.e.,
\begin{align}
\left|\Pr[\C\expb{\Sigma_{\tsf{pe-ce}},~ \qA}{ada}{ind}(\secp, 0) = 1] - \Pr[\C\expb{\Sigma_{\tsf{pe-ce}},~ \qA}{ada}{ind}(\secp, 1) = 1]\right| \le \negl(\secp).
\end{align}
We skip the formal description as it follows from the security of ABE and the security of $\CECCO$.
We can erase information about $R=(r_,\ldots,r_\ell)$ by the security of ABE. Then, we can apply the security of $\CECCO$.
This completes the proof.
\end{proof}

\fi

\ifnum\anonymous=1
\else
\section*{Acknowledgement}
TM is supported by JST Moonshot JPMJMS2061-5-1-1, JST FOREST, MEXT QLEAP, the Grant-in-Aid for Scientific Research (B) No.JP19H04066,
the Grant-in Aid for Transformative Research Areas (A) 21H05183, and the Grant-in-Aid for Scientific Research (A)
No.22H00522.
TH is supported by JSPS research fellowship and by JSPS KAKENHI No. JP22J21864.
\fi

\newcommand{\etalchar}[1]{$^{#1}$}

\ifnum\cameraready=0
	\ifnum\llncs=0
	\appendix

\section{Omitted Proofs for Collusion-Resistant FE}\label{appsec:omitted_proof_crfecd}

We prove the adaptive security of our collusion-resistant scheme $\CED$ in~\cref{sec:CRFE_CED_const}.
That is, we show
\begin{align}
    \left|\Pr[\C\expb{\CED,~ \qA}{ada}{ind}(\secp, 0) = 1] - \Pr[\C\expb{\CED,~ \qA}{ada}{ind}(\secp, 1) = 1]\right| \leq \negl(\secp).
\end{align}

Let $\qA$ be a QPT adversary against the adaptive security.
We consider the following sequence of hybrids.

\begin{description}
\item[$\hyb_0$:]
This is the original adaptive security experiment where the challenge bit is set to be $0$. Specifically, it works as follows:
\begin{enumerate}
    \item The challenger generates $(\fe.\MPK,\fe.\MSK) \leftarrow \FE.\setup(1^{\secp})$, sets $\MPK \seteq \fe.\MPK$ and $\MSK \seteq \fe.\MSK$, and sends $\MPK$ to $\qA$.
    \item $\qA$ can make arbitrarily many key queries at any point of the experiment. When it makes a key query $f$, the challenger generates  $\fe.\sk_{g[f]}\gets \FE.\keygen(\fe.\MSK,g[f])$ and returns $\sk_f=\fe.\sk_{g[f]}$ to $\qA$. 
    \item $\qA$ sends $(m^{(0)},m^{(1)})$ to the challenger.\footnote{We use $(m^{(0)},m^{(1)})$ instead of $(m_0,m_1)$ to denote a pair of challenge messages to avoid a notational collision.} It must satisfy $f(m^{(0)})=f(m^{(1)})$ for all key queries $f$ that are made before or after sending $(m^{(0)},m^{(1)})$. 
    \item The challenger generates $(\qct,\vk)\gets \qenc(\MPK,m^{(0)})$. Specifically,
    \begin{enumerate}
     \item Generate $z_{i}, \theta_{i} \leftarrow \{0, 1\}^{\secp}$ for every $i\in[2\msglen+1]$.
    \item Generate $u_{i,j,b}\la\bit^\secp$ and compute $v_{i,j,b}\la\PRG(u_{i,j,b})$ for every $i\in[2\msglen+1]$, $j\in[\secp]$ and $b\in\bit$ and set $U=(u_{i,j,b})_{i\in[2\msglen+1],j\in[\secp],b\in\bit}$ and $V\seteq (v_{i,j,b})_{i\in[2\msglen+1],j\in[\secp],b\in\bit}$. 
    \item 
    Generate a state 
    \begin{align}
        \ket{\psi_{i,j}}\seteq \begin{cases}
    \ket{z_{i,j}}\ket{u_{i,j,z_{i,j}}} &\text{~if~} \theta_{i,j}=0\\ 
    \ket{0}\ket{u_{i,j,0}}+(-1)^{z_{i,j}} \ket{1}\ket{u_{i,j,1}} &\text{~if~} \theta_{i,j}=1
    \end{cases}
    \end{align}
    where $\theta_{i,j}$ (resp. $z_{i,j}$) is the $j$-th bit of $\theta_i$ (resp. $z_i$)
    for every $i\in [2\msglen+1]$ and $j\in [\secp]$. 
    \item  \label{step:generate_beta}
    Generate
      \begin{align}
        \beta_i\seteq \begin{cases}
        m_i^{(0)}\oplus \bigoplus_{j: \theta_{i,j} = 0} z_{i,j} &\text{if~}i\in[\msglen]\\
        0\oplus \bigoplus_{j: \theta_{i,j} = 0} z_{i,j} &\text{if~}i\in[\msglen+1, 2\msglen+1]
        \end{cases}.
    \end{align}
\item Generate $\fe.\ct\la\FE.\enc(\fe.\MPK,V\|\theta_1\|\ldots\|\theta_{2\msglen+1}\|\beta_1\|\ldots \|\beta_{2\msglen+1})$.
    \item Set $\qct = (\fe.\ct, \bigotimes_{i\in[2\msglen+1],j\in[\secp]}\ket{\psi_{i,j}})$ and $\vk = (U,(z_i,\theta_i)_{i\in[2\msglen+1]})$.
    \end{enumerate}
    The challenger sends $\qct$ to $\qA$.
    \item If $\qA$ sends a certificate of deletion $\cert$, the challenger computes $\vrfy(\vk_0, \cert)$ and sends the result to $\qA$.
    \item Again, the challenger answers key queries from $\qA$.
\item  When $\qA$ outputs a bit $b'$, the experiment outputs $b'$ if $f_\ell(m_0)=f_\ell(m_1)$ holds for all key queries $f_\ell$.
\end{enumerate}
\item[$\hyb_1$:] This is identical to $\hyb_0$  except that $v_{i,j,1\oplus z_{i,j}}$ is uniformly chosen from $\bit^{2\secp}$ instead of being set to be $\PRG(u_{i,j,1\oplus z_{i,j}})$ for all $i\in [2\msglen+1]$ and $j\in [\secp]$ such that $\theta_{i,j}=0$. 
\item[$\hyb_2$:] This is identical to $\hyb_1$  except that $(\beta_i)_{i\in [2\msglen+1]}$ is generated as       \begin{align}
        \beta_i\seteq \begin{cases}
        m_i^{(1)}\oplus \bigoplus_{j: \theta_{i,j} = 0} z_{i,j} &\text{if~}i\in[\msglen]\\
        0\oplus \bigoplus_{j: \theta_{i,j} = 0} z_{i,j} &\text{if~}i\in[\msglen+1, 2\msglen+1]
        \end{cases}.
    \end{align}
\item[$\hyb_3$:] 
This is identical to $\hyb_2$ except that $v_{i,j,b}$ is set to be $\PRG(u_{i,j,b})$ for all $i\in [2\msglen+1]$, $j\in [\secp]$, and $b\in \bit$. 
\end{description}
Note that $\hyb_3$ is identical to the original adaptive security experiment where the challenge bit is set to be $1$. Thus, we only have to prove 
\begin{align} \label{eq:goal_fe_adaptive}
    |\Pr[\hyb_0=1]-\Pr[\hyb_3=1]|\le \negl(\secp).
\end{align}
We prove \cref{eq:goal_fe_adaptive} by the following lemmata.
\begin{lemma}\label{lem:fe_adaptive_one}
If $\PRG$ is a secure PRG, 
\begin{align}
    |\Pr[\hyb_0=1]-\Pr[\hyb_1=1]|\le \negl(\secp).
\end{align}
\begin{proof}
Noting that $u_{i,j,1\oplus z_{i,j}}$ for $i\in[2\msglen+1]$ and $j\in [\secp]$ such that $\theta_{i,j}=0$ is used only for generating $v_{i,j,1\oplus z_{i,j}}$ in $\hyb_0$, \cref{lem:fe_adaptive_one} directly follows from the security of $\PRG$.
Note that we can simulate $\Vrfy(\vk_0,\cert)$ where $\cert=(c_{i,j},d_{i,j})_{i,j}$ since we need $\setbk{z_{i,j}}_{i,j}$ and $\setbk{u_{i,j,b}}_{i,j,b}$ such that $\theta_{i,j}=1$ for verification.
\end{proof}
\begin{lemma}\label{lem:fe_adaptive_two}
If $\FE$ is adaptively secure, 
\begin{align}
    |\Pr[\hyb_1=1]-\Pr[\hyb_2=1]|\le \negl(\secp).
\end{align}
\end{lemma}
\begin{proof}
For each $i\in[2\msglen+1]$ and $j\in [\secp]$ such that $\theta_{i,j}=0$, there is no $u$ such that $\PRG(u)=v_{i,j,1\oplus z_{i,j}}$ except for probability $2^{-\secp}$. Let $\mathsf{Good}$ be the event that the above holds for all $i\in[2\msglen+1]$ and $j\in [\secp]$. We have $\Pr[\mathsf{Good}]\ge 1-(2\msglen+1) \secp 2^{-\secp}=1-\negl(\secp)$. 
We prove that whenever $\mathsf{Good}$ occurs, we have 
\begin{align} \label{eq:equivalence_function_normal_security}
&g[f]((V,\theta_1,\ldots,\theta_{2\msglen+1},\beta_1^{(0)},\ldots,\beta_{2\msglen+1}^{(0)}),(b_{i,j},u_{i,j})_{i\in [2\msglen+1],j\in [\secp]})\\
=
&g[f]((V,\theta_1,\ldots,\theta_{2\msglen+1},\beta_1^{(1)},\ldots,\beta_{2\msglen+1}^{(1)}),(b_{i,j},u_{i,j})_{i\in [2\msglen+1],j\in [\secp]})
\end{align}
for all key queries $f$ and $(b_{i,j},u_{i,j})_{i\in [2\msglen+1],j\in [\secp]}$
where 
  \begin{align}
        \beta_i^{(a)}\seteq \begin{cases}
        m_i^{(a)}\oplus \bigoplus_{j: \theta_{i,j} = 0} z_{i,j} &\text{if~}i\in[\msglen]\\
        0\oplus \bigoplus_{j: \theta_{i,j} = 0} z_{i,j} &\text{if~}i\in[\msglen+1, 2\msglen+1]
        \end{cases}
    \end{align}
for $a\in \bit$. 
If this is proven, \cref{lem:fe_adaptive_two} directly follows from the adaptive security of $\FE$. 

Below, we prove \cref{eq:equivalence_function_normal_security}. 
We consider the following two cases.
\begin{itemize}
    \item If $\PRG(u_{i,j})=v_{i,j,b_{i,j}}$ holds for every $i\in[2\msglen+1]$ and $j\in[\secp]$, then by the assumption that $\mathsf{Good}$ occurs, we have $b_{i,j}=z_{i,j}$ for all $i\in[2\msglen+1]$ and $j\in [\secp]$ such that $\theta_{i,j}=0$. Then we have $\beta_i^{(a)}\oplus \bigoplus_{j: \theta_{i,j} = 0} b_{i,j}=m^{(a)}_i$ for $i\in [\msglen]$ and  $\beta_{2\msglen+1}^{(a)}\oplus \bigoplus_{j: \theta_{2\msglen+1,j} = 0} b_{2\msglen+1,j}=0$ for $a\in \bit$. Then the LHS of  \cref{eq:equivalence_function_normal_security} is equal to $f(m^{(0)})$ and the RHS of \cref{eq:equivalence_function_normal_security} is equal to $f(m^{(1)})$. By the restriction on $\qA$ in the adaptive security experiment, we have  $f(m^{(0)})=f(m^{(1)})$. Therefore, both sides of \cref{eq:equivalence_function_normal_security} are equal to  $f(m^{(0)})=f(m^{(1)})$. 
    \item Otherwise, both sides of \cref{eq:equivalence_function_normal_security} are equal to $\bot$.
\end{itemize}
In either case, \cref{eq:equivalence_function_normal_security} holds.
Note that we can simulate $\Vrfy(\vk_b,\cert)$ where $\cert=(c_{i,j},d_{i,j})_{i,j}$ since we need $\setbk{z_{i,j}}_{i,j}$ and $\setbk{u_{i,j,b}}_{i,j,b}$ such that $\theta_{i,j}=1$ for verification.
This completes the proof of \cref{lem:fe_adaptive_two}.
\end{proof}
\begin{lemma}
If $\PRG$ is a secure PRG, 
\begin{align}
    |\Pr[\hyb_2=1]-\Pr[\hyb_3=1]|\le \negl(\secp).
\end{align}
\end{lemma}
\begin{proof}
Noting that $u_{i,j,1\oplus z_{i,j}}$ for $i\in[2\msglen+1]$ and $j\in [\secp]$ such that $\theta_{i,j}=0$ is used only for generating $v_{i,j,1\oplus z_{i,j}}$ in $\hyb_3$, \cref{lem:fe_adaptive_one} directly follows from the security of $\PRG$.
Note that we can simulate $\Vrfy(\vk_1,\cert)$ where $\cert=(c_{i,j},d_{i,j})_{i,j}$ since we need $\setbk{z_{i,j}}_{i,j}$ and $\setbk{u_{i,j,b}}_{i,j,b}$ such that $\theta_{i,j}=1$ for verification.
\end{proof}
\end{lemma}

\newcommand{\okskfe}[0]{\algo{1KeySKFE}}
\newcommand{\okskfes}[0]{\keys{1keyskfe}}
\newcommand{\okcskfe}[0]{\algo{SKFE}}
\newcommand{\pkfe}[0]{\keys{pkfe}}
\newcommand{\okcskfes}[0]{\keys{skfe}}

\section{Adaptively Secure Public-Slot PKFE}\label{sec:adaptive_PKFE_public_slot}
In this section, we present an adaptively secure public-slot PKFE scheme based on
\begin{itemize}
\item Selectively secure PKFE,
\item Selectively single-key function private SKFE, and
\item Adaptively single-key single-ciphertext public-slot SKFE.
\end{itemize}

We need to show how to achieve adaptively single-key single-ciphertext public-slot SKFE since it is an essentail building block.
Our adaptively secure public-slot PKFE scheme is presented in~\cref{sec:adaptive_PKFE_pub_slot}.

We present an adaptively single-key single-ciphertext public-slot SKFE scheme based on
\begin{itemize}
\item Selectively single-key single-ciphertext public-slot SKFE and
\item Receiver non-committing encryption
\end{itemize}
in~\cref{sec:constructions_ada_one_one_SKFE_public_slot}.

We also present a selectively secure single-ciphertext SKFE with public scheme based on IO and OWFs. This construction uses an MIFE scheme whose arity is $2$ (i.e., $2$-input FE) by Goldwasser et al.~\cite{EC:GGGJKL14}. We introduce necessary tools and definitions in~\cref{sec:tools_for_adaptive_single_key_ct_SKFE_public_slot,sec:SKFE_public_slot_selective_security}.

\subsection{Building Blocks}\label{sec:tools_for_adaptive_single_key_ct_SKFE_public_slot}
We introduce building blocks for our adaptively single-key single-ciphertext public-slot SKFE scheme.

\paragraph{Non-committing encryption.}
We recall the notion of (secret-key) receiver non-committing encryption (NCE)~\cite{STOC:CFGN96,EC:JarLys00,TCC:CanHalKat05}.

\begin{definition}[Secret-Key RNCE (Syntax)]\label{def:sk_nce_syntax}
A secret-key NCE scheme is a tuple of PPT algorithms $(\keygen,\Enc,\Dec,\allowbreak\Fake,\Reveal)$ with plaintext space $\Ms$.
\begin{description}
    \item [$\keygen(1^\secp)\ra (\ek,\dk,\aux)$:] The key generation algorithm takes as input the security 
    parameter $1^\secp$ and outputs a key pair $(\ek,\dk)$ and an auxiliary information $\aux$.
    \item [$\Enc(\ek,m)\ra \ct$:] The encryption algorithm takes as input $\ek$ and a plaintext $m\in\cM$ and outputs a ciphertext $\ct$.
    \item [$\Dec(\dk,\ct)\ra m^\prime \mbox{ or }\bot$:] The decryption algorithm takes as input $\dk$ and $\ct$ and outputs a plaintext $m^\prime$ or $\bot$.
    \item [$\Fake(\ek,\aux)\ra \tlct$:] The fake encryption algorithm takes $\dk$ and $\aux$, and outputs a fake ciphertext $\tlct$.
    \item [$\Reveal(\ek,\aux,\tlct,m)\ra \tldk $:] The reveal algorithm takes $\ek,\aux,\tlct$ and $m$, and outputs a fake secret key $\tldk$.
\end{description}  
\end{definition}

\begin{definition}[Correctness of secret-key NCE]\label{def:correctness_sk_nce}
There exists a negligible function $\negl$ such that for any $\secp\in \N$, $m\in\Ms$,
\begin{align}
\Pr\left[
m^\prime \ne m
\ \middle |
\begin{array}{ll}
(\ek,\dk,\aux)\la\keygen(1^\secp)\\
\ct\lrun \Enc(\ek,m)\\
m^\prime\la\Dec(\dk,\ct)
\end{array}
\right] 
\le\negl(\secp).
\end{align}
\end{definition}

\begin{definition}[Receiver Non-Committing (RNC) Security for SKE]\label{def:sk_nce_security}
A secret-key NCE scheme is RNC secure if it satisfies the following.
Let $\Sigma=(\keygen, \Enc, \Dec, \Fake,\Reveal)$ be a secret-key NCE scheme.
We consider the following security experiment $\expb{\Sigma,\qA}{sk}{rec}{nc}(\secp,b)$.

\begin{enumerate}
    \item The challenger computes $(\ek,\dk,\aux) \lrun \keygen(1^\secp)$ and sends $1^\secp$ to the adversary $\qA$.
    \item $\qA$ sends an encryption query $m$ to the challenger. The challenger computes and returns $\ct\lrun \Enc(\ek,m)$ to $\qA$. This process can be repeated polynomially many times.
    \item $\qA$ sends a query $m \in \Ms$ to the challenger.
    \item The challenger does the following.
    \begin{itemize}
    \item If $b =0$, the challenger generates $\ct \lrun \Enc(\ek,m)$ and returns $(\ct,\dk)$ to $\qA$.
    \item If $b=1$, the challenger generates $\tlct \lrun \Fake(\ek,\aux)$ and\\
    $\tldk \lrun \Reveal(\ek,\aux,\tlct,m)$ and returns $(\tlct,\tldk)$ to $\qA$.
    \end{itemize}
    \item Again $\qA$ can send encryption queries.
    \item $\qA$ outputs $b'\in \bit$.
\end{enumerate}
Let $\advc{\Sigma,\qA}{sk}{rec}{nc}(\secp)$ be the advantage of the experiment above.
We say that the $\Sigma$ is RNC secure if for any QPT adversary, it holds that
\begin{align}
\advc{\Sigma,\qA}{sk}{rec}{nc}(\secp)\seteq \abs{\Pr[ \expc{\Sigma,\qA}{sk}{rec}{nc}(\secp, 0)=1] - \Pr[ \expc{\Sigma,\qA}{sk}{rec}{nc}(\secp, 1)=1] }\leq \negl(\secp).
\end{align}
\end{definition}

\begin{theorem}[{\cite[Section 7.2 in the eprint version]{C:KNTY19}}]\label{thm:indcpa-pke_to_rnc-pke}
If there exists an IND-CPA secure SKE scheme (against QPT adversaries), there exists an RNC secure secret-key NCE scheme (against QPT adversaries) with plaintext space $\bit^{\ell}$, where $\ell$ is some polynomial, respectively.
\end{theorem}

\paragraph{Functional Encryption.}
\begin{definition}[Public-Key FE (Syntax)]\label{def:pkfe}
A public-key functional encryption (PKFE) scheme for a class $\cF$ of functions is a tuple of PPT algorithms
$\Sigma=(\Setup,\keygen,\Enc,\Dec)$ with plaintext space $\Ms$, ciphertext space $\Cs$, master public key space $\mathcal{MPK}$, master secret key space $\mathcal{MSK}$, and secret key space $\mathcal{SK}$,
that work as follows.
\begin{description}
    \item[$\Setup(1^\secp)\ra(\MPK,\MSK)$:]
    The setup algorithm takes the security parameter $1^\secp$ as input, and outputs a master public key $\MPK\in\mathcal{MPK}$
    and a master secret key $\MSK\in\mathcal{MSK}$.
    \item[$\keygen(\MSK,f)\ra\sk_f$:]
    The key generation algorithm takes $\MSK$ and $f\in\mathcal{F}$ as input,
    and outputs a secret key $\sk_f\in\mathcal{SK}$.
    \item[$\Enc(\MPK,m)\ra \CT$:]
    The encryption algorithm takes $\MPK$ and $m\in\Ms$ as input,
     and outputs  a ciphertext $\CT\in \Cs$.
    \item[$\Dec(\sk_f,\CT)\ra y$ $\mbox{\bf  or }\bot$:]
    The decryption algorithm takes $\sk_f$ and $\CT$ as input, and outputs $y$ or $\bot$.
\end{description}
\end{definition}

We require that a PKFE scheme satisfies correctness defined below.
\begin{definition}[Correctness of PKFE]\label{def:correctness_pkfe}
There exists a negligible function $\negl$ such that for any $\secp\in \N$, $m\in\Ms$ and $f\in\mathcal{F}$
\begin{align}
\Pr\left[
y\ne f(m)
\ \middle |
\begin{array}{ll}
(\MPK,\MSK)\la\Setup(1^\secp)\\
\sk_f\la\keygen(\MSK,f)\\
\CT\lrun \Enc(\MPK,m)\\
y\la\Dec(\sk_f,\ct)
\end{array}
\right] 
\le\negl(\secp).
\end{align}
\end{definition}

\begin{definition}[Selective Security of PKFE]\label{def:selective_security_pkfe}
Let $\Sigma=(\Setup,\keygen,\Enc,\Dec)$ be a PKFE scheme.
We consider the following security experiment $\expb{\Sigma,\qA}{sel}{ind}(\secp,b)$ against a QPT adversary $\qA$. 
\begin{enumerate}
    \item $\qA$ sends $(m_0,m_1)\in\Ms^2$ to the challenger.
    \item The challenger runs $(\MPK,\MSK)\la \Setup(1^\secp)$, computes $\CT\la \Enc(\MPK,m_b)$, and sends $(\MPK,\CT)$ to $\qA$.
    \item $\qA$ is allowed to make arbitrarily many key queries.
    For the $\ell$-th key query, the challenger receives $f_\ell\in\mathcal{F}$, 
    computes $\sk_{f_\ell}\la\keygen(\MSK,f_\ell)$,
    and sends $\sk_{f_\ell}$ to $\qA$.
    \item $\qA$ outputs $b'\in\bit$. 
    If $f_\ell(m_0)=f_\ell(m_1)$ holds for all key queries $f_\ell$, the experiment outputs $b'$. Otherwise, it outputs $\bot$.
\end{enumerate}
We say that $\Sigma$ is adaptively secure if for any $QPT$ adversary $\qA$ it holds that
\begin{align}
\advb{\Sigma,\qA}{sel}{ind}(\secp) \seteq \abs{\Pr[\expb{\Sigma,\qA}{sel}{ind}(\secp,0)=1]  - \Pr[\expb{\Sigma,\qA}{sel}{ind}(\secp,1)=1]} \leq \negl(\secp).
\end{align}
\end{definition}

\begin{definition}[Secret-Key FE (Syntax)]\label{def:skfe}
A secret-key functional encryption (SKFE) scheme for a class $\cF$ of functions is a tuple of PPT algorithms
$\Sigma=(\Setup,\keygen,\Enc,\Dec)$ with plaintext space $\Ms$, ciphertext space $\Cs$, master secret key space $\mathcal{MSK}$, and secret key space $\mathcal{SK}$,
that work as follows.
\begin{description}
    \item[$\Setup(1^\secp)\ra \MSK$:]
    The setup algorithm takes the security parameter $1^\secp$ as input, and outputs a master secret key $\MSK\in\mathcal{MSK}$.
    \item[$\keygen(\MSK,f)\ra\sk_f$:]
    The key generation algorithm takes $\MSK$ and $f\in\mathcal{F}$ as input,
    and outputs a secret key $\sk_f\in\mathcal{SK}$.
    \item[$\Enc(\MSK,m)\ra \CT$:]
    The encryption algorithm takes $\MSK$ and $m\in\Ms$ as input,
     and outputs  a ciphertext $\CT\in \Cs$.
    \item[$\Dec(\sk_f,\CT)\ra y$ $\mbox{\bf  or }\bot$:]
    The decryption algorithm takes $\sk_f$ and $\CT$ as input, and outputs $y$ or $\bot$.
\end{description}
\end{definition}

We require that an SKFE scheme satisfies correctness defined below.
\begin{definition}[Correctness of SKFE]\label{def:correctness_skfe}
There exists a negligible function $\negl$ such that for any $\secp\in \N$, $m\in\Ms$ and $f\in\mathcal{F}$
\begin{align}
\Pr\left[
y\ne f(m)
\ \middle |
\begin{array}{ll}
\MSK\la\Setup(1^\secp)\\
\sk_f\la\keygen(\MSK,f)\\
\CT\lrun \Enc(\MSK,m)\\
y\la\Dec(\sk_f,\ct)
\end{array}
\right] 
\le\negl(\secp).
\end{align}
\end{definition}

\begin{definition}[Adaptive Single-Key Single-Ciphertext Security of SKFE]\label{def:adaptive_single-key-ct_security_skfe}
Let $\Sigma=(\Setup,\keygen,\Enc,\Dec)$ be an SKFE scheme.
We consider the following security experiment $\expc{\Sigma,\qA}{ada}{1key}{1ct}(\secp,b)$ against a QPT adversary $\qA$. 
\begin{enumerate}
    \item The challenger runs $\MSK\la \Setup(1^\secp)$.
    \item The adversary makes the following encryption query and key query in no particular order.
\begin{itemize}
    \item $\qA$ sends $f\in\mathcal{F}$ to the challenger. The challenger computes $\sk_{f}\la\keygen(\MSK,f)$, and returns $\sk_{f}$ to $\qA$. $\qA$ can do this process one once.
    \item $\qA$ sends $(m_0,m_1)\in\mathcal{M}^2$ to the challenger. The challenger computes $\CT\la\enc(\MSK,m_{b})$, and returns $\CT$ to $\qA$. $\qA$ can do this process one once.
\end{itemize}
    \item $\qA$ outputs $b'\in\bit$. 
    If $f(m_{0})=f(m_{1})$ holds, the experiment outputs $b'$. Otherwise, it outputs $\bot$.
\end{enumerate}
We say that $\Sigma$ is adaptively single-key single-ciphertext secure if for any $QPT$ adversary $\qA$ it holds that
\begin{align}
\advc{\Sigma,\qA}{ada}{1key}{1ct}(\secp) \seteq \abs{\Pr[\expc{\Sigma,\qA}{ada}{1key}{1ct}(\secp,0)=1]  - \Pr[\expc{\Sigma,\qA}{ada}{1key}{1ct}(\secp,1)=1]} \leq \negl(\secp).
\end{align}
\end{definition}

\begin{definition}[Selective Single-Key Function Privacy of SKFE]\label{def:selective_single-key_function-privacy_skfe}
Let $\Sigma=(\Setup,\keygen,\Enc,\Dec)$ be an SKFE scheme.
We consider the following security experiment $\expc{\Sigma,\qA}{sel}{1key}{fp}(\secp,b)$ against a QPT adversary $\qA$. 
\begin{enumerate}
    \item The challenger sends $(m_{1,0},m_{1,1}),\ldots, (m_{q,0},m_{q,1}) \in \cM^{2q}$ to the challenger.
    \item The challenger runs $\MSK\la \Setup(1^\secp)$, computes $\CT_i \la \enc(\MSK, m_{i,b})$ for all $i \in [q]$, and returns $(\CT_1,\ldots,\CT_q)$ to $\qA$.
    \item $\qA$ sends $(f_0,f_1)\in\mathcal{F}^2$ to the challenger. The challenger computes $\sk_{f_b}\la\keygen(\MSK,f_b)$, and returns $\sk_{f_b}$ to $\qA$. $\qA$ can do this process only once.
    \item $\qA$ outputs $b'\in\bit$.
    If $f_0(m_{i,0})=f_1(m_{i,1})$ holds for all $i \in [q]$, the experiment outputs $b'$. Otherwise, it outputs $\bot$.
\end{enumerate}
We say that $\Sigma$ is selectively single-key function private if for any $QPT$ adversary $\qA$ it holds that
\begin{align}
\advc{\Sigma,\qA}{sel}{1key}{fp}(\secp) \seteq \abs{\Pr[\expc{\Sigma,\qA}{sel}{1key}{fp}(\secp,0)=1]  - \Pr[\expc{\Sigma,\qA}{sel}{1key}{fp}(\secp,1)=1]} \leq \negl(\secp).
\end{align}
\end{definition}

\paragraph{2-input FE.}
We recall the notion of $2$-input FE. The following definitions are special cases of multi-input functional encryption (MIFE) by Goldwasser et al.~\cite{EC:GGGJKL14}.

\begin{definition}[$2$-input FE (Syntax)]\label{def:two-input_fe}
A $2$-input FE scheme for a class $\cF$ of functions is a tuple of PPT algorithms
$\Sigma=(\Setup,\keygen,\Enc,\Dec)$ with plaintext space $\Ms$, ciphertext space $\Cs$, master secret key space $\mathcal{MSK}$, and secret key space $\mathcal{SK}$,
that work as follows.
\begin{description}
    \item[$\Setup(1^\secp)\ra \MSK$:]
    The setup algorithm takes the security parameter $1^\secp$ as input, and outputs a master secret key $\MSK\in\mathcal{MSK}$ and two encryption keys $\EK_1$ and $\EK_2$.
    \item[$\keygen(\MSK,f)\ra\sk_f$:]
    The key generation algorithm takes $\MSK$ and $f\in\mathcal{F}$ as input,
    and outputs a secret key $\sk_f\in\mathcal{SK}$.
    \item[$\Enc(\EK_i,x)\ra \CT_i$:]
    The encryption algorithm takes $\EK_i$ and $x\in\Ms$ as input,
     and outputs a ciphertext $\CT_i\in \Cs$.
    \item[$\Dec(\sk_f,\CT_1,\CT_2)\ra z$ $\mbox{\bf  or }\bot$:]
    The decryption algorithm takes $\sk_f$ and $(\CT_1,\CT_2)$ as input, and outputs $y$ or $\bot$.
\end{description}
\end{definition}

We require that a $2$-input FE scheme satisfies correctness defined below.
\begin{definition}[Correctness of $2$-input FE]\label{def:correctness_two-input_fe}
There exists a negligible function $\negl$ such that for any $\secp\in \N$, $(x,y)\in\Ms^2$ and $f\in\mathcal{F}$
\begin{align}
\Pr\left[
z\ne f(x,y)
\ \middle |
\begin{array}{ll}
(\MSK,\EK_1,\EK_2)\la\Setup(1^\secp)\\
\sk_f\la\keygen(\MSK,f)\\
\CT_1\lrun \Enc(\EK_1,x),\CT_2\lrun \Enc(\EK_2,y) \\
z\la\Dec(\sk_f,\CT_1,\CT_2)
\end{array}
\right] 
\le\negl(\secp).
\end{align}
\end{definition}

\begin{definition}[$(1,1)$-$\sel$-$\ind$ Security of $2$-FE]\label{def:selective_single-key-ct_security_two_fe}
Let $\Sigma=(\Setup,\keygen,\Enc,\Dec)$ be a $2$-input FE scheme.
We consider the following security experiment $\expb{\Sigma,\qA}{sel}{1ct}(\secp,b)$ against a QPT adversary $\qA$. 
\begin{enumerate}
	\item The adversary sends $((x_0,y_0),(x_1,y_1))$ to the challenger.
    \item The challenger runs $(\MSK,\EK_1,\EK_2)\la \Setup(1^\secp)$, computes $\CT_1 \gets \Enc(\EK_1,x_b)$ and $\CT_2\gets \Enc(\EK_2,y_b)$, and sends $(\EK_{2},\CT_1,\CT_2)$ to $\qA$.
    \item $\qA$ can send a key query $f_i\in\mathcal{F}$ to the challenger. The challenger computes $\sk_{f_i}\la\keygen(\MSK,f_i)$, and returns $\sk_{f_i}$ to $\qA$. $\qA$ can send unbounded polynomially many key queries. Let $q_{k}$ be the total number of the key queries.
    \item $\qA$ outputs $b'\in\bit$.
    If $f_i(x_{0},y)=f_i(x_{1},y)$ holds for all $y\in\Ms$ and $i\in [q_{k}]$ (we call $\qA$ is valid), the experiment outputs $b'$. Otherwise, it outputs $\bot$.
\end{enumerate}
We say that $\Sigma$ is $(1,1)$-$\sel$-$\ind$ secure if for any $QPT$ adversary $\qA$ it holds that
\begin{align}
\advb{\Sigma,\qA}{sel}{1ct}(\secp) \seteq \abs{\Pr[\expb{\Sigma,\qA}{sel}{1ct}(\secp,0)=1]  - \Pr[\expb{\Sigma,\qA}{sel}{1ct}(\secp,1)=1]} \leq \negl(\secp).
\end{align}
Here, $(1,1)$ means that $\qA$ is given \emph{one} encryption key $\EK_2$ and \emph{one} challenge ciphertext vector $(\CT_1,\CT_2)$.
\end{definition}
The security definition above is a special case of $(t,q)$-$\sel$-$\ind$ security by Goldwasser et al.~\cite{EC:GGGJKL14}, where $t$ is the number of corrupted encryption keys and $q$ is the number of challenge ciphertext vectors.

\begin{theorem}[\cite{EC:GGGJKL14}]
If there exist IO and OWFs, there exists $(1,1)$-$\sel$-$\ind$ secure $2$-input FE for all polynomial-size circuits.
\end{theorem}
Although Goldwasser et al. proved a more general theorem ($(t,q)$-$\sel$-$\ind$ secure $n$-input FE where $t\le n$ and $q$ is an a-priori bounded polynomial), the simplified version above is sufficient for our purpose.

\subsection{Variants of Security Definitions}\label{sec:SKFE_public_slot_selective_security}

We can consider the secret-key variant of~\cref{def:fe_public_slot}, where $\setup$ generates only a master secret key $\MSK$ and $\enc$ uses $\MSK$ instead of $\MPK$. Correctness of public-slot SKFE is a natural extension of~\cref{def:correctness_fe_public_slot}.
We omit syntax and correctness for public-slot SKFE.

Below, we introduce variants of security definitions for public-slot SKFE.
\begin{definition}[Single-Key Security of Public-Slot SKFE]\label{def:selective_security_skfe_public_slot}
Let $\Sigma=(\Setup,\keygen,\Enc,\Dec)$ be a public-slot SKFE scheme.
We consider the following security experiment $\expb{\Sigma,\qA}{ada}{1ct}(\secp,b)$ against a QPT adversary $\qA$. 
\begin{enumerate}
    \item The challenger runs $\MSK\la \Setup(1^\secp)$.
    \item $\qA$ is allowed to make arbitrarily many key queries.
    For the $\ell$-th key query, the challenger receives $f_\ell\in\mathcal{F}$, 
    computes $\sk_{f_\ell}\la\keygen(\MSK,f_\ell)$,
    and sends $\sk_{f_\ell}$ to $\qA$.
    \item $\qA$ can send a challenge plaintext pair $(m_0,m_1)\in\Ms^2$ to the challenger.
    \item The challenger computes $\CT\la \Enc(\MSK,m_b)$ and sends $\CT$ to $\qA$.
    \item Again, $\qA$ is allowed to make arbitrarily many key queries.
    \item $\qA$ outputs $b'\in\bit$. 
    If $f_\ell(m_0,\pub)=f_\ell(m_1,\pub)$ holds for all key queries $f_\ell$ and public inputs $\pub\in \mathcal{P}$, the experiment outputs $b'$. Otherwise, it outputs $\bot$.
\end{enumerate}
We say that $\Sigma$ is adaptively single-ciphertext secure if for any $QPT$ adversary $\qA$ it holds that
\begin{align}
\advb{\Sigma,\qA}{ada}{1ct}(\secp) \seteq \abs{\Pr[\expb{\Sigma,\qA}{ada}{1ct}(\secp,0)=1]  - \Pr[\expb{\Sigma,\qA}{ada}{1ct}(\secp,1)=1]} \leq \negl(\secp).
\end{align}
If the adversary must declare the challenge plaintext pair $(m_0,m_1)$ at the very beginning of the experiment, we say that $\Sigma$ is selectively single-ciphertext secure and denote the advantage and experiment by $\advb{\Sigma,\qA}{sel}{1ct}(\secp)$ and $\expb{\Sigma,\qA}{sel}{1ct}(\secp)$, respectively.

If the adversary is allowed to make only one key query, we say $\Sigma$ is adaptively/selectively single-key single-ciphertext secure and denote the advantage and experiment by $\advc{\Sigma,\qA}{xxx}{1key}{1ct}(\secp)$ and $\expc{\Sigma,\qA}{xx}{1key}{1ct}(\secp)$, respectively, where $\mathsf{xxx} \in \setbk{\mathsf{sel},\mathsf{ada}}$.
\end{definition}

\subsection{Adaptively Single-Key Single-Ciphertext Public-Slot SKFE Scheme}\label{sec:constructions_ada_one_one_SKFE_public_slot}

\paragraph{Selectively single-ciphertext public-slot SKFE.}
First, we present our selectively single-ciphertext public-slot SKFE scheme $\oneselFEps$.
\paragraph{Ingredients.}
\begin{itemize}
    \item $(1,1)$-$\sel$-$\ind$ secure $2$-input FE $\twoFE = \twoFE.(\setup,\keygen, \enc, \dec)$ for all polynomial-size circuits.
\end{itemize}
\paragraph{Scheme description.} Our scheme $\oneselFEps=\oneselFEps.(\setup,\keygen,\enc,\dec)$ is as follows.
\begin{description}
\item[$\oneselFEps.\setup(1^{\secp})$:]$ $
\begin{enumerate}
    \item Generate $(\twofe.\msk,\twofe.\ek_1,\twofe.\ek_2) \leftarrow \twoFE.\setup(1^{\secp})$.
    \item Output $\MSK \seteq (\twofe.\msk,\twofe.\ek_1,\twofe.\ek_2)$.
\end{enumerate}

\item[$\oneselFEps.\keygen(\MSK, f)$:]$ $
\begin{enumerate}
    \item Parse $\MSK =(\twofe.\msk,\twofe.\ek_1,\twofe.\ek_2)$.
    \item Generate $\twofe.\SK_f \gets \twoFE.\keygen(\twofe.\msk,f)$.
    \item Output $\SK_f \seteq \twofe.\sk_f$.
\end{enumerate}

\item[$\oneselFEps.\enc(\MSK, m)$:] $ $
\begin{enumerate}
    \item Parse $\MSK =(\twofe.\msk,\twofe.\ek_1,\twofe.\ek_2)$.
    \item Generate $\twofe.\ct_1 \gets \twoFE.\Enc(\twofe.\ek_1,m)$.
    \item Output $\CT \seteq (\twofe.\ct_1,\twofe.\ek_2)$.
\end{enumerate}

\item[$\oneselFEps.\dec(\SK_f, \CT,\pub)$:] $ $
\begin{enumerate}
\item Parse $\SK_f= \twofe.\sk_f$ and $\CT = (\twofe.\ct_1, \twofe.\ek_2)$.
    \item Compute $\twofe.\ct_2 \gets \twoFE.\Enc(\twofe.\ek_2,\pub)$.
    \item Compute and output $y\seteq \twoFE.\Dec(\twofe.\sk_f,\twofe.\ct_1,\twofe.\ct_2)$.
\end{enumerate}
\end{description}

\paragraph{Correctness.} It is easy to see correctness holds due to correctness of $\twoFE$.

\begin{theorem}\label{thm:selective_single-key-ct_SKFE_public_slot}
If $\twoFE$ is $(1,1)$-$\sel$-$\ind$ secure $2$-input FE for all polynomial-size circuits, $\oneselFEps$ is selectively single-ciphertext public-slot SKFE for all polynomial-size circuits.
\end{theorem}
This theorem immediately yields a selectively single-key single-ciphertext public-slot SKFE for all polynomial-size circuits.

\begin{proof}
Let $\expb{\oneselFEps,\qA}{sel}{1ct}(\secp,b)$ denote the selective single-ciphertext security of public-slot SKFE.
We construct an algorithm $\qB$ that breaks $(1,1)$-$\sel$-$\ind$ security of $\twoFE$ by using an adversary $\qA$ that breaks selectively single-ciphertext security of $\oneselFEps$. $\qB$ does the following.
\begin{enumerate}
\item First, $\qA$ sends $(m_0,m_1)$. Then, $\qB$ chooses a random $y\chosen \Ms$, sets $(x_0,x_1)\seteq (m_0,m_1)$ and $(y_0,y_1)\seteq (y,y)$, and passes $((x_0,y_0),(x_1,y_1))$ to its challenger.
\item When $\qB$ receives $(\twofe.\ek_2,\twofe.\ct_1,\twofe.\ct_2)$ from its challenger, $\qB$ sets $\CT \seteq (\twofe.\ct_1,\twofe.\ek_2)$ and passes $\CT$ to $\qA$
\item When $\qA$ sends a key query $f_i$, $\qB$ passes $f_i$ to its challenger, receives $\twofe.\sk_{f_i}\gets \twoFE.\keygen(\twofe.\msk,f_i)$, and passes $\SK_{f_i} \seteq \twofe.\sk_{f_i}$ to $\qA$.
\item When $\qA$ outputs $b^\prime$, $\qB$ outputs $b^\prime$.
\end{enumerate}
If $\qA$ is valid in the experiment of selective single-ciphertext security for public-slot SKFE $\oneselFEps$, it holds $f_i(x_0,y^\prime)=f_i(x_1,y^\prime)$ for all $i \in [q]$ and $y^\prime \in \Ms$. Then, $\qB$ is also a valid adversary in the experiment of $(1,1)$-$\sel$-$\ind$ security for $\twoFE$ since $\qB$ received $\twofe.\ek_2$.
It is easy to see the following.
\begin{itemize}
\item If $\twofe.\ct_1 \gets \twoFE.\Enc(\twofe.\msk,x_0)$ and $\twofe.\ct_2\gets \twoFE.\Enc(\twofe.\msk,y)$, $\qB$ perfectly simulates $\expb{\oneselFEps,\qA}{sel}{1ct}(\secp,0)$.
\item If $\twofe.\ct_1 \gets \twoFE.\Enc(\twofe.\msk,x_1)$ and $\twofe.\ct_2\gets \twoFE.\Enc(\twofe.\msk,y)$, $\qB$ perfectly simulates $\expb{\oneselFEps,\qA}{sel}{1ct}(\secp,1)$.
\end{itemize}
Thus, if $\qA$ distinguishes $\expb{\oneselFEps,\qA}{sel}{1ct}(\secp,0)$ from $\expb{\oneselFEps,\qA}{sel}{1ct}(\secp,1)$, $\qB$ distinguishes $\expb{\twoFE,\qB}{sel}{1ct}(\secp,0)$ from $\expb{\twoFE,\qB}{sel}{1ct}(\secp,1)$.
This completes the proof.
\end{proof}

\paragraph{Adaptively single-key single-ciphertext public-slot SKFE.}
Next, we present our adaptively single-key single-ciphertext public-slot SKFE scheme $\oneadaFEps$.
\paragraph{Ingredients.}
\begin{itemize}
    \item Selectively single-key single-ciphertext public-slot SKFE $\oneselFEps = \oneselFEps.(\setup,\keygen, \enc, \dec)$ for all polynomial-size circuits.
    \item Receiver non-committing encryption $\NCE = \NCE.(\keygen,\Enc,\Dec,\Fake,\Reveal)$.
\end{itemize}

\paragraph{Scheme description.} Our scheme $\oneadaFEps=\oneadaFEps.(\setup,\keygen,\enc,\dec)$ is as follows.
\begin{description}
\item[$\oneadaFEps.\setup(1^{\secp})$:]$ $
\begin{enumerate}
    \item Generate $\sel.\msk \leftarrow \oneselFEps.\setup(1^{\secp})$.
    \item Generate $(\nce.\ek,\nce.\dk,\nce.\aux)\gets \NCE.\keygen(1^\secp)$.
    \item Output $\MSK \seteq (\sel.\msk,\nce.\ek,\nce.\dk,\nce.\aux)$.
\end{enumerate}

\item[$\oneadaFEps.\keygen(\MSK, f)$:]$ $
\begin{enumerate}
    \item Parse $\MSK =(\sel.\MSK,\nce.\ek,\nce.\dk,\nce.\aux)$.
    \item Generate $\sel.\sk_f \gets \oneselFEps.\keygen(\sel.\msk,f)$.
    \item Generate $\nce.\ct \gets \NCE.\Enc(\nce.\ek,\fe.\sk_f)$.
    \item Output $\SK_f \seteq \nce.\ct$.
\end{enumerate}

\item[$\oneadaFEps.\enc(\MSK, m)$:] $ $
\begin{enumerate}
    \item Parse $\MSK =(\sel.\msk,\nce.\ek,\nce.\dk,\nce.\aux)$.
    \item Generate $\sel.\ct \gets \oneselFEps.\Enc(\sel.\msk,m)$.
    \item Output $\CT \seteq (\sel.\ct,\nce.\dk)$.
\end{enumerate}

\item[$\oneadaFEps.\dec(\SK_f, \CT,\pub)$:] $ $
\begin{enumerate}
\item Parse $\SK_f= \nce.\ct$ and $\CT = (\sel.\ct, \nce.\dk)$.
    \item Compute $\sk_f^\prime \gets \NCE.\Dec(\nce.\dk,\nce.\ct)$.
    \item Compute and output $y\seteq \oneselFEps.\Dec(\sk_f^\prime,\sel.\ct,\pub)$.
\end{enumerate}
\end{description}

\paragraph{Correctness.} It is easy to see correctness holds due to correctness of $\oneselFEps$ and $\NCE$.

\begin{theorem}\label{thm:adaptive_single-key-ct_SKFE_public_slot}
If $\oneselFEps$ is selectively single-key single-ciphertext secure public-slot SKFE for all polynomial-size circuits and $\NCE$ is RNC secure, $\oneadaFEps$ is adaptively single-key single-ciphertext public-slot SKFE for all polynomial-size circuits.
\end{theorem}

\begin{proof}
Let $\hyb_0(b)$ denote $\expc{\oneadaFEps,\qA}{ada}{1key}{1ct}(\secp,b)$.
We define a hybrid game $\hyb_1(b)$ as follows.
\begin{description}
\item[$\hyb_1(b)$:] This is the same as $\expc{\oneadaFEps,\qA}{ada}{1key}{1ct}(\secp,b)$ except that:
\begin{enumerate}
 \item if $\qA$ sends a key query $f$ before an encryption query $(m_0,m_1)$, we generate $\nce.\tlct \gets \NCE.\Fake(\nce.\ek,\nce.\aux)$ instead of $\nce.\ct \gets \NCE.\Enc(\nce.\ek,\sel.\sk_f)$ and return $\SK_f \seteq \nce.\tlct$ for the key query.
 \item when $\qA$ sends an encryption query $(m_0,m_1)$ after the key query $f$ above, we generate $\sel.\ct\gets \oneselFEps.\Enc(\sel.\msk,m_b)$, $\sel.\sk_f \gets \oneselFEps.\keygen(\sel.\msk,f)$, and $\nce.\tldk \gets \NCE.\Reveal(\nce.\ek,\allowbreak \nce.\aux,\nce.\tlct,\sel.\sk_f)$, and return $\CT \seteq (\sel.\ct,\nce.\tldk)$ for the encryption query.
 \end{enumerate} 
\end{description}

First, we show the following.
\begin{proposition}\label{prop:hyb_one_ada_single_ct_public_slot}
It holds $\abs{\Pr[\hyb_0(b)=1] -\Pr[\hyb_1(b)=1]}\le \negl(\secp)$ if $\NCE$ is RNC secure.
\end{proposition}
We construct an algorithm $\qB$ that breaks RNC security of $\NCE$ by using an adversary $\qA$ that breaks adaptive single-key single-ciphertext security of $\oneadaFEps$. Note that if $\qA$ sends an encryption query $(m_0,m_1)$ before a key query $f$, these two games are the same. We focus on the case where $\qA$ sends a key query $f$ before an encryption query $(m_0,m_1)$. $\qB$ does the following.
\begin{enumerate}
\item First, $\qB$ generates $\sel.\msk \gets \oneselFEps.\Setup(1^\secp)$.
\item When $\qA$ sends a key query $f$, $\qB$ generates $\sel.\sk_f \gets \oneselFEps.\keygen(\sel.\msk,f)$, sends $\sel.\sk_f$ to its challenger, and receives $(\nce.\ct^\ast,\nce.\dk^\ast)$. $\qB$ passes $\SK_f \seteq \nce.\ct^\ast$ to $\qA$. 
\item After the key query $f$ above, for an encryption query $(m_0,m_1)$, $\qB$ generates $\sel.\ct\gets \oneselFEps.\Enc(\sel.\msk,m_b)$ and returns $\CT \seteq (\sel.\ct,\nce.\dk^\ast)$ to $\qA$.
\item When $\qA$ outputs $b^\prime$, $\qB$ outputs $b^\prime$.
\end{enumerate}
It is easy to see the following.
\begin{itemize}
\item If $\nce.\ct^\ast \gets \NCE.\Enc(\nce.\ek,\sel.\sk_f)$ and $\nce.\dk^\ast = \nce.\dk$ where $(\nce.\ek.\nce.\dk,\nce.\aux)\gets \NCE.\keygen(1^\secp)$, $\qB$ perfectly simulates $\hyb_0(b)$.
\item If $\nce.\ct^\ast \seteq \nce.\tlct \gets \NCE.\Fake(\nce.\ek,\nce.\aux)$ and $\nce.\dk^\ast \gets \NCE.\Reveal(\nce.\ek,\nce.\aux,\nce.\tlct,\sel.\sk_f)$, $\qB$ perfectly simulates $\hyb_1(b)$.
\end{itemize}
Thus, if $\qA$ distinguishes $\hyb_0(b)$ from $\hyb_1(b)$, $\qB$ distinguishes $\expc{\NCE.\qB}{sk}{rec}{nc}(\secp,b)$.

Next, we show the following.
\begin{proposition}\label{prop:prop:hyb_two_ada_single_ct_public_slot}
It holds $\abs{\Pr[\hyb_1(0)=1] -\Pr[\hyb_1(1)=1]}\le \negl(\secp)$ if $\oneselFEps$ is selectively single-key single-ciphertext secure.
 \end{proposition} 
We construct an algorithm $\qB$ that breaks selective single-key single-ciphertext security of $\oneselFEps$ by using an adversary $\qA$ that breaks adaptive single-key single-ciphertext security of $\oneadaFEps$.
$\qB$ does the following.
\begin{enumerate}
\item First, $\qB$ generates $(\nce.\ek,\nce.\dk,\nce.\aux)\gets \NCE.\keygen(1^\secp)$.
\item There are the following two cases:
\begin{itemize}
\item When $\qA$ sends a key query $f$ before an encryption query $(m_0,m_1)$, $\qB$ generates $\nce.\tlct \gets \NCE.\Fake(\nce.\ek,\allowbreak \nce.\aux)$ passes $\SK_f \seteq \nce.\tlct$ to $\qA$. After the key query $f$ above, for an encryption query $(m_0,m_1)$, $\qB$ passes $(m_0,m_1)$ to its challenger and receives $\sel.\ct^\ast$. Then, $\qB$ sends $f$ to its challenger and receives $\sel.\sk_f \gets \oneselFEps(\sel.\msk,f)$. Finally, $\qB$ generates $\nce.\tldk \gets \NCE.\Reveal(\nce.\ek,\nce.\dk,\nce.\tlct,\sel.\sk_f)$ and sends $\CT\seteq (\sel.\ct^\ast,\nce.\tldk)$ to $\qA$.
\item When $\qA$ sends an encryption query $(m_0,m_1)$ before a key query $f$, $\qB$ passes $(m_0,m_1)$ to its challenger, receives $\sel.\ct^\ast$, and returns $\CT\seteq (\sel.\ct^\ast,\nce.\dk)$ to $\qA$. After the encryption query $(m_0,m_1)$ above, for a key query $f$, $\qB$ passes $f$ to its challenger and receives $\sel.\sk_f \gets \oneselFEps.\keygen(\sel.\msk,f)$. $\qB$ returns $\SK_f \seteq \nce.\ct \gets \NCE.\Enc(\nce.\ek,\sel.\sk_f)$ to $\qA$.
\end{itemize}
\item When $\qA$ outputs $b^\prime$, $\qB$ outputs $b^\prime$.
\end{enumerate}
Note that $\qB$ sends $(m_0,m_1)$ to its challenger before it sends $f$ as a key query in the both cases above.
If $\qA$ is a valid adversary in the experiment of adaptive single-key single-ciphertext security for public-lot SKFE $\oneadaFEps$, it holds $f(m_0,y^\prime)=f(m_1,y^\prime)$ for all $y^\prime\in\Ms$.
Then, $\qB$ is also a valid adversary in the experiment of selective single-key single-ciphertext security for public-slot SKFE $\oneselFEps$.
In addition, it is easy to see the following.
\begin{itemize}
\item If $\sel.\ct^\ast \gets \oneselFEps.\Enc(\sel.\msk,m_0)$, $\qB$ perfectly simulates $\hyb_1(0)$.
\item If $\sel.\ct^\ast \gets \oneselFEps.\Enc(\sel.\msk,m_1)$, $\qB$ perfectly simulates $\hyb_1(1)$.
\end{itemize}
Thus, if $\qA$ distinguishes $\hyb_1(0)$ from $\hyb_1(1)$, $\qB$ distinguishes $\expc{\oneselFEps,\qB}{sel}{1key}{1ct}(\secp,0)$ from $\expc{\oneselFEps,\qB}{sel}{1key}{1ct}(\secp,1)$.

Therefore, we obtain $\abs{\Pr[\hyb_0(0)=1]-\Pr[\hyb_0(1)=1]}\le \negl(\secp)$, which is our goal.
\end{proof}

\subsection{Adaptively Secure Public-Slot PKFE Scheme}\label{sec:adaptive_PKFE_pub_slot}

Note that the construction in this section is bassically the same as that by Ananth and Sahai~\cite{TCC:AnaSah16} except that we use single-key single-ciphertext public-slot SKFE as a building block istead of single-key single-ciphertext standard SKFE.
\paragraph{Ingredients.}
\begin{itemize}
	\item Selectively secure PKFE $\PKFE = \PKFE.(\setup,\keygen, \enc, \dec)$ for all polynonial-size circuits. 
	\item Selectively single-key function private SKFE $\okskfe = \okskfe.(\setup, \keygen, \enc, \dec)$ for polynomial-size circuits.
	\item  Adaptively single-key single-ciphertext public-slot SKFE $\okcskfe =  \okcskfe.(\setup, \keygen, \enc, \dec)$ for polynomial-size circuits.
	\item A PRF $\PRF: \mcl{K} \times \{0, 1\}^{\secp} \rightarrow \{0, 1\}^{\secp}$.
	\item SKE with pseudorandom ciphertext $\SKE = \SKE.(\setup, \enc, \dec)$.
\end{itemize}

\paragraph{Scheme description.} The adaptively secure public-slot PKFE scheme $\FE = \FE.(\setup,\keygen,\enc,\dec)$ is as follows.
\begin{description}
\item[$\setup(1^{\secp})$:]$\vspace{0.01cm}$
\begin{enumerate}
	\item Generate $(\pkfe.\MPK,\pkfe.\MSK) \leftarrow \PKFE.\setup(1^{\secp})$.
	\item Output $\MPK \seteq \pkfe.\MPK$ and $\MSK \seteq \pkfe.\MSK$.
\end{enumerate}

\item[$\keygen(\msk, f)$:]$\vspace{0.01cm}$
\begin{enumerate}
	\item Parse $\MSK =\pkfe.\MSK$.
	\item Sample $C_{\ske} \leftarrow \{0, 1\}^{\ell_{\ske}(\secp)}$ where $\ell_{\ske}(\secp)$ is the length of a SKE ciphertext that encrypts a string of length $\ell_{\okcskfes}(\secp) + \ell_{\okskfes}(\secp)$. We denote $\ell_{\okcskfes}(\secp)$ by the length of a SKFE secret key and $\ell_{\okskfes}(\secp)$ by the length of a $\okskfe$ ciphertext.
	\item Sample $\tau \leftarrow \{0, 1\}^{4\secp}$
	\item Generate $\pkfe.\sk_{g[f, C_{\ske}, \tau]}\gets \PKFE.\keygen(\pkfe.\MSK,g[f, C_{\ske}, \tau])$ where $g[f, C_{\ske}, \tau]$ is a function described in \cref{fig:gf}. 
	\item Output $\sk_f=\pkfe.\sk_{g[f, C_{\ske}, \tau]}$.
\end{enumerate}
	
\begin{figure}
	\begin{framed}
		\begin{center}
			\underline{$g[f, C_{\ske}, \tau]$}
		\end{center}
		\textbf{Input:}~ $\okskfes.\MSK, \tsf{K}, \ske.\SK, \beta$
		\begin{enumerate}
			\item Parse $\tau = (\tau_0\concat \tau_1\concat \tau_2\concat \tau_3)$. 
			\item If $\beta = 0$ then
			\begin{itemize}
				\item Compute $R_i \leftarrow \PRF(\tsf{K}, \tau_i)$ for $i \in \{0, 1, 2, 3\}$.
				\item Generate $\okcskfes.\MSK \leftarrow \okcskfe.\setup(1^{\secp}; R_0)$.
				\item Compute $\okcskfes.\sk_f \leftarrow \okcskfe.\KeyGen(\okcskfes.\MSK, f; R_1)$.
				\item Compute $\okskfes.\ct \leftarrow \okskfe.\Enc(\okskfes.\MSK, (\okcskfes.\MSK, R_2, 0); R_3)$.
				\item Output $(\okcskfes.\sk_f, \okskfes.\ct)$.
			\end{itemize}	
            \item Else,
            \begin{itemize}		
            \item Compute $(\okcskfes.\sk_f, \okskfes.\ct) \leftarrow \SKE.\Dec(\ske.\SK, C_{\ske})$.
			\item Output $(\okcskfes.\sk_f, \okskfes.\ct)$.
			\end{itemize}
		\end{enumerate}
	\end{framed}
	\caption{The description of the function $g[f, C_E, \tau]$}
	\label{fig:gf}
\end{figure}

\item[$\enc(\MPK, m)$:]$\vspace{0.01cm}$
\begin{enumerate}
	\item Parse $\MPK = \pkfe.\MPK$.
	\item Sample $\tsf{K} \leftarrow \mcl{K}$.
	\item Generate $\okskfes.\MSK \leftarrow \okskfe.\Setup(1^{\secp})$.
	\item Generate $\okskfes.\sk_{h[m]} \gets \okskfe.\keygen(\okskfes.\MSK,h[m])$ where $h[m]$ is a function described in \cref{fig:hm}. 
	\item Compute $\pkfe.\ct \leftarrow \PKFE.\Enc(\pkfe.\MPK, (\okskfes.\MSK, \tsf{K}, \bot, 0))$.
	\item Output $\ct = (\okskfes.\sk_{h[m]} ,\pkfe.\ct)$.
\end{enumerate}

\begin{figure}
	\begin{framed}
		\begin{center}
			\underline{$h[m]$}
		\end{center}
		\textbf{Input:}~ $\okcskfes.\MSK, R, \alpha$
		\begin{enumerate}
			\item If $\alpha = 0$ then
			\begin{itemize}
				\item Compute $\okcskfes.\ct \leftarrow \okcskfe.\Enc(\okcskfes.\MSK, m; R)$.
				\item Output $\okcskfes.\ct$.
			\end{itemize}	
			\item Else, output $\bot$.
		\end{enumerate}
	\end{framed}
	\caption{The description of the function $h[m]$}
	\label{fig:hm}
\end{figure}

\item[$\dec(\sk_f, \ct, y)$:]$\vspace{0.01cm}$
\begin{enumerate}
	\item Parse $\sk_f = \pkfe.\sk_{g[f]}$ and $\ct =(\okskfes.\sk_{h[m]} ,\pkfe.\ct)$.
	\item Compute $(\okcskfes.\sk_f, \okskfes.\ct) \leftarrow \PKFE.\dec(\pkfe.\sk_{g[f]}, \pkfe.\ct)$.
	\item Compute $\okcskfes.\ct \leftarrow \okskfe.\Dec( \okskfes.\sk_{h[m]}, \okskfes.\ct)$.
	\item Compute $m' \leftarrow \okcskfe.\Dec(\okcskfes.\sk_f, \okcskfes.\ct , y)$
	\item Output $m'$.
\end{enumerate}		
\end{description}

The security proof is almost the same as that of Ananth and Sahai~\cite{TCC:AnaSah16}. We provide the proof for confirmation since we use adaptively single-key single-ciphertext public-slot SKFE.

\begin{theorem}\label{thm:AS_adaptive_FE_public_slot}
If $\PKFE$ is selectively secure PKFE for $\Ppoly$, $\okskfe$ is selectively single-key function private SKFE for $\Ppoly$, $\okcskfe$ is adaptively single-key single-ciphertext public-slot SKFE for $\Ppoly$, $\PRF$ is a secure PRF, and $\SKE$ is ciphertext pseudorandom, $FE$ is adaptively indistinguishable-secure public-slot PKFE for $\Ppoly$.
\end{theorem}

We immediately obtain~\cref{thm:adaptive_PKFE_pub_slot_from_IO} from the thereom above.

\paragraph{Correctness.} Let $\ct =(\okskfes.\sk_{h[m]} ,\pkfe.\ct)$ be an honestly generated ciphertext encrypting a message $m$ and $\sk_f = \pkfe.\sk_{g[f]}$ be an honestly generated secret key corresponding to a function $f$. Firstly, we note that $\pkfe.\ct$ is an encryption of the message $(\okskfes.\MSK, \tsf{K}, \bot, 0)$ and $g[f]$ is a function that takes $(\okskfes.\MSK, \tsf{K}, \bot, 0)$ as input and outputs a SKFE secret key $\okcskfes.\sk_f$ corresponding to the function $f$  and a single key SKFE ciphertext $\okskfes.\ct$. Therefore, by the correctness of PKFE, the decryption algorithm $\PKFE.\dec(\pkfe.\sk_{g[f]}, \pkfe.\ct)$ yields $g[f](\okskfes.\MSK, \tsf{K}, \bot, 0) = (\okcskfes.\sk_f, \okskfes.\ct)$. Secondly, we observe that $\okskfes.\ct$ encrypts a message $(\okcskfe.\MSK, R_2, 0)$ and $h[m]$ is a function that takes $(\okcskfe.\MSK, R_2, 0)$ an input and outputs a SKFE ciphertext $\okcskfes.\ct$. Therefore, by the correctness of SKFE, the decryption algorithm $\okskfe.\dec(\okskfes.\sk_{h[m]}, \okskfes.\ct)$ yields $h[m](\okcskfe.\MSK, R_2, 0) = \okcskfes.\ct$. Finally, we note that $\okcskfes.\ct$ encrypts the message $m$ and $f$ is a function that takes $(m, y)$ as input and outputs $f(m, y)$ where $y$ is an input to the public slot. Thus, by the correctness of public-slot SKFE, we obtain $\okcskfe.\Dec(\okcskfes.\sk_f, \okcskfes.\ct , y) = m' = f(m, y)$.

\paragraph{Adaptive Security.} 
We prove~\cref{thm:AS_adaptive_FE_public_slot}.
\begin{proof}[Proof of~\cref{thm:AS_adaptive_FE_public_slot}]
Let $\mcl{A}$ be a PPT adversary against the adaptive security of the public-slot PKFE. We use the following sequence of hybrids to prove the security. Let $\Pr[\hyb_{i}=1]$ be the winning probability of $\mcl{A}$ in $\hyb_i$ for all $i$. 

\begin{description}
	\item[$\hyb_0$:] This is the original adaptive security experiment where the challenge bit set to 0. Specifically, it works as follows:
	\begin{enumerate}
		\item The challenger generates $(\pkfe.\MPK,\pkfe.\MSK) \leftarrow \PKFE.\setup(1^{\secp})$, sets $\MPK \seteq \pkfe.\MPK$ and $\MSK \seteq \pkfe.\MSK$, and sends $\MPK$ to $\mcl{A}$.
		\item The challenger samples a PRF key $\tsf{K}^* \leftarrow \mcl{K}$ and generate a master secret key $\okskfes.\MSK^* \leftarrow \okskfe.\Setup(1^{\secp})$.
		\item The challenger computes $\pkfe.\ct^* \leftarrow \PKFE.\Enc(\pkfe.\MPK, (\okskfes.\MSK^*, \tsf{K}^*, \bot, 0))$.
		\item $\mcl{A}$ can make arbitrarily many key queries at any point of the experiment. When it makes the $j$-th key query for a function $f_j$, the challenger works as follows:
		\begin{enumerate}
			\item Sample $C_{j, \ske}, \tau_j = (\tau_{j, 0} \concat \tau_{j, 1} \concat \tau_{j, 2} \concat \tau_{j, 3})$ uniformly at random.
			\item Generate $\pkfe.\sk_{g[f_j, C_{j, \ske}, \tau_j]}\gets \PKFE.\keygen(\pkfe.\MSK,g[f_j, C_{j, \ske}, \tau_j])$ where $g[f_j, C_{j, \ske}, \tau_j]$ is a function described in \cref{fig:gf}.
			\item Set $\sk_{f_j}\seteq \pkfe.\sk_{g[f_j,C_{j, \ske}, \tau_j]}$.
		\end{enumerate}	
	The challenger sends $\sk_{f_j}$ to $\mcl{A}$.
	\item $\mcl{A}$ sends $(m_0, m_1)$ to the challenger. It must satisfy $f(m_0, y) = f(m_1, y)$ for any public input $y$ and for all key queries $f$ that are made before or after sending $(m_0, m_1)$.
	\item The challenger computes the ciphertext as follows:
	\begin{enumerate}
		\item Generate $\okskfes.\sk_{h[m_0]}^* \gets \okskfe.\keygen(\okskfes.\MSK^*,h[m_0])$ where $h[m_0]$ is a function described in \cref{fig:hm}. 
		\item Set $\ct^* \seteq (\okskfes.\sk_{h[m_0]}^* ,\pkfe.\ct^*)$ where $\pkfe.\ct^*$ is computed in Step 3.
	\end{enumerate}	
	The challenger sends $\ct^*$ to $\mcl{A}$.
	\item $\mcl{A}$ outputs a bit $b'$ which is the final output of the experiment.
	\end{enumerate}	
Note that, the challenger can sample a PRF key $K^*$, a master secret key for the single key function-private SKFE $\okskfes.\MSK^*$ 	before it answers any secret key query. Moreover, the challenger can also compute the part of the challenge ciphertext $\pkfe.\ct^*$ before the key query phase.
	
\item[$\hyb_1$:] This hybrid is identical to $\hyb_0$ except the challenger samples a SKE key $\ske.\SK^*$ before it answers any key query and sets $C_{j, \ske}$ to be the ciphertext of SKE which corresponds to the challenge ciphertext. More specifically, the challenger answers to the $j$-th key query for a function $f_j$ as follows:    	
\begin{enumerate} 
	\item[(a)] Sample $ \tau_j = (\tau_{j, 0} \concat \tau_{j, 1} \concat \tau_{j, 2} \concat \tau_{j, 3})$ uniformly at random.
	\item[(b)] Compute $R_{j, i} = \PRF(\tsf{K}^*, \tau_{j, i})$ for all $i \in \{0, 1, 2, 3\}$. 
	\item[(c)] Generate $\okcskfes.\MSK_j \leftarrow \okcskfe.\setup(1^{\secp}; R_{j, 0})$.
	\item[(d)] Compute $\okcskfes.\sk_{f_j} \leftarrow \okcskfe.\KeyGen(\okcskfes.\MSK_j, f_j; R_{j, 1})$.
	\item[(e)]  Compute $\okskfes.\ct_j \leftarrow \okskfe.\Enc(\okskfes.\MSK^*, (\okcskfes.\MSK_j, R_{j, 2}, 0); R_{j, 3})$.
	\item[(f)] Compute $C_{j, \ske} \leftarrow \SKE.\Enc(\ske.\SK^*, u_j)$ where $u_j = (\okcskfes.\sk_{f_j}, \okskfes.\ct_j)$.
	\item[(g)] Generate $\pkfe.\sk_{g[f_j, C_{j, \ske}, \tau_j]}\gets \PKFE.\keygen(\pkfe.\MSK,g[f_j, C_{j, \ske}, \tau_j])$ where $g[f_j, C_{j, \ske}, \tau_j]$ is a function described in \cref{fig:gf}.
	\item[(h)] Set $\sk_f\seteq \pkfe.\sk_{g[f_j,C_{j, \ske}, \tau_j]}$.
\end{enumerate}		
The challenger sends $\sk_{f_j}$ to $\mcl{A}$. The indistinguishability between $\hyb_0$ and $\hyb_1$ follows from the security of $\SKE$ since the view of the adversary $\mcl{A}$ can be simulated without the knowledge of $\ske.\SK^*$ and using the challenger of the security experiment of SKE. In particular, consider $\mcl{B}_1$ to be an adversary against the security of SKE. When $\mcl{A}$ queries a secret key for a function $f_j$, $\mcl{B}_1$ proceeds as in Step (a) to (f) and sends the message $u_j = (\okcskfes.\sk_{f_j}, \okskfes.\ct_j)$ to it's challenger. Upon receiving a ciphertext $C_{j, \ske}$ from the challenger, $\mcl{B}_1$ computes $\pkfe.\sk_{g[f_j, C_{j, \ske}, \tau_j]}$ and sends it to $\mcl{A}$. If $\mcl{B}_1$ receives a random string then it simulates $\hyb_0$, otherwise, if $\mcl{B}_1$ is sent an  encryption of $u_j$ then it simulates $\hyb_1$. Therefore, the wining probability of $\mcl{B}_1$ is the same as  $\left|\Pr[\hyb_{0}=1]-\Pr[\hyb_{1}=1]\right|\le \negl$. Hence, by the security of SKE, it holds that   
	\begin{align}
	\left|\Pr[\hyb_{0}=1]-\Pr[\hyb_{1}=1]\right|\le \negl(\secp).
\end{align}

\item[$\hyb_2$:] This hybrid is identical to $\hyb_1$ except the challenger computes 
$$\pkfe.\ct^* \gets \PKFE.\Enc(\pkfe.\MPK, (\bot, \bot, \ske.\SK^*, 1)) $$	
instead of computing $\pkfe.\ct^* \leftarrow \PKFE.\Enc(\pkfe.\MPK, (\okskfes.\MSK^*, \tsf{K}^*, \bot, 0))$. In particular, the mode of decryption is changed to $\beta = 1$ from $\beta = 0$ meaning that $\ske.\SK^*$ is used to decrypt $C_{j, \ske}$ to get an output of $(\okcskfes.\sk_{f_j}, \okskfes.\ct_j)$ while decryption of $\ct^*$ is performed by the $j$-th secret key $\sk_f\seteq \pkfe.\sk_{g[f_j,C_{j, \ske}, \tau_j]}$. The indistinguishability between $\hyb_1$ and $\hyb_2$ follows from the security of $\PKFE$ since the view of the adversary $\mcl{A}$ can be simulated without the knowledge of $\pkfe.\MSK$ and using the challenger of the security experiment of PKFE. Let us consider an adversary $\mcl{B}_2$ against the security of PKFE. Firstly, $\mcl{B}_2$ sends a pair of challenge message $((\okskfes.\MSK^*, \tsf{K}^*, \bot, 0), (\bot, \bot, \ske.\SK^*, 1))$ to it's challenger and receives the public key $\pkfe.\MPK$ and a ciphertext $\pkfe.\ct^*$. Note that $\mcl{B}_2$ can choose the challenge message independent of all the queries of $\mcl{A}$. Therefore, a selectively secure PKFE is sufficient for arguing the indistinguishability between the hybrids. Whenever $\mcl{B}_2$ receives a secret key query from $\mcl{A}$ for a function $f_j$, it queries for a secret key to it's challenger for a function $g[f_j, C_{j, \ske}, \tau_j]$ and returns the output to $\mcl{A}$. Firstly, $\mcl{B}_2$ is an admissible adversary as 
$$
g[f_j, C_{j, \ske}, \tau_j](\okskfes.\MSK^*, \tsf{K}^*, \bot, 0) = g[f_j, C_{j, \ske}, \tau_j](\bot, \bot, \ske.\SK^*, 1)
$$
holds for all $j$. If $\mcl{B}_2$ receives an encryption of $(\okskfes.\MSK^*, \tsf{K}^*, \bot, 0)$ then it simulates $\hyb_1$, otherwise, if $\mcl{B}_2$ receives an encryption of $(\bot, \bot, \ske.\SK^*, 1)$ then it simulates $\hyb_2$. Therefore, the winning probability of $\mcl{B}_2$ is essentially the same as $\left|\Pr[\hyb_{1}=1]-\Pr[\hyb_{2}=1]\right|\le \negl$. Hence, by the security of PKFE, it holds that   
\begin{align}
	\left|\Pr[\hyb_{1}=1]-\Pr[\hyb_{2}=1]\right|\le \negl(\secp).
\end{align}

\item[$\hyb_3$:] This hybrid is identical to $\hyb_2$ except the challenger samples $R_{j, i}$ uniformly at random for all $j, i$, while answering the secret key queries instead of computing these values using the PRF key $\tsf{K}^*$. The indistinguishability between $\hyb_2$ and $\hyb_3$ follows from the security of $\PRF$ since the view of the adversary $\mcl{A}$ can be simulated without the knowledge of $\tsf{K}^*$ and using the challenger of the security experiment of PRF. In other words, if $\mcl{B}_3$ is an adversary against the security of PRF then the winning probability of $\mcl{B}_3$ is the same as $\left|\Pr[\hyb_{2}=1]-\Pr[\hyb_{3}=1]\right|\le \negl$. Hence, by the security of PRF, it holds that   
\begin{align}
	\left|\Pr[\hyb_{2}=1]-\Pr[\hyb_{3}=1]\right|\le \negl(\secp).
\end{align}   

\item[$\hyb_4$:]  This hybrid is identical to $\hyb_3$ except the challenger generates
$$\okskfes.\sk_{h[m_0,m_1, v]}^* \gets \okskfe.\keygen(\okskfes.\MSK^*,h[m_0, m_1, v])$$ 
and sets the challenge ciphertext as $\ct^* \seteq (\okskfes.\sk_{h[m_0,m_1, v]}^*, \pkfe.\ct^*)$ where $h[m_0, m_1, v]$ is a function described in \cref{fig:hmv} and $v$ is a random string. The indistinguishability between $\hyb_3$ and $\hyb_4$ follows from the security of $\okskfe$ since the view of the adversary $\mcl{A}$ can be simulated without the knowledge of $\okskfes.\MSK^*$ and using the challenger of the security experiment of $\okskfe$. Let us consider an adversary $\mcl{B}_4$ against the security of $\okskfe$. We assume that $Q$ be the total number of secret key queries the adversary $\mcl{A}$ makes in the experiment. At first, $\mcl{B}_4$ prepares a list of $Q$ challenge messages $(M_1, \ldots, M_Q)$ where $M_j = (\okcskfes.\MSK_j, R_{j, 2}, 0)$ for $j \in [Q]$. More precisely, $\mcl{B}$ sends $(M_j, M_j)$ for all $j \in [Q]$ and receives the ciphertexts as $\{\okskfes.\ct_j\}_{j \in [Q]}$ which is used in answering $\mcl{A}$'s secret key queries as in $\hyb_3$ or $\hyb_4$. When $\mcl{A}$ sends the challenge message tuple $(m_0, m_1)$, $\mcl{B}_4$ queries for a secret key with the pair of functions $(h[m_0], h[m_0, m_1, v])$ to it's challenger. Then, $\mcl{B}_4$ uses the output from it's challenger to create the challenge ciphertext for $\mcl{A}$. It is easy to see that $\mcl{B}_4$ is an admissible adversary for the security experiment of $\okskfe$ since 
$$
h[m_0](\okcskfes.\MSK_j, R_{j, 2}, 0) = h[m_0, m_1, v](\okcskfes.\MSK_j, R_{j, 2}, 0)
$$
holds for all $j \in [Q]$. If $\mcl{B}_4$ receives a secret key $\okskfes.\sk_{h[m_0]}^*$ then it simulates $\hyb_3$, otherwise, if $\mcl{B}_4$ receives a secret key $\okskfes.\sk_{h[m_0,m_1, v]}^*$ then it simulates $\hyb_4$. Therefore, the winning probability of $\mcl{B}_4$ is essentially the same as $\left|\Pr[\hyb_{3}=1]-\Pr[\hyb_{4}=1]\right|\le \negl$. Hence, by the security of $\okskfe$, it holds that   
\begin{align}
	\left|\Pr[\hyb_{3}=1]-\Pr[\hyb_{4}=1]\right|\le \negl(\secp).
\end{align}

\begin{figure}
	\begin{framed}
		\begin{center}
			\underline{$h[m, m', v]$}
		\end{center}
		\textbf{Input:}~ $\okcskfes.\MSK, R, \alpha$
		\begin{enumerate}
			\item If $\alpha = 0$ then
			\begin{itemize}
				\item Compute $\okcskfes.\ct \leftarrow \okcskfe.\Enc(\okcskfes.\MSK, m; R)$.
				\item Output $\okcskfes.\ct$.
			\end{itemize}	
		    \item If $\alpha = 1$ then
		    \begin{itemize}
		    	\item Compute $\okcskfes.\ct \leftarrow \okcskfe.\Enc(\okcskfes.\MSK, m'; R)$.
		    	\item Output $\okcskfes.\ct$.
		    \end{itemize}	 
			\item Else, output $v$.
		\end{enumerate}
	\end{framed}
	\caption{The description of the function $h[m,m',v]$}
	\label{fig:hmv}
\end{figure}

\item[$\hyb_5$:]  This hybrid is identical to $\hyb_4$ except the challenger computes
$$\okskfes.\ct_j \leftarrow \okskfe.\Enc(\okskfes.\MSK^*, (\okcskfes.\MSK_j, R_{j, 2}, 1); R_{j, 3})$$	
while generating the $j$-th secret key corresponding to a function $f_j $ for all $j$. In Lemma \ref{lem:hybr_five}, we show that 
\begin{align}
	\left|\Pr[\hyb_{4}=1]-\Pr[\hyb_{5}=1]\right|\le \negl(\secp).
\end{align}  

\item[$\hyb_6$:]  This hybrid is identical to $\hyb_5$ except the challenger generates
$$\okskfes.\sk_{h[m_1,m_1, v]}^* \gets \okskfe.\keygen(\okskfes.\MSK^*,h[m_1, m_1, v])$$ 
and sets the challenge ciphertext as $\ct^* \seteq (\okskfes.\sk_{h[m_1,m_1, v]}^*, \pkfe.\ct^*)$ where $h[m_1, m_1, v]$ is a function as described in \cref{fig:hmv}. The indistinguishability between $\hyb_5$ and $\hyb_6$ follows from the security of $\okskfe$ since the view of the adversary $\mcl{A}$ can be simulated without the knowledge of $\okskfes.\MSK^*$ and using the challenger of the security experiment of $\okskfe$. The simulation strategy is similar to $\hyb_4$. By the security of $\okskfe$, it holds that   
\begin{align}
	\left|\Pr[\hyb_{5}=1]-\Pr[\hyb_{6}=1]\right|\le \negl(\secp).
\end{align}

\item[$\hyb_7$:]  This hybrid is identical to $\hyb_6$ except the challenger computes
$$\okskfes.\ct_j \leftarrow \okskfe.\Enc(\okskfes.\MSK^*, (\okcskfes.\MSK_j, R_{j, 2}, 0); R_{j, 3})$$	
while generating the $j$-th secret key corresponding to a function $f_j $ for all $j$. Moreover,  the challenger generates
$$\okskfes.\sk_{h[m_1]}^* \gets \okskfe.\keygen(\okskfes.\MSK^*,h[m_1])$$ 
and sets the challenge ciphertext as $\ct^* \seteq (\okskfes.\sk_{h[m_1]}^*, \pkfe.\ct^*)$ where $h[m_1]$ is a function as described in \cref{fig:hm}. The indistinguishability between $\hyb_6$ and $\hyb_7$ follows from the security of $\okskfe$ since the view of the adversary $\mcl{A}$ can be simulated without the knowledge of $\okskfes.\MSK^*$ and using the challenger of the security experiment of $\okskfe$. Let us consider an adversary $\mcl{B}_7$ against the security of $\okskfe$.

At first, $\mcl{B}_7$ prepares a list of $Q$ challenge message pairs $( (M_1^{(0)}, M_1^{(1)}), \ldots, (M_Q^{(0)}, M_Q^{(1)}))$ where 
$$M_j^{(0)} = (\okcskfes.\MSK_j, R_{j, 2}, 1)~~ \text{ and }~~M_j^{(1)} = (\okcskfes.\MSK_j, R_{j, 2}, 0)$$ 

for $j \in [Q]$. More precisely, $\mcl{B}_7$ sends $(M_j^{(0)}, M_j^{(1)})$ for all $j \in [Q]$ and receives the ciphertexts as $\{\okskfes.\ct_j\}_{j \in [Q]}$ which is used in answering $\mcl{A}$'s secret key queries. When $\mcl{A}$ sends the challenge message tuple $(m_0, m_1)$, $\mcl{B}_7$ queries for a secret key with the pair of functions $(h[m_1, m_1, v], h[ m_1])$ to it's challenger. Then, $\mcl{B}_7$ uses the output from it's challenger to create the challenge ciphertext for $\mcl{A}$. It is easy to see that $\mcl{B}_7$ is an admissible adversary for the security experiment of $\okskfe$ since 
$$
h[m_1, m_1, v](\okcskfes.\MSK_j, R_{j, 2}, 1) = h[m_1](\okcskfes.\MSK_j, R_{j, 2}, 0)
$$
holds for all $j \in [Q]$. If $\mcl{B}_7$ receives ciphertexts for the messages $M_j^{(0)}$ and a secret key $\okskfes.\sk_{h[m_1, m_1, v]}^*$ then it simulates $\hyb_6$, otherwise, if $\mcl{B}_7$ receives ciphertexts for the messages $M_j^{(1)}$  and a secret key $\okskfes.\sk_{h[m_1]}^*$ then it simulates $\hyb_7$. Therefore, the winning probability of $\mcl{B}_7$ is essentially the same as $\left|\Pr[\hyb_{6}=1]-\Pr[\hyb_{7}=1]\right|\le \negl$. Hence, by the security of $\okskfe$, it holds that   
\begin{align}
	\left|\Pr[\hyb_{6}=1]-\Pr[\hyb_{7}=1]\right|\le \negl(\secp).
\end{align}


\item[$\hyb_8$:] This hybrid is identical to $\hyb_7$ except the challenger samples 
$$R_{j, i} = \PRF(\tsf{K}^*, \tau_{j, i}) ~\text{ for all } i \in \{0, 1, 2, 3\}$$ 
instead of sampling these values uniformly at random for all $j, i$, while answering the secret key queries.	 The indistinguishability between $\hyb_7$ and $\hyb_8$ follows from the security of $\PRF$ since the view of the adversary $\mcl{A}$ can be simulated without the knowledge of $\tsf{K}^*$ and using the challenger of the security experiment of PRF. In other words, it holds that   
\begin{align}
	\left|\Pr[\hyb_{7}=1]-\Pr[\hyb_{8}=1]\right|\le \negl(\secp).
\end{align}

\item[$\hyb_{9}$:] This hybrid is identical to $\hyb_8$ except the challenger computes 
$$\pkfe.\ct^* \leftarrow \PKFE.\Enc(\pkfe.\MPK, (\okskfes.\MSK^*, \tsf{K}^*, \bot, 0)) $$	
instead of computing $\pkfe.\ct^* \gets \PKFE.\Enc(\pkfe.\MPK, (\bot, \bot, \ske.\SK^*, 1))$. The indistinguishability between $\hyb_8$ and $\hyb_9$ follows from the security of $\PKFE$ since the view of the adversary $\mcl{A}$ can be simulated without the knowledge of $\pkfe.\MSK$ and using the challenger of the security experiment of PKFE. In other words, it holds that   
\begin{align}
	\left|\Pr[\hyb_{8}=1]-\Pr[\hyb_{9}=1]\right|\le \negl(\secp).
\end{align}

\item[$\hyb_{10}$:] This hybrid is identical to $\hyb_{9}$ except the challenger does not sample $\ske.\SK^*$ and chooses $C_{j, \ske}$ uniformly at random while answering to the $j$-th secret key query of $\mcl{A}$ for all $j$.  More specifically, the challenger answers to the $j$-th key query for a function $f_j$ as follows:    	
\begin{enumerate}
	\item[(a)] Sample $C_{j, \ske}, \tau_j = (\tau_{j, 0} \concat \tau_{j, 1} \concat \tau_{j, 2} \concat \tau_{j, 3})$ uniformly at random.
	\item[(b)] Generate $\pkfe.\sk_{g[f_j, C_{j, \ske}, \tau_j]}\gets \PKFE.\keygen(\pkfe.\MSK,g[f_j, C_{j, \ske}, \tau_j])$ where $g[f_j, C_{j, \ske}, \tau_j]$ is a function described in \cref{fig:gf}.
	\item[(c)] Set $\sk_{f_j}\seteq \pkfe.\sk_{g[f_j,C_{j, \ske}, \tau_j]}$.
\end{enumerate}	
The challenger sends $\sk_{f_j}$ to $\mcl{A}$. The indistinguishability between $\hyb_9$ and $\hyb_{10}$ follows from the security of $\SKE$ since the view of the adversary $\mcl{A}$ can be simulated without the knowledge of $\ske.\SK$ and using the challenger of the security experiment of SKE. In other words, it holds that   
\begin{align}
	\left|\Pr[\hyb_{9}=1]-\Pr[\hyb_{10}=1]\right|\le \negl(\secp).
\end{align}   
We observe that $\hyb_{10}$ is identical to original adaptive security experiment where the challenge bit set to $1$. Combining the advantage of $\mcl{A}$ in all the consecutive hybrids and applying the triangular inequality, we have $\left|\Pr[\hyb_{0}=1]-\Pr[\hyb_{10}=1]\right|\le \negl(\secp)$.
\end{description}
 This completes the proof of~\cref{thm:AS_adaptive_FE_public_slot} if we prove~\cref{lem:hybr_five}.
\end{proof}

\begin{lemma}\label{lem:hybr_five}
	If $\okskfe$ is selectively single key function-private secure and public-slot SKFE $\okcskfe$ is adaptively single-key single-ciphertext secure then for any $\secp \in [\msglen]$, 
	\begin{align}
	\left|\Pr[\hyb_{4}=1]-\Pr[\hyb_{5}=1]\right|\le \negl(\secp).
	\end{align}
\end{lemma}
\begin{proof}[Proof of~\cref{lem:hybr_five}]
We prove this lemma using a sequence of hybrids $\hyb_{4,q,1}, \hyb_{4,q,2}, \hyb_{4,q,3}, \hyb_{4,q,4}$ for $q \in [Q]$ where $Q$ denotes the total number of secret key queried made by the adversary $\mcl{A}$. Let us denote $\hyb_{4, Q+1, 1} = \hyb_{5}$.
\begin{description}
	\item[$\hyb_{4,q, 1}:$] This is exactly the same as $\hyb_{4}$ except the challenger sets $v$ to be the output of $\okcskfe.\enc(\okcskfes.\MSK_q, m_0; R_{q, 2})$. More precisely, the hybrid works as follows:
	\begin{enumerate}
		\item The challenger generates $(\pkfe.\MPK,\pkfe.\MSK) \leftarrow \PKFE.\setup(1^{\secp})$, sets $\MPK \seteq \pkfe.\MPK$ and $\MSK \seteq \pkfe.\MSK$, and sends $\MPK$ to $\mcl{A}$.
		\item The challenger generates a master secret key $\okskfes.\MSK^* \leftarrow \okskfe.\Setup(1^{\secp})$.
		\item The challenger generates a secret key $\ske.\SK^* \leftarrow \SKE.\Setup(1^{\secp})$.
		\item The challenger computes $\pkfe.\ct^* \leftarrow \PKFE.\Enc(\pkfe.\MPK, (\bot, \bot, \ske.\SK^*, 1))$.
		\item The challenger sets $u_q$ as follows:
		\begin{enumerate}
		\item[(a)] Sample $\tau_q = (\tau_{q, 0} \concat \tau_{q, 1} \concat \tau_{q, 2} \concat \tau_{q, 3})$ and $R_{q, i}$ uniformly at random for all $i \in \{0, 1, 2, 3\}$. 
		\item[(b)] Generate $\okcskfes.\MSK_q \leftarrow \okcskfe.\setup(1^{\secp}; R_{q, 0})$.
		\item[(c)] Compute $\okcskfes.\sk_{f_q} \leftarrow \okcskfe.\KeyGen(\okcskfes.\MSK_q, f_q; R_{q, 1})$.
		\item[(d)]  Compute $\okskfes.\ct_q \leftarrow \okskfe.\Enc(\okskfes.\MSK^*, (\okcskfes.\MSK_q, R_{q, 2}, 0); R_{q, 3})$.
		\item[(e)] Set $u_q = (\okcskfes.\sk_{f_q}, \okskfes.\ct_q)$.
		\end{enumerate}
		\item $\mcl{A}$ can make arbitrarily many key queries at any point of the experiment. When it makes the $j$-th key query for a function $f_j$, the challenger works as follows:
		\begin{enumerate}
			\item[If $j\neq q:$]
			\item[(a)] Sample $\tau_j = (\tau_{j, 0} \concat \tau_{j, 1} \concat \tau_{j, 2} \concat \tau_{j, 3})$ and $R_{j, i}$ uniformly at random for all $i \in \{0, 1, 2, 3\}$. 
			\item[(b)] Generate $\okcskfes.\MSK_j \leftarrow \okcskfe.\setup(1^{\secp}; R_{j, 0})$.
			\item[(c)] Compute $\okcskfes.\sk_{f_j} \leftarrow \okcskfe.\KeyGen(\okcskfes.\MSK_j, f_j; R_{j, 1})$.
			\item[(d)]  Compute 
			$$\begin{array}{l l}
			\okskfes.\ct_j \leftarrow & \okskfe.\Enc(\okskfes.\MSK^*, (\okcskfes.\MSK_j, R_{j, 2}, 1); R_{j, 3}) ~~\text{ if }~j < q\\
			\okskfes.\ct_j \leftarrow & \okskfe.\Enc(\okskfes.\MSK^*, (\okcskfes.\MSK_j, R_{j, 2}, 0); R_{j, 3}) ~~\text{ if }~j > q\\
			\end{array}$$
			\item[(e)] Compute $C_{j, \ske} \leftarrow \SKE.\Enc(\ske.\SK^*, u_j)$ where $u_j = (\okcskfes.\sk_{f_j}, \okskfes.\ct_j)$.
			\item[(f)] Generate $\pkfe.\sk_{g[f_j, C_{j, \ske}, \tau_j]}\gets \PKFE.\keygen(\pkfe.\MSK,g[f_j, C_{j, \ske}, \tau_j])$ where $g[f_j, C_{j, \ske}, \tau_j]$ is a function described in \cref{fig:gf}.
			\item[(g)] Set $\sk_f\seteq \pkfe.\sk_{g[f_j,C_{j, \ske}, \tau_j]}$.
			\item[If $j = q:$]
			\item[(a)] Set $C_{q, \ske} \leftarrow \SKE.\enc(\ske.\SK^*, u_{q}).$
			\item[(b)] Generate $\pkfe.\sk_{g[f_q, C_{q, \ske}, \tau_q]}\gets \PKFE.\keygen(\pkfe.\MSK,g[f_q, C_{q, \ske}, \tau_q])$ where $g[f_q, C_{q, \ske}, \tau_q]$ is a function described in \cref{fig:gf}.
		\end{enumerate}	
		The challenger sends $\sk_{f_j}$ to $\mcl{A}$.
		\item $\mcl{A}$ sends $(m_0, m_1)$ to the challenger. It must satisfy $f(m_0, y) = f(m_1, y)$ for any public input $y$ and for all key queries $f$ that are made before or after sending $(m_0, m_1)$.
		\item The challenger computes the ciphertext as follows:
		\begin{enumerate}
			\item Set $v \seteq \okcskfe.\Enc(\okcskfes.\MSK_q, m_0; R_{q, 2})$
			\item Generate $\okskfes.\sk_{h[m_0,m_1, v]}^* \gets \okskfe.\keygen(\okskfes.\MSK^*,h[m_0,m_1, v])$ where $h[m_0, m_1, v]$ is a function described in \cref{fig:hmv}. 
			\item Set $\ct^* \seteq (\okskfes.\sk_{h[m_0,m_1, v]}^* ,\pkfe.\ct^*)$ where $\pkfe.\ct^*$ is computed in Step 3.
		\end{enumerate}	
		The challenger sends $\ct^*$ to $\mcl{A}$.
		\item $\mcl{A}$ outputs a bit $b'$ which is the final output of the experiment.
	\end{enumerate}	
	The indistinguishability between $\hyb_{4}$ and $\hyb_{4, 1, 1}$ follows from the security of $\okskfe$ since the view of the adversary $\mcl{A}$ can be simulated without the knowledge of $\okskfes.\MSK^*$ and using the challenger of the security experiment of $\okskfe$. Let us consider an adversary $\mcl{B}_{4, 1}$ against the security of $\okskfe$.\par 
	
	At first, $\mcl{B}_{4, 1}$ prepares a list of $Q$ challenge message pairs $(M_1, \ldots,  M_Q)$ where 
	$$\begin{array}{l l l}
		M_j = & (\okcskfes.\MSK_j, R_{j, 2}, 1) & \text{ if } 1\leq j<q\\
		M_j = & (\okcskfes.\MSK_j, R_{j, 2}, 0) & \text{ if } q\leq j \leq Q.\\
	\end{array}$$
	for $j \in [Q]$. More precisely, $\mcl{B}_{4, 1}$ sends $(M_j, M_j)$ for all $j \in [Q]$ and receives the ciphertexts as $\{\okskfes.\ct_j\}_{j \in [Q]}$ which is used in answering $\mcl{A}$'s secret key queries. When $\mcl{A}$ sends the challenge message tuple $(m_0, m_1)$, $\mcl{B}_{4, 1}$ queries for a secret key with the pair of functions $(h[m_0, m_1, v], h[m_0, m_1, v'])$ to it's challenger where $v$ is a random string of appropriate length and $v' = \okcskfe.\enc(\okcskfes.\MSK_q, m_0; R_{q, 2})$. Then, $\mcl{B}_{4, 1}$ uses the output from it's challenger to create the challenge ciphertext for $\mcl{A}$. It is easy to see that $\mcl{B}_{4, 1}$ is an admissible adversary for the security experiment of $\okskfe$ since $ h[m_0, m_1, v](M_j) = h[m_0, m_1, v'](M_j)$ holds for all $j \in [Q]$. This is because $ h[m_0, m_1, v](*, *, k) = h[m_0, m_1, v'](*, *, k)$ if $k \neq 2$. If $\mcl{B}_{4, 1}$ receives a secret key $\okskfes.\sk_{h[m_0, m_1, v]}^*$ then it simulates $\hyb_{4}$, otherwise, if $\mcl{B}_{4, 1}$ receives a secret key $\okskfes.\sk_{h[m_0, m_1, v']}^*$ then it simulates $\hyb_{4, 1, 1}$. Therefore, the winning probability of $\mcl{B}_{4, 1}$ is essentially the same as $\left|\Pr[\hyb_{4}=1]-\Pr[\hyb_{4, 1, 1}=1]\right|\le \negl$. Hence, by the security of $\okskfe$, it holds that   
	\begin{align}
		\left|\Pr[\hyb_{4}=1]-\Pr[\hyb_{4, 1, 1}=1]\right|\le \negl(\secp).
	\end{align}

	\item[$\hyb_{4,q, 2}:$] This is exactly the same as $\hyb_{4, q, 1}$ except the challenger changes the mode from $\alpha = 0$ to $\alpha  = 2$ while decrypting the challenge ciphertext using the $q$-th secret key. More precisely, the challenger computes $\okskfes.\ct_q$ as follows:
	$$
	\okskfes.\ct_q \leftarrow \okskfe.\Enc(\okskfes.\MSK^*, (0, 0, 2); R_{q, 3}).
	$$
	The indistinguishability between $\hyb_{4,q, 1}$ and $\hyb_{4, q, 2}$ follows from the security of $\okskfe$ since the view of the adversary $\mcl{A}$ can be simulated without the knowledge of $\okskfes.\MSK^*$ and using the challenger of the security experiment of $\okskfe$. This can be shown similarly as we discussed in the previous hybrid since 
	$$
	h[m_0, m_1, v](\okcskfes.\MSK_q, R_{q, 2}, 0) = h[m_0, m_1, v](0, 0, 2)
	$$
	holds where $v = \okcskfe.\enc(\okcskfes.\MSK_q, m_0; R_{q, 2})$. Hence, by the security of $\okskfe$, it holds that   
	\begin{align}
		\left|\Pr[\hyb_{4,q, 1}=1]-\Pr[\hyb_{4, q, 2}=1]\right|\le \negl(\secp).
	\end{align}

	\item[$\hyb_{4,q, 3}:$] This is exactly the same as $\hyb_{4, q, 2}$ except the challenger changes $v$ to be the encryption of $m_1$, that is, it sets $v$ as follows:
	$$
	v \seteq \okcskfe.\Enc(\okcskfes.\MSK_q, m_1; R_{q, 2})
	$$ 
	The indistinguishability between $\hyb_{4, q, 2}$ and $\hyb_{4, q, 3}$ follows from the security of $\SKFE$ since the view of the adversary $\mcl{A}$ can be simulated without the knowledge of $\okcskfes.\MSK_q$ and using the challenger of the security experiment of SKFE. Let us consider an adversary $\mcl{B}_{4, 3}$ against the security of adaptively single-key single-ciphertext secure public-slot SKFE. Note that, $\mcl{B}_{4, 3}$ queries only a single secret key $\okcskfes.\sk_{f_q}$ corresponding to the function $f_q$ and a single ciphertext $v$ corresponding to the challenge message pair $(m_0, m_1)$. In particular, $\mcl{B}_{4, 3}$ can adaptively query the secret key $\okcskfes.\sk_{f_q}$ at any point whenever $\mcl{A}$ asks for a secret key for $f_q$ and sets $u_q = (\okcskfes.\sk_{f_q}, \okskfes.\ct_q)$. We observe that $\mcl{B}_{4, 3}$ is an admissible adversary since $\mcl{A}$ is only allowed to query for a secret key for $f_q$ and challenge message pair $(m_0, m_1)$ such that $f_q(m_0, y) = f_q(m_1, y)$ holds for is an arbitrary input $y$ to the public slot of $f$. If $\mcl{B}_{4, 3}$ receives a ciphertext $v = \okcskfe.\Enc(\okcskfes.\MSK_q, m_0)$ then it simulates $\hyb_{4, q, 2}$, otherwise, if $\mcl{B}_{4, 3}$ receives a ciphertext $v = \okcskfe.\Enc(\okcskfes.\MSK_q, m_1)$ then it simulates $\hyb_{4, q, 3}$. Therefore, the winning probability of $\mcl{B}_{4, 3}$ is essentially the same as $\left|\Pr[\hyb_{4,q, 2}=1]-\Pr[\hyb_{4, q, 3}=1]\right|\le \negl$. Hence, by the security of SKFE, it holds that   
	\begin{align}
		\left|\Pr[\hyb_{4, q, 2}=1]-\Pr[\hyb_{4, q, 3}=1]\right|\le \negl(\secp).
	\end{align}

	\item[$\hyb_{4,q, 4}:$] This is exactly the same as $\hyb_{4, q, 3}$ except the challenger changes the mode from $\alpha = 2$ to $\alpha  = 1$ while decrypting the challenge ciphertext using the $q$-th secret key. More precisely, the challenger computes $\okskfes.\ct_q$ as follows:
	$$
	\okskfes.\ct_q \leftarrow \okskfe.\Enc(\okskfes.\MSK^*, (\okcskfes.\ct_q, R_{q, 2}; 1); R_{q, 3}).
	$$
	The indistinguishability between $\hyb_{4,q, 3}$ and $\hyb_{4, q, 4}$ follows from the security of $\okskfe$ since the view of the adversary $\mcl{A}$ can be simulated without the knowledge of $\okskfes.\MSK^*$ and using the challenger of the security experiment of $\okskfe$. This can be shown similarly as we discussed in the previous hybrid since 
	$$
	h[m_0, m_1, v](0, 0, 2) = h[m_0, m_1, v](\okcskfes.\MSK_q, R_{q, 2}, 1)
	$$
	holds where $v = \okcskfe.\enc(\okcskfes.\MSK_q, m_1; R_{q, 2})$. Hence, by the security of $\okskfe$, it holds that   
	\begin{align}
		\left|\Pr[\hyb_{4,q, 3}=1]-\Pr[\hyb_{4, q, 4}=1]\right|\le \negl(\secp).
	\end{align} 
	
\end{description}	
Combining the advantage of $\mcl{A}$ in all the consecutive hybrids and applying the triangular inequality, we have $\left|\Pr[\hyb_{4}=1]-\Pr[\hyb_{5}=1]\right|\le \negl(\secp)$.
 This completes the proof of~\cref{lem:hybr_five}.
 \end{proof}	

\section{Secret and Public Key Encryption with Certified Everlasting Deletion}\label{sec:cd_ske_pke}
In \cref{sec:def_cd_ske_pke}, we define SKE and PKE with certified everlasting deletion.
In \cref{sec:const_ske_rom} and \cref{sec:const_ske_wo_rom}, 
we construct a certified everlasting secure SKE scheme with and without QROM, respectively.
In \cref{sec:const_pke_rom} and \cref{sec:const_pke_wo_rom}, we construct a certified everlasting secure PKE scheme with and without QROM, respectively.

\subsection{Definition}\label{sec:def_cd_ske_pke}
\begin{definition}[SKE with Certified Everlasting Deletion (Syntax)]\label{def:cert_ever_ske}
Let $\lambda$ be a security parameter and let $p$, $q$, $r$ and $s$ be some polynomials.
An SKE with certified everlasting deletion scheme is a tuple of algorithms $\Sigma=(\keygen,\qEnc,\qDec,\qDelete,\Vrfy)$ with plaintext space $\Ms\seteq\{0,1\}^n$, ciphertext space $\Cs\seteq \cQ^{\otimes p(\lambda)}$,
secret key space $\mathcal{SK}\seteq\{0,1\}^{q(\lambda)}$,
verification key space $\mathcal{VK}\seteq\bit^{r(\secp)}$,
and deletion certificate space $\mathcal{D}\seteq \cQ^{\otimes s(\lambda)}$.
\begin{description}
    \item[$\keygen (1^\secp) \ra \sk$:]
    The key generation algorithm takes the security parameter $1^\secp$ as input and outputs a secret key $\sk \in \mathcal{SK}$.
    \item[$\qEnc(\sk,m) \ra (\vk,\qct)$:]
    The encryption algorithm takes $\sk$ and a plaintext $m\in\Ms$ as input, and outputs a verification key $\vk\in\mathcal{VK}$ and a ciphertext $\qct\in \Cs$.
    \item[$\qDec(\sk,\qct) \ra m^\prime~or~\bot$:]
    The decryption algorithm takes $\sk$ and $\qct$ as input, and outputs a plaintext $m^\prime \in \Ms$ or $\bot$.
    \item[$\qDelete(\ct) \ra \cert$:] The deletion algorithm takes $\qct$ as input, and outputs a certification $\cert\in\mathcal{D}$.
    \item[$\Vrfy(\vk,\cert)\ra \top~\mbox{\bf or}~\bot$:] The verification algorithm takes $\vk$ and $\cert$ as input, and outputs $\top$ or $\bot$.
\end{description}
\end{definition}
\begin{remark}
Although we consider quantum certificates in~\cref{sec:const_ske_wo_rom}, we consider classical certificates by default. In the quantum certificate case, we need to use $\qcert$ and $\qVrfy$ in the syntax.
\end{remark}

We require that an SKE with certified everlasting deletion scheme satisfies correctness defined below.
\begin{definition}[Correctness for SKE with Certified Everlasting Deletion]\label{def:correctness_cert_ever_ske}
There are three types of correctness, namely,
decryption correctness,
verification correctness, and special correctness.

\paragraph{Decryption Correctness:} There exists a negligible function $\negl$ such that for any $\secp\in \N$ and $m\in\Ms$, 
\begin{align}
\Pr\left[
m'\neq m
\ \middle |
\begin{array}{ll}
\sk\lrun \keygen(1^\secp)\\
(\vk,\qct) \lrun \qEnc(\sk,m)\\
m'\la\qDec(\sk,\qct)
\end{array}
\right] 
\leq\negl(\secp).
\end{align}

\paragraph{Verification Correctness:} There exists a negligible function $\negl$ such that for any $\secp\in \N$ and $m\in\Ms$, 
\begin{align}
\Pr\left[
\Vrfy(\vk,\cert)=\bot
\ \middle |
\begin{array}{ll}
\sk\lrun \keygen(1^\secp)\\
(\vk,\qct) \lrun \qEnc(\sk,m)\\
\cert \lrun \qDelete(\qct)
\end{array}
\right] 
\leq
\negl(\secp).
\end{align}
\end{definition}

Minimum requirements for correctness are decryption correctness and verification correctness.
However, we also require special correctness and verification correctness with QOTP in this work
because we need special correctness for the construction of the garbling scheme in \cref{sec:const_garbling},
and verification correctness with QOTP for the construction of FE in \cref{sec:const_fe_adapt}.

\begin{definition}[Special Correctness]\label{def:cd_ske_special_correctness}
There exists a negligible function $\negl$ such that for any $\secp\in \N$ and $m\in\Ms$, 
\begin{align}
\Pr\left[
\Dec(\sk_2,\qct)\neq \bot
\ \middle |
\begin{array}{ll}
\sk_2,\sk_1\lrun \keygen(1^\secp)\\
(\vk,\qct) \lrun \qEnc(\sk_1,m)
\end{array}
\right] 
\leq\negl(\secp).
\end{align}

\end{definition}

\begin{definition}[Verification Correctness with QOTP]\label{def:cd_ske_ver_correctness_QOTP}
There exists a negligible function $\negl$ and a PPT algorithm $\Modify$ such that for any $\secp\in \N$ and $m\in\Ms$, 
\begin{align}
\Pr\left[
\Vrfy(\vk,\cert^*)=\bot
\ \middle |
\begin{array}{ll}
\sk\lrun \keygen(1^\secp)\\
(\vk,\qct) \lrun \qEnc(\sk,m)\\
a,b\la\bit^{p(\lambda)}\\
\wtl{\cert} \lrun \qDelete(Z^bX^a\qct X^aZ^b)\\
\cert^*\lrun \Modify(a,b,\wtl{\cert})  
\end{array}
\right] 
\leq
\negl(\secp).
\end{align}
\end{definition}

As security, we consider two definitions, 
\cref{def:IND-CPA_security_cert_ever_ske} 
and
\cref{def:cert_ever_security_cert_ever_ske} given below.
The former is just the standard IND-CPA security 
and the latter is the certified everlasting security
that we newly define in this paper.
Roughly, the everlasting security guarantees that any QPT adversary cannot obtain plaintext information even if it becomes computationally
unbounded and obtains the secret key after it issues a valid certificate.

\begin{definition}[IND-CPA Security for SKE with Certified Everlasting Deletion]\label{def:IND-CPA_security_cert_ever_ske}
Let $\Sigma=(\keygen,\qEnc,\qDec,\qDelete,\allowbreak \Vrfy)$ be an SKE with certified everlasting deletion scheme.
We consider the following security experiment $\expb{\Sigma,\qA}{ind}{cpa}(\secp,b)$ against a QPT adversary $\qA$.
\begin{enumerate}
    \item The challenger computes $\sk \la \keygen(1^\secp)$.
    \item $\qA$ sends an encryption query $m$ to the challenger.
    The challenger computes $(\vk,\qct)\la\qEnc(\sk,m)$, and returns $(\vk,\qct)$ to $\qA$.
    $\qA$ can repeat this process polynomially many times.
    \item $\qA$ sends $(m_0,m_1)\in\cM^2$ to the challenger.
    \item The challenger computes $(\vk,\qct) \la \qEnc(\sk,m_b)$, and sends $\qct$ to $\qA$.
    \item $\qA$ sends an encryption query $m$ to the challenger.
    The challenger computes $(\vk,\qct)\la\qEnc(\sk,m)$, and returns $(\vk,\qct)$ to $\qA$.
    $\qA$ can repeat this process polynomially many times.
    \item $\qA$ outputs $b'\in \bit$. This is the output of the experiment.
\end{enumerate}
We say that $\Sigma$ is IND-CPA secure if, for any QPT $\qA$, it holds that
\begin{align}
\advb{\Sigma,\qA}{ind}{cpa}(\secp)
\seteq \abs{\Pr[ \expb{\Sigma,\qA}{ind}{cpa}(\secp, 0)=1] - \Pr[ \expb{\Sigma,\qA}{ind}{cpa}(\secp, 1)=1] }\leq \negl(\secp).
\end{align}
\end{definition}

\begin{definition}[Certified Everlasting IND-CPA Security for SKE]\label{def:cert_ever_security_cert_ever_ske}
Let $\Sigma=(\keygen,\qEnc,\qDec,\allowbreak \qDelete,\Vrfy)$ be a certified everlasting SKE scheme.
We consider the following security experiment $\expd{\Sigma,\qA}{cert}{ever}{ind}{cpa}(\secp,b)$ against a QPT adversary $\qA_1$ and an unbounded adversary $\qA_2$.
\begin{enumerate}
    \item The challenger computes $\sk \la \keygen(1^\secp)$.
    \item $\qA_1$ sends an encryption query $m_i$ to the challenger.
    The challenger computes $(\vk_i,\qct_i)\la\qEnc(\sk,m_i)$, and returns $(\vk_i,\qct_i)$ to $\qA_1$.
    $\qA_1$ can repeat this process polynomially many times.
    \item $\qA_1$ sends $(m_0,m_1)\in\cM^2$ to the challenger.
    \item The challenger computes $(\vk,\qct) \la \qEnc(\sk,m_b)$, and sends $\qct$ to $\qA_1$.
    \item $\qA_1$ sends an encryption query $m_i$ to the challenger.
    The challenger computes $(\vk_i,\qct_i)\la\Enc(\sk,m_i)$, and returns $(\vk_i,\qct_i)$ to $\qA_1$.
    $\qA_1$ can repeat this process polynomially many times.
    \item At some point, $\qA_1$ sends $\cert$ to the challenger and sends the internal state to $\qA_2$.
    \item The challenger computes $\Vrfy(\vk,\cert)$.
    If the output is $\bot$, the challenger outputs $\bot$, and sends $\bot$ to $\qA_2$.
    Otherwise, the challenger outputs $\top$, and sends $\sk$ to $\qA_2$.
    \item $\qA_2$ outputs $b'\in\{0,1\}$.
    \item If the challenger outputs $\top$, then the output of the experiment is $b'$. Otherwise, the output of the experiment is $\bot$.
\end{enumerate}
We say that $\Sigma$ is certified everlasting IND-CPA secure if, for any QPT $\qA_1$ and any unbounded $\qA_2$, it holds that
\begin{align}
\advd{\Sigma,\qA}{cert}{ever}{ind}{cpa}(\secp)
\seteq \abs{\Pr[ \expd{\Sigma,\qA}{cert}{ever}{ind}{cpa}(\secp, 0)=1] - \Pr[ \expd{\Sigma,\qA}{cert}{ever}{ind}{cpa}(\secp, 1)=1] }\leq \negl(\secp).
\end{align}
\end{definition}

\begin{definition}[PKE with Certified Everlasting Deletion (Syntax)]\label{def:cert_ever_pke}
Let $\secp$ be a security parameter and let $p$, $q$, $r$, $s$ and $t$ be polynomials. 
A PKE with certified everlasting deletion scheme is a tuple of algorithms $\Sigma=(\keygen,\qEnc,\qDec,\qDelete,\Vrfy)$
with plaintext space $\Ms\seteq\bit^n$,
ciphertext space $\Cs\seteq \cQ^{\otimes p(\lambda)}$,
public key space $\mathcal{PK}\seteq\{0,1\}^{q(\lambda)}$, secret key space $\mathcal{SK}\seteq \bit^{r(\secp)}$,
verification key space $\mathcal{VK}\seteq \bit^{s(\secp)}$
and deletion certificate space $\mathcal{D}\seteq \cQ^{\otimes t(\lambda)}$.
\begin{description}
    \item[$\keygen (1^\secp) \ra (\pk,\sk)$:]
    The key generation algorithm takes the security parameter $1^\secp$ as input and outputs a public key $\pk\in\mathcal{PK}$ and a secret key $\sk \in \mathcal{SK}$.
    \item[$\qEnc(\pk,m) \ra (\vk,\qct)$:]
    The encryption algorithm takes $\pk$ and a plaintext $m\in\Ms$ as input,
    and outputs a verification key $\vk\in\mathcal{VK}$ and a ciphertext $\qct\in \Cs$.
    \item[$\qDec(\sk,\qct) \ra m^\prime~or~\bot$:]
    The decryption algorithm takes $\sk$ and $\qct$ as input, and outputs a plaintext $m^\prime \in \Ms$ or $\bot$.
    \item[$\qDelete(\qct) \ra \cert$:]
    The deletion algorithm takes $\qct$ as input and outputs a certification $\cert\in\mathcal{D}$.
    \item[$\Vrfy(\vk,\cert)\ra \top$ or $\bot$:]
    The verification algorithm takes $\vk$ and $\cert$ as input, and outputs $\top$ or $\bot$.
    \end{description}
\end{definition}
\begin{remark}
Although we consider quantum certificates in~\cref{sec:const_pke_wo_rom}, we consider classical certificates by default. In the quantum certificate case, we need to use $\qcert$ and $\qVrfy$ in the syntax.
\end{remark}

We require that a PKE with certified everlasting deletion scheme satisfies correctness defined below.
\begin{definition}[Correctness for PKE with Certified Everlasting Deletion]\label{def:correctness_cert_ever_pke}
There are two types of correctness, namely, decryption correctness and verification correctness.

\paragraph{Decryption Correctness:}
There exists a negligible function $\negl$ such that for any $\secp\in \N$ and $m\in\Ms$,
\begin{align}
\Pr\left[
m'\neq m
\ \middle |
\begin{array}{ll}
(\pk,\sk)\lrun \keygen(1^\secp)\\
(\vk,\qct) \lrun \qEnc(\pk,m)\\
m'\la\qDec(\sk,\qct)
\end{array}
\right] 
\le\negl(\secp).
\end{align}

\paragraph{Verification Correctness:}
There exists a negligible function $\negl$ such that for any $\secp\in \N$ and $m\in\Ms$,
\begin{align}
\Pr\left[
\Vrfy(\vk,\cert)=\bot
\ \middle |
\begin{array}{ll}
(\pk,\sk)\lrun \keygen(1^\secp)\\
(\vk,\qct) \lrun \qEnc(\pk,m)\\
\cert \lrun \qDelete(\qct)
\end{array}
\right] 
\leq
\negl(\secp).
\end{align}
\end{definition}

Minimum requirements for correctness are decryption correctness and verification correctness.
However, we also require verification correctness with QOTP in this work because we need it for the construction of FE in~\cref{sec:const_fe_adapt}.

\begin{definition}[Verification Correctness with QOTP]\label{def:cd_pke_ver_correctness_QOTP}
There exists a negligible function $\negl$ and a PPT algorithm $\Modify$ such that for any $\secp\in \N$ and $m\in\Ms$, 
\begin{align}
\Pr\left[
\Vrfy(\vk,\cert^*)=\bot
\ \middle |
\begin{array}{ll}
(\pk,\sk)\lrun \keygen(1^\secp)\\
(\vk,\qct) \lrun \qEnc(\pk,m)\\
a,b\la\bit^{p(\secp)}\\
\wtl{\cert} \lrun \qDelete(Z^bX^a\qct X^aZ^b)\\
\cert^*\lrun \Modify(a,b,\wtl{\cert})  
\end{array}
\right] 
\leq
\negl(\secp).
\end{align}
\end{definition}

As security, we consider two definitions, 
\cref{def:IND-CPA_security_cert_ever_pke} 
and
\cref{def:cert_ever_security_cert_ever_pke} given below.
The former is just the standard IND-CPA security 
and the latter is the certified everlasting security
that we newly define in this paper.
Roughly, the everlasting security guarantees that any QPT adversary cannot obtain plaintext information even if it becomes computationally unbounded and obtains the secret key after it issues a valid certificate.

\begin{definition}[IND-CPA Security for PKE with Certified Everlasting Deletion]\label{def:IND-CPA_security_cert_ever_pke}
Let $\Sigma=(\keygen,\qEnc,\qDec,\qDelete,\allowbreak \Vrfy)$ be a PKE with certified everlasting deletion scheme.
We consider the following security experiment $\expb{\Sigma,\qA}{ind}{cpa}(\secp,b)$ against a QPT adversary $\qA$.
\begin{enumerate}
    \item The challenger generates $(\pk,\sk)\lrun \keygen(1^{\secp})$, and sends $\pk$ to $\qA$.
    \item $\qA$ sends $(m_0,m_1)\in\cM^2$ to the challenger.
    \item The challenger computes $(\vk,\qct) \lrun \qEnc(\pk,m_b)$, and sends $\qct$ to $\qA$.
    \item $\qA$ outputs $b'\in\bit$. This is the output of the experiment.
\end{enumerate}
We say that the $\Sigma$ is IND-CPA secure if, for any QPT $\qA$, it holds that
\begin{align}
\advb{\Sigma,\qA}{ind}{cpa}(\secp) \seteq \abs{\Pr[\expb{\Sigma,\qA}{ind}{cpa}(\secp,0)=1]  - \Pr[\expb{\Sigma,\qA}{ind}{cpa}(\secp,1)=1]} \leq \negl(\secp).
\end{align}
\end{definition}

\begin{definition}[Certified Everlasting IND-CPA Security for PKE]\label{def:cert_ever_security_cert_ever_pke}
Let $\Sigma=(\keygen,\qEnc,\qDec,\allowbreak \qDelete,\Vrfy)$ be a PKE with certified everlasting deletion scheme.
We consider the following security experiment $\expd{\Sigma,\qA}{cert}{ever}{ind}{cpa}(\secp,b)$ against a QPT adversary $\qA_1$ and an unbounded adversary $\qA_2$.
\begin{enumerate}
    \item The challenger computes $(\pk,\sk) \la \keygen(1^\secp)$, and sends $\pk$ to $\qA_1$.
    \item $\qA_1$ sends $(m_0,m_1)\in\cM^2$ to the challenger.
    \item The challenger computes $(\vk,\qct) \la \qEnc(\pk,m_b)$, and sends $\qct$ to $\qA_1$.
    \item At some point, $\qA_1$ sends $\cert$ to the challenger, and sends the internal state to $\qA_2$.
    \item The challenger computes $\Vrfy(\vk,\cert)$.
    If the output is $\bot$, the challenger outputs $\bot$, and sends $\bot$ to $\qA_2$.
    Otherwise, the challenger outputs $\top$, and sends $\sk$ to $\qA_2$.
    \item $\qA_2$ outputs $b'\in\{0,1\}$.
    \item If the challenger outputs $\top$, then the output of the experiment is $b'$.
    Otherwise, the output of the experiment is $\bot$.
\end{enumerate}
We say that the $\Sigma$ is certified everlasting IND-CPA secure if for any QPT $\qA_1$ and any unbounded $\qA_2$, it holds that
\begin{align}
\advd{\Sigma,\qA}{cert}{ever}{ind}{cpa}(\secp)
\seteq \abs{\Pr[ \expd{\Sigma,\qA}{cert}{ever}{ind}{cpa}(\secp, 0)=1] - \Pr[ \expd{\Sigma,\qA}{cert}{ever}{ind}{cpa}(\secp, 1)=1] }\leq \negl(\secp).
\end{align}
\end{definition}

\subsection{SKE Scheme with QROM}\label{sec:const_ske_rom}
In this section, we construct an SKE with certified everlasting deletion scheme with QROM.
Our construction is similar to that of the certified everlasting commitment scheme in \cite{C:HMNY22}.
The difference is that we use SKE instead of commitment.

\paragraph{Our certified everlasting secure SKE scheme.}
We construct a certified everlasting secure SKE scheme $\Sigma_{\mathsf{cesk}}=(\keygen,\qEnc,\qDec,\allowbreak \qDelete,\Vrfy)$ from the following primitives.
\begin{itemize}
    \item A one-time SKE with certified deletion scheme~(\cref{def:sk_cert_del}) $\Sigma_{\skcd}=\sfCD.(\keygen,\qEnc,\qDec,\qDelete,\Vrfy)$. 
    \item A SKE scheme~(\cref{def:ske}) $\Sigma_{\mathsf{sk}}=\mathsf{SKE}.(\keygen,\Enc,\Dec)$ with plaintext space $\bit^\secp$.
    \item A hash function $H$ 
    modeled as a quantum random oracle.
\end{itemize}

\begin{description}
    \item[$\keygen(1^{\secp})$:] $ $
    \begin{itemize}
        \item Generate $\ske.\sk\la\SKE.\keygen(1^\secp)$.
        \item Output $\sk\seteq\ske.\sk$.
    \end{itemize}
    \item[$\qEnc(\sk,m)$:] $ $
    \begin{itemize}
        \item Parse $\sk=\ske.\sk$.
        \item Generate $\sfcd.\sk\la\sfCD.\keygen(1^\secp)$ and $R\la\bit^\lambda$.
        \item Compute $\ske.\ct\la \SKE.\Enc(\ske.\sk,R)$.
        \item Compute $h\seteq H(R)\oplus \sfcd.\sk$ and $\sfcd.\qct\la \sfCD.\qEnc(\sfcd.\sk,m)$.
        \item Output $\qct\seteq (h,\ske.\ct,\sfcd.\qct)$ and $\vk\seteq \sfcd.\sk$.
    \end{itemize}
    \item[$\Dec(\sk,\qct)$:] $ $
    \begin{itemize}
        \item Parse $\sk=\ske.\sk$ and $\qct= (h,\ske.\ct,\sfcd.\qct)$.
        \item Compute $R' {\bf\,\, or \,\,}\bot\la\SKE.\Dec(\ske.\sk,\ske.\ct)$.
        If it outputs $\bot$, $\Dec(\sk,\ct)$ outputs $\bot$.
        \item Compute $\sfcd.\sk'\seteq H(R')\oplus h$.
        \item Compute $m' \la\sfCD.\qDec(\sfcd.\sk',\sfcd.\qct)$.
        \item Output $m'$.
    \end{itemize}
    \item[$\Delete(\qct)$:] $ $
    \begin{itemize}
        \item Parse $\qct=(h,\ske.\ct,\sfcd.\qct)$.
        \item Compute $\sfcd.\cert\la\sfCD.\qDelete(\sfcd.\qct)$.
        \item Output $\cert\seteq\sfcd.\cert$.
    \end{itemize}
    \item[$\Vrfy(\vk,\cert)$:] $ $
    \begin{itemize}
        \item Parse $\vk=\sfcd.\sk$ and $\cert=\sfcd.\cert$.
        \item Compute $b\la\sfCD.\Vrfy(\sfcd.\sk,\sfcd.\cert)$.
        \item Output $b$.
    \end{itemize}
\end{description}

\paragraph{Correctness:}
It is easy to see that correctness of $\Sigma_{\mathsf{cesk}}$ comes from those of $\Sigma_{\mathsf{sk}}$ and $\Sigma_{\skcd}$.
Special correctness holds due to that of $\Sigma_{\sk}$. Verifcation correctness with QOTP holds due to that of $\Sigma_{\skcd}$.

\paragraph{Security:}
The following two theorems hold.
\begin{theorem}\label{thm:sk_comp_security}
If $\Sigma_{\mathsf{sk}}$ satisfies the OW-CPA security~(\cref{def:OW-CPA_security_ske}) and $\Sigma_{\skcd}$ satisfies the OT-CD security~(\cref{def:sk_cert_del}), 
$\Sigma_{\mathsf{cesk}}$ satisfies the IND-CPA security~(\cref{def:IND-CPA_security_cert_ever_ske}).
\end{theorem}
Its proof is similar to that of \cref{thm:sk_ever_security}, and therefore we omit it.

\begin{theorem}\label{thm:sk_ever_security}
If $\Sigma_{\mathsf{sk}}$ satisfies the OW-CPA security~(\cref{def:OW-CPA_security_ske})
and $\Sigma_{\skcd}$ satisfies the OT-CD security~(\cref{def:sk_cert_del}),
$\Sigma_{\mathsf{cesk}}$ satisfies the certified everlasting IND-CPA security~(\cref{def:cert_ever_security_cert_ever_ske}).
\end{theorem}
Its proof is similar to that of \cite[Theorem~5.8]{C:HMNY22}.

\subsection{SKE Scheme without QROM}\label{sec:const_ske_wo_rom}
In this section, we construct an SKE with certified everlasting deletion scheme without QROM.
Note that unlike the construction with QROM (\cref{sec:const_ske_rom}), in this construction
the plaintext space is of constant size.
However, the size can be extended to the polynomial size via the standard hybrid argument.
Our construction is similar to that of revocable quantum timed-release encryption in \cite{JACM:Unruh15}.
The difference is that we use SKE instead of timed-release encryption.

\paragraph{Our certified everlasting secure SKE scheme without QROM.}
Let $k_1$ and $k_2$ be constants such that $k_1>k_2$.
Let $p$, $q$, $r$, $s$ and $t$ be polynomials.
Let $(C_1,C_2)$ be a CSS code with parameters $q,k_1,k_2,t$.
We construct a certified everlasting secure SKE scheme
$\Sigma_{\mathsf{cesk}}=(\keygen,\qEnc,\qDec,\qDelete,\Vrfy)$ with plaintext space $\Ms=C_1/C_2$ (isomorphic to $\bit^{k_1-k_2}$), 
ciphertext space $\Cs=\cQ^{\otimes \left(p(\lambda)+q(\lambda)\right)}\times \{0,1\}^{r(\lambda)}\times \bit^{q(\lambda)}/C_1\times C_1/C_2$,
secret key space $\mathcal{SK}= \bit^{s(\lambda)}$,
verification key space $\mathcal{VK}=\{0,1\}^{p(\lambda)}\times [p(\lambda)+q(\lambda)]_{p(\lambda)}\times \bit^{p(\lambda)}$ 
and deletion certificate space $\mathcal{D}=\mathcal{Q}^{\otimes \left(p(\lambda)+q(\lambda)\right)}$
from the following primitive.
\begin{itemize}
    \item An SKE scheme~(\cref{def:ske}) $\Sigma_{\mathsf{sk}}=\SKE.(\keygen,\Enc,\Dec)$ with plaintext space $\Ms=\{0,1\}^{p(\lambda)}\times[p(\lambda)+q(\lambda)]_{p(\lambda)}\times \bit^{p(\lambda)} \times C_1/C_2$, secret key space $\mathcal{SK}=\{0,1\}^{s(\secp)}$ and ciphertext space $\Cs=\bit^{r(\lambda)}$. 
\end{itemize}
The construction is as follows. (We will omit the security parameter below.)
\begin{description}
\item[$\keygen(1^\secp)$:]$ $
\begin{itemize}
    \item Generate $\ske.\sk\la\SKE.\keygen(1^\secp)$.
    \item Output $\sk\seteq\ske.\sk$.
\end{itemize}
\item[$\Enc(\sk,m)$:] $ $
\begin{itemize}
    \item Parse $\sk=\ske.\sk$.
    \item Generate $B\la\bit^{p}$, $Q\la[p+q]_{p}$, $y\la C_1/C_2$, $u\la\bit^q/C_1$, $r\la\bit^p$,
    $x\la C_1/C_2$, $w\la C_2$.
    \item Compute $\ske.\ct\la\SKE.\Enc\left(\ske.\sk,(B,Q,r,y)\right)$.
    \item Let $U_Q$ be the unitary that permutes the qubits in $Q$ into the first half of the system.
    (I.e., $U_Q\ket{x_1x_2\cdots x_{p+q}}=\ket{x_{a_1}x_{a_2}\cdots x_{a_p}x_{b_1}x_{b_2}\cdots x_{b_q}}$ with $Q\seteq \{a_1,a_2,\cdots, a_p\}$ and $\{1,2,\cdots, p+q\}\backslash Q\seteq \{b_1,b_2,\cdots,b_q\}$.)
    \item Construct a quantum state $\ket{\Psi} \seteq U_{Q}^{\dagger}(H^B\otimes I^{\otimes q})(\ket{r}\otimes \ket{x\oplus w\oplus u})$.
    \item Compute $h\seteq m\oplus x\oplus y$.
    \item Output $\qct\seteq (\ket{\Psi},\ske.\ct,u,h)$ and $\vk\seteq (B,Q,r)$.
\end{itemize}
\item[$\Dec(\sk,\qct)$:] $ $
\begin{itemize}
    \item Parse $\sk=\ske.\sk$, $\qct=(\ket{\Psi},\ske.\ct,u,h)$.
    \item Compute $(B,Q,r,y)/\bot\la\SKE.\Dec(\ske.\sk,\ske.\ct)$.
    If $\bot\la\SKE.\Dec(\ske.\sk,\ske.\ct)$, $\Dec(\sk,\ct)$ outputs $\bot$ and aborts.
    \item Apply $U_Q$ to $\ket{\Psi}$, measure the last $q$-qubits in the computational basis and obtain the measurement outcome $\gamma\in\bit^q$. 
    \item Compute $x\seteq \gamma\oplus u$ mod $C_2$.
    \item Output $m'\seteq h\oplus x\oplus y$.
\end{itemize}
\item[$\Delete(\qct)$:] $ $
\begin{itemize}
    \item Parse $\qct=(\ket{\Psi},\ske.\ct,u,h)$.
    \item Output $\qcert\seteq\ket{\Psi}$.
\end{itemize}
\item[$\qVrfy(\vk,\qcert)$:] $ $
\begin{itemize}
    \item Parse $\vk=(B,Q,r)$ and $\qcert=\ket{\Psi}$.
    \item Apply $(H^B\otimes I^{\otimes q})U_Q$ to $\ket{\Psi}$, measure the first $p$-qubits in the computational basis and obtain the measurement outcome $r'\in\bit^p$.
    \item Output $\top$ if $r=r'$ and output $\bot$ otherwise.
\end{itemize}
\end{description}

\paragraph{Correctness.}
Correctness easily follows from that of
$\Sigma_{\sk}$.
Special correctness holds due to that of $\Sigma_{\sk}$. Verifcation correctness with QOTP holds since $\Modify$ is the decryption algorithm of QOTP.

\paragraph{Security.}
The following two theorems hold.

\begin{theorem}\label{thm:sk_comp_security_wo_rom}
If $\Sigma_{\mathsf{sk}}$ is IND-CPA secure~(\cref{def:IND-CPA_security_ske}), then $\Sigma_{\mathsf{cesk}}$ is IND-CPA secure~(\cref{def:IND-CPA_security_cert_ever_ske}).
\end{theorem}

Its proof is straightforward, so we omit it.

\begin{theorem}\label{thm:sk_ever_security_wo_rom}
If $\Sigma_{\mathsf{sk}}$ is IND-CPA secure~(\cref{def:IND-CPA_security_ske}) and $tp/(p+q)-4(k_1-k_2){\rm ln}2$ is superlogarithmic, then $\Sigma_{\mathsf{cesk}}$ is certified everlasting IND-CPA secure~(\cref{def:cert_ever_security_cert_ever_ske}).
\end{theorem}

Its proof is similar to that of \cite[Theorem~3]{JACM:Unruh15}.

Note that the plaintext space is of constant size in our construction. However,
via the standard hybrid argument 
, we can extend it to the polynomial size.

\subsection{PKE Scheme with QROM}\label{sec:const_pke_rom}
In this section, we construct a certified everlasting secure PKE scheme with QROM.
Our construction is similar to that of the certified everlasting commitment scheme in \cite{C:HMNY22}.
The difference is that we use PKE instead of commitment.

\paragraph{Our certified everlasting secure PKE scheme.}
We construct a certified everlasting secure PKE scheme $\Sigma_{\mathsf{cepk}}=(\keygen,\qEnc,\qDec,\allowbreak\qDelete,\Vrfy)$
from a one-time SKE with certified deletion scheme $\Sigma_{\skcd}=\SKE.(\keygen,\qEnc,\qDec,\allowbreak \qDelete,\Vrfy)$~(\cref{def:sk_cert_del}),
a PKE scheme $\Sigma_{\mathsf{pk}}=\PKE.(\keygen,\Enc,\Dec)$ with plaintext space $\bit^\secp$~(\cref{def:pke})
and a hash function $H$ modeled as quantum random oracle.

\begin{description}
    \item[$\keygen(1^{\secp})$:] $ $
    \begin{itemize}
        \item Generate $(\pke.\pk,\pke.\sk)\la\keygen(1^\secp)$.
        \item Output $\pk\seteq \pke.\pk$ and $\sk\seteq\pke.\sk$.
    \end{itemize}
    \item[$\qEnc(\pk,m)$:] $ $
    \begin{itemize}
        \item Parse $\pk=\pke.\pk$.
        \item Generate $\ske.\sk\la\SKE.\keygen(1^\secp)$.
        \item Randomly generate $R\la\{0,1\}^\secp$.
        \item Compute $\pke.\ct\la\PKE.\Enc(\pke.\pk,R)$.
        \item Compute $h\seteq H(R)\oplus \ske.\sk$ and $\ske.\qct\la\SKE.\qEnc(\ske.\sk,m)$.
        \item Output $\qct\seteq (h,\ske.\qct,\pke.\ct)$ and $\vk\seteq \ske.\sk$.
    \end{itemize}
    \item[$\Dec(\sk,\qct)$:] $ $
    \begin{itemize}
        \item Parse $\sk=\pke.\sk$ and $\qct= (h,\ske.\qct,\pke.\ct)$.
        \item Compute $R'\la\PKE.\Dec(\pke.\sk,\pke.\ct)$.
        \item Compute $\ske.\sk'\seteq h\oplus H(R')$.
        \item Compute $m'\la\SKE.\qDec(\ske.\sk',\ske.\qct)$.
        \item Output $m'$.
    \end{itemize}
    \item[$\Delete(\qct)$:] $ $
    \begin{itemize}
        \item Parse $\qct=(h,\ske.\qct,\pke.\ct)$.
        \item Compute $\ske.\cert\la\SKE.\qDelete(\ske.\qct)$.
        \item Output $\cert\seteq\ske.\cert$.
    \end{itemize}
    \item[$\Vrfy(\vk,\cert)$:] $ $
    \begin{itemize}
        \item Parse $\vk=\ske.\sk$ and $\cert=\ske.\cert$.
        \item Compute $b\la\SKE.\Vrfy(\ske.\sk,\ske.\cert)$.
        \item Output $b$.
    \end{itemize}
\end{description}

\paragraph{Correctness:}
Correctness easily follows from those of
$\Sigma_{\mathsf{pk}}$ and $\Sigma_{\skcd}$.
Verifcation correctness with QOTP holds due to that of $\Sigma_{\skcd}$.

\paragraph{Security:}
The following two theorems hold.
Their proofs are similar to those of \cref{thm:sk_comp_security,thm:sk_ever_security}, and therefore we omit them.
\begin{theorem}\label{thm:pk_comp_seucirty}
If $\Sigma_{\mathsf{pk}}$ satisfies the OW-CPA security~(\cref{def:OW-CPA_security_pke}) and $\Sigma_{\skcd}$ satisfies the OT-CD security~(\cref{def:security_sk_cert_del}), 
$\Sigma_{\mathsf{cepk}}$ is IND-CPA secure~(\cref{def:IND-CPA_security_cert_ever_pke}).
\end{theorem}

\begin{theorem}\label{thm:pk_ever_security}
If $\Sigma_{\mathsf{pk}}$ satisfies the OW-CPA security~(\cref{def:OW-CPA_security_pke}) and $\Sigma_{\skcd}$ satisfies the OT-CD security~(\cref{def:security_sk_cert_del}),
$\Sigma_{\mathsf{cepk}}$ is certified everlasting IND-CPA secure~(\cref{def:cert_ever_security_cert_ever_pke}).
\end{theorem}

\subsection{PKE Scheme without QROM}\label{sec:const_pke_wo_rom}
In this section, we construct a certified everlasting secure PKE scheme 
without QROM.
Our construction is similar to that of quantum timed-release encryption presented in \cite{JACM:Unruh15}.
The difference is that we use PKE instead of timed-release encryption.

\paragraph{Our certified everlasting secure PKE scheme without QROM.}
Let $k_1$ and $k_2$ be some constant such that $k_1>k_2$.
Let $p$, $q$, $r$, $s$, $t$ and $u$ be some polynomials.
Let $(C_1,C_2)$ be a CSS code with parameters $q,k_1,k_2,t$.
We construct a certified everlasting secure PKE scheme
$\Sigma_{\mathsf{cepk}}=(\keygen,\qEnc,\qDec,\qDelete,\Vrfy)$, with plaintext space $\Ms=C_1/C_2$ (isomorphic $\bit^{(k_1-k_2)}$), 
ciphertext space $\Cs=\cQ^{\otimes \left(p(\lambda)+q(\lambda)\right)}\times \{0,1\}^{r(\lambda)}\times \bit^{q(\lambda)}/C_1\times C_1/C_2$,
public key space $\mathcal{PK}=\bit^{u(\lambda)}$,
secret key space $\mathcal{SK}= \bit^{s(\lambda)}$,
verification key space $\mathcal{VK}=\{0,1\}^{p(\lambda)}\times [p(\lambda)+q(\lambda)]_{p(\lambda)}\times \bit^{p(\lambda)}$ 
and deletion certificate space $\mathcal{D}=\mathcal{Q}^{\otimes \left(p(\lambda)+q(\lambda)\right)}$
from the following primitive.
\begin{itemize}
    \item A PKE scheme~(\cref{def:pke}) $\Sigma_{\mathsf{pk}}=\PKE.(\keygen,\Enc,\Dec)$ with
    plaintext space $\Ms=\{0,1\}^{p(\lambda)}\times[p(\lambda)+q(\lambda)]_{p(\lambda)}\times \bit^{p(\lambda)} \times C_1/C_2$,
    public key space $\mathcal{PK}=\bit^{u(\lambda)}$,
    secret key space $\mathcal{SK}=\{0,1\}^{s(\secp)}$ and ciphertext space $\Cs=\bit^{r(\lambda)}$.
\end{itemize}
The construction is as follows. (We will omit the security parameter below.)
\begin{description}
\item[$\keygen(1^\secp)$:]$ $
\begin{itemize}
    \item Generate $(\pke.\pk,\pke.\sk)\la\PKE.\keygen(1^\secp)$.
    \item Output $\pk\seteq\pke.\pk$ and $\sk\seteq\pke.\sk$.
\end{itemize}
\item[$\qEnc(\pk,m)$:] $ $
\begin{itemize}
    \item Parse $\pk=\pke.\pk$.
    \item Generate $B\la\bit^{p}$, $Q\la[p+q]_{p}$, $y\la C_1/C_2$, $u\la\bit^q/C_1$, $r\la\bit^p$,
    $x\la C_1/C_2$, $w\la C_2$.
    \item Compute $\pke.\ct\la\PKE.\Enc\left(\pke.\pk,(B,Q,r,y)\right)$.
    \item Let $U_Q$ be the unitary that permutes the qubits in $Q$ into the first half of the system.
    (I.e., $U_Q\ket{x_1x_2\cdots x_{p+q}}=\ket{x_{a_1}x_{a_2}\cdots x_{a_p}x_{b_1}x_{b_2}\cdots x_{b_q}}$ with $Q\seteq \{a_1,a_2,\cdots, a_p\}$ and $\{1,2,\cdots, p+q\}\setminus Q\seteq \{b_1,b_2,\cdots,b_q\}$.)
    \item Generate a quantum state $\ket{\Psi} \seteq U_{Q}^{\dagger}(H^B\otimes I^{\otimes q})(\ket{r}\otimes \ket{x\oplus w\oplus u})$.
    \item Compute $h\seteq m\oplus x\oplus y$.
    \item Output $\qct\seteq (\ket{\Psi},\pke.\ct,u,h)$ and $\vk\seteq (B,Q,r)$.
\end{itemize}
\item[$\qDec(\sk,\qct)$:] $ $
\begin{itemize}
    \item Parse $\sk=\pke.\sk$ and $\qct=(\ket{\Psi},\pke.\ct,u,h)$.
    \item Compute $(B,Q,r,y)\la\PKE.\Dec(\pke.\sk,\pke.\ct)$.
    \item Apply $U_Q$ to $\ket{\Psi}$, measure the last $q$-qubits in the computational basis and obtain the measurement outcome $\gamma$. 
    \item Compute $x\seteq \gamma\oplus u$ mod $C_2$.
    \item Output $m'\seteq h\oplus x\oplus y$.
\end{itemize}
\item[$\qDelete(\qct)$:] $ $
\begin{itemize}
    \item Parse $\qct=(\ket{\Psi},\pke.\ct,u,h)$.
    \item Output $\qcert\seteq\ket{\Psi}$.
\end{itemize}
\item[$\qVrfy(\vk,\qcert)$:] $ $
\begin{itemize}
    \item Parse $\vk=(B,Q,r)$ and $\qcert=\ket{\Psi}$.
    \item Apply $(H^B\otimes I^{\otimes q})U_Q$ to $\ket{\Psi}$, measure the first $p$-qubits in the computational basis and obtain the measurement outcome $r'$.
    \item Output $\top$ if $r=r'$ and output $\bot$ otherwise.
\end{itemize}
\end{description}

\paragraph{Correctness.}
Correctness easily follows from that of $\Sigma_{\pk}$.
Verifcation correctness with QOTP holds since $\Modify$ is the decryption algorithm of QOTP.

\paragraph{Security.}
The following two theorems hold.

\begin{theorem}\label{thm:pk_comp_security_wo_rom}
If $\Sigma_{\mathsf{pk}}$ is IND-CPA secure~(\cref{def:IND-CPA_pke}), then $\Sigma_{\mathsf{cepk}}$ is IND-CPA secure~(\cref{def:IND-CPA_security_cert_ever_pke}).
\end{theorem}
Its proof is straightforward, and therefore we omit it.

\begin{theorem}\label{thm:pk_ever_security_wo_rom}
If $\Sigma_{\mathsf{pk}}$ is IND-CPA secure~(\cref{def:IND-CPA_pke}) and $tp/(p+q)-4(k_1-k_2){\rm ln}2$ is superlogarithmic,
then $\Sigma_{\mathsf{cepk}}$ is certified everlasting IND-CPA secure~(\cref{def:cert_ever_security_cert_ever_pke}).
\end{theorem}
Its proof is similar to that of \cite[Theorem~3]{JACM:Unruh15}.
Note that the plaintext space is of constant size in our construction. However,
via the standard hybrid argument, we can extend it to the polynomial size.


\section{Receiver Non-Committing Encryption with Certified Everlasting Deletion}\label{sec:RNCE}
In this section, we define and construct receiver non-committing encryption with certified everlasting deletion.
In \cref{sec:def_rnce}, we define RNCE with certified everlasting deletion.
In \cref{sec:const_rnce_classic}, we construct a certified everlasting RNCE scheme from certified everlasting secure PKE (\cref{sec:cd_ske_pke}).

\subsection{Definition}\label{sec:def_rnce}
\begin{definition}[RNCE with Certified Everlasting Deletion (Syntax)]\label{def:cert_ever_rnce_classic}
Let $\secp$ be the security parameter and let $p$, $q$, $r$, $s$, $t$, $u$, and $v$ be polynomials. 
An RNCE with certified everlasting deletion scheme is a tuple of algorithms $\Sigma=(\Setup,\keygen,\qEnc,\qDec,\qFake,\Reveal,\allowbreak \qDelete,\Vrfy)$
with plaintext space $\Ms\seteq\bit^{n}$,
ciphertext space $\Cs\seteq \cQ^{\otimes p(\lambda)}$,
public key space $\mathcal{PK}\seteq\{0,1\}^{q(\lambda)}$,
master secret key space $\mathcal{MSK}\seteq \bit^{r(\lambda)}$,
secret key space $\mathcal{SK}\seteq \bit^{s(\secp)}$,
verification key space $\mathcal{VK}\seteq \bit^{t(\secp)}$,
deletion certificate space $\mathcal{D}\seteq \cQ^{u(\lambda)}$,
and auxiliary state space $\mathcal{AUX}\seteq\bit^{v(\lambda)}$.
\begin{description}
    \item[$\Setup(1^\secp)\ra(\pk,\MSK)$:]The setup algorithm takes the security parameter $1^\secp$ as input, and outputs a public key $\pk\in\mathcal{PK}$ and a master secret key $\MSK\in\mathcal{MSK}$.
    \item[$\keygen (\MSK) \ra \sk$:]
    The key generation algorithm takes the master secret key $\MSK$ as input,
    and outputs a secret key $\sk\in\mathcal{SK}$.
    \item[$\qEnc(\pk,m) \ra (\vk,\qct)$:]
    The encryption algorithm takes $\pk$ and a plaintext $m\in\Ms$ as input,
    and outputs a verification key $\vk\in\mathcal{VK}$ and a ciphertext $\qct\in \Cs$.
    \item[$\qDec(\sk,\qct) \ra m^\prime~or~\bot$:]
    The decryption algorithm takes $\sk$ and $\ct$ as input, and outputs a plaintext $m^\prime \in \Ms$ or $\bot$. 
    \item[$\qFake(\pk)\ra (\vk,\widetilde{\qct},\aux)$:]
    The fake encryption algorithm takes $\pk$ as input, and outputs a verification key $\vk\in\mathcal{VK}$, a fake ciphertext $\widetilde{\qct}\in \Cs$ and an auxiliary state $\aux\in\mathcal{AUX}$.
    \item[$\Reveal(\pk,\MSK,\aux,m)\ra \widetilde{\sk}$:]
    The reveal algorithm takes $\pk,\MSK,\aux$ and $m$ as input, and outputs a fake secret key $\widetilde{\sk}\in\mathcal{SK}$.
    \item[$\qDelete(\qct) \ra \cert$:]
    The deletion algorithm takes $\qct$ as input and outputs a certification $\cert\in\mathcal{D}$.
    \item[$\Vrfy(\vk,\cert)\ra \top$ or $\bot$:]
    The verification algorithm takes $\vk$ and $\cert$ as input, and outputs $\top$ or $\bot$.
    \end{description}
\end{definition}

We require that an RNCE with certified everlasting deletion scheme satisfies correctness defined below.
\begin{definition}[Correctness for RNCE with Certified Everlasting Deletion]\label{def:correctness_cert_ever_rnce_classic}
There are two types of correctness, namely, decryption correctness and verification correctness.

\paragraph{Decryption Correctness:}
There exists a negligible function $\negl$ such that for any $\secp\in \N$ and $m\in\Ms$,
\begin{align}
\Pr\left[
m'\neq m
\ \middle |
\begin{array}{ll}
(\pk,\MSK)\la\Setup(1^\secp)\\
(\vk,\qct) \lrun \qEnc(\pk,m)\\
\sk\lrun\keygen(\MSK)\\
m'\la\qDec(\sk,\qct)
\end{array}
\right] 
\leq \negl(\lambda).
\end{align}

\paragraph{Verification Correctness:}
There exists a negligible function $\negl$ such that for any $\secp\in \N$ and $m\in\Ms$,
\begin{align}
\Pr\left[
\Vrfy(\vk,\cert)=\bot
\ \middle |
\begin{array}{ll}
(\pk,\MSK)\la\Setup(1^\secp)\\
(\vk,\qct) \lrun \qEnc(\pk,m)\\
\cert \lrun \qDelete(\qct)
\end{array}
\right] 
\leq
\negl(\secp).
\end{align}


\end{definition}

As security, we consider two definitions, 
\cref{def:rec_nc_security_classic} 
and
\cref{def:cert_ever_rec_nc_security_classic} given below.
The former is just the standard receiver non-committing security 
and the latter is the certified everlasting security that we newly define in this paper.
Roughly, the everlasting security guarantees that any QPT adversary cannot distinguish whether the ciphertext and the secret key are properly generated or not even if it becomes computationally unbounded and obtains the master secret key after it issues a valid certificate.

\begin{definition}[Receiver Non-Committing Security for RNCE with Certified Everlasting Deletion]\label{def:rec_nc_security_classic}
Let $\Sigma=(\Setup,\keygen,\allowbreak \qEnc,\qDec,\qFake,\Reveal,\qDelete,\Vrfy)$ be an RNCE with certified everlasting deletion scheme.
We consider the following security experiment $\expb{\Sigma,\qA}{rec}{nc}(\secp,b)$ against a QPT adversary $\qA$.
\begin{enumerate}
    \item The challenger runs $(\pk,\MSK)\la\Setup(1^\secp)$ and sends $\pk$ to $\qA$.
    \item $\qA$ sends $m\in\Ms$ to the challenger.
    \item The challenger does the following:
    \begin{itemize}
        \item If $b=0$, the challenger generates $(\vk,\qct)\la\qEnc(\pk,m)$ and $\sk\la\keygen(\MSK)$, and sends $(\qct,\sk)$ to $\qA$.
        \item If $b=1$, the challenger generates $(\vk,\widetilde{\qct},\aux)\la \qFake(\pk)$ and $\widetilde{\sk}\la \Reveal(\pk,\MSK,\aux,m)$, and sends $(\widetilde{\qct},\widetilde{\sk})$ to $\qA$.
    \end{itemize}
    \item $\qA$ outputs $b'\in\bit$.
\end{enumerate}
We say that $\Sigma$ is receiver non-committing (RNC) secure if, for any QPT $\qA$, it holds that
\begin{align}
    \advb{\Sigma,\qA}{rec}{nc}(\secp) \seteq \abs{\Pr[\expb{\Sigma,\qA}{rec}{nc}(\secp,0)=1]  - \Pr[\expb{\Sigma,\qA}{rec}{nc}(\secp,1)=1]} \leq \negl(\secp).
\end{align}
\end{definition}

\begin{definition}[Certified Everlasting RNC Security for RNCE]\label{def:cert_ever_rec_nc_security_classic}
Let $\Sigma=(\Setup,\keygen,\qEnc,\allowbreak \qDec,\qFake, \Reveal,\allowbreak  \qDelete,\Vrfy)$ be a certified everlasting RNCE scheme.
We consider the following security experiment $\expd{\Sigma,\qA}{cert}{ever}{rec}{nc}(\secp,b)$ against a QPT adversary $\qA_1$ and an unbounded adversary $\qA_2$.
\begin{enumerate}
    \item The challenger runs $(\pk,\MSK)\la\Setup(1^\lambda)$ and sends $\pk$ to $\qA_1$.
    \item $\qA_1$ sends $m\in\Ms$ to the challenger.
    \item The challenger does the following:
    \begin{itemize}
        \item If $b=0$, the challenger generates $(\vk,\qct)\la\qEnc(\pk,m)$ and $\sk\la\keygen(\MSK)$, and sends $(\qct,\sk)$ to $\qA_1$.
        \item If $b=1$, the challenger generates $(\vk,\widetilde{\qct},\aux)\la \qFake(\pk)$ and $\widetilde{\sk}\la\Reveal(\pk,\MSK,\aux,m)$,
        and sends $(\widetilde{\qct},\widetilde{\sk})$ to $\qA_1$.
    \end{itemize}
    \item At some point, $\qA_1$ sends $\cert$ to the challenger and its internal state to $\qA_2$.
    \item The challenger computes $\Vrfy(\vk,\cert)$.
    If the output is $\top$, the challenger outputs $\top$ and sends $\MSK$ to $\qA_2$. 
    If the output is $\bot$, the challenger outputs $\bot$ and sends $\bot$ to $\qA_2$.
    \item $\qA_2$ outputs $b'\in\bit$.
    \item If the challenger outputs $\top$, then the output of the experiment is $b'$.
    Otherwise, the output of the experiment is $\bot$.
\end{enumerate}
We say that $\Sigma$ is certified everlasting RNC secure if for any QPT $\qA_1$ and any unbounded $\qA_2$,
it holds that
\begin{align}
    \advd{\Sigma,\qA}{cert}{ever}{rec}{nc}(\secp) \seteq \abs{\Pr[\expd{\Sigma,\qA}{cert}{ever}{rec}{nc}(\secp,0)=1]  - \Pr[\expd{\Sigma,\qA}{cert}{ever}{rec}{nc}(\secp,1)=1]} \leq \negl(\secp).
\end{align}
\end{definition}


\subsection{Construction}\label{sec:const_rnce_classic}
In this section, we construct a certified everlasting RNCE scheme from a certified everlasting PKE scheme~(\cref{def:cert_ever_pke}).
Our construction is similar to that of the secret-key RNCE scheme presented in \cite{C:KNTY19}.
The difference is that we use a certified everlasting secure PKE scheme instead of a standard SKE scheme.

\paragraph{Our certified everlasting secure RNCE scheme.} 
We construct a certified everlasting secure RNCE scheme
$\Sigma_{\mathsf{cence}}=(\Setup,\keygen,\allowbreak \qEnc,\qDec, \qFake,\Reveal,\qDelete,\Vrfy)$
from a certified everlasting secure PKE scheme $\Sigma_{\mathsf{cepk}}=\PKE.(\keygen,\qEnc,\qDec,\qDelete,\Vrfy)$, which was introduced in~\cref{def:cert_ever_pke}.
    
\begin{description}
    \item[$\Setup(1^\secp)$:]$ $
    \begin{itemize}
    \item Generate $(\pke.\pk_{i,\alpha},\pke.\sk_{i,\alpha})\la\PKE.\keygen(1^\secp)$ for all $i\in[n]$ and $\alpha\in\bit$.
    \item Output $\pk\seteq\{\pke.\pk_{i,\alpha}\}_{i\in[n],\alpha\in\bit}$ and $\MSK\seteq\{\pke.\sk_{i,\alpha}\}_{i\in[n],\alpha\in\bit}$.
    \end{itemize}
    \item[$\keygen(\MSK)$:]$ $
    \begin{itemize}
        \item Parse $\MSK=\{\pke.\sk_{i,\alpha}\}_{i\in[n],\alpha\in\bit}$.
        \item Generate $x\la \bit^n$.
        \item Output $\sk\seteq(x,\{\pke.\sk_{i,x[i]}\}_{i\in[n]})$.
    \end{itemize}
    \item[$\qEnc(\pk,m)$:]$ $
    \begin{itemize}
        \item Parse $\pk= \{\pke.\pk_{i,\alpha}\}_{i\in[n],\alpha\in\bit}$.
        \item Compute $(\pke.\vk_{i,\alpha},\pke.\qct_{i,\alpha})\la\PKE.\qEnc(\pke.\pk_{i,\alpha},m[i])$ for all $i\in[n]$ and $\alpha\in\bit$.
        \item Output $\vk\seteq \{\pke.\vk_{i,\alpha}\}_{i\in[n],\alpha\in\bit}$ and $\qct\seteq\{\pke.\qct_{i,\alpha}\}_{i\in[n],\alpha\in\bit}$.
    \end{itemize}
    \item[$\qDec(\sk,\qct)$:]$ $
    \begin{itemize}
        \item Parse $\sk= (x,\{\pke.\sk_{i}\}_{i\in[n]})$ and $\qct=\{\pke.\qct_{i,\alpha}\}_{i\in[n],\alpha\in\bit}$.
        \item Compute $m[i]\la\PKE.\qDec(\pke.\sk_{i},\pke.\qct_{i,x[i]})$ for all $i\in[n]$.
        \item Output $m\seteq m[1]||m[2]||\cdots|| m[n]$.
    \end{itemize}
    \item[$\qFake(\pk)$:]$ $
    \begin{itemize}
        \item Parse $\pk= \{\pke.\pk_{i,\alpha}\}_{i\in[n],\alpha\in\bit}$.
        \item Generate $x^*\la\bit^n$.
        \item Compute $(\pke.\vk_{i,x^*[i]},\pke.\qct_{i,x^*[i]})\la \PKE.\qEnc(\pke.\pk_{i,x^*[i]},0)$ and  $(\pke.\vk_{i,x^*[i]\oplus 1},\allowbreak\pke.\qct_{i,x^*[i]\oplus 1})\la \PKE.\qEnc(\pke.\pk_{i,x^*[i]\oplus 1},1)$ for all $i\in[n]$.
        \item Output $\vk\seteq \{\pke.\vk_{i,\alpha}\}_{i\in[n],\alpha\in\bit}$,
        $\widetilde{\qct}\seteq \{\pke.\qct_{i,\alpha}\}_{i\in[n],\alpha\in\bit}$ and $\aux=x^*$.
    \end{itemize}
    \item[$\Reveal(\pk,\MSK,\aux,m)$:]$ $
    \begin{itemize}
        \item Parse $\pk= \{\pke.\pk_{i,\alpha}\}_{i\in[n],\alpha\in\bit}$, $\MSK= \{\pke.\sk_{i,\alpha}\}_{i\in[n],\alpha\in\bit}$ and $\aux=x^*$.
        \item Output $\widetilde{\sk}\seteq \left(x^*\oplus m,\{\pke.\sk_{i,x^*[i]\oplus m[i]}\}_{i\in[n]}\right)$. 
    \end{itemize}
    \item[$\qDelete(\qct)$:]$ $
    \begin{itemize}
        \item Parse $\qct=\{\pke.\qct_{i,\alpha}\}_{i\in[n],\alpha\in\bit}$.
        \item Compute $\pke.\cert_{i,\alpha}\la\PKE.\qDelete(\pke.\qct_{i,\alpha})$ for all $i\in[n]$ and $\alpha\in\bit$.
        \item Output $\cert\seteq\{\pke.\cert_{i,\alpha}\}_{i\in[n],\alpha\in\bit}$.
    \end{itemize}
    \item[$\Vrfy(\vk,\cert)$:]$ $
    \begin{itemize}
        \item Parse $\vk=\{\pke.\vk_{i,\alpha}\}_{i\in[n],\alpha\in\bit}$ and $\cert=\{\pke.\cert_{i,\alpha}\}_{i\in[n],\alpha\in\bit}$.
        \item Compute $\top/\bot\la\PKE.\Vrfy(\pke.\vk_{i,\alpha},\pke.\cert_{i,\alpha})$ for all $i\in[n]$ and $\alpha\in\bit$.
        If all results are $\top$, $\Vrfy(\vk,\cert)$ outputs $\top$. Otherwise, it outputs $\bot$. 
    \end{itemize}
\end{description}

\paragraph{Correctness:}
Correctness easily follows from that of $\Sigma_{\mathsf{cepk}}$.

\paragraph{Security:}
The following two theorems hold.
\begin{theorem}\label{thm:comp_security_rnce_classic}
If $\Sigma_{\mathsf{cepk}}$ is IND-CPA secure~(\cref{def:IND-CPA_security_cert_ever_pke}), $\Sigma_{\mathsf{cence}}$ is RNC secure~(\cref{def:rec_nc_security_classic}).
\end{theorem}
Its proof is similar to that of \cref{thm:ever_security_rnce_classic}, and therefore we omit it.

\begin{theorem}\label{thm:ever_security_rnce_classic}
If $\Sigma_{\mathsf{cepk}}$ is certified everlasting IND-CPA secure~(\cref{def:cert_ever_security_cert_ever_pke}), $\Sigma_{\mathsf{cence}}$ is certified everlasting RNC secure~(\cref{def:cert_ever_rec_nc_security_classic}).
\end{theorem}

\begin{proof}[Proof of \cref{thm:ever_security_rnce_classic}]
To prove the theorem, let us introduce the sequence of hybrids.
\begin{description}
    \item[$\hyb_{0}$:] This is identical to $\expd{\Sigma_{\mathsf{cence}},\qA}{cert}{ever}{rec}{nc}(\secp,0)$.
    For clarity, we describe the experiment against any adversary $\qA=(\qA_1,\qA_2)$, where $\qA_1$ is any QPT adversary and $\qA_2$ is any unbounded adversary. 
\begin{enumerate}
    \item The challenger generates $(\pke.\pk_{i,\alpha},\pke.\sk_{i,\alpha})\la\PKE.\keygen(1^\lambda)$ for all $i\in[n]$ and $\alpha\in\bit$.
    \item The challenger sends $\{\pke.\pk_{i,\alpha}\}_{i\in[n],\alpha\in\bit}$ to $\qA_1$.
    \item $\qA_1$ sends $m\in\Ms$ to the challenger.
    \item\label{step:enc_rnce_classical} The challenger generates $x\la\bit^n$, computes $(\pke.\vk_{i,\alpha},\pke.\qct_{i,\alpha})\la\allowbreak\PKE.\qEnc(\pke.\pk_{i,\alpha},m[i])$ for all $i\in[n]$ and $\alpha\in\bit$, and sends \allowbreak$(\{\pke.\qct_{i,\alpha}\}_{i\in[n],\alpha\in\bit},(x,\{\pke.\sk_{i,x[i]}\}_{i\in[n]}))$ to $\qA_1$.
    \item $\qA_1$ sends $\{\pke.\cert_{i,\alpha}\}_{i\in[n],\alpha\in\bit}$ to the challenger and its internal state to $\qA_2$.
    \item The challenger computes $\PKE.\Vrfy(\pke.\vk_{i,\alpha},\pke.\cert_{i,\alpha})$ for all $i\in[n]$ and $\alpha\in\bit$.
    If all results are $\top$, the challenger outputs $\top$ and sends $\{\pke.\sk_{i,\alpha}\}_{i\in[n],\alpha\in\bit}$ to $\qA_2$. Otherwise, the challenger outputs $\bot$ and sends $\bot$ to $\qA_2$.
    \item $\qA_2$ outputs $b'\in\bit$. 
    \item If the challenger outputs $\top$, then the output of the experiment is $b'$. Otherwise, the output of the experiment is $\bot$.
\end{enumerate}
\item[$\hyb_1$:] 
This is identical to $\hyb_{0}$ except that
the challenger generates $(\pke.\vk_{i,x[i]\oplus 1},\allowbreak\pke.\qct_{i,x[i]\oplus 1})\la\allowbreak\PKE.\qEnc(\pke.\pk_{i,x[i]\oplus 1}, \allowbreak m[i]\oplus 1)$ for all $i\in[n]$ instead of computing $(\pke.\vk_{i,x[i]\oplus 1},\allowbreak\pke.\qct_{i,x[i]\oplus 1})\la\PKE.\qEnc(\pke.\pk_{i,x[i]\oplus 1},m[i])$ for all $i\in[n]$.
\item[$\hyb_{2}$:]
This is identical to $\hyb_{1}$ except for the following three points.
\begin{enumerate}
\item The challenger generates $x^*\la\bit^n$ instead of generating $x\la\bit^n$.
\item For all $i\in[n]$, the challenger generates $(\pke.\vk_{i,x^*[i]},\allowbreak \pke.\qct_{i,x^*[i]})\la \PKE.\qEnc(\pke.\pk_{i,x^*[i]},0)$
and $(\pke.\vk_{i,x^*[i]\oplus 1},\pke.\qct_{i,x^*[i]\oplus 1})\la \PKE.\qEnc(\pke.\pk_{i,x^*[i]\oplus 1},1)$
instead of computing $(\pke.\vk_{i,x[i]},\pke.\qct_{i,x[i]})\allowbreak \la\PKE.\qEnc(\pke.\pk_{i,x[i]},m[i])$
and $(\pke.\vk_{i,x[i]\oplus 1},\pke.\qct_{i,x[i]\oplus 1})\la\PKE.\qEnc(\pke.\pk_{i,x[i]\oplus 1},m[i]\oplus 1)$.
\item The challenger sends $(\{\pke.\qct_{i,\alpha}\}_{i\in[n],\alpha\in\bit},(x^*\oplus m,\{\pke.\sk_{i,x^*[i]\oplus m[i]}\}_{i\in[n]}))$ to $\qA_1$
instead of sending $(\{\pke.\qct_{i,\alpha}\}_{i\in[n],\alpha\in\bit},(x,\{\pke.\sk_{i,x[i]}\}_{i\in[n]}))$ to $\qA_1$.
\end{enumerate}
\end{description}
It is clear that $\hyb_{0}$ is identical to $\expd{\Sigma,\qA}{cert}{ever}{rec}{nc}(\secp,0)$ and $\hyb_{2}$ is identical to $\expd{\Sigma,\qA}{cert}{ever}{rec}{nc}(\secp,1)$.
Hence, \cref{thm:ever_security_rnce_classic} easily follows from the following \cref{prop:hyb_0_hyb_1_rnce_classic,prop:hyb_1_hyb_2_rnce_classic}
(whose proof is given later.).
\end{proof}

\begin{proposition}\label{prop:hyb_0_hyb_1_rnce_classic}
If $\Sigma_{\mathsf{cepk}}$ is certified everlasting IND-CPA secure,
it holds that
$\abs{\Pr[\hyb_{0}=1]-\Pr[\hyb_{1}=1]}\leq\negl(\lambda)$.
\end{proposition}
\begin{proposition}\label{prop:hyb_1_hyb_2_rnce_classic}
$\abs{\Pr[\hyb_{1}=1]-\Pr[\hyb_{2}=1]}\leq\negl(\lambda)$.    
\end{proposition}

\begin{proof}[Proof of \cref{prop:hyb_0_hyb_1_rnce_classic}]

For the proof, we use \cref{lem:cut_and_choose_pke}.
We assume that $\abs{\Pr[\hyb_{0}=1]-\Pr[\hyb_{1}=1]}$ is non-negligible,
and construct an adversary $\qB$ that breaks the security experiment $\expc{\Sigma_{\mathsf{cepk}},\qB}{multi}{cert}{ever}(\secp,b)$ defined in~\cref{lem:cut_and_choose_pke}.
This contradicts the certified everlasting IND-CPA security of $\Sigma_{\mathsf{cepk}}$ from~\cref{lem:cut_and_choose_pke}.
Let us describe how $\qB$ works below.
\begin{enumerate}
    \item $\qB$ receives $\{\pke.\pk_{i,\alpha}\}_{i\in[n],\alpha\in\bit}$ from the challenger of
    $\expc{\Sigma_{\mathsf{cepk}},\qB}{multi}{cert}{ever}(\secp,b)$.
    \item $\qB$ sends $\{\pke.\pk_{i,\alpha}\}_{i\in[n],\alpha\in\bit}$ to $\qA_1$.
    \item $\qA_1$ chooses $m\in\Ms$ and sends $m$ to $\qB$.
    \item $\qB$ generates $x\la\bit^n$ and sends $(x,m[1],\cdots, m[n], m[1]\oplus 1,\cdots,m[n]\oplus 1)$ to the challenger of $\expc{\Sigma_{\mathsf{cepk}},\qB}{multi}{cert}{ever}(\secp,b)$.
    \item $\qB$ receives $(\{\pke.\sk_{i,x[i]}\}_{i\in[n]},\{\pke.\qct_{i,x[i]\oplus 1}\}_{i\in[n]})$ from the challenger of $\expc{\Sigma_{\mathsf{cepk}},\qB}{multi}{cert}{ever}(\secp,b)$.
    \item $\qB$ computes $(\{\pke.\vk_{i,x[i]}\}_{i\in[n]},\{\pke.\qct_{i,x[i]}\}_{i\in[n]})\la\PKE.\qEnc(\pke.\pk_{i,x[i]},m[i])$ for $i\in[n]$. 
    \item $\qB$ sends $(\{\pke.\qct_{i,\alpha}\}_{i\in[n],\alpha\in\bit},(x,\{\pke.\sk_{i,x[i]}\}_{i\in[n]}))$ to $\qA_1$.
    \item $\qA_1$ sends $\{\pke.\cert_{i,\alpha}\}_{i\in[n],\alpha\in\bit}$ to $\qB$ and the internal state to $\qA_2$.
    \item $\qB$ sends $\{\pke.\cert_{i,x[i]\oplus 1}\}_{i\in[n]}$ to the challenger, and receives $\{\pke.\sk_{i,x[i]\oplus 1}\}_{i\in[n]}$ or $\bot$.
    If $\qB$ receives $\bot$, it outputs $\bot$ and aborts.
    \item $\qB$ sends $\{\pke.\sk_{i,\alpha}\}_{i\in[n],\alpha\in\bit}$ to $\qA_2$.
    \item $\qA_2$ outputs $b'$.
    \item $\qB$ computes $\PKE.\Vrfy(\pke.\vk_{i,x[i]},\pke.\cert_{i,x[i]})$ for all $i\in[n]$.
    If all results are $\top$, $\qB$ outputs $b'$.
    Otherwise, $\qB$ outputs $\bot$.
\end{enumerate}
It is clear that $\Pr[1\la\qB \mid b=0]=\Pr[\hyb_{0}=1]$ and $\Pr[1\la\qB \mid b=1]=\Pr[\hyb_{1}=1]$.
By assumption,\\ $\abs{\Pr[\hyb_{0}=1]-\Pr[\hyb_{1}=1]}$ is non-negligible.
Therefore, $\abs{\Pr[1\la\qB \mid b=0]-\Pr[1\la\qB \mid b=1]}$ is also non-negligible,
which contradicts the certified everlasting IND-CPA security of $\Sigma_{\mathsf{cepk}}$ from~\cref{lem:cut_and_choose_pke}.
\end{proof}

\begin{proof}[Proof of \cref{prop:hyb_1_hyb_2_rnce_classic}]
It is obvious that the joint probability distribution that $\qA_1$ receives $(\{\pke.\qct_{i,\alpha}\}_{i\in[n],\alpha\in\bit},\allowbreak(x,\{\pke.\sk_{i,x[i]}\}_{i\in[n]}))$ in $\hyb_{1}$ is identical to the joint probability distribution that $\qA_1$ receives $(\{\pke.\qct_{i,\alpha}\}_{i\in[n],\alpha\in\bit},\allowbreak (x^*\oplus m,\{\pke.\sk_{i,x^*[i]\oplus m[i]}\}_{i\in[n]}))$ in $\hyb_{2}$.
Hence, \cref{prop:hyb_1_hyb_2_rnce_classic} follows.
\end{proof}

We use the following lemma for the proof of \cref{thm:ever_security_rnce_classic} and \cref{thm:garble_ever_security}.
The proof is shown with the standard hybrid argument.
It is also easy to see that a similar lemma holds for IND-CPA security.
\begin{lemma}\label{lem:cut_and_choose_pke}
Let $s$ be some polynomial of the security parameter $\lambda$.
Let $\Sigma\seteq(\keygen,\qEnc,\qDec,\qDelete,\Vrfy)$ be a certified everlasting secure PKE scheme.
Let us consider the following security experiment 
$\expc{\Sigma,\qA}{multi}{cert}{ever}(\secp,b)$ against $\qA$ consisting of any QPT adversary $\qA_1$ and any unbounded adversary $\qA_2$.
\begin{enumerate}
    \item The challenge generates $(\pk_{i,\alpha},\sk_{i,\alpha})\la\keygen(1^\secp)$ for all $i\in[s]$ and $\alpha\in\bit$,
    and sends $\{\pk_{i,\alpha}\}_{i\in[s],\alpha\in\bit}$ to $\qA_1$.
    \item $\qA_1$ chooses $f\in\bit^s$ and $(m_0[1],m_0[2],\cdots, m_0[s],m_1[1],m_1[2], \cdots, m_1[s]) \in\Ms^{2s}$,
    and sends $(f,m_0[1],m_0[2],\\ \cdots ,m_0[s],m_1[1],m_1[2],\cdots, m_1[s] )$ to the challenger.
    \item The challenger computes 
    $
    (\vk_{i,f[i]\oplus 1},\qct_{i,f[i]\oplus 1})\la \qEnc(\pk_{i,f[i]\oplus 1},m_b[i])$ for all $i\in[s]$,
    and sends $(\{\sk_{i,f[i]}\}_{i\in [s]},\\ \{\qct_{i,f[i]\oplus 1}\}_{i\in[s]})$ to $\qA_1$.
    \item At some point, $\qA_1$ sends $\{\cert_{i,f[i]\oplus 1}\}_{i\in[s]}$ to the challenger,
    and sends its internal state to $\qA_2$.
    \item The challenger computes $\Vrfy(\vk_{i,f[i]\oplus 1},\cert_{i,f[i]\oplus 1})$ for every $i\in[s]$.
    If all results are $\top$, the challenger outputs $\top$, and sends $\{\sk_{i,f[i]\oplus 1}\}_{i\in [s]}$ to $\qA_2$.
    Otherwise, the challenger outputs $\bot$, and sends $\bot$ to $\qA_2$.
    \item $\qA_2$ outputs $b'$.
    \item If the challenger outputs $\top$, then the output of the experiment is $b'$. Otherwise, the output of the experiment is $\bot$.
\end{enumerate}
If the $\Sigma$ satisfies the certified everlasting IND-CPA security,
\begin{align}
\advc{\Sigma,\qA}{multi}{cert}{ever}(\secp)
\seteq \abs{\Pr[ \expc{\Sigma,\qA}{multi}{cert}{ever}(\secp, 0)=1] - \Pr[ \expc{\Sigma,\qA}{multi}{cert}{ever}(\secp, 1)=1] }\leq \negl(\secp)
\end{align}
for any QPT adversary $\qA_1$ and any unbounded adversary $\qA_2$.
\end{lemma}

\begin{proof}[Proof of \cref{lem:cut_and_choose_pke}]
Let us consider the following hybrids for $j\in\{0,1,\cdots ,s\}$.
\begin{description}
\item[$\hyb_{j}$:] $ $
\begin{enumerate}
    \item The challenger generates $(\pk_{i,\alpha},\sk_{i,\alpha})\la\keygen(1^\secp)$ for every $i\in[s]$ and $\alpha\in\bit$, and sends $\{\pk_{i,\alpha}\}_{i\in[s],\alpha\in\bit}$ to $\qA_1$.
    \item $\qA_1$ chooses $f\in\bit^s$ and $(m_0[1],m_0[2],\cdots, m_0[s],m_1[1],m_1[2],\cdots, m_1[s]) \in\Ms^{2s}$,
    and sends $(f,m_0[1],m_0[2],\cdots, m_0[s],m_1[1],m_1[2],\cdots, m_1[s] )$ to the challenger.
    \item The challenger computes 
    \begin{align}
    (\vk_{i,f[i]\oplus 1},\qct_{i,f[i]\oplus 1})\la \qEnc(\pk_{i,f[i]\oplus 1},m_1[i])
    \end{align}
    for $i\in[j]$ and 
    \begin{align}
    (\vk_{i,f[i]\oplus 1},\qct_{i,f[i]\oplus 1})\la \qEnc(\pk_{i,f[i]\oplus 1},m_0[i])
    \end{align}
    for $i\in\{j+1,j+2,\cdots,s\}$,
    and sends $(\{\sk_{i,f[i]}\}_{i\in [s]}, \{\qct_{i,f[i]\oplus 1}\}_{i\in[s]})$ to $\qA_1$.
    \item At some point, $\qA_1$ sends $\{\cert_{i,f[i]\oplus 1}\}_{i\in[s]}$ to the challenger,
    and sends its internal state to $\qA_2$.
    \item The challenger computes $\Vrfy(\vk_{i,f[i]\oplus 1},\cert_{i,f[i]\oplus 1})$ for every $i\in[s]$.
    If all results are $\top$, the challenger outputs $\top$, and sends $\{\sk_{i,f[i]\oplus 1}\}_{i\in [s]}$ to $\qA_2$.
    Otherwise, the challenger outputs $\bot$, and sends $\bot$ to $\qA_2$.
    \item $\qA_2$ outputs $b'$.
    \item If the challenger outputs $\top$, then the output of the experiment is $b'$. Otherwise, the output of the experiment is $\bot$.
\end{enumerate}
\end{description}
It is clear that $\Pr[\hyb_{0}=1]=\Pr[\expc{\Sigma,\qA}{multi}{cert}{ever}(\secp,0)=1]$ and $\Pr[\hyb_{s}=1]=\Pr[\expc{\Sigma,\qA}{multi}{cert}{ever}(\secp,1)=1]$.
Furthermore, we can show 
\begin{align}
    \abs{\Pr[\hyb_{j}=1]-\Pr[\hyb_{j+1}=1]}\leq\negl(\lambda)
\end{align}
for each $j\in\{0,1,\cdots,s-1\}$.
(Its proof is given below.) From these facts, we obtain \cref{lem:cut_and_choose_pke}.

Let us show the remaining one. To show it, let us assume that $\abs{\Pr[\hyb_{j}=1]-\Pr[\hyb_{j+1}=1]}$ is non-negligible.
Then, we can construct an adversary $\qB$ that can break the certified everlasting IND-CPA security of $\Sigma$ as follows. 

\begin{enumerate}
    \item $\qB$ receives $\pk$ from the challenger of $\expd{\Sigma,\qB}{cert}{ever}{ind}{cpa}(\secp,b)$.
    \item $\qB$ generates $\beta\la\bit$ and sets $\pk_{j+1,\beta}\seteq \pk$. 
    \item $\qB$ generates $(\pk_{i,\alpha},\sk_{i,\alpha})\la\keygen(1^\secp)$ for $i\in\{1,\cdots,j,j+2,\cdots, s\}$ and $\alpha\in\bit$, and $(\pk_{j+1,\beta\oplus 1},\sk_{j+1,\beta\oplus 1})\la\keygen(1^\secp)$.
    \item $\qB$ sends $\{\pk_{i,\alpha}\}_{i\in[s],\alpha\in\bit}$ to $\qA_1$.
    \item $\qA_1$ chooses $f\in\bit^s$ and $(m_0[1],m_0[2],\cdots, m_0[s],m_1[1],m_1[2],\cdots, m_1[s]) \in\Ms^{2s}$,
    and sends $(f,m_0[1],m_0[2],\\ \cdots, m_0[s],m_1[1],m_1[2],\cdots, m_1[s] )$ to the challenger.
    \item If $f[j+1]=\beta$, $\qB$ aborts the experiment, and outputs $\bot$.
    \item $\qB$ computes
    \begin{align}
        (\vk_{i,f[i]\oplus 1},\qct_{i,f[i]\oplus 1})\la \qEnc(\pk_{i,f[i]\oplus 1},m_1[i])
    \end{align}
    for $i\in[j]$ and
    \begin{align}
        (\vk_{i,f[i]\oplus 1},\qct_{i,f[i]\oplus 1})\la \qEnc(\pk_{i,f[i]\oplus 1},m_0[i])
    \end{align}
    for $i\in\{j+2,\cdots, s\}$.
    \item $\qB$ sends $(m_0[j+1],m_1[j+1])$ to the challenger of $\expd{\Sigma,\qB}{cert}{ever}{ind}{cpa}(\secp,b)$.
    The challenger computes $(\vk_{j+1,f[j+1]\oplus1},\qct_{j+1,f[j+1]\oplus 1})\la \qEnc(\pk_{j+1,f[j+1]\oplus 1},m_b[j+1])$ and sends $\qct_{j+1,f[j+1]\oplus 1}$ to $\qB$.
    \item $\qB$ sends $(\{\sk_{i,f[i]}\}_{i\in [s]}, \{\qct_{i,f[i]\oplus 1}\}_{i\in[s]})$ to $\qA_1$.
    \item $\qA_1$ sends $\{\cert_i\}_{i\in[s]}$ to $\qB$, and sends its internal state to $\qA_2$.
    \item $\qB$ sends $\cert_{j+1}$ to the challenger, and receives $\sk_{j+1,f[j+1]\oplus1}$ or $\bot$ from the challenger.
    If $\qB$ receives $\bot$ from the challenger, it outputs $\bot$ and aborts.
    \item $\qB$ sends $\{\sk_{i,f[i]\oplus 1}\}_{i\in[s]}$ to $\qA_2$.
    \item $\qA_2$ outputs $b'$.
    \item $\qB$ computes $\Vrfy$ for all $\cert_i$, and outputs $b'$ if all results are $\top$. Otherwise, $\qB$ outputs $\bot$.
\end{enumerate}
Since $\pk_{j+1,\beta}$ and $\pk_{j+1,\beta\oplus 1}$ are identically distributed, 
it holds that
$ \Pr[f[j+1]=\beta]=\Pr[f[j+1]=\beta\oplus 1]=\frac{1}{2}$.
If $b=0$ and $f[j+1]=\beta\oplus 1$, $\qB$ simulates $\hyb_{j}$.
Therefore, we have
\begin{align}
\Pr[1\la\qB \mid b=0]&=\Pr[1\la\qB \wedge f[j+1]=\beta\oplus 1 \mid b=0]\\
                &=\Pr[1\la\qB \mid b=0, f[j+1]=\beta \oplus 1]\cdot\Pr[f[j+1]=\beta \oplus 1]\\
                &=\frac{1}{2}\Pr[\hyb_{j}=1].
\end{align}
If $b=1$ and $f[j+1]=\beta\oplus 1$, $\qB$ simulates $\hyb_{j+1}$.
Similarly, we have 
$
\Pr[1\la\qB \mid b=1]=\frac{1}{2}\Pr[\hyb_{j+1}=1]
$.
By assumption, $\abs{\Pr[\hyb_{j}=1]-\Pr[\hyb_{j+1}=1]}$ is non-negligible, and therefore
$\abs{\Pr[1\la\qB \mid b=0]-\Pr[1\la\qB \mid b=1]}$ is non-negligible, which contradicts the certified everlasting IND-CPA security of $\Sigma$.
\end{proof}


\section{Garbling Scheme with Certified Everlasting Deletion}\label{sec:garbling_scheme}
In \cref{sec:def_garbled}, we define a garbling scheme with certified everlasting deletion.
In \cref{sec:const_garbling}, we construct a certified everlasting secure garbling scheme from a certified everlasting secure SKE scheme.

\subsection{Definition}\label{sec:def_garbled}
We define a garbling scheme with certified everlasting deletion below. An important difference from a standard classical garbling scheme is that the garbled circuit $\wtl{\qC}$ (i.e., an output of $\Garble$) 
is a quantum state.
\begin{definition}[Garbling Scheme with Certified Everlasting Deletion (Syntax)]\label{def:cert_ever_garbled}
Let $\lambda$ be a security parameter and $p,q,r$ and $s$ be polynomials. 
Let $\Cs_n$ be a family of circuits that take $n$-bit inputs.
A garbling scheme with certified everlasting deletion is a tuple of algorithms $\Sigma=(\Setup,\qGarble,\qEval,\qDelete,\Vrfy)$
with label space $\mathcal{L}\seteq\bit^{p(\lambda)}$, garbled circuit space $\mathcal{C}\seteq \cQ^{\otimes q(\lambda)}$,
verification key space $\mathcal{VK}\seteq \bit^{r(\lambda)}$ and deletion certificate space $\mathcal{D}\seteq\cQ^{\otimes s(\lambda)}$.
\begin{description}
\item[$\Setup(1^\secp)\ra\{L_{i,\alpha}\}_{i\in[n],\alpha\in\bit}$:]
The sampling algorithm takes a security parameter $1^\lambda$ as input, and outputs $2n$ labels $\{L_{i,\alpha}\}_{i\in[n],\alpha\in\bit}$ with $L_{i,\alpha}\in\mathcal{L}$ for each $i\in[n]$ and $\alpha\in\bit$.
\item[$\qGarble(1^{\secp},C,\{L_{i,\alpha}\}_{i\in[n],\alpha\in\{0,1\}})\ra (\wtl{\qC},\vk)$:] 
The garbling algorithm takes $1^\lambda$, a circuit $C\in\Cs_n$ and $2n$ labels $\{L_{i,\alpha}\}_{i\in[n],\alpha\in\bit}$ as input, and outputs a garbled circuit $\wtl{\qC}\in\mathcal{C}$ and a verification key $\vk\in\mathcal{VK}$.
\item[$\qEval(\wtl{\qC},\{L_{i,x_i}\}_{i\in[n]})\ra y$:]
The evaluation algorithm takes $\wtl{\qC}$ and $n$ labels $\{L_{i,x_i}\}_{i\in[n]}$ where $x_i\in\{0,1\}$ as input,
and outputs $y$.
\item[$\qDelete(\wtl{\qC})\ra \cert$:]
The deletion algorithm takes $\wtl{\qC}$ as input, and outputs a certificate $\cert\in\mathcal{D}$.
\item[$\Vrfy(\vk,\cert)\ra \top$ or $\bot$:]
The verification algorithm takes $\vk$ and $\cert$ as input, and outputs $\top$ or $\bot$.
\end{description}
\end{definition}

We require that a garbling scheme with certified everlasting deletion satisfies correctness defined below.
\begin{definition}[Correctness for Garbling Scheme with Certified Everlasting Deletion]\label{def:correctness_cert_ever_garbled}
There are two types of correctness, namely, evaluation correctness and verification correctness.

\paragraph{Evaluation Correctness:}
There exists a negligible function $\negl$ such that for any $\secp\in\N$, $C\in\Cs_n$ and $x\in \bit^n$,
\begin{align}
\Pr\left[
y\neq C(x)
\ \middle |
\begin{array}{ll}
\{L_{i,\alpha}\}_{i\in[n],\alpha\in\bit}\la\Setup(1^\secp)\\
(\wtl{\qC},\vk)\lrun \qGarble(1^\secp,C,\{L_{i,\alpha}\}_{i\in[n],\alpha\in\{0,1\}})\\
y\la\qEval(\wtl{\qC},\{L_{i,x_i}\}_{i\in[n]})
\end{array}
\right] 
\leq\negl(\secp).
\end{align}

\paragraph{Verification Correctness:} There exists a negligible function $\negl$ such that for any $\secp\in\N$,
\begin{align}
\Pr\left[
\Vrfy(\vk,\cert)=\bot
\ \middle |
\begin{array}{ll}
\{L_{i,\alpha}\}_{i\in[n],\alpha\in\bit}\la\Setup(1^\secp)\\
(\wtl{\qC},\vk)\la\qGarble(1^{\secp},C,\{L_{i,\alpha}\}_{i\in[n],\alpha\in\{0,1\}}) \\
\cert\la\qDelete(\wtl{\qC})
\end{array}
\right] 
\leq
\negl(\secp).
\end{align}
\end{definition}

Minimum requirements for correctness are evaluation correctness and verification correctness.
However, we also require verification correctness with QOTP in this work because we need it for the construction of FE in \cref{sec:const_fe_adapt}.

\begin{definition}[Verification Correctness with QOTP]\label{def:cd_gc_ver_correctness_QOTP}
There exists a negligible function $\negl$ and a PPT algorithm $\Modify$ such that for any $\secp\in \N$, 
\begin{align}
\Pr\left[
\Vrfy(\vk,\cert^*)=\bot
\ \middle |
\begin{array}{ll}
\{L_{i,\alpha}\}_{i\in[n],\alpha\in\bit}\la\Setup(1^\secp)\\
(\wtl{\qC},\vk)\la\qGarble(1^{\secp},C,\{L_{i,\alpha}\}_{i\in[n],\alpha\in\{0,1\}}) \\
a,b\la\bit^{q(\lambda)}\\
\wtl{\cert} \lrun \qDelete(Z^bX^a\wtl{\qC} X^aZ^b)\\
\cert^*\lrun \Modify(a,b,\wtl{\cert})  
\end{array}
\right] 
\leq
\negl(\secp).
\end{align}
\end{definition}

As security, we consider two definitions, 
\cref{def:sel_sec_ever_garb} 
and
\cref{def:ever_sec_ever_garb} given below.
The former is just the standard selective security 
and the latter is the certified everlasting security
that we newly define in this paper.
Roughly, the everlasting security guarantees that any QPT adversary with the garbled circuit $\wtl{\qC}$ and the labels $\{L_{i,x[i]}\}_{i\in[n]}$ cannot obtain any information
beyond $C(x)$ even if it becomes  computationally unbounded after it issues a valid certificate.

\begin{definition}[Selective Security for Garbling Scheme with Certified Everlasting Deletion]\label{def:sel_sec_ever_garb}
Let $\Sigma=(\Setup,\qGarble,\allowbreak \qEval, \qDelete,\Vrfy)$ be a garbling scheme with certified everlasting deletion.
We consider the following security experiment $\expb{\Sigma,\qA}{sel}{gbl}(1^\secp,b)$ against a QPT adversary $\qA$. Let $\qSim$ be a QPT algorithm.
\begin{enumerate}
    \item $\qA$ sends a circuit $C\in\Cs_n$ and an input $x\in\bit^n$ to the challenger.
    \item The challenger computes $\{L_{i,\alpha}\}_{i\in[n],\alpha\in\bit}\la\Setup(1^\secp)$.
    \item If $b=0$, the challenger computes $(\wtl{\qC},\vk)\la\qGarble(1^\secp,C,\{L_{i,\alpha}\}_{i\in[n],\alpha\in\bit})$, and returns $(\wtl{\qC},\{L_{i,x_i}\}_{i\in[n]})$ to $\qA$.
    If $b=1$, the challenger computes $\wtl{\qC}\la\qSim(1^\secp,1^{|C|},C(x),\{L_{i,x_i}\}_{i\in[n]})$,
    and returns $(\widetilde{C},\{L_{i,x_i}\}_{i\in[n]})$ to $\qA$.
    \item $\qA$ outputs $b'\in\bit$. The experiment outputs $b'$.
\end{enumerate}
We say that $\Sigma$ is selective secure if there exists a QPT simulator $\qSim$ such that for any QPT adversary $\qA$ it holds that
\begin{align}
    \advb{\Sigma,\qA}{sel}{gbl}(\secp)\seteq \abs{\Pr[\expb{\Sigma,\qA}{sel}{gbl}(1^\secp,0)=1]-\Pr[\expb{\Sigma,\qA}{sel}{gbl}(1^\secp,1)=1]}\leq\negl(\secp).
\end{align}
\end{definition}

\begin{definition}[Selective Certified Everlasting Security for Garbling Scheme]\label{def:ever_sec_ever_garb}
Let $\Sigma=(\Setup,\qGarble,\qEval,\qDelete,\allowbreak \Vrfy)$ be a garbling scheme with certified everlasting deletion.
We consider the following security experiment $\expd{\qA,\Sigma}{cert}{ever}{sel}{gbl}(1^\secp,b)$ against a QPT adversary $\qA_1$ and an unbounded adversary $\qA_2$. Let $\qSim$ be a QPT algorithm.
\begin{enumerate}
    \item $\qA_1$ sends a circuit $C\in\Cs_n$ and an input $x\in\bit^n$ to the challenger.
    \item The challenger computes $\{L_{i,\alpha}\}_{i\in[n],\alpha\in\bit}\la\Setup(1^\secp)$.
    \item If $b=0$, the challenger computes $(\wtl{\qC},\vk)\la\qGarble(1^\secp,C,\{L_{i,\alpha}\}_{i\in[n],\alpha\in\bit})$, and returns $(\wtl{\qC},\{L_{i,x_i}\}_{i\in[n]})$ to $\qA_1$.
    If $b=1$, the challenger computes $(\wtl{\qC},\vk)\la\qSim(1^\secp,1^{|C|},C(x),\{L_{i,x_i}\}_{i\in[n]})$,
    and returns $(\wtl{\qC},\{L_{i,x_i}\}_{i\in[n]})$ to $\qA_1$.
    \item At some point, $\qA_1$ sends $\cert$ to the challenger, and sends the internal state to $\qA_2$.
    \item The challenger computes $\Vrfy(\vk,\cert)$.
    If the output is $\bot$, then the challenger outputs $\bot$, and sends $\bot$ to $\qA_2$.
    Otherwise, the challenger outputs $\top$, and sends $\top$ to $\qA_2$.
    \item $\qA_2$ outputs $b'\in\bit$.
    \item If the challenger outputs $\top$, then the output of the experiment is $b'$. 
    Otherwise, the output of the experiment is $\bot$.
\end{enumerate}
We say that $\Sigma$ is selective certified everlasting secure if there exists a QPT simulator $\qSim$ such that for any QPT $\qA_1$ and any unbounded $\qA_2$ it holds that
\begin{align}
    \advd{\Sigma,\qA}{cert}{ever}{sel}{gbl}(\secp)\seteq
    \abs{\Pr[\expd{\Sigma,\qA}{cert}{ever}{sel}{gbl}(1^\secp,0)=1]-\Pr[\expd{\Sigma,\qA}{cert}{ever}{sel}{gbl}(1^\secp,1)=1]}\leq\negl(\secp).
\end{align}
\end{definition}

\subsection{Construction}\label{sec:const_garbling}
In this section, we construct a certified everlasting secure garbling scheme from a certified everlasting secure SKE scheme~(\cref{def:cert_ever_ske}).
Our construction is similar to Yao's construction of a standard garbling scheme \cite{FOCS:Yao86}, but there are two important differences.
First, we use a certified everlasting secure SKE scheme instead of a standard SKE scheme.
Second, we use XOR secret sharing, although \cite{FOCS:Yao86} used double encryption. The reason why we cannot use double encryption is
that our certified everlasting SKE scheme has quantum ciphertext and classical plaintext.

Before introducing our construction, let us quickly review notations for circuits.
Let $C$ be a boolean circuit. 
A boolean circuit $C$ consists of gates, $\mathsf{gate}_1,\mathsf{gate}_2,\cdots,\mathsf{gate}_q$,
where $q$ is the number of gates in the circuit.
Here, $\mathsf{gate}_i\seteq (g,w_a,w_b,w_c)$, where $g:\bit^2\ra\bit$ is a function, $w_a$, $w_b$ are the incoming wires, and $w_c$ is the outgoing wire.
(The number of outgoing wires is not necessarily one. There can be many outgoing wires, but we use the same label $w_c$ for all outgoing wires.) 
We say $C$ is leveled if each gate has an associated level and any gate at level $\ell$ has incoming wires only from gates at level $\ell-1$ and outgoing wires only to gates at level $\ell+1$. 
Let $\mathsf{out}_1,\mathsf{out}_2,\cdots,\mathsf{out}_m$ be the $m$ output wires.
For any $x\in\bit^n$, $C(x)$ is the output of the circuit $C$ on input $x$.
We consider that $\mathsf{gate_1},\mathsf{gate_2},\cdots,\mathsf{gate}_q$ are arranged in the ascending order of the level.

\paragraph{Our certified secure everlasting garbling scheme.}
We construct a certified everlasting secure garbling scheme $\Sigma_{\mathsf{cegc}}=(\Setup,\qGarble,\allowbreak \qEval,\qDelete,\Vrfy)$
from a certified everlasting secure SKE scheme
$\Sigma_{\mathsf{cesk}}=\SKE.(\keygen,\qEnc,\qDec,\allowbreak \qDelete,\Vrfy)$~(\cref{def:cert_ever_ske}).
Let $\mathcal{K}$ be the key space of $\Sigma_{\mathsf{cesk}}$.
Let $C$ be a leveled boolean circuit.
Let $n,m,q$, and $p$ be the input size,
the output size,
the number of gates,
and
the total number of wires
of $C$, 
respectively.

\begin{description}
\item[$\Setup(1^\secp)$:] $ $
\begin{itemize}
    \item For each $i\in[n]$ and $\sigma\in\bit$, generate $\ske.\sk_{i}^\sigma\la\SKE.\keygen(1^\secp)$.
    \item Output $\{L_{i,\sigma}\}_{i\in[n],\sigma\in\bit}\seteq\{ \ske.\sk_i^\sigma\}_{i\in [n],\sigma\in\bit}$.
\end{itemize}
\item[$\qGarble(1^\secp,C,\{L_{i,\sigma}\}_{i\in[n],\sigma\in\bit})$:] $ $
\begin{itemize}
    \item For each $i\in\{n+1,\cdots, p\}$ and $\sigma\in\bit$, generate $\ske.\sk_{i}^\sigma\la\SKE.\keygen(1^\secp)$.
    \item For each $i\in [q]$, 
    compute 
    \begin{align}
    (\vk_i,\qggate_i)\la \qGateGrbl(\mathsf{gate}_i, \{\ske.\sk_a^\sigma,\ske.\sk_b^\sigma,\ske.\sk_c^\sigma\}_{\sigma\in\bit}), 
    \end{align}
    where $\mathsf{gate}_i=(g,w_a,w_b,w_c)$ and $\qGateGrbl$ is described in Fig~\ref{fig:Garble_circuit}.
    \item For each $i\in[m]$,
    set $\widetilde{d}_i\seteq[(\ske.\sk_{\mathsf{out}_i}^0, 0),(\ske.\sk_{\mathsf{out}_i}^1, 1)]$.
    \item Output $\widetilde{\qC}\seteq (\{\qggate_i\}_{i\in [q]},\{\widetilde{d}_i\}_{i\in [m]})$
    and 
    $\vk\seteq\{\vk_i\}_{i\in[q]}$.
\end{itemize}
\item[$\Eval(\widetilde{C},\{L_{i,x_i}\}_{i\in [n]})$:] $ $
\begin{itemize}
    \item Parse $\widetilde{C}= (\{\qggate_i\}_{i\in [q]},\{\widetilde{d}_i\}_{i\in [m]})$ and $\{L_{i,x_i}\}_{i\in[n]}=\{\ske.\sk'_i\}_{i\in [n]}$.
    \item For each $i\in[q]$, 
    compute $\ske.\sk'_{c}\la\qGateEval(\qggate_i, \ske.\sk'_a,\ske.\sk'_b)$ in the ascending order of the level, where $\qGateEval$ is described in Fig~\ref{fig:Gate_Eval}.
    If $\ske.\sk'_{c}=\bot$, output $\bot$ and abort.
    \item For each $i\in[m]$, set $y[i]=\sigma$ if $\ske.\sk'_{\mathsf{out}_i}=\ske.\sk_{\mathsf{out}_i}^\sigma$. Otherwise, set $y[i]=\bot$, and abort.
    \item Output $y\seteq y[1]||y[2]||\cdots|| y[m]$.
\end{itemize}
\item[$\Delete(\widetilde{\qC})$:] $ $
\begin{itemize}
    \item Parse $\widetilde{\qC}= (\{\qggate_i\}_{i\in [q]},\{\widetilde{d}_i\}_{i\in [m]})$.
    \item For each $i\in[q]$, compute $\cert_i\la\qGateDel(\qggate_i)$, where $\qGateDel$ is described in Fig~\ref{fig:Gate_Del}.
    \item Output $\cert\seteq\{\cert_i\}_{i\in [q]}$.
\end{itemize}
\item[$\Vrfy(\vk,\cert)$:] $ $
\begin{itemize}
    \item Parse
    $\vk=\{\vk_i\}_{i\in[q]}$
    and $\cert= \{\cert_i\}_{i\in[q]}$.
    \item For each $i\in[q]$,
    compute $\bot/\top\la\mathsf{GateVrfy}(\vk_i,\cert_i)$, where $\mathsf{GateVrfy}$ is described in Fig~\ref{fig:Gate_Verify}.
    \item If $\top\la\mathsf{GateVrfy}(\vk_i,\cert_i)$ for all $i\in[q]$, then output $\top$.
    Otherwise, output $\bot$. 
\end{itemize}
\end{description}

\protocol{Gate Garbling Circuit $\qGateGrbl$
}
{The description of $\qGateGrbl$}
{fig:Garble_circuit}
{
\begin{description}
\item[Input:] $\mathsf{gate}_i,\{\ske.\sk_a^\sigma,\ske.\sk_b^\sigma,\ske.\sk_c^\sigma\}_{\sigma\in\bit}$.
\item[Output:] $\qggate_i$ and $\vk_i$.
\end{description}
\begin{enumerate}
\item Parse $\mathsf{gate}_i=(g,w_a,w_b,w_c)$.
\item Sample $\gamma_i\la\mathsf{S}_4$.\footnote{$\mathsf{S}_4$ is the symmetric group of order $4$.}
\item For each $\sigma_a,\sigma_b\in\bit$, sample $p^{\sigma_{a},\sigma_b}_{c}\la\Ks$.
\item For each $\sigma_a,\sigma_b\in\bit$, compute $(\ske.\vk_{a}^{\sigma_a,\sigma_b},\ske.\qct_{a}^{\sigma_a,\sigma_b})\la\SKE.\qEnc(\ske.\sk_a^{\sigma_a},p^{\sigma_a,\sigma_b}_{c})$
and 
$(\ske.\vk_{b}^{\sigma_a,\sigma_b},\ske.\qct_b^{\sigma_a,\sigma_b})\la
\SKE.\qEnc(\ske.\sk_b^{\sigma_b},p^{\sigma_a,\sigma_b}_{c}\oplus \ske.\sk_c^{g(\sigma_a,\sigma_b)})$.
\item Output $\qggate_i\seteq \{\ske.\qct_a^{\sigma_a,\sigma_b},\ske.\qct_b^{\sigma_a,\sigma_b}\}_{\sigma_a,\sigma_b\in\bit}$ in the permutated order of $\gamma_i$ and 
$\vk_i\seteq\{\ske.\vk_{a}^{\sigma_a,\sigma_b},\ske.\vk_{b}^{\sigma_a,\sigma_b}\}_{\sigma_a,\sigma_b\in\bit}$ in the permutated order of $\gamma_i$.
\end{enumerate}
}

\protocol{Gate Evaluating Circuit $\qGateEval$
}
{The description of $\qGateEval$}
{fig:Gate_Eval}
{
\begin{description}
\item[Input:] A garbled gate $\qggate_i$ and $(\ske.\sk'_a,\ske.\sk'_b)$.
\item[Output:] $\ske.\sk_c$ or $\bot$.
\end{description}
\begin{enumerate}
\item Parse $\qggate_i= \{\ske.\qct_a^{\sigma_a,\sigma_b},\ske.\qct_b^{\sigma_a,\sigma_b}\}_{\sigma_a,\sigma_b\in\bit}$.
\item For each $\sigma_a,\sigma_b\in\bit$, compute 
      $q_a^{\sigma_a,\sigma_b}\la\SKE.\qDec(\ske.\sk'_a,\ske.\qct_a^{\sigma_a,\sigma_b})$ and
      $q_b^{\sigma_a,\sigma_b}\la\SKE.\qDec(\ske.\sk'_b,\ske.\qct_b^{\sigma_a,\sigma_b})$.
\item If there exists a unique pair $(\sigma_a,\sigma_b)\in\bit^2$ such that $q_a^{\sigma_a,\sigma_b}\neq \bot$ and $q_b^{\sigma_a,\sigma_b}\neq\bot$, then
      compute $\ske.\sk_{c}^{'\sigma_a,\sigma_b}\seteq q_a^{\sigma_a,\sigma_b}\oplus q_b^{\sigma_a,\sigma_b}$ and output $\ske.\sk'_c\seteq \ske.\sk_c^{'\sigma_a,\sigma_b}$.
      Otherwise, output $\ske.\sk'_c\seteq\bot$.
\end{enumerate}
}

\protocol{Gate Deletion Circuit $\qGateDel$
}
{The description of $\qGateDel$}
{fig:Gate_Del}
{
\begin{description}
\item[Input:] A garbled gate $\qggate_i$.
\item[Output:] $\cert_i$
\end{description}
\begin{enumerate}
\item Parse $\qggate_i= \{\ske.\qct_a^{\sigma_a,\sigma_b},\ske.\qct_b^{\sigma_a,\sigma_b}\}_{\sigma_a,\sigma_b\in\bit}$.
\item For each $\sigma_a,\sigma_b\in\bit$, compute $\ske.\cert_a^{\sigma_a,\sigma_b}\la\SKE.\qDelete(\ske.\qct_a^{\sigma_a,\sigma_b})$.
\item For each $\sigma_a,\sigma_b\in\bit$, compute $\ske.\cert_b^{\sigma_a,\sigma_b}\la\SKE.\qDelete(\ske.\qct_b^{\sigma_a,\sigma_b})$.
\item Output $\cert_i\seteq \{\ske.\cert_a^{\sigma_a,\sigma_b},\ske.\cert_b^{\sigma_a,\sigma_b}\}_{\sigma_a,\sigma_b\in\bit}$.
\end{enumerate}
}

\protocol{Gate Verification Circuit $\mathsf{GateVrfy}$
}
{The description of $\mathsf{GateVrfy}$}
{fig:Gate_Verify}
{
\begin{description}
\item[Input:] $\vk_i$ and $\cert_i$.
\item[Output:] $\top$ or $\bot$.
\end{description}
\begin{enumerate}
\item Parse
$\vk_i=\{\ske.\vk_{a}^{\sigma_a,\sigma_b},\ske.\vk_{b}^{\sigma_a,\sigma_b}\}_{\sigma_a,\sigma_b\in\bit}$ and
$\cert_i= \{\ske.\cert_a^{\sigma_a,\sigma_b},\ske.\cert_b^{\sigma_a,\sigma_b}\}_{\sigma_a,\sigma_b\in\bit}$.
\item For each $\sigma_a,\sigma_b\in\bit$, compute $\top/\bot\la\SKE.\Vrfy(\ske.\vk_{a}^{\sigma_a,\sigma_b},\ske.\cert_a^{\sigma_a,\sigma_b})$.
\item For each $\sigma_a,\sigma_b\in\bit$, compute $\top/\bot\la\SKE.\Vrfy(\ske.\vk_b^{\sigma_a,\sigma_b},\ske.\cert_b^{\sigma_a,\sigma_b})$.
\item If all the outputs are $\top$, then output $\top$.
Otherwise, output  $\bot$.
\end{enumerate}
}

\paragraph{Correctness:}
Correctness easily follows from that of $\Sigma_{\mathsf{cesk}}$.

\paragraph{Security:}
The following two theorems hold.
\begin{theorem}\label{thm:garble_comp_seucirty}
If $\Sigma_{\mathsf{cesk}}$ satisfies the IND-CPA security~(\cref{def:IND-CPA_security_cert_ever_ske}), 
$\Sigma_{\mathsf{cegc}}$ satisfies the selective security~(\cref{def:sel_sec_ever_garb}).
\end{theorem}
Its proof is similar to that of \cref{thm:garble_ever_security},
and therefore we omit it.

\begin{theorem}\label{thm:garble_ever_security}
If $\Sigma_{\mathsf{cesk}}$ satisfies the certified everlasting IND-CPA security~(\cref{def:cert_ever_security_cert_ever_ske}),
$\Sigma_{\mathsf{cegc}}$ satisfies the selective certified everlasting security~(\cref{def:ever_sec_ever_garb}).
\end{theorem}

Let $\widehat{\mathsf{gate}}_1,\widehat{\mathsf{gate}}_2,\cdots ,\widehat{\mathsf{gate}}_q$ be the topology of the gates $\mathsf{gate}_1,\mathsf{gate}_2,\cdots ,\mathsf{gate}_q$ which indicates how gates are connected.
In other words, if $\mathsf{gate}_i=(g,w_a,w_b,w_c)$, then $\widehat{\mathsf{gate}}_i=(\bot,w_a,w_b,w_c)$.
\begin{proof}[Proof of \cref{thm:garble_ever_security}]
First, let us define a simulator $\qSim$ as follows.
\begin{description}
\item[The simulator $\qSim(1^\secp,1^{|C|},C(x),\{L_{i,x_i}\}_{i\in[n]})$:] $ $
\begin{enumerate}
    \item Parse $\{L_{i,x_i}\}_{i\in[n]}\seteq\{\ske.\sk_{i}^{x_i}\}_{i\in[n]}$.
    \item For $i\in[n]$, generate $\ske.\sk_{i}^{x_i\oplus 1}\la \SKE.\keygen(1^{\secp})$. 
    \item For $i\in\{n+1,n+2,\cdots,p\}$ and $\sigma\in\bit$, generate $\ske.\sk_i^{\sigma}\la\SKE.\keygen(1^\secp)$.
    \item For each $i\in[q]$, compute $(\vk_i,\qggate_i)\la \qSim.\qGateGrbl(\widehat{\mathsf{gate}_i}, \{\ske.\sk_a^\sigma,\ske.\sk_b^\sigma,\ske.\sk_c^\sigma\}_{\sigma\in\bit})$, where $\qSim.\qGateGrbl$ is described in Fig~\ref{fig:Sim_Garble_circuit} and $\widehat{\mathsf{gate}_i}=(\bot,w_a,w_b,w_c)$.
    \item For each $i\in[m]$, generate $\widetilde{d_i}\seteq \left[\left(\ske.\sk^{0}_{\mathsf{out}_i}, C(x)_i\right),\left(\ske.\sk^{1}_{\mathsf{out}_i}, C(x)_i\oplus 1 \right)\right]$.
    \item Output $\widetilde{C}\seteq (\{\qggate_i\}_{i\in[q]},\{\widetilde{d_i}\}_{i\in[m]})$ and $\vk\seteq\{\vk_i\}_{i\in[q]}$.
\end{enumerate}
\end{description}

\protocol{Simulation Gate Garbling Circuit $\qSim.\qGateGrbl$
}
{The description of $\qSim.\qGateGrbl$}
{fig:Sim_Garble_circuit}
{
\begin{description}
\item[Input:] $(\widehat{\mathsf{gate}_i},\{\ske.\sk_a^\sigma,\ske.\sk_b^\sigma,\ske.\sk_c^\sigma\}_{\sigma\in\bit})$.
\item[Output:] $\qggate_i$ and $\vk_i$.
\end{description}
\begin{enumerate}
\item For each $\sigma_a,\sigma_b\in\bit$, sample $p^{\sigma_{a},\sigma_b}_{a,b}\la\Ks$.
\item Sample $\gamma_i\la\mathsf{S}_4$.
\item For each $\sigma_a,\sigma_b\in\bit$, compute $(\ske.\vk_{a}^{\sigma_a,\sigma_b},\ske.\qct_{a}^{\sigma_a,\sigma_b})\la\SKE.\qEnc(\ske.\sk_a^{\sigma_a},p^{\sigma_a,\sigma_b}_{c})$
and 
$(\ske.\vk_{b}^{\sigma_a,\sigma_b},\ske.\qct_b^{\sigma_a,\sigma_b})\la
\SKE.\qEnc(\ske.\sk_b^{\sigma_b},p^{\sigma_a,\sigma_b}_{c}\oplus \ske.\sk_c^{0})$.
\item Output $\qggate_i\seteq \{\ske.\qct_a^{\sigma_a,\sigma_b},\ske.\qct_b^{\sigma_a,\sigma_b}\}_{\sigma_a,\sigma_b\in\bit}$ in permutated order of $\gamma_i$ and 
$\vk_i\seteq\{\ske.\vk_{a}^{\sigma_a,\sigma_b},\ske.\vk_{b}^{\sigma_a,\sigma_b}\}_{\sigma_a,\sigma_b\in\bit}$ 
in permutated order of $\gamma_i$.
\end{enumerate}
}

For each $j\in[q]$, we define a QPT algorithm (a simulator) $\mathsf{InputDep}.\qSim_j$ as follows.
\begin{description}
\item[The simulator $\mathsf{InputDep}.\qSim_j(1^\secp,C,x,\{L_{i,x_i}\}_{i\in[n]})$:] $ $
\begin{enumerate}
    \item Parse $\{L_{i,x_i}\}_{i\in[n]}=\{\ske.\sk_{i}^{x_i}\}_{i\in[n]}$. 
    \item For $i\in[n]$, generate $\ske.\sk_{i}^{x_i\oplus 1}\la \SKE.\keygen(1^{\secp})$.
    \item For $i\in\{n+1,n+2,\cdots,p\}$ and $\sigma\in\bit$, generate $\ske.\sk_i^{\sigma}\la\SKE.\keygen(1^\secp)$.
    \item
    For $i\in[j]$, compute $(\vk_i,\qggate_i)\la \mathsf{InputDep}.\qGateGrbl(\mathsf{gate}_i, \{\ske.\sk_a^\sigma,\ske.\sk_b^\sigma,\ske.\sk_c^\sigma\}_{\sigma\in\bit})$, where $\mathsf{InputDep}\qGateGrbl$ is described in Fig.~\ref{fig:Dep_Garble_circuit} and $\mathsf{gate}_i=(g,w_a,w_b,w_c)$
    \item For each $i\in \{j+1,j+2,\cdots,q\}$, compute $(\vk_i,\qggate_i)\la \qGateGrbl(\mathsf{gate}_i,\allowbreak \{\ske.\sk_a^\sigma,\ske.\sk_b^\sigma,\ske.\sk_c^\sigma\}_{\sigma\in\bit})$, where $\qGateGrbl$ is described in Fig~\ref{fig:Garble_circuit} and $\mathsf{gate}_i=(g,w_a,w_b,w_c)$.
    \item For each $i\in[m]$, generate $\widetilde{d_i}\seteq \left[\left(\ske.\sk^{0}_{\mathsf{out}_i}, 0\right),\left(\ske.\sk^{1}_{\mathsf{out}_i}, 1 \right)\right]$.
    \item Output $\widetilde{C}\seteq(\{\qggate_i\}_{i\in[q]},\{\widetilde{d_i}\}_{i\in[m]})$
    and $\vk\seteq\{\vk_i\}_{i\in[q]}$.
\end{enumerate}
\end{description}

\protocol{Input Dependent Gate Garbling Circuit $\mathsf{InputDep}.\qGateGrbl$
}
{The description of $\mathsf{InputDep}.\qGateGrbl$}
{fig:Dep_Garble_circuit}
{
\begin{description}
\item[Input:] $\mathsf{gate}_i,\{\ske.\sk_a^\sigma,\ske.\sk_b^\sigma,\ske.\sk_c^\sigma\}_{\sigma\in\bit}$.
\item[Output:] $\qggate_i$ and $\vk_i$.
\end{description}
\begin{enumerate}
\item For each $\sigma_a,\sigma_b\in\bit$, sample $p^{\sigma_{a},\sigma_b}_{c}\la\Ks$.
\item Sample $\gamma_i\leftarrow \mathsf{S}_4$.
\item For each $\sigma_a,\sigma_b\in\bit$, compute $(\ske.\vk_{a}^{\sigma_a,\sigma_b},\ske.\qct_{a}^{\sigma_a,\sigma_b})\la\SKE.\qEnc(\ske.\sk_a^{\sigma_a},p^{\sigma_a,\sigma_b}_{c})$
and 
$(\ske.\vk_{b}^{\sigma_a,\sigma_b},\ske.\qct_b^{\sigma_a,\sigma_b})\la
\SKE.\qEnc(\ske.\sk_b^{\sigma_b},p^{\sigma_a,\sigma_b}_{c}\oplus \ske.\sk_c^{v(c)})$.
Here, $v(c)$ is the correct value of the bit going over the wire $c$ during the computation of $C(x)$.
\item Output $\qggate_i\seteq \{\ske.\qct_a^{\sigma_a,\sigma_b},\ske.\qct_b^{\sigma_a,\sigma_b}\}_{\sigma_a,\sigma_b\in\bit}$ in permutated order of $\gamma_i$ and 
$\vk_i\seteq\{\ske.\vk_{a}^{\sigma_a,\sigma_b},\ske.\vk_{b}^{\sigma_a,\sigma_b}\}_{\sigma_a,\sigma_b\in\bit}$ in permutated order of $\gamma_i$.
\end{enumerate}
}

For each $j\in\{0,1,\cdots,q\}$, let us define a sequence of hybrid games $\{\mathsf{Hyb}_j\}_{j\in\{0,1,\cdots,q\}}$ against any adversary $\qA\seteq(\qA_1,\qA_2)$, where $\qA_1$ is any QPT adversary and $\qA_2$ is any unbounded adversary.
Note that  
\begin{align}
\mathsf{InputDep}.\qSim_0(1^\secp,C,x,\{L_{i,x_i}\}_{i\in[n]})= \qGarble(1^\secp,C,\{L_{i,\alpha}\}_{i\in[n],\alpha\in\bit}).  
\end{align}

\begin{description}
\item[The hybrid game $\mathsf{Hyb}_j$:] $ $
\begin{enumerate}
    \item $\qA_1$ sends a circuit $C\in\Cs_n$ and an input $x\in\bit^n$ to the challenger.
    \item The challenger computes $\{L_{i,\alpha}\}_{i\in[n],\alpha\in\bit}\la\Samp(1^\secp)$.
    \item  The challenger computes $(\widetilde{\qC},\vk)\la\GC.\mathsf{InputDep}.\qSim_j(1^\secp,C,x,\{L_{i,x_i}\}_{i\in[n]})$,
    and sends $(\widetilde{\qC},\{L_{i,x_i}\}_{i\in[n]})$ to $\qA_1$.
    \item At some point, $\qA_1$ sends $\cert$ to the challenger and the internal state to $\qA_2$.
    \item The challenger computes $\Vrfy(\vk,\cert)\ra\top/\bot$. If the output is $\bot$, then the challenger outputs $\bot$ and sends $\bot$ to $\qA_2$.
    Else, the challenger outputs $\top$ and sends $\top$ to $\qA_2$.
    \item $\qA_2$ outputs $b'\in\bit$.
    \item If the challenger outputs $\top$, then the output of the experiment is $b'$. 
    Otherwise, the output of the experiment is $\bot$.
\end{enumerate}
\end{description}

Note that $\mathsf{Hyb}_0$ is identical to $\expd{\Sigma_{\mathsf{cegc}},\qA}{cert}{ever}{sel}{gbl}(1^\secp,0)$ by definition.
Therefore, \cref{thm:garble_ever_security} easily follows from the following \cref{prop:garble_hyb_j,prop:garble_hyb_q} (whose proofs are given later).
\end{proof}

\begin{proposition}\label{prop:garble_hyb_j}
If $\Sigma_{\mathsf{cesk}}$ satisfies the certified everlasting IND-CPA security, 
it holds that
$
\abs{\Pr[\mathsf{Hyb}_{j-1}=1]-\Pr[\mathsf{Hyb}_{j}=1]} \leq \negl(\secp)
$
for all $j\in[q]$.
\end{proposition}

\begin{proposition}\label{prop:garble_hyb_q}
$\abs{\Pr[\mathsf{Hyb}_q=1]-\Pr[\expd{\Sigma_{\mathsf{cegc}},\qA}{cert}{ever}{sel}{gbl}(1^\secp,1)=1]} \leq \negl(\secp)$.
\end{proposition}

\begin{proof}[Proof of \cref{prop:garble_hyb_j}]
For the proof, we use \cref{lem:parallel} whose statement and proof are given later.
We construct an adversary $\qB$ that breaks the security experiment of $\expc{\Sigma_{\mathsf{cesk}},\qB}{parallel}{cert}{ever}(\secp,b)$,
which is described in \cref{lem:parallel}, assuming that 
$\abs{\Pr[\mathsf{Hyb}_{j-1}=1]-\Pr[\mathsf{Hyb}_{j}=1]}$ is non-negligible.
This contradicts the certified everlasting IND-CPA security of $\Sigma_{\mathsf{cesk}}$
from \cref{lem:parallel}.
Let us describe how $\qB$ works below.
\begin{enumerate}
    \item $\qB$ receives $C\in \Cs_n$ and $x\in\bit^n$ from $\qA_1$. Let $\mathsf{gate_{j}}=(g,w_{\alpha},w_{\beta},w_{\gamma})$.
    \item 
    The challenger of $\expb{\Sigma_{\mathsf{cesk}},\qB}{parallel}{cert}{ever}(\secp,b)$ generates $\ske.\sk_\alpha^{v(\alpha)\oplus 1}\la\SKE.\keygen(1^\secp)$ and $\ske.\sk_\beta^{v(\beta)\oplus 1}\la\SKE.\keygen(1^\secp)$\footnote{Recall that $v(\alpha)$ is the correct value of the bit going over the wire $\alpha$ during the computation of $C(x)$.}.
    \item 
    For each $i\in[p]\setminus\{\alpha,\beta\}$ and $\sigma\in\bit$, $\qB$ generates $\ske.\sk_i^\sigma\la\SKE.\keygen(1^\secp)$. 
    $\qB$ generates $\ske.\sk_\alpha^{v(\alpha)}\la\SKE.\keygen(1^\secp)$ and $\ske.\sk_{\beta}^{v(\beta)}\la\SKE.\keygen(1^\secp)$.
    $\qB$ sets $\{L_{i,x_i}\}_{i\in[n]}\seteq \{\ske.\sk_{i}^{x_i}\}_{i\in[n]}$.
    \item For each $i\in[j-1]$, $\qB$ computes $(\vk_i,\qggate_i)\la \mathsf{InputDep}.\qGateGrbl(\mathsf{gate}_i, \{\ske.\sk_a^\sigma,\allowbreak\ske.\sk_b^\sigma,\ske.\sk_c^\sigma\}_{\sigma\in\bit})$, where $\mathsf{InputDep}.\qGateGrbl$ is described in Fig~\ref{fig:Dep_Garble_circuit} and $\mathsf{gate}_i=(g,w_a,w_b,w_c)$.
    $\qB$ calls the encryption query of $\expc{\Sigma_{\mathsf{cesk}},\qB}{parallel}{cert}{ever}(\secp,b)$ if it needs to use $\ske.\sk_\alpha^{v(\alpha)\oplus 1}$ or $\ske.\sk_\beta^{v(\beta)\oplus 1}$ to run $(\vk_i,\qggate_i)\la \mathsf{InputDep}.\qGateGrbl(\mathsf{gate}_i,\\ \{\ske.\sk_a^\sigma,\ske.\sk_b^\sigma,\ske.\sk_c^\sigma\}_{\sigma\in\bit})$.
    \item 
    $\qB$ samples $p_\gamma^{v(\alpha),v(\beta)}\la \Ks$.
    $\qB$ computes 
    \begin{align}
    &(\ske.\vk_{\alpha}^{v(\alpha),v(\beta)},\ske.\qct_{\alpha}^{v(\alpha),v(\beta)})\la\SKE.\qEnc(\ske.\sk_{\alpha}^{v(\alpha)},p^{v(\alpha),v(\beta)}_{\gamma}),\\
    &(\ske.\vk_{\beta}^{v(\alpha),v(\beta)},\ske.\qct_\beta^{v(\alpha),v(\beta)})\la\SKE.\qEnc(\ske.\sk_\beta^{v(\beta)},p^{v(\alpha),v(\beta)}_{\gamma}\oplus \ske.\sk_\gamma^{v(\gamma)}).
    \end{align}
    \item $\qB$ sets 
    \begin{align}
    &(x_0,y_0,z_0)\seteq (\ske.\sk_{\gamma}^{g(v(\alpha),v(\beta)\oplus 1)},\ske.\sk_{\gamma}^{g(v(\alpha)\oplus 1,v(\beta))},\ske.\sk_{\gamma}^{g(v(\alpha)\oplus 1,v(\beta)\oplus 1)}),\\
    &(x_1,y_1,z_1)\seteq(\ske.\sk_\gamma^{v(\gamma)},\ske.\sk_\gamma^{v(\gamma)},\ske.\sk_\gamma^{v(\gamma)}),
    \end{align}
    and sends $(\ske.\sk_{\alpha}^{v(\alpha)},\ske.\sk_{\beta}^{v(\beta)},\{x_\sigma,y_\sigma,z_\sigma\}_{\sigma\in\bit})$ to the challenger of $\expc{\Sigma_{\mathsf{cesk}},\qB}{parallel}{cert}{ever}(\secp,b)$.
    \item The challenger samples $(x,y,z)\la \Ks^3$ and $(\ske.\sk_\alpha^{v(\alpha)\oplus 1},\ske.\sk_\beta^{v(\beta)\oplus 1})\la\keygen(1^\secp)$.
    The challenger computes 
    \begin{align}
    &(\ske.\vk_{\alpha}^{v(\alpha),v(\beta)\oplus 1},\ske.\qct_{\alpha}^{v(\alpha),v(\beta)\oplus 1})\la\qEnc(\ske.\sk_\alpha^{v(\alpha)},x),\\
    &(\ske.\vk_{\beta}^{v(\alpha),v(\beta)\oplus 1},\ske.\qct_{\beta}^{v(\alpha),v(\beta)\oplus 1})\la\qEnc(\ske.\sk_\beta^{v(\beta)\oplus 1},x\oplus x_b),\\
    &(\ske.\vk_{\alpha}^{v(\alpha)\oplus 1,v(\beta)},\ske.\qct_{\alpha}^{v(\alpha)\oplus 1,v(\beta)})\la\qEnc(\ske.\sk_{\alpha}^{v(\alpha)\oplus 1},y),\\
    &(\ske.\vk_{\beta}^{v(\alpha)\oplus 1,v(\beta)},\ske.\qct_\beta^{v(\alpha)\oplus 1,v(\beta)})\la\qEnc(\ske.\sk_{\beta}^{v(\beta)},y\oplus y_b),\\
    &(\ske.\vk_\alpha^{v(\alpha)\oplus 1,v(\beta)\oplus 1},\ske.\qct_\alpha^{v(\alpha)\oplus 1,v(\beta)\oplus 1})\la \qEnc(\ske.\sk_\alpha^{v(\alpha)\oplus 1},z),\\
    &(\ske.\vk_\beta^{v(\alpha)\oplus 1,v(\beta)\oplus 1},\ske.\qct_\beta^{v(\alpha)\oplus 1,v(\beta)\oplus 1})\la \qEnc(\ske.\sk_\beta^{v(\beta)\oplus 1},z\oplus z_b),
    \end{align}
    and sends
    \ifnum\submission=0
    \begin{align}
    (\ske.\qct_\alpha^{v(\alpha),v(\beta)\oplus 1},
    \ske.\qct_{\beta}^{v(\alpha),v(\beta)\oplus 1},
    \ske.\qct_\alpha^{v(\alpha)\oplus 1,v(\beta)},
    \ske.\qct_{\beta}^{v(\alpha)\oplus 1,v(\beta)},
    \ske.\qct_{\alpha}^{v(\alpha)\oplus 1,v(\beta)\oplus 1},
    \ske.\qct_{\beta}^{v(\alpha)\oplus 1,v(\beta)\oplus 1})
    \end{align}
    \else
    \begin{align}
    &\Big(\ske.\qct_\alpha^{v(\alpha),v(\beta)\oplus 1},
    \ske.\qct_{\beta}^{v(\alpha),v(\beta)\oplus 1},
    \ske.\qct_\alpha^{v(\alpha)\oplus 1,v(\beta)},\\
    &\ske.\qct_{\beta}^{v(\alpha)\oplus 1,v(\beta)},
    \ske.\qct_{\alpha}^{v(\alpha)\oplus 1,v(\beta)\oplus 1},
    \ske.\qct_{\beta}^{v(\alpha)\oplus 1,v(\beta)\oplus 1}\Big)
    \end{align}
    \fi 
    to $\qB$.
    \item $\qB$ samples $\gamma_j\la \mathsf{S}_4$.
    $\qB$ sets $\qggate_j\seteq \{\ske.\qct_\alpha^{\sigma_{\alpha},\sigma_\beta},\ske.\qct_\beta^{\sigma_\alpha,\sigma_\beta}\}_{\sigma_\alpha,\sigma_\beta\in\bit}$ in the permutated order of $\gamma_j$.
    \item For each $i\in\{j+1,j+2,\cdots, q\}$, $\qB$ computes $(\vk_i,\qggate_i)\la \qGateGrbl(\mathsf{gate}_i,\allowbreak \{\ske.\sk_a^\sigma,\ske.\sk_b^\sigma,\ske.\sk_c^\sigma\}_{\sigma\in\bit})$, where $\qB$ calls the encryption query of $\expc{\Sigma_{\mathsf{cesk}},\qB}{parallel}{cert}{ever}(\secp,b)$ if $\qB$ needs to use $\ske.\sk_\alpha^{v(\alpha)\oplus 1}$ or $\ske.\sk_\beta^{v(\beta)\oplus 1}$ to run $(\vk_i,\qggate_i)\la \qGateGrbl(\mathsf{gate}_i,\allowbreak \{\ske.\sk_a^\sigma,\ske.\sk_b^\sigma,\ske.\sk_c^\sigma\}_{\sigma\in\bit})$.
    \item $\qB$ computes $\widetilde{d_i}\seteq [(\ske.\sk^0_{\mathsf{out}_i},0),(\ske.\sk^1_{\mathsf{out}_i},1)]$ for $i\in[m]$,
    sets $\widetilde{\qC}\seteq (\{\qggate_{i}\}_{i\in[q]},\{\widetilde{d_i}\}_{i\in[m]})$,
    and sends $(\widetilde{\qC},\{L_{i,x_i}\}_{i\in[n]})$ to $\qA_1$.
    \item At some point, $\qA_1$ sends $\cert\seteq\{\cert_i\}_{i\in [q]}$ to $\qB$ and the internal state to $\qA_2$, respectively.
    \item 
    $\qB$ re-sorts $\cert_{j}=\{\ske.\cert_{\alpha}^{\sigma_{\alpha},\sigma_{\beta}},\ske.\cert_{\beta}^{\sigma_{\alpha},\sigma_{\beta}}\}_{\sigma_{\alpha},\sigma_{\beta}\in\bit}$ according to $\gamma_j$.
    $\qB$ sends 
    \ifnum\submission=0
    \begin{align}
    (\ske.\cert_{\alpha}^{v(\alpha),v(\beta)\oplus 1},
    \ske.\cert_{\beta}^{v(\alpha),v(\beta)\oplus 1},
    \ske.\cert_\alpha^{v(\alpha)\oplus 1,v(\beta)},
    \ske.\cert_\beta^{v(\alpha)\oplus 1,v(\beta)},
    \ske.\cert_\alpha^{v(\alpha)\oplus 1,v(\beta)\oplus 1},
    \ske.\cert_\beta^{v(\alpha)\oplus 1,v(\beta)\oplus 1})
    \end{align}
    \else
    \begin{align}
    &\Big(\ske.\cert_{\alpha}^{v(\alpha),v(\beta)\oplus 1},
    \ske.\cert_{\beta}^{v(\alpha),v(\beta)\oplus 1},
    \ske.\cert_\alpha^{v(\alpha)\oplus 1,v(\beta)},\\
    &\ske.\cert_\beta^{v(\alpha)\oplus 1,v(\beta)},
    \ske.\cert_\alpha^{v(\alpha)\oplus 1,v(\beta)\oplus 1},
    \ske.\cert_\beta^{v(\alpha)\oplus 1,v(\beta)\oplus 1}\Big)
    \end{align}
    \fi
    to the challenger of $\expc{\Sigma_{\mathsf{cesk}},\qB}{parallel}{cert}{ever}(\secp,b)$ and receives $\bot$ or $(\ske.\sk_\alpha^{'v(\alpha)\oplus 1},\ske.\sk_\beta^{'v(\beta)\oplus 1})$ from the challenger.
    $\qB$ computes $\SKE.\Vrfy(\ske.\vk_{\alpha}^{v(\alpha),v(\beta)},\ske.\cert_{\alpha}^{v(\alpha),v(\beta)})$
    and $\SKE.\Vrfy(\ske.\vk_{\beta}^{v(\alpha),v(\beta)},\ske.\cert_{\beta}^{v(\alpha),v(\beta)})$.
    $\qB$ computes $\mathsf{GateVrfy}(\vk_i,\cert_i)$ for each $i\in\{1,2,\cdots,j-1,j+1,j+2,\cdots,q\}$,
    where $\mathsf{GateVrfy}$ is described in Fig\ref{fig:Gate_Verify}.
    If $\qB$ receives $(\ske.\sk_\alpha^{'v(\alpha)\oplus 1},\ske.\sk_\beta^{'v(\beta)\oplus 1})$ from the challenger,
    $\top\la\SKE.\Vrfy(\ske.\vk_{\alpha}^{v(\alpha),v(\beta)},\ske.\cert_{\alpha}^{v(\alpha),v(\beta)})$,
    $\top\la\SKE.\Vrfy(\ske.\vk_{\beta}^{v(\alpha),v(\beta)},\ske.\cert_{\beta}^{v(\alpha),v(\beta)})$,
    and $\top\la\mathsf{GateVrfy}(\cert_i,\vk_i)$ for all 
    $i\in\{1,2,\cdots,j-1,j+1,j+2,\cdots,q\}$, then $\qB$ sends $\top$ to $\qA_2$.
    Otherwise, $\qB$ sends $\bot$ to $\qA_2$, and aborts.
    \item $\qB$ outputs the output of $\qA_2$.
\end{enumerate}
It is clear that $\Pr[1\la\qB \mid b=0]=\Pr[\hyb_{j-1}=1]$ and $\Pr[1\la\qB \mid b=1]=\Pr[\hyb_{j}=1]$.
Therefore, if for an adversary $\qA$,
$
\abs{\Pr[\mathsf{Hyb}_{j-1}=1]-\Pr[\mathsf{Hyb}_{j}=1]}
$
is non-negligible, then
\[
\left|\Pr[\expc{\Sigma_{\mathsf{cesk}},\qB}{parallel}{cert}{ever}(\secp,0)=1]-\Pr[\expc{\Sigma_{\mathsf{cesk}},\qB}{parallel}{cert}{ever}\allowbreak(\secp,1)=1]\right|
\]
is non-negligible.
From \cref{lem:parallel}, it contradicts the certified everlasting IND-CPA security of $\Sigma_{\mathsf{cesk}}$ , which completes the proof.
\end{proof}

\begin{proof}[Proof of \cref{prop:garble_hyb_q}]
To show \cref{prop:garble_hyb_q}, it is sufficient to prove that the probability distribution of $\widetilde{\qC}$ in $\expc{\Sigma_{\mathsf{cegc}},\qA}{cert}{ever}{select}(1^\secp,1)$ is statistically identical to
that of $\widetilde{\qC}$ in $\mathsf{Hyb}_{q}$.

First, let us remind the difference between $\mathsf{Hyb}_{q}$ and $\expc{\Sigma_{\mathsf{cegc}},\qA}{cert}{ever}{select}(1^\secp,1)$. 
In both experiments , $\widetilde{\qC}$ consists of $\{\qggate_i\}_{i\in[q]}$ and $\{\widetilde{d_i}\}_{i\in[m]}$.
On the other hand the contents of $\{\qggate_i\}_{i\in[q]}$ and $\{\widetilde{d_i}\}_{i\in[m]}$ are different in each experiments.
In $\hyb_{q}$,
$\qggate_i$ consists of $(\ske.\qct_a^{\sigma_a,\sigma_b},\ske.\qct_{b}^{\sigma_a,\sigma_b})$
where 
\begin{align}
&(\ske.\vk_{a}^{\sigma_a,\sigma_b},\ske.\qct_{a}^{\sigma_a,\sigma_b})\la\SKE.\qEnc(\ske.\sk_a^{\sigma_a},p^{\sigma_a,\sigma_b}_{c}),\\
&(\ske.\vk_{b}^{\sigma_a,\sigma_b},\ske.\qct_b^{\sigma_a,\sigma_b})\la\SKE.\qEnc(\ske.\sk_b^{\sigma_b},p^{\sigma_a,\sigma_b}_{c}\oplus \ske.\sk_c^{v(c)}),
\end{align}
and
$\widetilde{d_i}$ is
\begin{align}
    [(\ske.\sk_{\mathsf{out_i}}^0, 0),(\ske.\sk_{\mathsf{out_i}}^1,1)].
\end{align}

In $\expc{\Sigma_{\mathsf{cegc}},\qA}{cert}{ever}{select}(1^\secp,1)$, 
$\qggate_i$ consists of $(\ske.\qct_a^{\sigma_a,\sigma_b},\ske.\qct_{b}^{\sigma_a,\sigma_b})$
where
\begin{align}
&(\ske.\vk_{a}^{\sigma_a,\sigma_b},\ske.\qct_{a}^{\sigma_a,\sigma_b})\la\SKE.\qEnc(\ske.\sk_a^{\sigma_a},p^{\sigma_a,\sigma_b}_{c}),\\
&(\ske.\vk_{b}^{\sigma_a,\sigma_b},\ske.\qct_b^{\sigma_a,\sigma_b})\la\SKE.\qEnc(\ske.\sk_b^{\sigma_b},p^{\sigma_a,\sigma_b}_{c}\oplus \ske.\sk_c^{0}),
\end{align}
and
$\widetilde{d_i}$ is 
\begin{align}
[(\ske.\sk_{\mathsf{out_i}}^0, C(x)_i),(\ske.\sk_{\mathsf{out_i}}^1, C(x)_i\oplus 1)].
\end{align}

The resulting distribution of $(\{\qggate_i\}_{i\in[q]},\{\widetilde{d_i}\}_{i\in[m]})$ in $\mathsf{Hyb}_{q}$ is statistically identical to
the resulting distribution of $(\{\qggate_i\}_{i\in[q]},\{\widetilde{d_i}\}_{i\in[m]})$ in $\expb{\Sigma_{\mathsf{cegc}},\qA}{cert}{ever}{select}(1^\secp,1)$.
This is because, at any level that is not output, the keys $\ske.\sk_c^0,\ske.\sk_c^1$ are used completely identically in the subsequent level so there is no difference
between always encrypting $\ske.\sk_c^{v(c)}$ and $\ske.\sk_c^0$.
At the output level, there is no difference between encrypting $\ske.\sk_c^{v(c)}$ and giving the real mapping $\ske.\sk_c^{v(c)}\ra v(c)$ 
or encrypting $\ske.\sk_c^0$ and giving the programming mapping $\ske.\sk_c^0\ra C(x)_i$, which completes the proof.
\end{proof}

We use the following lemma for the proof of \cref{prop:garble_hyb_j}.
The proof is shown with the standard hybrid argument. It is also easy to see that a similar lemma holds for IND-CPA security.
\begin{lemma}\label{lem:parallel}
Let $\Sigma\seteq(\keygen,\qEnc,\qDec,\qDelete,\Vrfy)$ be a certified everlasting secure SKE scheme.
Let us consider the following security experiment 
$\expc{\Sigma,\qA}{parallel}{cert}{ever}(\secp,b)$ against $\qA$ consisting of any QPT adversary $\qA_1$ and any unbounded adversary $\qA_2$.
\begin{enumerate}
    \item The challenger generates $(\sk'^{0},\sk'^1)\la\keygen(1^\secp)$.
    \item $\qA_1$ can call encryption queries. More formally, it can do the followings: $\qA_1$ chooses $\beta\in\bit$, $\sk\in\mathcal{SK}$ and $m\in\Ms$.
    $\qA_1$ sends $(\beta,\sk,m)$ to the challenger. 
    \begin{itemize}
    \item 
     If $\beta=0$, the challenger generates $m^*\la\Ms$, computes $(\vk^{0}_{m},\qct^{0}_{m})\la\qEnc(\sk'^{0},m^*)$ and $(\vk^{1}_m,\qct^{1}_m)\la\qEnc(\sk,m\oplus m^*)$, and sends $\{\vk^\sigma_m,\qct^\sigma_m\}_{\sigma\in\bit}$ to $\qA_1$.
     \item If $\beta=1$, the challenger generates $m^*\la\Ms$, computes $(\vk^{1}_{m},\qct^{1}_{m})\la\qEnc(\sk'^{1},m\oplus m^*)$ and
     $(\vk^{0}_{m},\qct^{0}_{m})\la\qEnc(\sk,m^*)$, and sends $\{\vk^\sigma_m,\qct^\sigma_m\}_{\sigma\in\bit}$ to $\qA_1$.
    \end{itemize}
    $\qA_1$ can repeat this process polynomially many times.
    \item $\qA_1$ generates $(\sk^0,\sk^1)\la\keygen(1^\secp)$ and chooses two triples of messages $(x_0,y_0,z_0)\in\Ms^3$ and $(x_1,y_1,z_1)\in\Ms^3$, and sends $(\sk^0,\sk^1,\{x_\sigma,y_\sigma,z_\sigma\}_{\sigma\in\bit})$ to the challenger.
    \item The challenger generates $(x,y,z)\la\Ms^3$.
    The challenger computes
    \begin{align}
    &(\vk_{x}^0,\qct_{x}^0)\la\qEnc(\sk^0,x),\,\,\, (\vk_{x}^1,\qct_{x}^1)\la\qEnc(\sk'^1,x\oplus x_b)\\
    &(\vk_{y}^0,\qct_{y}^0)\la\qEnc(\sk'^0,y),\,\,\, (\vk_y^1,\qct_y^1)\la\qEnc(\sk^1,y\oplus y_b)\\
    &(\vk_z^0,\qct_z^0)\la \qEnc(\sk'^0,z), \,\,\, (\vk_z^1,\qct_z^1)\la \qEnc(\sk'^1,z\oplus z_b)
    \end{align}
    and sends $\{\qct_{x}^\sigma,\qct_y^\sigma,\qct_z^\sigma\}_{\sigma\in\bit}$ to $\qA_1$.
    \item $\qA_1$ can call encryption queries. More formally, it can do the followings:
    $\qA_1$ chooses $\beta\in\bit$, $\sk\in\mathcal{SK}$ and $m\in\Ms$.
    $\qA_1$ sends $(\beta,\sk,m)$ to the challenger. 
    \begin{itemize}
    \item 
     If $\beta=0$, the challenger generates $m^*\la\Ms$, computes $(\vk^{0}_{m},\qct^{0}_{m})\la\qEnc(\sk'^{0},m^*)$ and $(\vk^{1}_m,\qct^{1}_m)\la\qEnc(\sk,m\oplus m^*)$, and sends $\{\vk^\sigma_m,\qct^\sigma_m\}_{\sigma\in\bit}$ to $\qA_1$.
     \item If $\beta=1$, the challenger generates $m^*\la\Ms$, computes $(\vk^{1}_{m},\qct^{1}_{m})\la\qEnc(\sk'^{1},m\oplus m^*)$ and
     $(\vk^{0}_{m},\qct^{0}_{m})\la\qEnc(\sk,m^*)$, and sends $\{\vk^\sigma_m,\qct^\sigma_m\}_{\sigma\in\bit}$ to $\qA_1$.
    \end{itemize}
    $\qA_1$ can repeat this process polynomially many times.
    \item $\qA_1$ sends $\{\cert_x^\sigma,\cert_y^\sigma,\cert_z^\sigma\}_{\sigma\in\bit}$ to the challenger, and sends the internal state to $\qA_2$.
    \item The challenger computes $\Vrfy(\vk_x^\sigma,\cert_x^\sigma)$, $\Vrfy(\vk_y^\sigma,\cert_y^\sigma)$ and $\Vrfy(\vk_z^\sigma,\cert_z^\sigma)$ for each $\sigma\in\bit$. If all results are $\top$, then the challenger outputs $\top$, and sends $\{\sk'^\sigma\}_{\sigma\in\{0,1\}}$ to $\qA_2$.
    Otherwise, the challenger outputs $\bot$, and sends $\bot$ to $\qA_2$.
    \item $\qA_2$ outputs $b'\in\bit$.
    \item If the challenger outputs $\top$, then the output of the experiment is $b'$. Otherwise, the output of the experiment is $\bot$.
\end{enumerate}
If the $\Sigma$ satisfies the certified everlasting IND-CPA security,
\begin{align}
\advc{\Sigma,\qA}{parallel}{cert}{ever}(\secp)
\seteq \abs{\Pr[ \expc{\Sigma,\qA}{parallel}{cert}{ever}(\secp, 0)=1] - \Pr[ \expc{\Sigma,\qA}{parallel}{cert}{ever}(\secp, 1)=1] }\leq \negl(\secp)
\end{align}
for any QPT adversary $\qA_1$ and any unbounded adversary $\qA_2$.
\end{lemma}

\begin{proof}[Proof of \cref{lem:parallel}]
We define the following hybrid experiment.
\begin{description}
\item[$\hyb_{1}$:] This is identical to $\expc{\Sigma,\qA}{parallel}{cert}{ever}(\secp, 0)$ except that the challenger encrypts $(x_0,y_0,z_1)$ instead of encrypting $(x_0,y_0,z_0)$.
\item[$\hyb_{2}$:] This is identical to $\hyb_{1}$ except that the challenger encrypts $(x_0,y_1,z_1)$ instead of encrypting $(x_0,y_0,z_1)$.
\end{description}

\cref{lem:parallel} easily follows from the following \cref{prop:Exp_Hyb_1_parallel,prop:Hyb_1_Hyb_2_parallel,prop:Hyb_2_Hyb_3_parallel} (whose proof is given later.).
\end{proof}

\begin{proposition}\label{prop:Exp_Hyb_1_parallel}
If $\Sigma$ is certified everlasting IND-CPA secure, it holds that
\[
\abs{\Pr[\expc{\Sigma,\qA}{parallel}{cert}{ever}(\secp, 0)=1]-\Pr[\hyb_{1}=1]}\leq \negl(\lambda).
\]
\end{proposition}

\begin{proposition}\label{prop:Hyb_1_Hyb_2_parallel}
If $\Sigma$ is certified everlasting IND-CPA secure, it holds that
\[
\abs{\Pr[\hyb_{1}=1]-\Pr[\hyb_{2}=1]}\leq \negl(\lambda).
\]
\end{proposition}

\begin{proposition}\label{prop:Hyb_2_Hyb_3_parallel}
If $\Sigma$ is certified everlasting IND-CPA secure, it holds that
\[
\abs{\Pr[\hyb_{2}=1]-\Pr[\expc{\Sigma,\qA}{parallel}{cert}{ever}(\secp, 1)=1]}\leq \negl(\lambda).
\]
\end{proposition}

\begin{proof}[Proof of \cref{prop:Exp_Hyb_1_parallel}]
We assume that
$
\abs{\Pr[\expc{\Sigma,\qA}{parallel}{cert}{ever}(\secp, 0)=1]-\Pr[\hyb_{1}(1)=1]}
$
is non-negligible, and construct an adversary $\qB$ that breaks the security experiment of $\expd{\Sigma,\qB}{cert}{ever}{ind}{cpa}(\secp, b)$. 
This contradicts the certified everlasting IND-CPA security of $\Sigma$.
Let us describe how $\qB$ works.
\begin{enumerate}
    \item The challenger of $\expd{\Sigma,\qB}{cert}{ever}{ind}{cpa}(\secp,b)$ generates $\sk'^{0}\la\keygen(1^\secp)$, and $\qB$ generates $\sk'^{1}\la\keygen(1^\secp)$.
    \item $\qA_1$ chooses $\beta\in\bit$, $\sk\in\Ks$ and $m\in\Ms$.
    $\qA_1$ sends $(\beta,\sk,m)$ to $\qB$.
    \begin{itemize} 
    \item If $\beta=0$, $\qB$ generates $m^*\la\Ms$, sends $m^*$ to the challenger, receives $(\vk_m^0,\qct_m^0)$ from the challenger, computes $(\vk^{1}_m,\qct^{1}_m)\la\qEnc(\sk,m\oplus m^*)$, and sends $\{\vk^\sigma_m,\qct^\sigma_m\}_{\sigma\in\bit}$ to $\qA_1$.
    \item If $\beta=1$, $\qB$ generates $m^*\la\Ms$, computes $(\vk^{1}_{m},\qct^{1}_{m})\la\qEnc(\sk'^{1},m\oplus m^*)$ and $(\vk^{0}_{m},\qct^{0}_{m})\la\qEnc(\sk,m^*)$ and sends $\{\vk^\sigma_m,\qct^\sigma_m\}_{\sigma\in\bit}$ to $\qA_1$.
    \end{itemize}
    $\qB$ repeats this process when $(\beta,\sk,m)$ is sent from $\qA_1$.
    \item $\qB$ receives $(\sk^0,\sk^1,\{x_\sigma,y_\sigma,z_\sigma\}_{\sigma\in\bit})$ from $\qA_1$.
    \item $\qB$ generates $(x,y,z)\la\Ms^3$.
          $\qB$ computes 
          \begin{align}
          &(\vk_{x}^0,\qct_{x}^0)\la\qEnc(\sk^0,x), (\vk_{x}^1,\qct_{x}^1)\la\qEnc(\sk'^1,x\oplus x_0),\\
          &\hspace{2cm}(\vk_y^1,\qct_y^1)\la\qEnc(\sk^1,y\oplus y_0),\\
          &\hspace{2cm}(\vk_z^1,\qct_z^1)\la \qEnc(\sk'^1,z\oplus z_0).
          \end{align}
    \item $\qB$ sets $m_0\seteq z$ and $m_1\seteq z\oplus z_0\oplus z_1$. 
          $\qB$ sends $(m_0,m_1)$ to the challenger.
    \item The challenger computes $(\vk_z^0,\qct_z^0)\la \qEnc(\sk'^0,m_{b})$, and sends $\qct_z^0$ to $\qB$.
    \item $\qB$ sends an encryption query $y$ to the challenger, and receives $(\vk_y^0,\qct_y^0)$.
    \item $\qB$ sends $\{\qct_{x}^\sigma,\qct_y^\sigma,\qct_z^\sigma\}_{\sigma\in\bit}$ to $\qA_1$.
    \item $\qA_1$ chooses $\beta\in\bit$, $\sk\in\Ks$ and $m\in\Ms$.
    $\qA_1$ sends $(\beta,\sk,m)$ to $\qB$.
    \begin{itemize} 
    \item If $\beta=0$, $\qB$ generates $m^*\la\Ms$, sends $m^*$ to the challenger, receives $(\vk_m^0,\qct_m^0)$ from the challenger, computes $(\vk^{1}_m,\qct^{1}_m)\la\qEnc(\sk,m\oplus m^*)$, and sends $\{\vk^\sigma_m,\qct^\sigma_m\}_{\sigma\in\bit}$ to $\qA_1$.
    \item If $\beta=1$, $\qB$ generates $m^*\la\Ms$, computes $(\vk^{1}_{m},\qct^{1}_{m})\la\qEnc(\sk'^{1},m\oplus m^*)$ and $(\vk^{0}_{m},\qct^{0}_{m})\la\qEnc(\sk,m^*)$ and sends $\{\vk^\sigma_m,\qct^\sigma_m\}_{\sigma\in\bit}$ to $\qA_1$.
    \end{itemize}
    $\qB$ repeats this process when $(\beta,\sk,m)$ is sent from $\qA_1$.
    \item $\qA_1$ sends $\{\cert^{\sigma}_x,\cert^{\sigma}_y,\cert^{\sigma}_z\}_{\sigma\in\bit}$ to $\qB$, and sends the internal state to $\qA_2$.
    \item $\qB$ sends $\cert_z^0$ to the challenger, and receives $\sk'^0$ or $\bot$ from the challenger.
    If $\qB$ receives $\bot$, it outputs $\bot$ and aborts.
    \item $\qB$ sends $\{\sk'^{\sigma}\}_{\sigma\in\bit}$ to $\qA_2$.
    \item $\qA_2$ outputs $b'$.
    \item $\qB$ computes $\Vrfy(\vk^\sigma_x,\cert_x^{\sigma})$ and $\Vrfy(\vk^\sigma_y,\cert_y^{\sigma})$ for each $\sigma\in\bit$, and $\Vrfy(\vk_z^{1},\cert_z^1)$.
    If all results are $\top$, $\qB$ outputs $b'$.
    Otherwise, $\qB$ outputs $\bot$.
\end{enumerate}
It is clear that $\Pr[1\la\qB \mid b=0]=\Pr[\expc{\Sigma,\qA}{parallel}{cert}{ever}(\secp, 0)=1]$.
Since $z$ is uniformly distributed, $(z,z\oplus z_1)$ and $(z\oplus z_0\oplus z_1,z\oplus z_0)$ are identically distributed. 
Therefore, it holds that $\Pr[1\la\qB \mid b=1]=\Pr[\hyb_{1}=1]$.
By assumption, $\abs{\Pr[\expc{\Sigma,\qA}{parallel}{cert}{ever}(\secp, 0)=1]-\Pr[\hyb_{1}=1]}$ is non-negligible, and therefore
\[
\abs{\Pr[1\la\qB \mid b=0]-\Pr[1\la\qB \mid b=1]}
\] is non-negligible, which contradicts the certified everlasting IND-CPA security of $\Sigma_{\mathsf{cesk}}$.
\end{proof}

\begin{proof}[Proof of \cref{prop:Hyb_1_Hyb_2_parallel}]
The proof is very similar to that of \cref{prop:Exp_Hyb_1_parallel}.
Therefore we skip the proof.
\end{proof}

\begin{proof}[Proof of \cref{prop:Hyb_2_Hyb_3_parallel}]
The proof is very similar to that of \cref{prop:Exp_Hyb_1_parallel}.
Therefore, we skip the proof.
\end{proof}

\else
 \fi
	\else
\fi

\end{document}